\documentclass[12pt]{article}
\pdfoutput=1
\usepackage{latexsym, graphicx,cite}
\usepackage{amsmath}
\usepackage{amssymb}
\usepackage{amsthm}
\usepackage{mathrsfs}
\usepackage{mathtools}
\usepackage{graphicx}
\usepackage{caption}
\usepackage{subcaption}
\usepackage{enumitem}
\usepackage{soul}
\usepackage{changepage}

\numberwithin{equation}{section}

\usepackage[top = 0.8 in,bottom = 0.8 in, left = 0.8 in, right=0.8 in]{geometry}

\newcommand{\arXiv}[1]{\href{http://www.arXiv.org/abs/#1}{arXiv:#1}}
\newcommand{\doi}[2]{\href{http://dx.doi.org/#2}{#1}}
\usepackage[colorlinks=true, linkcolor=blue, bookmarks=true]{hyperref}

\makeatletter
\renewcommand\section{\@startsection {section}{1}{\z@}%
	{-1.5ex \@plus -1ex \@minus -.2ex}
	{.3ex \@plus.2ex}%
	{\normalfont\large\bfseries}}
\renewcommand\subsection{\@startsection{subsection}{2}{\z@}%
	{-0.5ex\@plus -1ex \@minus -.2ex}%
	{0.2ex \@plus .2ex}%
	{\normalfont\bfseries}}
\makeatother

\newcommand{\w}{\omega}

\newcommand{\q}{a}
\newcommand{\h}{F}

\newcommand{\ax}{\hat{\mu}}
\newcommand{\bx}{\hat{\nu}}
\newcommand{\ah}{\nu_1}
\newcommand{\bh}{\nu_0}
\newcommand{\at}{\tilde{\mu}}
\newcommand{\bt}{\tilde{\nu}}

\newcommand{\aone}{\mu_1}
\newcommand{\bone}{\mu_0}
\newcommand{\atwo}{\mu_2}

\newcommand{\beq}{\begin{equation}}
	\newcommand{\eeq}{\end{equation}}
\newcommand{\beqnn}{\begin{equation*}}
	\newcommand{\eeqnn}{\end{equation*}}

\theoremstyle{proposition}
\newtheorem{proposition}{Proposition}[section]

\theoremstyle{corollary}
\newtheorem{corollary}{Corollary}[section]

\theoremstyle{remark}

\theoremstyle{lemma}
\newtheorem{lemma}{Lemma}[section]

\begin{document}
	
	\begin{center}
        {\large \bf 
			Coherent energy cascades in nonlinear disordered Hamiltonians			
		}

		\vskip 2mm
		{ Anxo Biasi$^1$, Brad Cownden$^1$, Oleg Evnin$^{2,3}$ and Alvaro Iturbe Jabaloyes$^1$}

		\vskip 1mm

		\vskip 0mm
	\end{center}
		
			\begin{adjustwidth}{-0.1cm}{-0.1cm}
				\begin{center}	\footnotesize{$^1$Instituto Galego de Física de Altas Enerxías, Universidade de Santiago de Compostela,\vspace{-.3mm}\\ Santiago de Compostela, Spain,\vspace{.5mm}\\	
					$^2$High Energy Physics Research Unit, Faculty of Science,
					Chulalongkorn University, Bangkok, Thailand,\vspace{.5mm}\\
					$^3$Theoretische Natuurkunde, Vrije Universiteit Brussel (VUB)
					\& International Solvay Institutes, Brussels, Belgium
				}	

                  \footnotesize{anxo.biasi@gmail.com, bradleystuart.cownden@usc.es, oleg.evnin@gmail.com, alvarojacobo.iturbe@rai.usc.gal}
                
                \end{center}
			\end{adjustwidth}

	\vspace{-0.2cm}
	\begin{center}
		\begin{adjustwidth}{2.2cm}{2.2cm}\footnotesize
The familiar wave kinetic theory revolves around turbulent energy cascades characterized by random phases of the individual normal modes.
By contrast, recent years have revealed turbulent cascades where the transfer of energy is mediated by phase-coherent dynamics.
Such coherent cascades have been previously established in deterministic systems with a tightly controlled algebraic structure of the mode interactions, such as nonlinear Schr\"odinger equations.\vspace{1.5mm}\\
Here, we show that the presence of phase-coherent cascades is robust with respect to substantial alterations of the mode couplings. We prove that, in infinite-dimensional resonant Hamiltonian systems with fully resonant normal mode spectra and random nonlinear couplings, a structured subset of couplings --- an infinitesimal fraction of the total --- suffices to support coherent cascades characterized by phase alignment, the formation of power-law spectra, and either unbounded growth of Sobolev norms or a finite-time blow-up. Numerical simulations further demonstrate the emergence of coherent cascades from random initial data through a process of phase locking, suggesting a mechanism for their emergence in broader classes of initial data and systems. We thus show that disordered mode interactions can sustain, rather than block, coherent energy transport across arbitrarily distant scales.
		\end{adjustwidth}
	\end{center}

	
	\section{Introduction}

    One of the most challenging aspects in the study of weak turbulence and long-time instabilities in nonlinear dispersive equations is understanding the different mechanisms that transfer energy to progressively smaller spatial scales, or equivalently, to modes of progressively higher frequencies. Such processes are broadly referred to as forward energy cascades and can occur in qualitatively different dynamical regimes, distinguished by the degree of phase correlation among the interacting modes. The classic regime is that of phase-incoherent cascades predicted by kinetic theories of weak turbulence \cite{Book_Nazarenko,Book_Zakharov,Book_Galtier}, in which phase correlations are assumed to be sufficiently weak for the transfer of energy to be described statistically within the random phase approximation. We will focus instead on the phase-coherent regime, where energy cascades are characterized by persistent correlations among the phases and are amenable to a fully deterministic treatment. Understanding these dynamical regimes is valuable because structured phases organize the transfer of energy in ways fundamentally different from those predicted by kinetic theories, leading to complementary forms of weak turbulent instabilities \cite{BR,BMR}, condensate formation \cite{Biasi,BG}, and rigorous descriptions of the growth of Sobolev norms \cite{Staffilani2010,GG}.

        Coherent cascades have been established primarily in models with a rigid analytic structure in the mode interactions, and are present  in a broad array of literature dealing with important equations of mathematical physics: nonlinear Schr\"odinger equations \cite{Staffilani2010,Maspero,Maspero2,Guardia,GHP,Hani,Hani2,Kuksin1,Kuksin2,HausProcesi}, nonlinear wave equations \cite{GGwaveequation,GerardLenzmannPocovnicuRaphael,Bourgain2}, related Hamiltonian systems \cite{GG,GG2,GGH,BE,Xu,Biasi,BG,GL,GP,Giuliani,GuardiaGiuliani,GiulianiScandone,minHamiltonian,Thomann,Pocovnicu}, water waves \cite{MasperoWaterWaves}, and relativistic waves in asymptotically anti-de Sitter (AdS) spacetimes and related geometries \cite{BMR,Jalmuzna,MBox1,MBox2,KehleMoschidis,E2} in relation to the AdS instability conjecture \cite{BR} of Bizo\'n and Rostworowski; see \cite{Freivogel} for an early targeted discussion of phase coherence in gravitational systems. (A number of works we mention do not emphasize the properties of phases explicitly and focus on the transfer of energy from low to high modes in deterministic Hamiltonian systems, more than on the phase-structure of those processes. However, the underlying dynamics revolves around phase-sensitive rather than random-phase trajectories.) Considering a quartic dispersive Hamiltonian as a representative example of the above models,
    \beq
        \mathcal{H} =  \sum_{n} \omega_n |a_n|^2 + \frac{1}{2}\sum_{n,m,k,j} C_{nmkj}\ \bar{a}_n\bar{a}_m a_k a_j,
        \label{eq:generic_Hamiltonian}
    \eeq
    coherent cascades have been established in systems whose dispersion relation $\omega_n$ and interaction coefficients $C_{nmkj}$ follow explicit functional dependences on the mode numbers, typically descending from the original PDEs that (\ref{eq:generic_Hamiltonian}) represents or approximates. These analytic specifications provide highly structured `interaction networks' through which phase correlations among the modes can persist and organize the transfer of energy across the spectrum.  Coherent cascades are often associated with asymptotically locked phases, with $\arg a_n$ becoming approximately linear in $n$ at large $n$.
    
In this work, we aim to understand whether nonlinear Hamiltonian systems with more complicated or arbitrary mode couplings compared to the previous studies can also exhibit coherent cascades in the way they tranfer energy toward arbitrarily small spatial scales. As a first approach, {\em we show that coherent energy cascades can arise in Hamiltonian systems with random nonlinear couplings, provided that an infinitesimal fraction of those couplings remain structured.} Introducing randomness in the Hamiltonian structure is a significant departure from small deformations of previously studied interaction networks concerning coherent cascades. It reveals that rigidly controlling the entire algebraic structure of the Hamiltonian is unnecessary for coherent energy transfer: a minimal prescribed core in the mode interactions may be sufficient.

This study differs fundamentally from the more familiar random-phase cascades of wave kinetic theory \cite{Book_Nazarenko,Book_Galtier,Book_Zakharov}. In that setting, the evolution preserves the uncorrelated phases imposed in the initial state and can be therefore effectively expressed in terms of the corresponding random-phase averages. Randomness in the interaction network can be and has been considered \cite{PicozziDramaticAceleration,EvninMelonicTurbulence,Frahm, HuRosenhaus}, having a venerable history in relation to turbulence and related problems \cite{rnd}; see also \cite{HaniLiNahmodStaffilani,FaouCarles,NazarenkoKrstulovicZhuSemisalov} for the inclusion of randomness through external sources. In such settings, interaction randomness can support energy transfer while contributing, heuristically, to the dynamical randomization of phases. Here, by contrast, randomness is introduced precisely to disrupt the global structure of the interaction network and thereby test how much organization is actually required for phase-coherent energy transport. The key question for us is whether a minimal structured component in the mode interactions can still organize the dynamics and sustain phase locking. Our results show that it can, suggesting that phase-coherent and random-phase cascades are not necessarily distinguished by the presence or absence of disorder, but by whether a residual organized structure is capable of supporting the alignment of phases during energy transport.

    We make our statement precise in the setting of infinite-dimensional resonant Hamiltonian systems, which often arise as effective models for weakly nonlinear Hamiltonian PDEs \cite{Kuksin3,E2}. We consider the cubic equation of motion derived from systems with fully resonant frequency spectra ($\omega_n$ linear in $n$), such as  nonlinear Schr\"odinger equations with harmonic potentials \cite{FaouGermainHani,GGT,GHT,BBCE}, nonlinear waves  in some special geometries \cite{CEL,CF,GGwaveequation}, and relativistic waves in anti-de Sitter space \cite{BMR,CEV1,BEF,E2}:
    \begin{equation}
		i\frac{d\alpha_n}{dt} = \underset{n+m=k+j}{\underbrace{\sum^{\infty}_{m=0}\sum^{\infty}_{k=0}\sum^{\infty}_{j=0}}}\  C_{nmkj}\ \bar{\alpha}_m \alpha_k \alpha_j,
		\label{eq:Resonant_Equation}
	\end{equation}
	where $\alpha_n\in\mathbb{C}$ are the dynamical variables written in the `interaction picture' ($a_n = \alpha_n \ e^{i\,\omega_n t}$) with respect to (\ref{eq:generic_Hamiltonian}), $n\in\mathbb{N}$ is the mode number, $n+m=k+j$ the resonance condition, and $C_{nmkj} \in \mathbb{R}$ the interaction coefficients, which satisfy the permutation symmetries:
    \beq
	C_{nmkj} = C_{mnkj} = C_{nmjk} = C_{kjnm}.
	\label{eq:C_symmetries}
	\eeq
    Such symmetries guarantee a Hamiltonian structure and two conserved quantities, which correspond to the `number of particles' (the
	$L^2$-norm) and the quadratic part of the energy when coming from nonlinear Schr\"odinger equations:
	\beq
	\mathcal{H} = \frac{1}{2}\sum_{\underset{n+m=k+j}{n,m,k,j}} C_{nmkj} \bar{\alpha}_n \bar{\alpha}_m \alpha_k \alpha_j, \qquad 	N = \sum_{n=0}^{\infty} |\alpha_n|^2, \qquad \text{and} \qquad E = \sum_{n=0}^{\infty} n |\alpha_n|^2.
	\label{eq:conserved_quantities}
	\eeq
    
    Previous studies of coherent cascades in these settings \cite{GG,Xu,GG2,GGH,BE,Biasi,BG,BMR, Jalmuzna,MBox2} considered nonlinear couplings $C_{nmkj}$ with a prescribed functional dependence on the indices.  We instead promote most entries of $C_{nmkj}$ to random variables, either sampled from independent probability distributions or evolved as stochastic processes. Under such randomization, we prove that retaining only a small subset of structured entries in $C_{nmkj}$ --- an infinitesimal fraction of the total --- is sufficient to sustain coherent transfer of energy from low to arbitrarily high modes. This is characterized by the formation of a power-law tail in the energy spectrum ($E_n \equiv n |\alpha_n|^2 $) with exponential suppression of high modes:
	\beq
	E_{n\gg 1}(t) \sim n^{\gamma} e^{-\rho(t) n}
	\ \underset{\rho\, \to\, 0}{\longrightarrow} \
	n^{\gamma},
	\eeq
    and asymptotic phase locking:
	\beq
	\arg\bigl(\alpha_{n\gg 1}(t)\bigr) \sim \phi_1(t)\,n + \phi_0(t),
	\eeq
	with $\rho>0$ and $\gamma,\phi_1,\phi_0\in\mathbb{R}$. 

	
	\subsection{Strategy}

    The difficulty in this paper lies in describing rigorously coherent energy cascades in the presence of an overwhelming majority of random couplings. Our strategy builds on an existing approach to describe coherent cascades in deterministic systems of the form (\ref{eq:Resonant_Equation}), whose idea is to capture cascades supported by a low-dimensional invariant manifold \cite{GG,GG2,GGH,Biasi,Xu,BE,BG,minHamiltonian}. {\em We develop a systematic construction of systems with the same kind of dynamically invariant manifold as in}  \cite{GG,GG2,GGH,Biasi,Xu,BE,BG,minHamiltonian} {\em and then randomize the interaction coefficients $C_{nmkj}$ that are not constrained by its dynamical invariance.} Then, we show that, while the invariant manifold is preserved even when most interaction coefficients are replaced by random variables, the coherent cascades are unaltered by this modification. 
    
    {\em To our knowledge, this construction results in the first rigorous description of coherent energy cascades in nonlinear random Hamiltonian systems with mostly random couplings.} This is however established on the invariant manifold, and we supplement this result with numerical simulations to investigate the robustness of coherent cascades beyond those analytically tractable conditions. We examine the emergence of phase locking and power-law spectra from random initial data both in systems with the invariant manifold underlying our construction and in systems that lack this analytic feature.

	
	\subsection{Results}
	\label{sec:Results}

    The results fall naturally into three parts that we summarize here:
	
	\noindent {\bf Part 1: Minimally structured Hamiltonians.} We develop a systematic construction of resonant Hamiltonian systems of the form (\ref{eq:Resonant_Equation}) that support an invariant manifold of the form
	\beq
	\alpha_0 = b,\qquad \alpha_{n\geq 1} = \sqrt{\h_n}\ c\,p^{n-1},
	\label{eq:invariant_manifold_RESULTS_SECTION}  
	\eeq
	where $b,c,p\in\mathbb{C}$ are the dynamical variables and $\h_n$ is a time-independent, positive sequence to be specified later. We shall see that this sequence presents an asymptotic behavior $F_{n\gg 1} \sim n^{\gamma} x_c^{-n}$, which requires $|p|^2 < x_c$, so that we consider scenarios where the total amount of energy is finite and primarily concentrated in the low modes. The specific form of the couplings $C_{nmkj}$ and the sequence $\h_n$ will be provided in Section~\ref{subsec:existence_invariant_manifold}. This treatment constitutes a major qualitative improvement over the previous related considerations in \cite{Biasi,BG}, as it gives a systematic construction of very large classes of resonant systems admitting invariant manifolds, and furthermore identifies precisely the large residual ambiguity in the mode couplings that we use essentially for coupling randomization.
	
	We show that this construction supports systems whose nonlinear couplings $C_{nmkj}$ are predominantly independent random variables (fixed in time or evolving as stochastic processes), while just an infinitesimal fraction are rigidly prescribed: if $\mathscr{N}_D$ and $\mathscr{N}_R$ denote the numbers of deterministic and random entries of $C_{nmkj}$, respectively, up to a resonant level with $n+m\leq M$, then their ratio vanishes in the limit $M\to \infty$,
    \beq
        \mathscr{F} := \frac{\mathscr{N}_D}{\mathscr{N}_R + \mathscr{N}_D} \xrightarrow[M\to\infty]{} 0.
    \eeq
   Yet, we prove that these largely disordered systems admit energy cascade solutions on the invariant manifold, independent of the probability distributions or stochastic processes followed by the unconstrained couplings.
	
	\noindent {\bf Part 2: Classification of energy cascades.} We characterize three types of energy cascades on the invariant manifold (\ref{eq:invariant_manifold_RESULTS_SECTION}); see Figure~\ref{fig:examples_coherent_energy_cascades_MANIFOLD} for their visual presentation. The transfer of energy to arbitrarily high modes is quantified by either the formation of a power-law tail in the energy spectrum or, equivalently, by unbounded growth or finite-time blow-up of Sobolev norms, denoted by:
	\beq
	H^{s} := \left(\sum_{n=0}^{\infty}(n+1)^{2s}|\alpha_n|^{2}\right)^{1/2}.
	\label{eq:Sobolev_norms}
	\eeq
	\begin{itemize}
		\item \textbf{Cascade I:} Infinite-time formation of a positive power law, inducing unbounded growth of Sobolev norms:
		\beq
		\frac{E_{n\geq 1}(t)}{|c(t)|^2} \xrightarrow[t\to\infty]{} n, \qquad \text{and} \qquad H^{s>1/2} \xrightarrow[t\to\infty]{} \infty. 
		\eeq
		
		\item \textbf{Cascade II:} 
		Finite-time formation of a power law $n^{-1/2}$, inducing blow-up of Sobolev norms:
		\beq
		\frac{E_{n\gg1}(t)}{|c(t)|^2} \xrightarrow[t\to T]{} n^{-1/2}, \qquad \text{and} \qquad 	H^{s>1/2} \xrightarrow[t\to T]{} \infty.
		\eeq
		
		\item \textbf{Cascade III:} 
		Finite-time formation of a power-law $n^{-3/2}$, inducing blow-up of Sobolev norms:
		\beq
		\frac{E_{n\gg1}(t)}{|c(t)|^2}
		\xrightarrow[t\to T]{} n^{-3/2}, \qquad \text{and} \qquad H^{s\geq 3/4} \xrightarrow[t\to T]{} \infty.
		\eeq
	\end{itemize}
	The first two cascades effectively ``dilute" ($E_n\to 0$) the energy over the space of modes: the energy becomes distributed across an increasing range of modes in such a way that the amount of energy stored in any individual mode becomes arbitrarily small. Specifically, as the power law is gradually approached, conservation of $E$ forces the amplitude $|c(t)|^2$ to decay to zero at the exact rate required for the total energy sum to remain finite. This is why Figure~\ref{fig:examples_coherent_energy_cascades_MANIFOLD} displays $E_n/|c|^2$ rather than $E_n$ itself, which would otherwise decrease as the power law develops. Such processes are interpreted as the coherent formation of a condensate as discussed in detail in \cite{Biasi, BG}: the system effectively approaches the ground state $|\alpha_n|^2 = N \delta_{n0}$, which has zero energy $E=0$ and stores the quantity $N$ entirely in this lowest mode. By contrast, the third cascade spreads the energy over the modes in such a way that all modes can retain a non-vanishing amount of energy.   
	
	\begin{figure}[h!]
		\centering
		\includegraphics[width = 7.6cm]{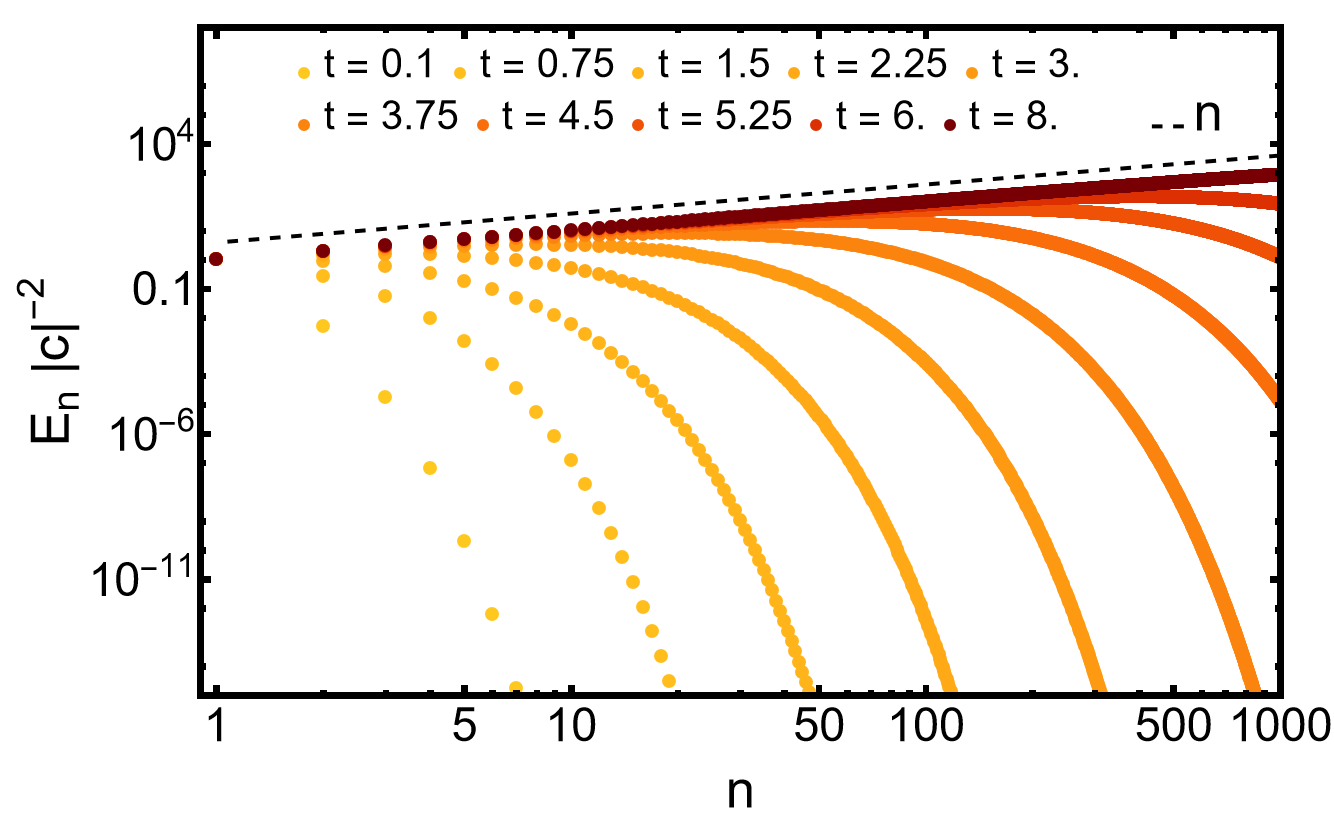}
        
		\includegraphics[width = 7.6cm]{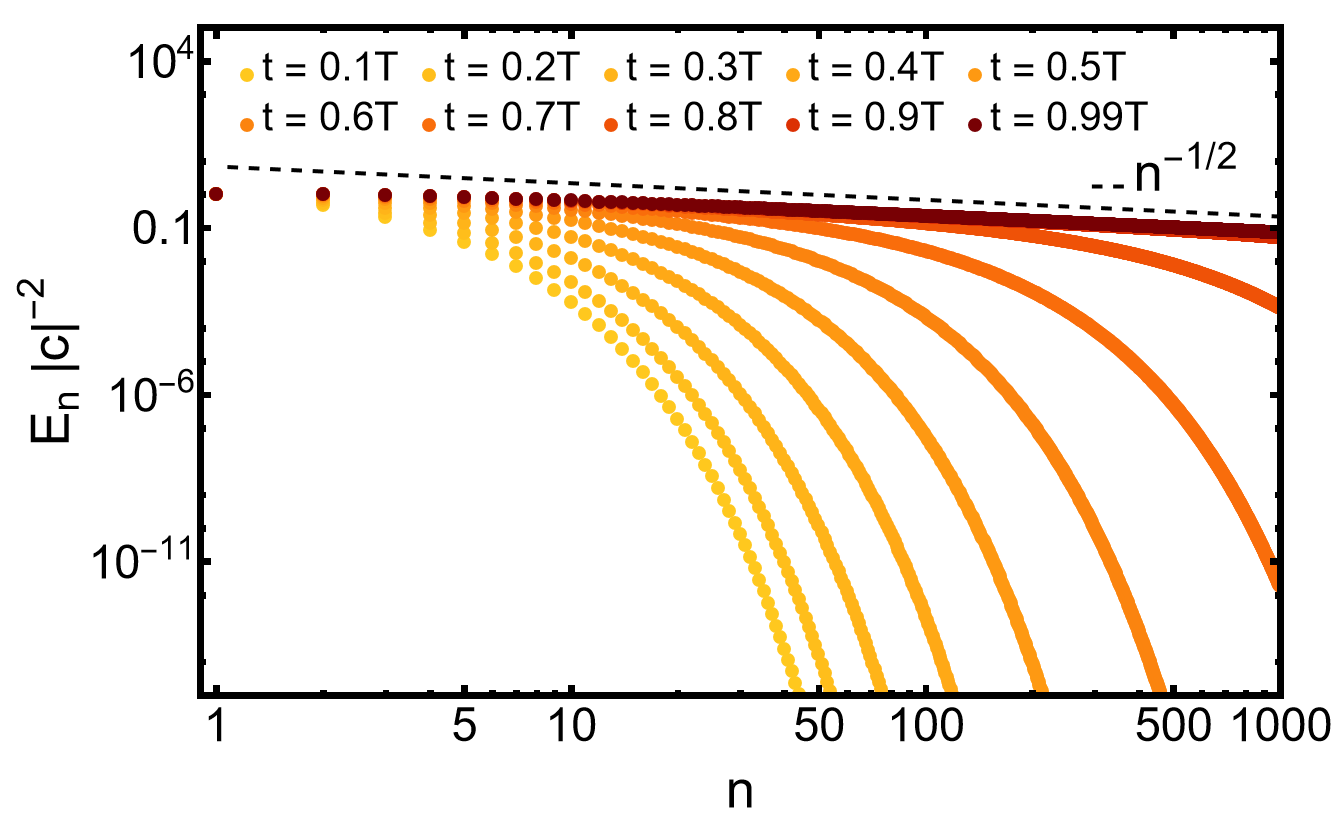}
		\includegraphics[width = 7.6cm]{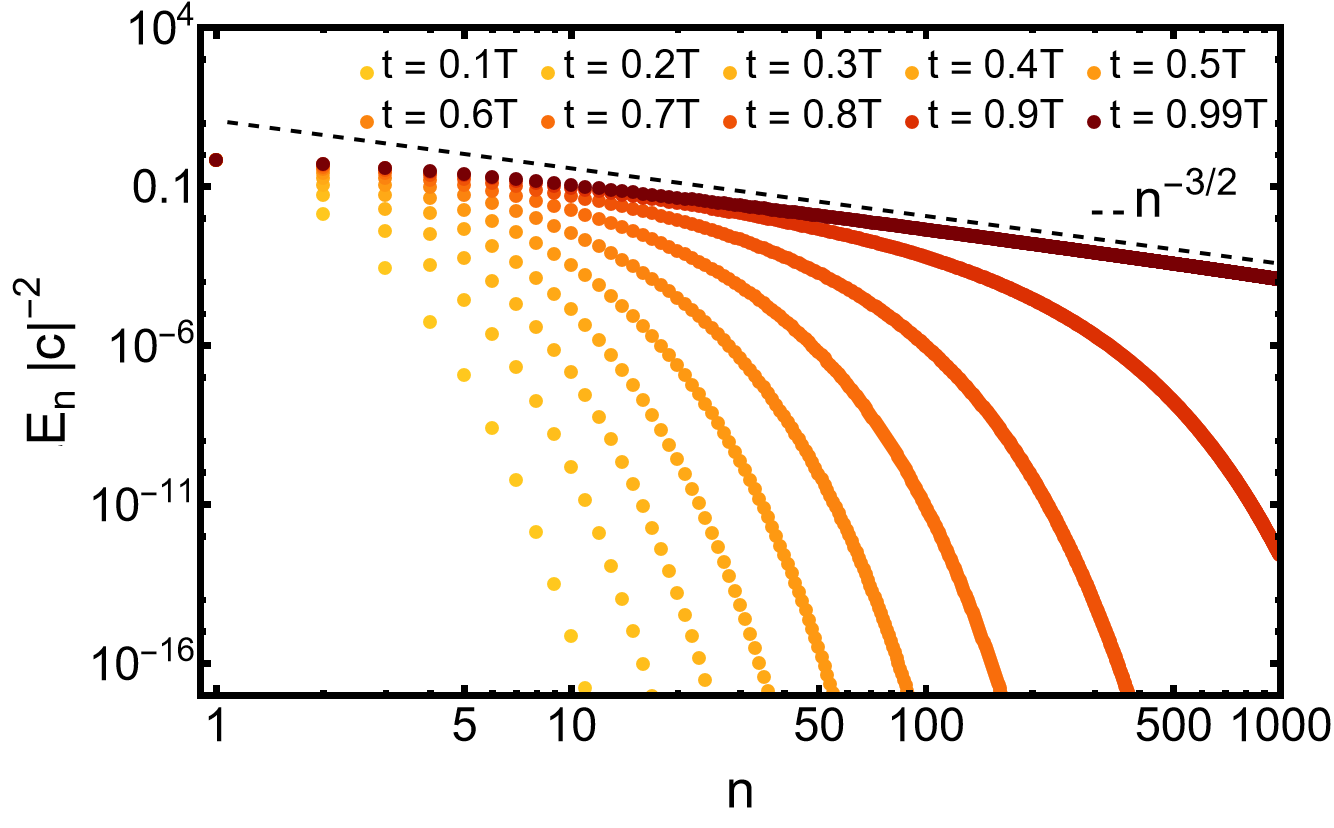}
		\caption{Analytic solutions for the three types of cascades present on the invariant manifold.}
		\label{fig:examples_coherent_energy_cascades_MANIFOLD}
	\end{figure}

    \newpage
	
	\noindent 	{\bf Part 3: Coherent cascades from random initial data.} We next perform numerical simulations of the above minimally structured systems but starting with random initial data (random phases and amplitudes) instead of the analytically tractable conditions on the invariant manifold. In this case, we consider systems with two classes of random structures for the couplings: those whose random couplings are drawn from independent identical Gaussian distributions and held fixed for each realization, and those whose random couplings evolve as a stochastic Ornstein--Uhlenbeck processes \cite{Stochastic_Processes}. Full details are provided in Section~\ref{sec:numerics}.
	
	Numerical simulations reveal that systems admitting Cascade III solutions exhibit qualitatively similar cascades starting from random initial data in a significant fraction of the simulations ($60\%-20\%$). As illustrated in Figure~\ref{fig:examples_coherent_energy_cascade_III}, they display:
	\begin{itemize}
		\item 	{\bf Power-law formation:} The energy spectrum develops a power-law tail close to the one exhibited by the analytic solutions on the manifold:
		\beq
		E_{n\gg 1}(t) \sim n^{-3/2}.
		\eeq
		
		\item 	{\bf Phase locking:}  Randomly distributed initial phases evolve to an approximately linear relation for sufficiently high modes,
		\beq
		\arg\bigl(\alpha_{n\gg 1}(t)\bigr) \sim \phi_1(t)\,n + \phi_0(t).
		\eeq
	\end{itemize}
	Systems admitting Cascades I and II do tend to exhibit phase locking when initialized with random data. Still, on most occasions, they are unable to maintain it or develop a clear power-law tail in the energy spectrum, as shown in Figure~\ref{fig:examples_coherent_energy_cascade_II}~and~\ref{fig:examples_coherent_energy_cascade_I}. This disparity between cascade types is consistent with the analysis of the invariant manifold carried out in Section~\ref{subsec:dynamics_within_the_manifold}: the domain of initial conditions on the manifold that lead to energy cascades is substantially larger for Cascade III than for Cascades I and II, and this difference appears to persist outside the manifold. 
    
    To conclude the numerical exploration, we assess the role of structured couplings in coherent cascade formation by simulating systems of the form (\ref{eq:Resonant_Equation}) both without a structured subset of couplings (fully nonlinear random systems) and with structured subsets distinct from those supporting the invariant manifold underlying our main construction. In fully random systems, the resonant condition ($n+m=k+j$) alone is insufficient to preserve or organize mode phases: the dynamics tends toward incoherence, although energy can still be transferred to higher modes, consistent with \cite{EvninMelonicTurbulence}. By contrast, systems with other structured subsets of the mode couplings can preserve or induce phase locking while maintaining energy transfer to high modes. These results show that resonance alone does not generate coherent cascades; rather, an organized subset of couplings provides the additional structure needed for phase-sensitive energy transfer. Details and visualizations are given in Section~\ref{sec:numerics}.

	
	\subsection{Structure of the paper}
	
	The paper is organized as follows. Section~\ref{subsec:existence_invariant_manifold} constructs Hamiltonian systems with explicit invariant manifolds, Section~\ref{subsec:Fn_sequence} studies structural properties of the manifold, Section~\ref{subsec:dynamics_within_the_manifold} provides a classification of dynamics and energy cascade solutions, Section~\ref{subsec:construction_predominantly_random_Hamiltonian_systems} constructs predominantly random and stochastic Hamiltonian systems with coherent energy cascade solutions, Section~\ref{sec:numerics} explores the dynamics of these systems beyond the invariant manifold, and Section~\ref{sec:Explicit_systems} provides a collection of explicit systems with the above properties. The paper concludes with a discussion in Section~\ref{sec:Discussion}. 
	
	\begin{figure}
		\centering
		\includegraphics[width = 5.6cm]{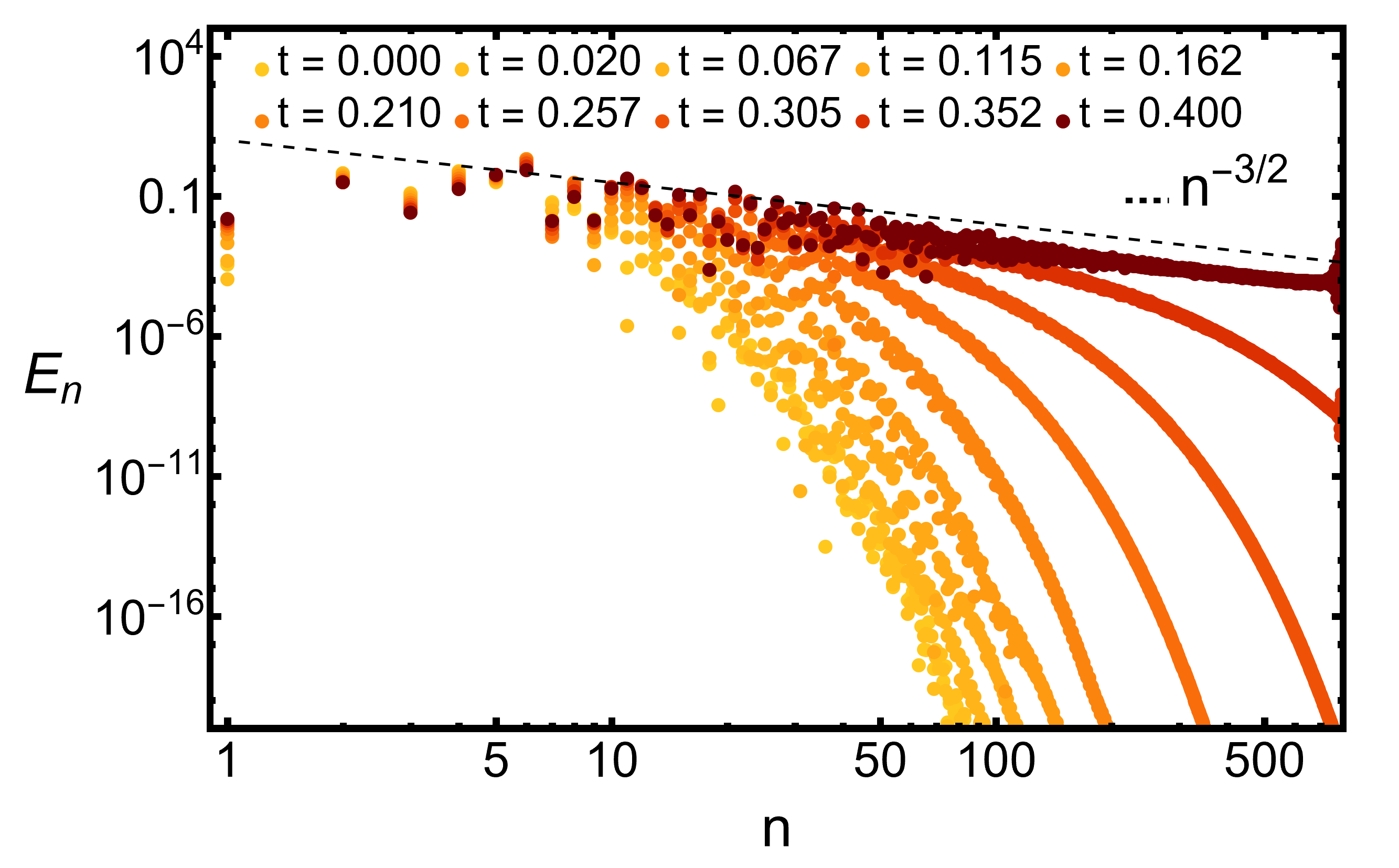}
		\includegraphics[width = 5.6cm]{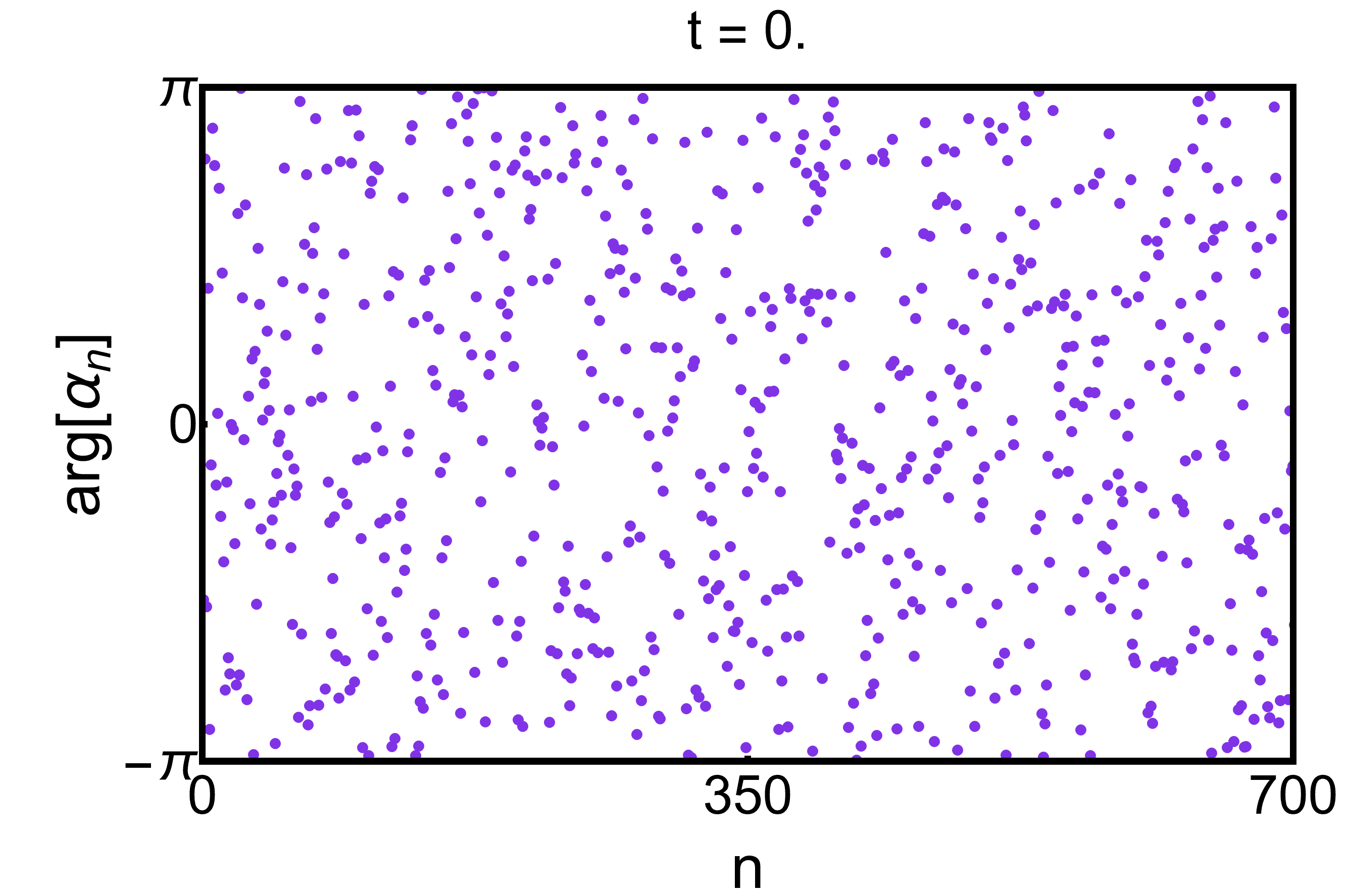}	
		\includegraphics[width = 5.6cm]{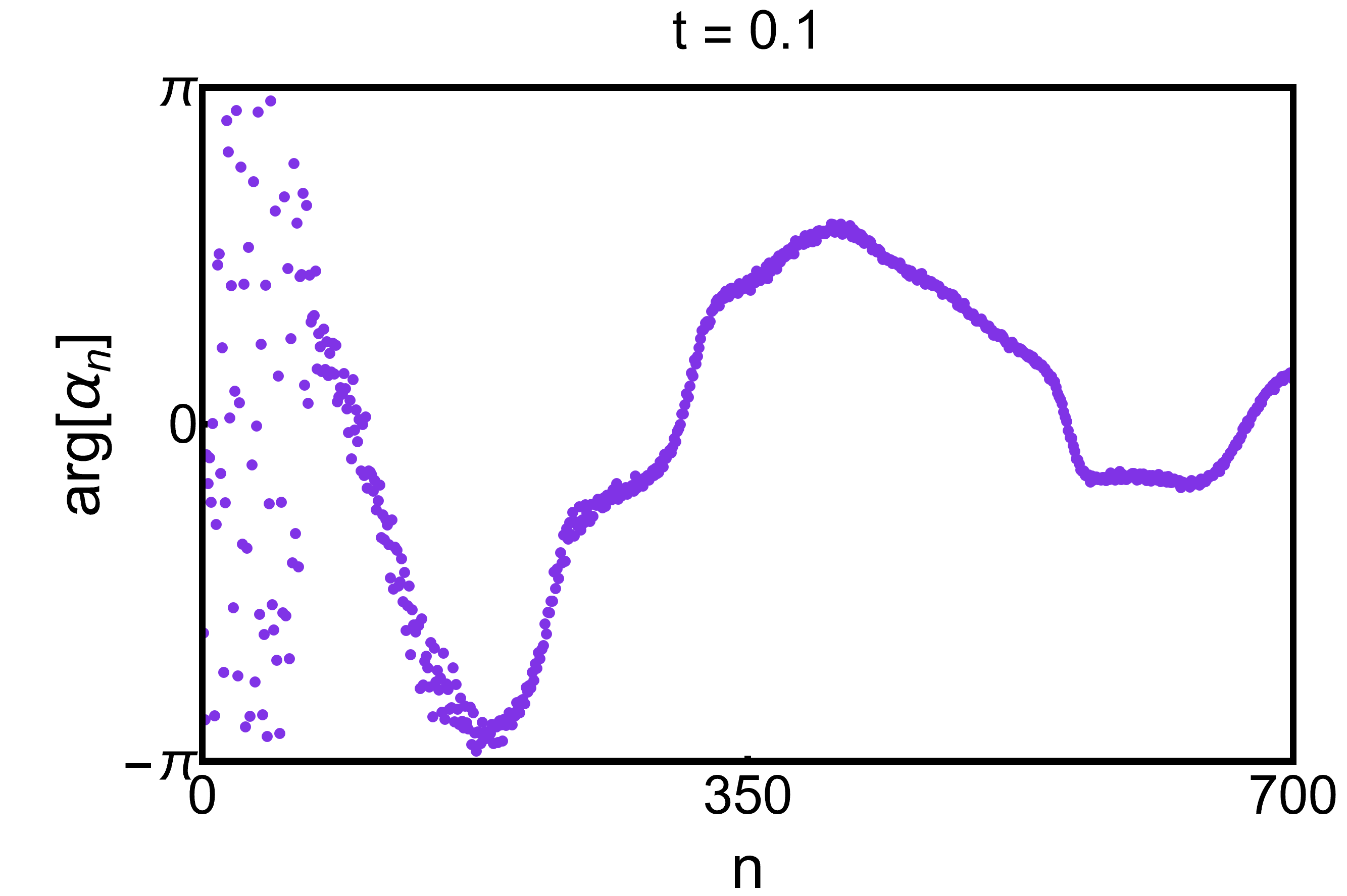}	
		
		\vspace{0.2cm}
		
		\includegraphics[width = 5.6cm]{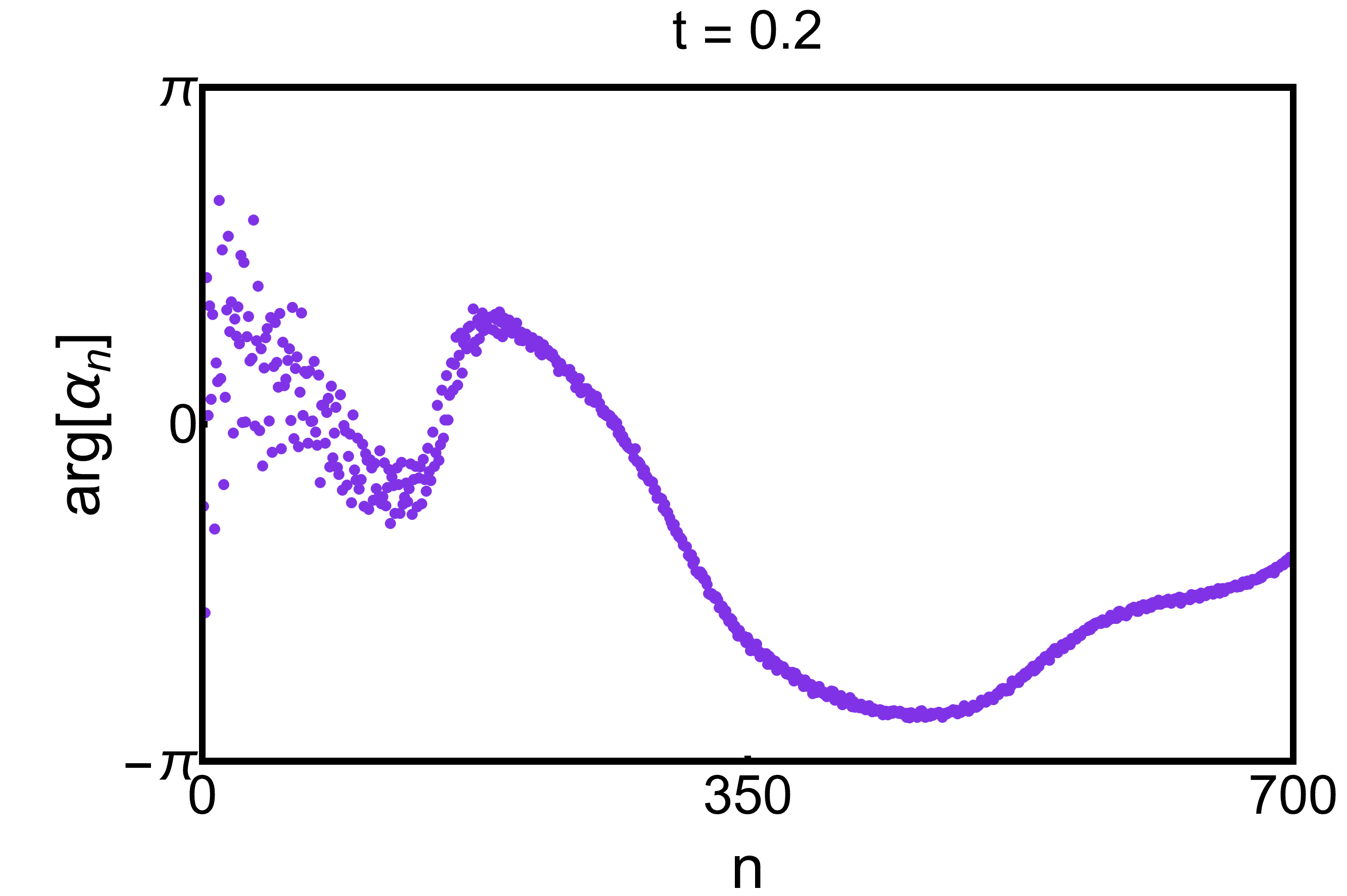}	
		\includegraphics[width = 5.6cm]{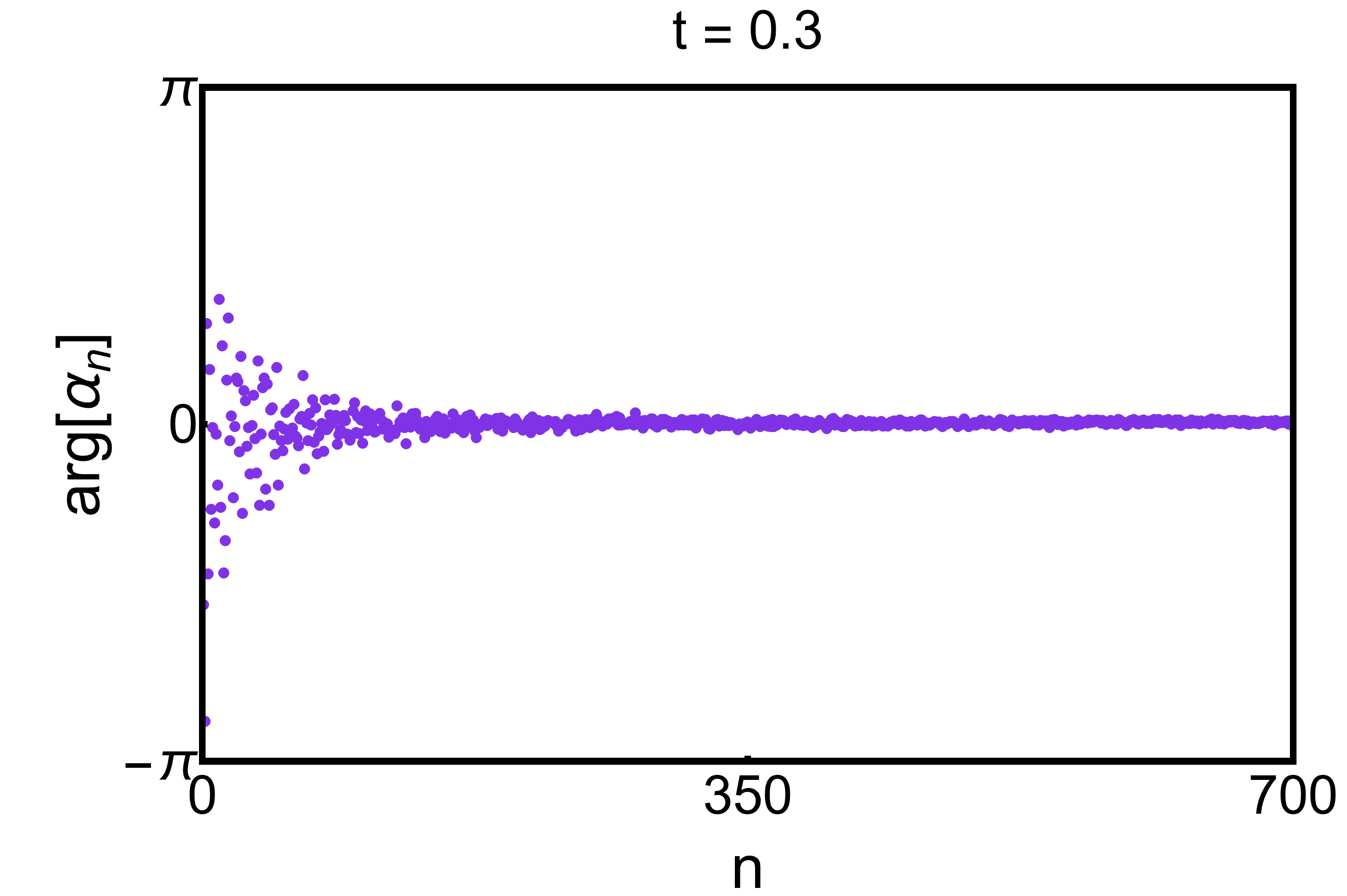}	
		\includegraphics[width = 5.6cm]{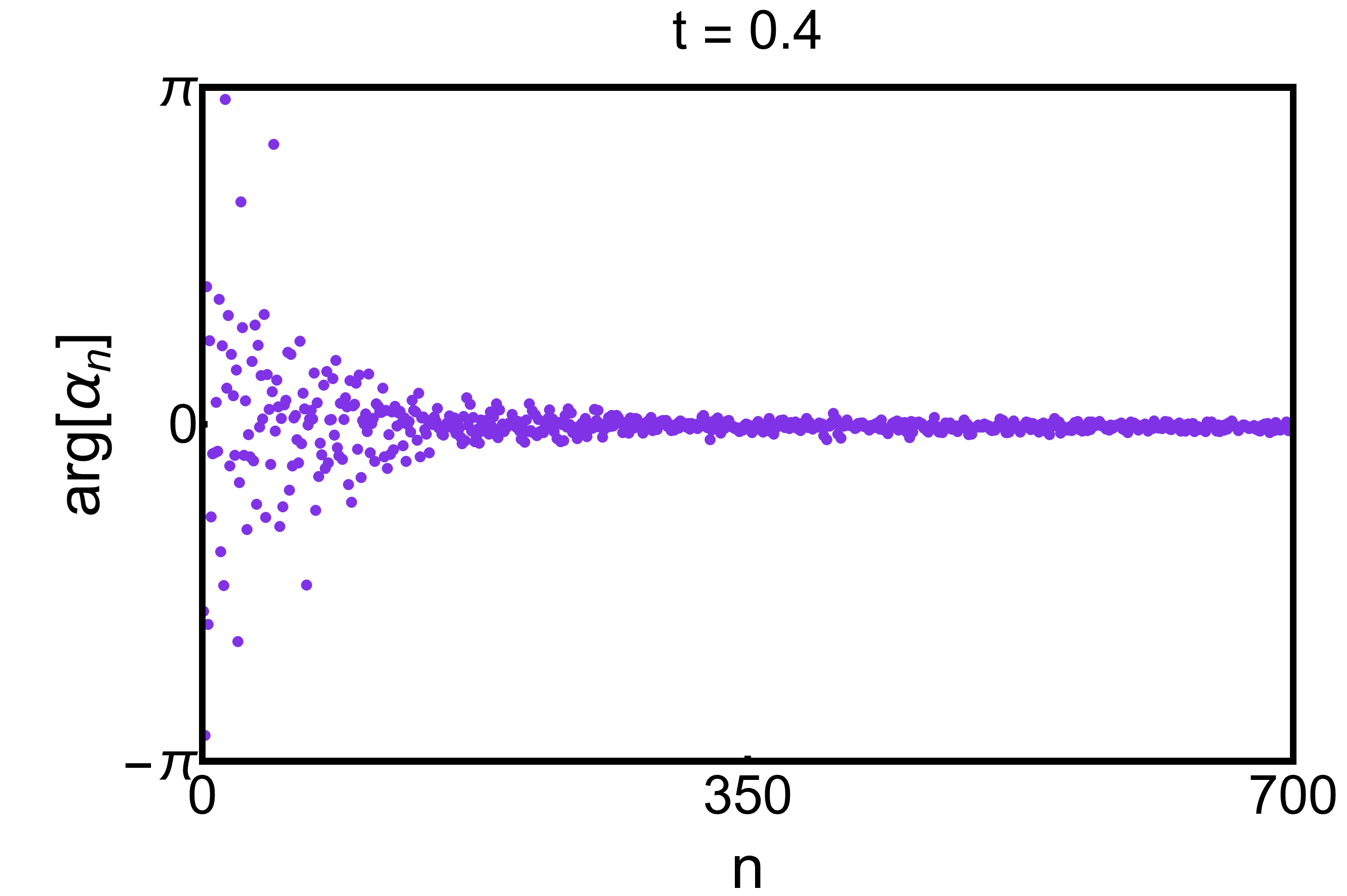}		
		\caption{Emergence of a coherent energy cascade from random initial data in a system admitting Cascade III solutions on the invariant manifold. {\em Upper-left corner}: power-law spectrum formation. {\em Remaining plots}: distribution of phases at successive times during the evolution. In these and subsequent plots, the phases have been rotated for visual purposes using the symmetry transformation $\alpha_n\to e^{i(\theta_1 n+\theta_0)}\alpha_n$ of resonant systems (\ref{eq:Resonant_Equation}), with $\theta_1,\theta_0\in\mathbb{R}$.}
		\label{fig:examples_coherent_energy_cascade_III}
	\end{figure}

	\begin{figure}
		\centering
		\includegraphics[width = 5.6cm]{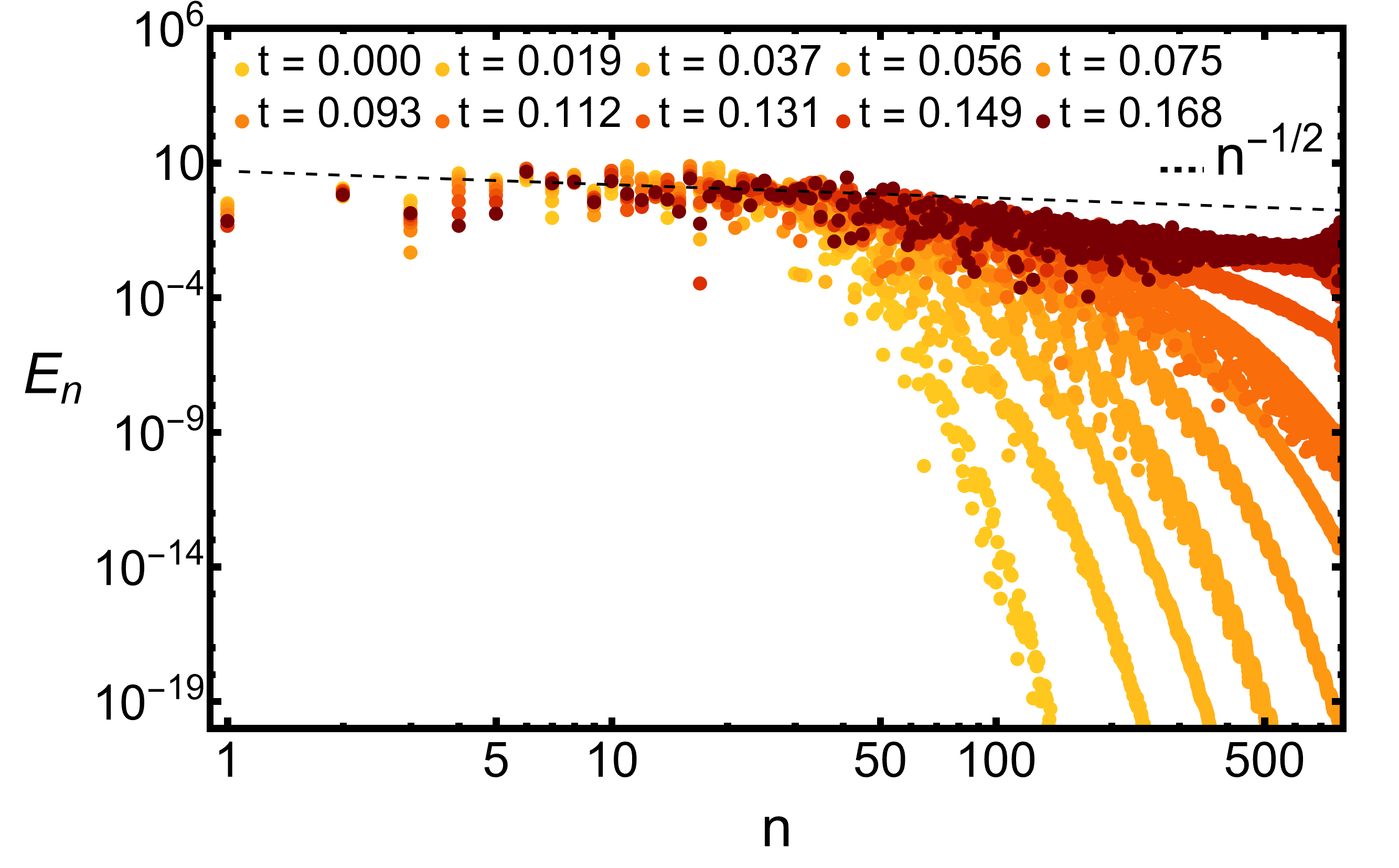}
		\includegraphics[width = 5.6cm]{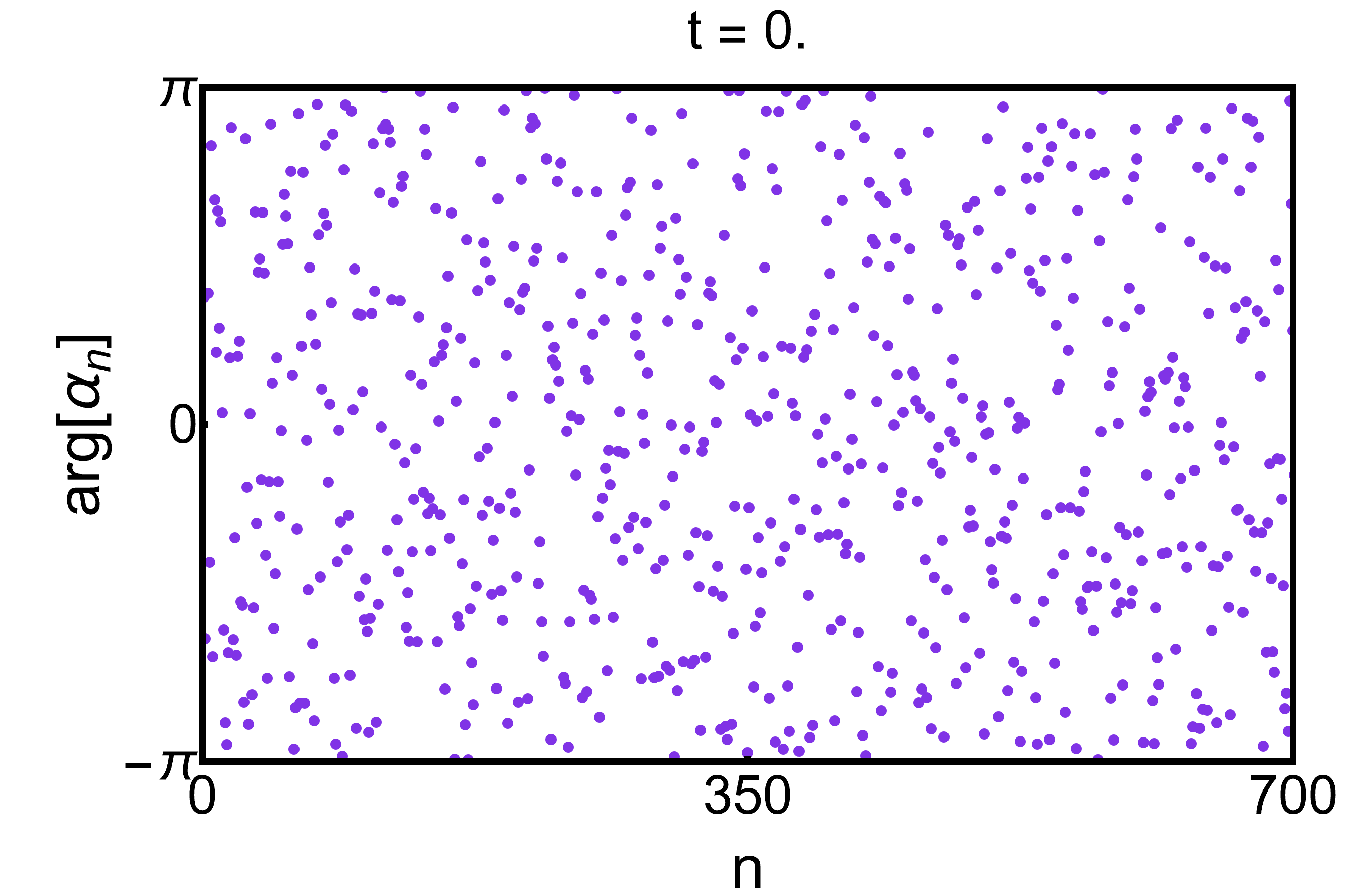}	
		\includegraphics[width = 5.6cm]{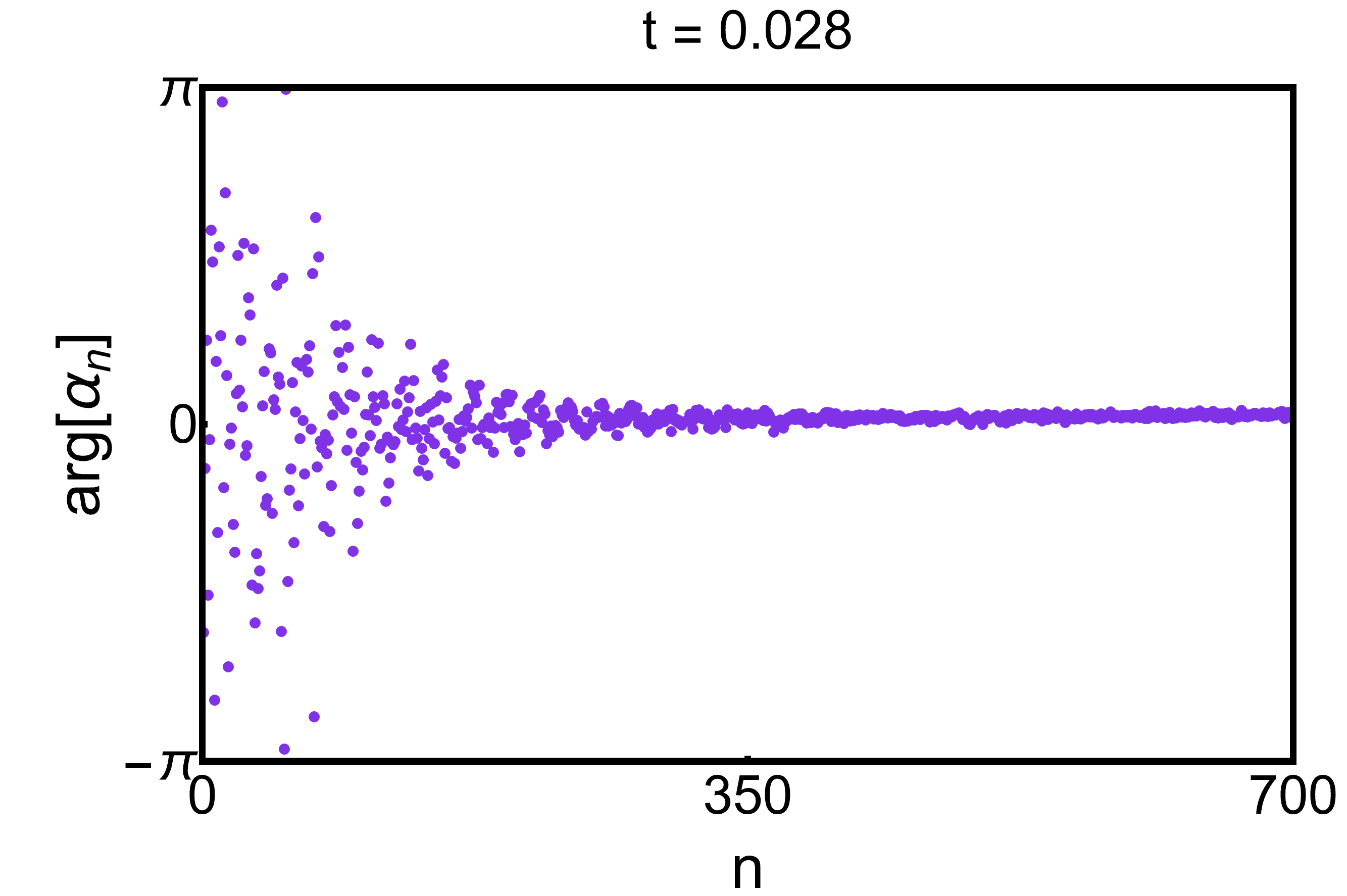}	
		
		\vspace{0.2cm}
		
		\includegraphics[width = 5.6cm]{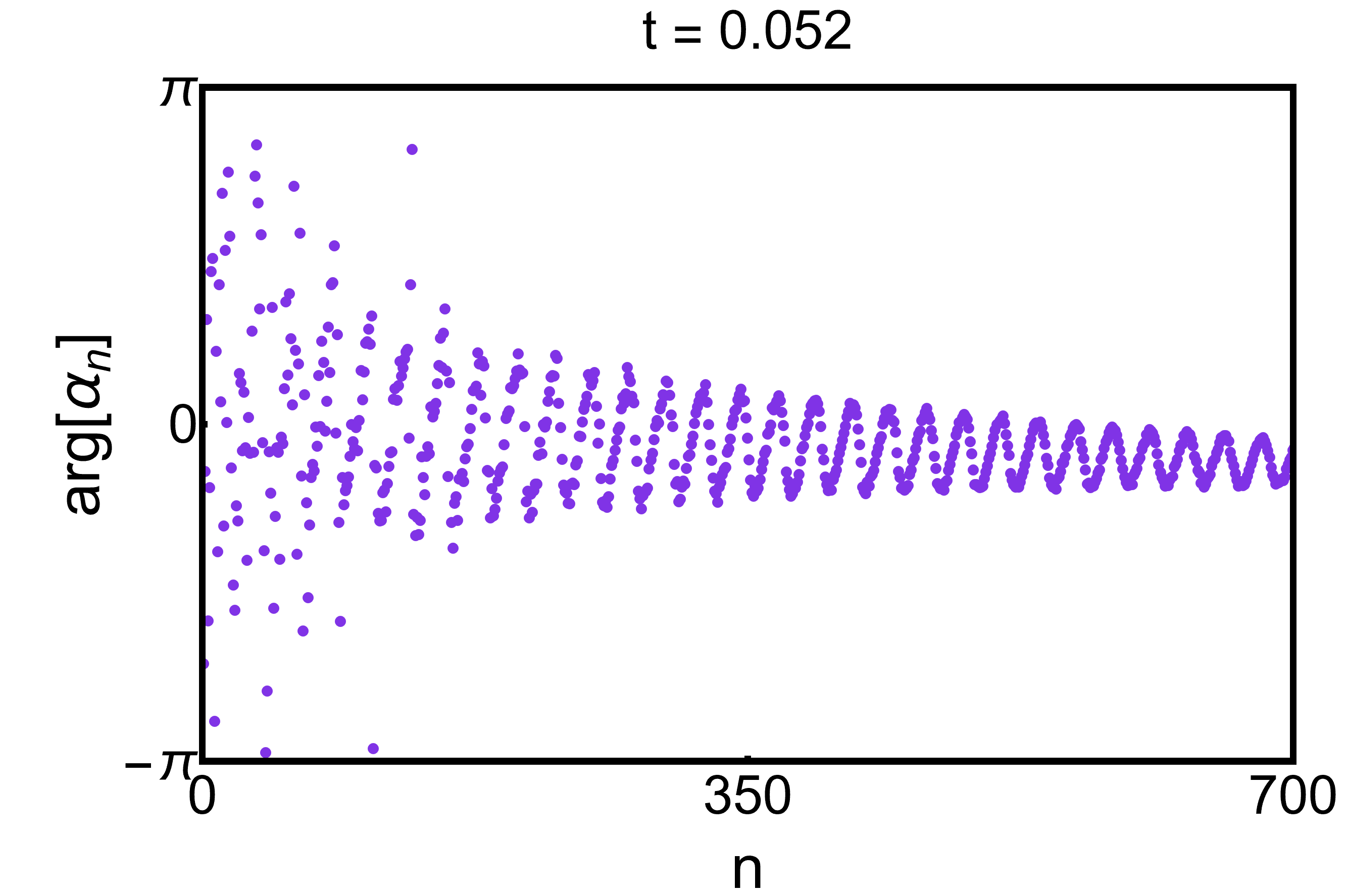}	
		\includegraphics[width = 5.6cm]{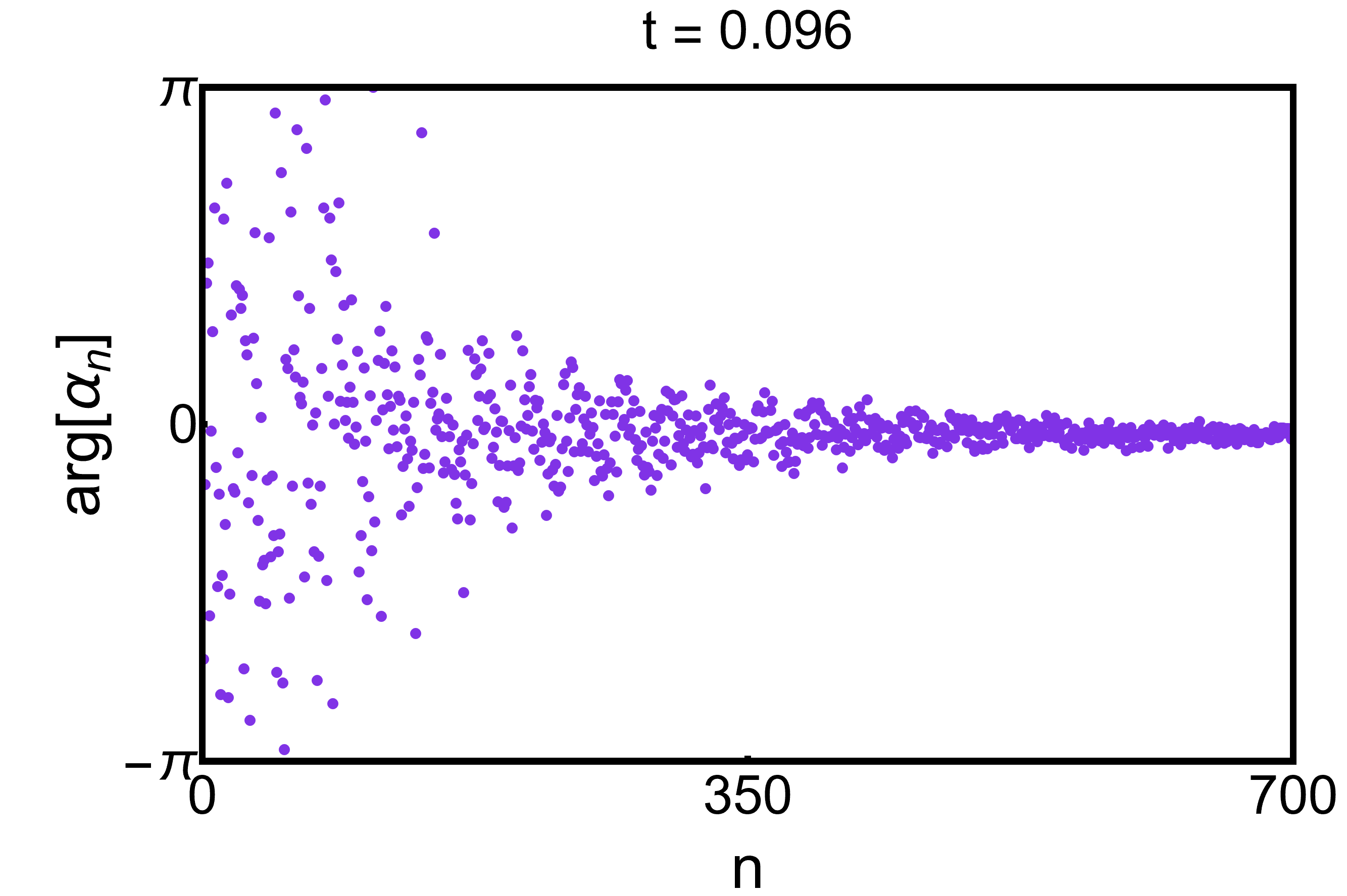}	
		\includegraphics[width = 5.6cm]{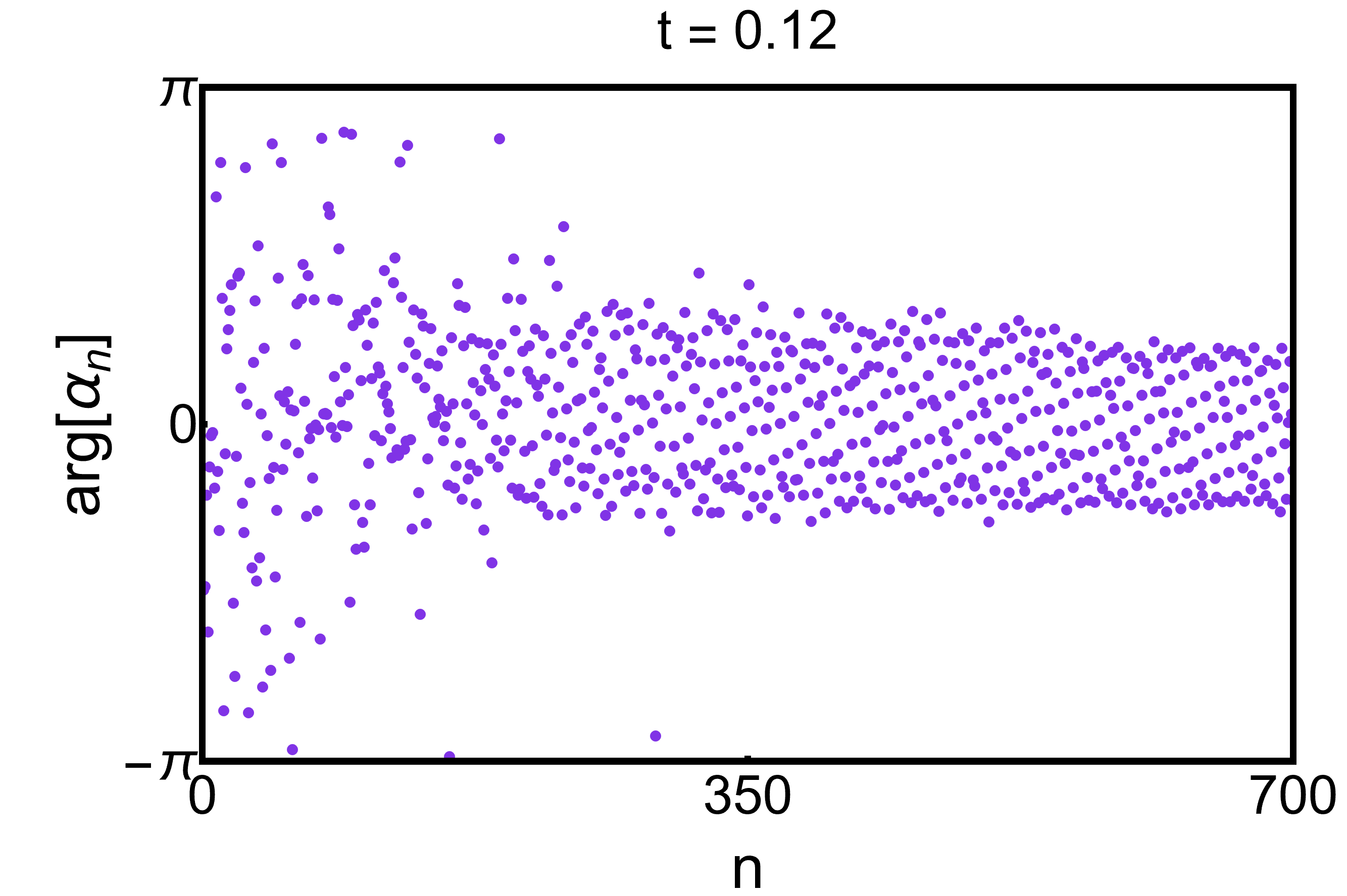}		
		\caption{Evolution of a system admitting Cascade II solutions on the invariant manifold, initialized with random initial data. {\em Upper-left corner}: energy spectrum at successive times. {\em Remaining plots}: phase distribution at the corresponding times.}
		\label{fig:examples_coherent_energy_cascade_II}
	\end{figure}
	
	\begin{figure}
		\centering
		\includegraphics[width = 5.6cm]{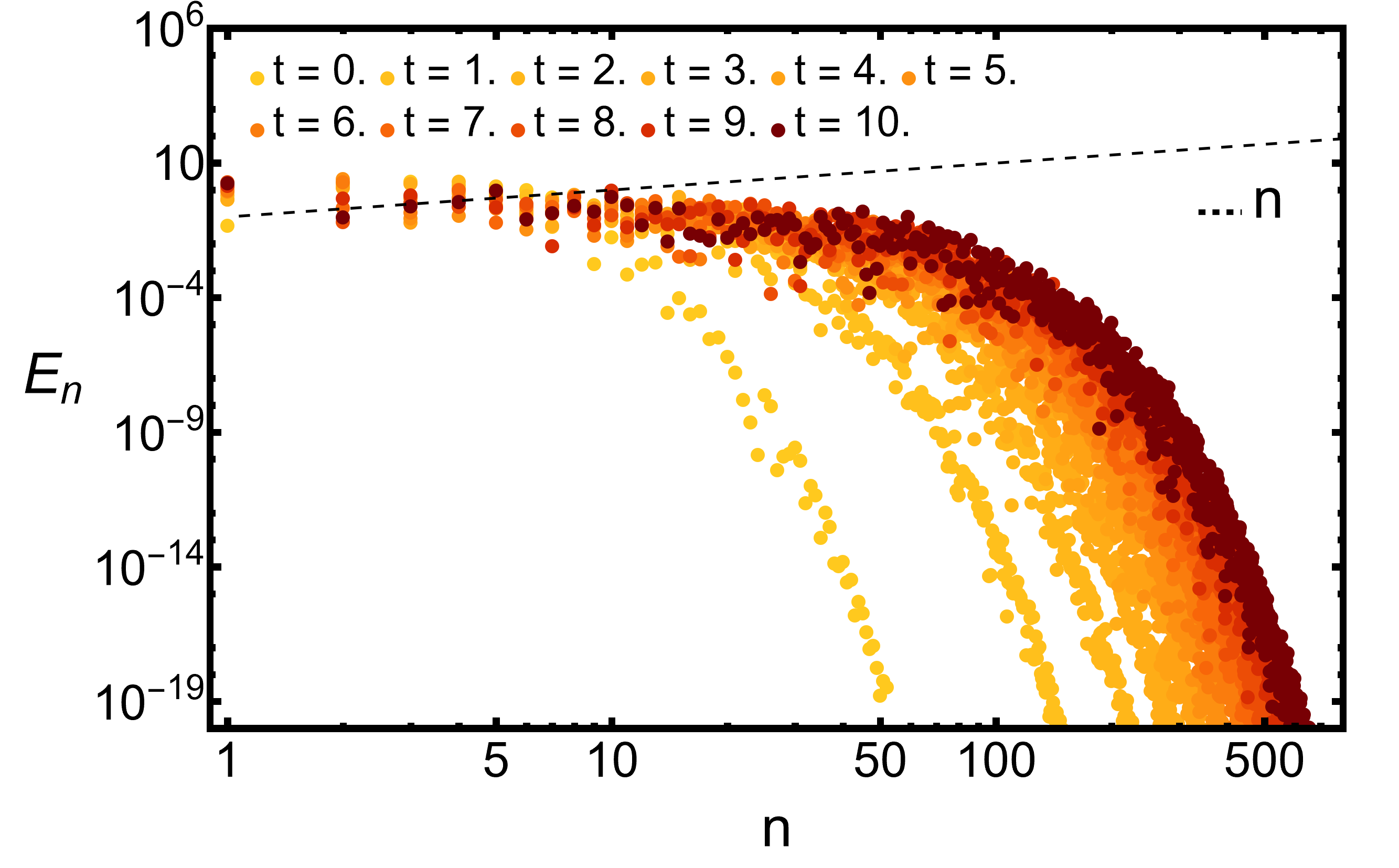}
		\includegraphics[width = 5.6cm]{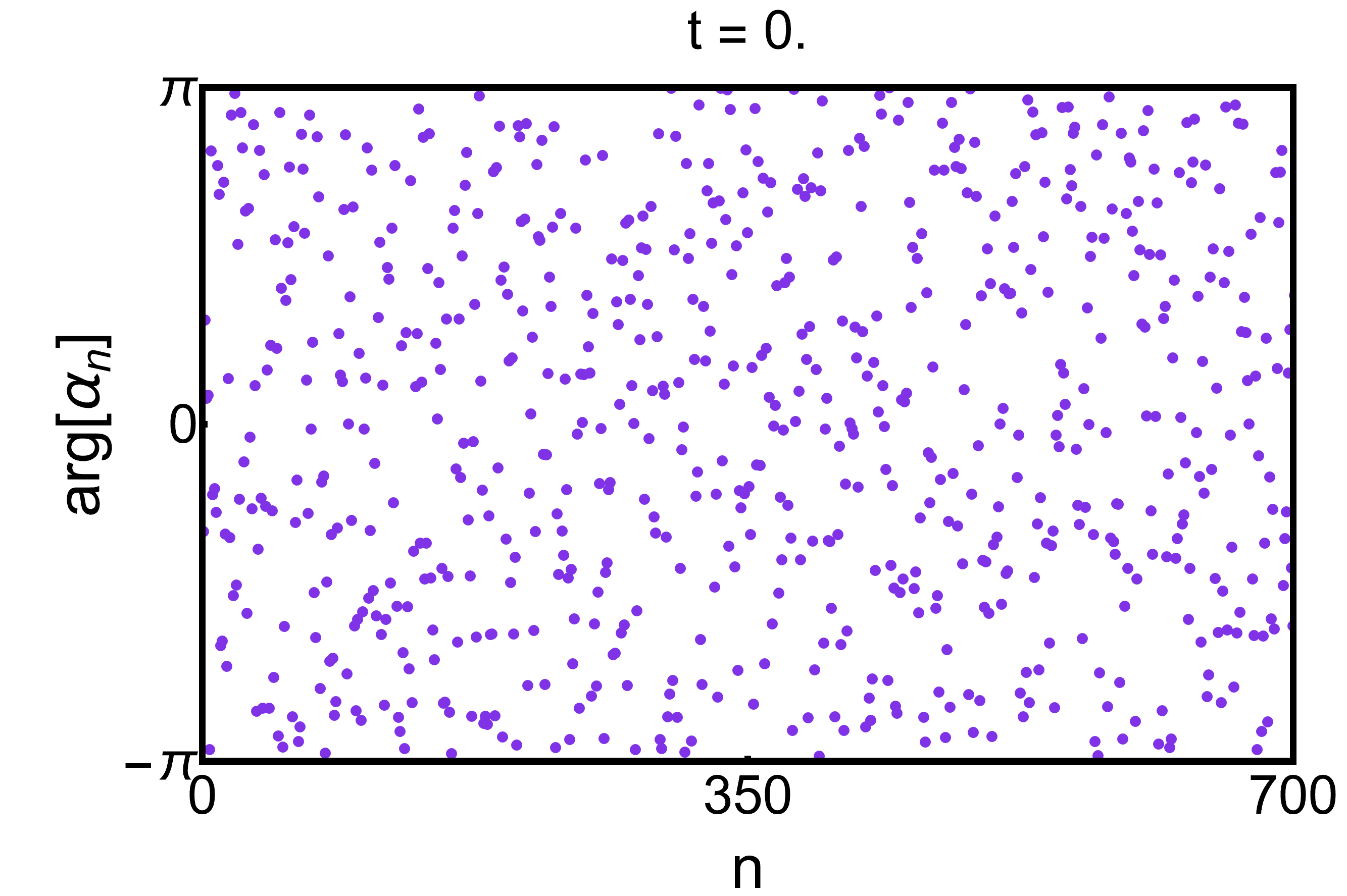}	
		\includegraphics[width = 5.6cm]{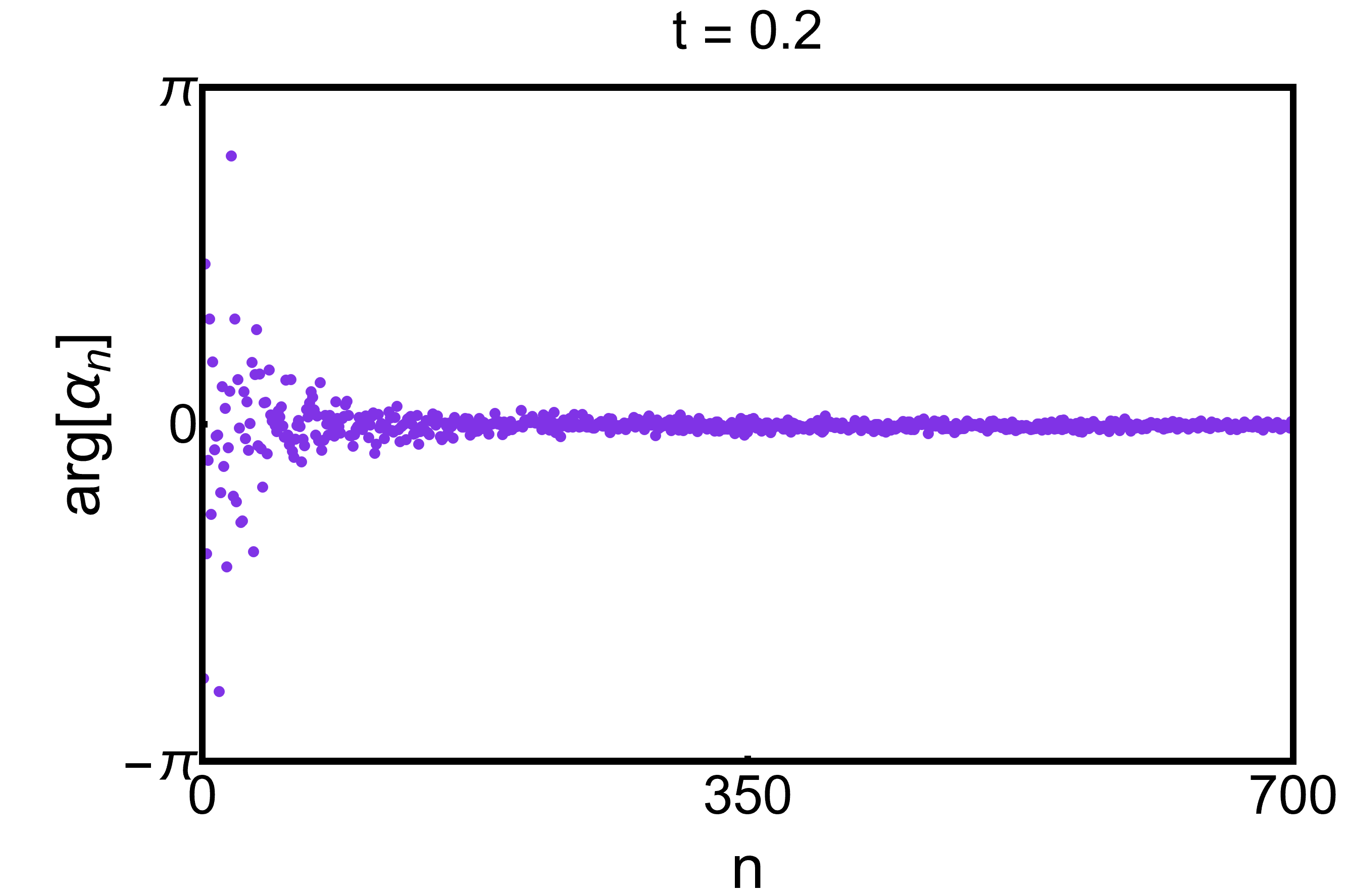}	
		
		\vspace{0.2cm}
		
		\includegraphics[width = 5.6cm]{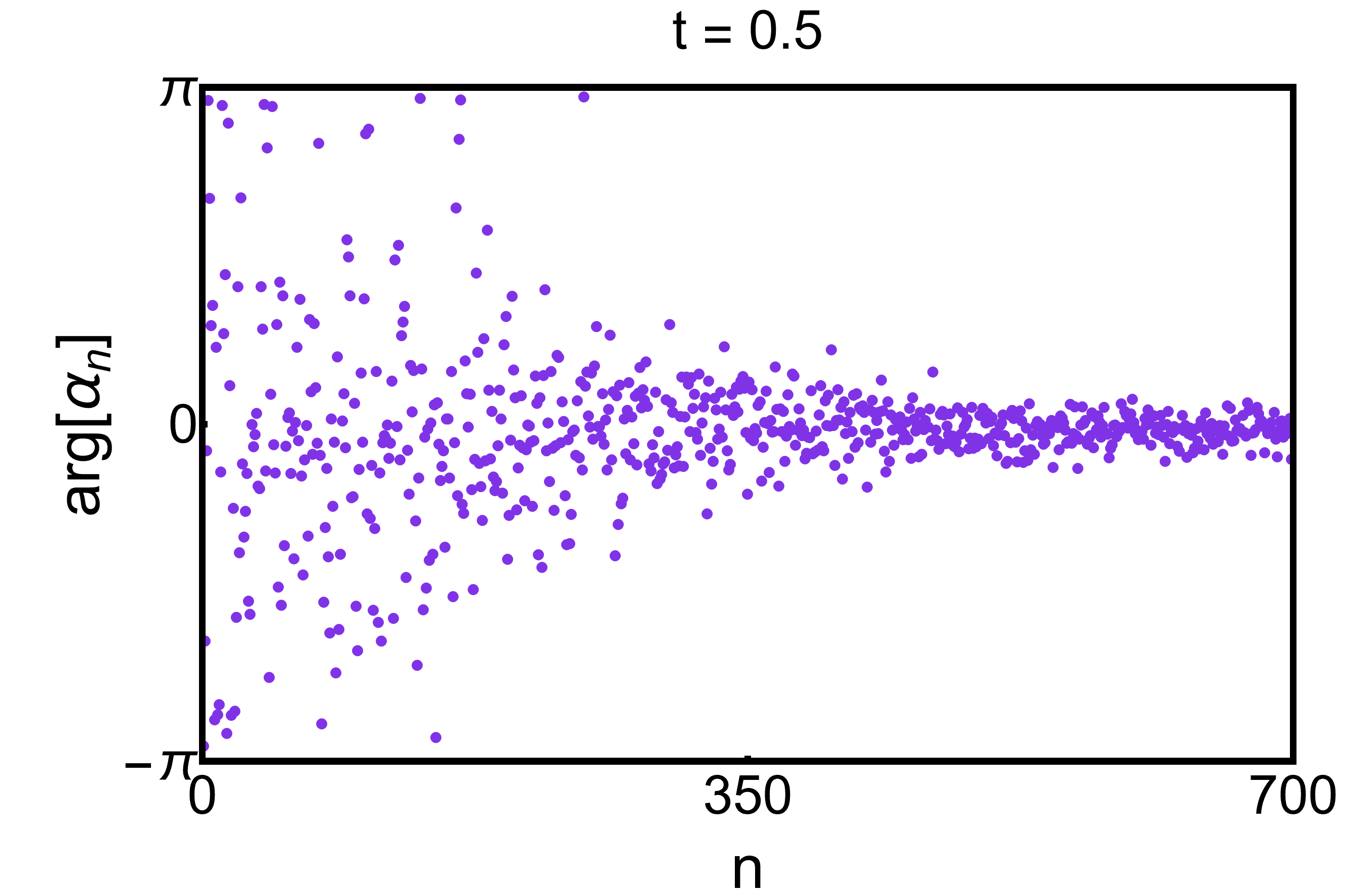}	
		\includegraphics[width = 5.6cm]{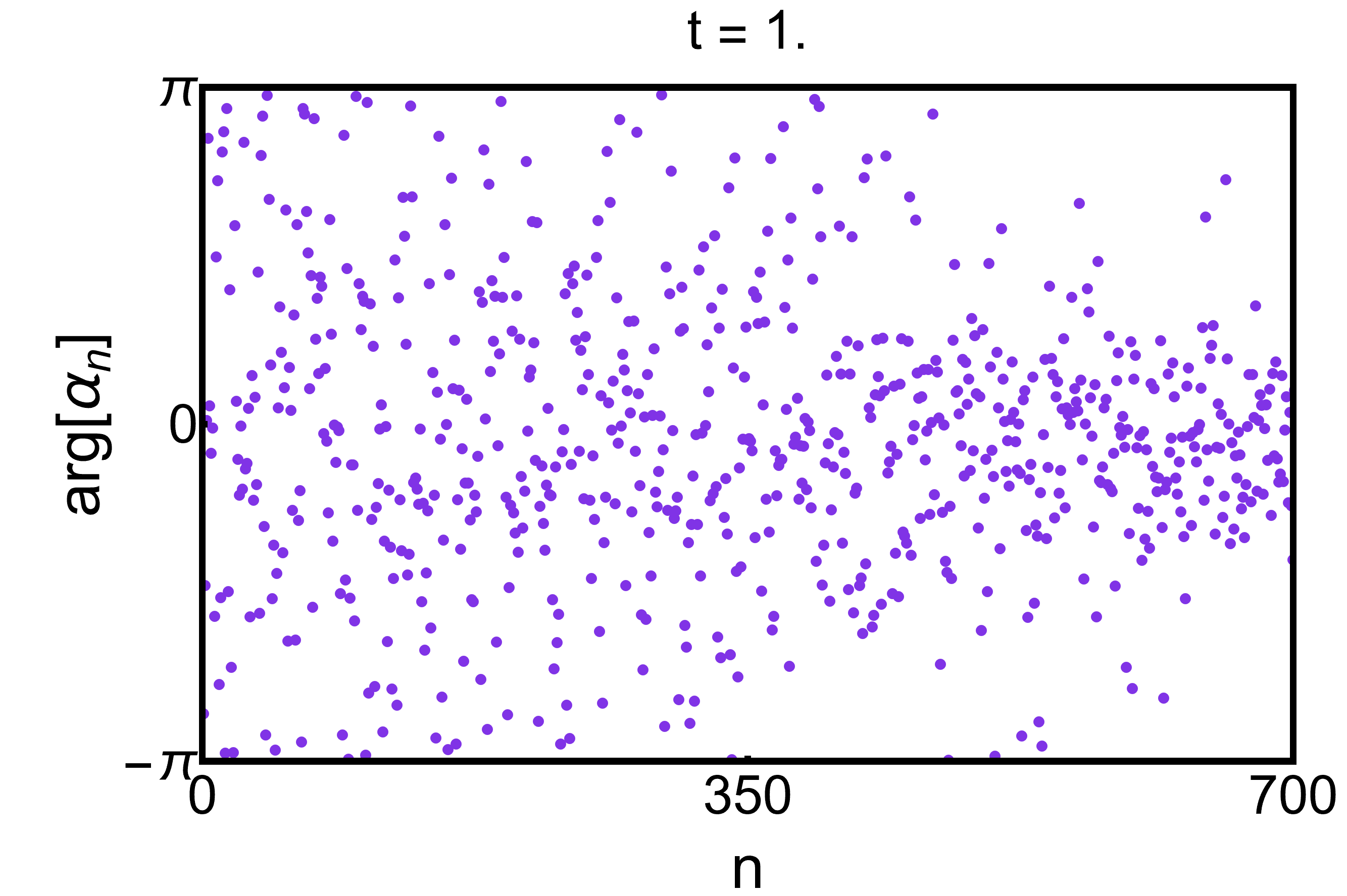}	
		\includegraphics[width = 5.6cm]{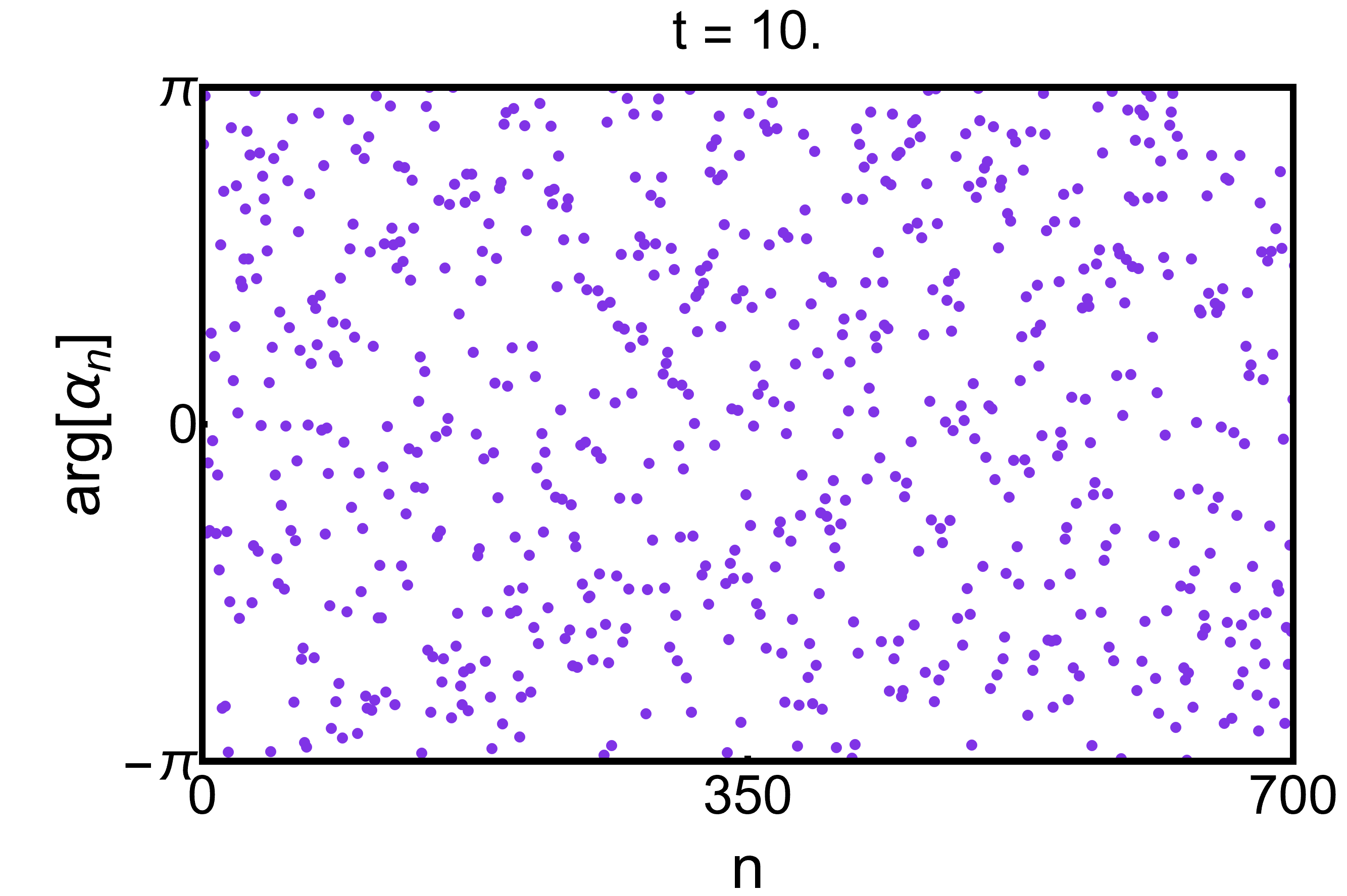}		
		\caption{Evolution of a system admitting Cascade I solutions on the invariant manifold, initialized with random initial data. {\em Upper-left corner}: energy spectrum at successive times. {\em Remaining plots}: phase distribution at the corresponding times.}
		\label{fig:examples_coherent_energy_cascade_I}
	\end{figure}


	\newpage

	\section{Construction of solvable Hamiltonian systems}
	\label{subsec:existence_invariant_manifold}
	
	We will first construct a family of Hamiltonian systems with a low-dimensional invariant manifold that is larger than the one alluded to in Section~\ref{sec:Results}. Starting with this general construction, we will specialize to particular cases of interest in the subsequent sections. The broader class of Hamiltonians corresponds to the equations of motion
	\beq
	i \h_n \frac{d{\q}_n}{dt} = \underset{n+m=k+j}{\underbrace{\sum_{m=0}^{\infty}\sum_{k=0}^{\infty}\sum_{j=0}^{\infty}}} S_{nmkj} \bar{\q}_m\q_k\q_j
	\label{eq:Resonant_System_S_form_section_construction}
	\eeq
	with $\q_n \in\mathbb{C}$, $\h_n\in\mathbb{R}$ an arbitrary sequence, and the couplings $S_{nmkj}\in\mathbb{R}$ maintain the symmetries in (\ref{eq:C_symmetries}): $S_{nmkj} = S_{nmjk} = S_{mnkj} = S_{kjnm} \in\mathbb{R}$. In this case, the conservation laws and the Hamiltonian take the form
	\beq
	N = \sum_{n=0}^{\infty} \h_n |\q_n|^2, \qquad E = \sum_{n=0}^{\infty} n \h_n |\q_n|^2, \qquad \text{and} \qquad 
	\mathcal{H} = \frac{1}{2}\sum_{\underset{n+m=k+j}{n,m,k,j}} S_{nmkj} \bar{\q}_n \bar{\q}_m \q_k \q_j. 
	\label{eq:Hamiltonian}
	\eeq
	The manifold previously introduced in (\ref{eq:invariant_manifold_RESULTS_SECTION}),
	now takes the form:
	\beq
	\q_0 = b,\qquad \q_{n\geq 1} = c\,p^{n-1},
	\label{eq:invariant_manifold_V2_section_construction}  
	\eeq
	with dynamical variables $b,c,p\in\mathbb{C}$ and the constraint $|p|^2 < x_c$, coming from the upper bound on the asymptotic behavior of the sequence $|F_{n\gg 1}| < C n^{\gamma} x_c^{-n}$, with $C>0$ and $x_c>0$ that we prove later in Corollary~\ref{corollary:asymptotic_behavior_Fn_any_parameters}. We shall find that the sequence $F_n$ can have both positive and negative entries for this broader class of systems. Of those, the ones associated with positive sequences, $\h_{n}>0$ for every $n$, correspond to the systems considered in (\ref{eq:Resonant_Equation}) through the transformation:
	\beq
	\alpha_n = \sqrt{\h_n}\ \q_n, \qquad \text{and} \qquad C_{nmkj} =  \frac{S_{nmkj}}{\sqrt{\h_n\h_m\h_k\h_j}}.
	\eeq

	A family of systems for which the above manifold remains invariant is summarized in the following proposition, for which it is useful to introduce the notation
	\beq
	x(t) := |p(t)|^2,
	\eeq
	and classify the coefficients according to the number of zero indices they have:
	\beq
	S_{0000}^{(4)} := S_{0000}, \qquad 
	S_{n0n0}^{(2)} := S_{n0n0}, \qquad 
	S_{nm0(n+m)}^{(1)} := S_{nm0(n+m)}, \qquad 
	S_{nmkj}^{(0)} := S_{nmkj},
	\eeq
	where in these expressions $n,m,k,j\neq 0$. Coefficients with three zero indices are excluded by the resonant condition $n+m=k+j$.

	\begin{proposition} \label{prop:existence_invariant_manifold} Hamiltonian systems of the form (\ref{eq:Resonant_System_S_form_section_construction})
		admit the invariant manifold given in Eq.~(\ref{eq:invariant_manifold_V2_section_construction}) if the coefficients $S_{nmkj}$ and the sequence $\h_n$ satisfy the following conditions: 
		
		\noindent {\bf (i)} Coefficients with one and two zero-indices, $S_{nm0(n+m)}^{(1)}$ and $S_{n0n0}^{(2)}$, have the form
		\beq
		S_{nmkj}
		=
		\left[
		\atwo (nm + kj) + \aone (n+m) + \bone
		\right]
		\frac{\h_n \h_m \h_k \h_j}{\h_{n+m}}.
		\eeq
		
		\noindent {\bf (ii)} Coefficients with nonzero indices satisfy the linear constraint
		\beq
		\sum_{k=1}^{n+m-1} S_{nmk(n+m-k)}^{(0)}
		=
		\left(
		\beta_0 + \beta_1 (n+m) + \beta_2 nm
		\right)
		\h_n \h_m.
		\label{eq:Condition_3_ONLY}
		\eeq
		
		\noindent {\bf (iii)} The sequence $\h_n$ satisfies the recursive relation:
		\beq
		\h_{n\geq 2}
		=
		\frac{1}{n-1}
		\sum_{k=1}^{n-1}
		\left[
		\atwo (n-k)k + \aone n + \bone
		\right]
		\h_k \h_{n-k}.
		\label{eq:fn_iterative_equation_consturction_method_ONLY}
		\eeq
		Where $\h_0=1$, $\h_1 \neq 0$, and the parameters $\bone,\aone,\atwo,\beta_0,\beta_1,\beta_2\in\mathbb{R}$ are arbitrary.
	\end{proposition}
	\begin{proof}
		This proof proceeds in the following way: (i) we first identify the necessary conditions on the coefficients $S_{nmkj}$ to leave the manifold invariant; (ii) we then solve such constraints by writing the coefficients in terms of the sequence $\h_n$; and (iii) an equation for the sequence $\h_n$ is derived. 
		
		\noindent \underline{\em Closure conditions:} In order to identify the necessary conditions on the couplings $S_{nmkj}$ for the manifold to be invariant, we substitute the manifold expressions for $\q_n$ (\ref{eq:invariant_manifold_V2_section_construction}) into the equations (\ref{eq:Resonant_System_S_form_section_construction}) and impose their closure. The equation for $\dot{b}$ follows from $\dot{\q}_0$:
		\beq
		i\dot{b} = S_{0000}^{(4)} |b|^2b + 2 b \frac{|c|^2}{x}\sum_{m=1}^{\infty} S_{0m0m}^{(2)} x^m + \frac{|c|^2}{x} \frac{c}{p} \sum_{m=1}^{\infty} \sum_{k=1}^{m-1} S_{0mk(m-k)}^{(1)} x^m.
		\label{eq:b_equation_uncomplete}
		\eeq	
		where we have set $\h_0=1$ for simplicity, without altering the general conclusions of this work.\footnote{We exclude the case $\h_0=0$ for being considerably different from the rest. It corresponds to systems where the lowest mode is decoupled from the rest, admitting an invariant manifold that only contains stationary solutions after the shift  $n\to n-1$: $a_{n} = c p^{n}$.} The equations for $\dot{c}$ and $\dot{p}$ come from $\dot{\q}_{n\geq 1}$. Canceling the factor $p^{n-1}$ on both sides, and using the symmetry $S_{nmkj}=S_{nmjk}$ yields
		\begin{multline}
			i\left(\dot{c} - c \frac{\dot{p}}{p} + n c \frac{\dot{p}}{p}\right) \h_n = 2 S_{n0n0}^{(2)}|b|^2 c + 2 b \frac{|c|^2}{\bar{p}}  \sum_{m=1}^{\infty} S_{nm0(n+m)}^{(1)} x^m \\ + \bar{b} \frac{c c}{p} \sum_{k=1}^{n-1} S_{n0k(n-k)}^{(1)} + \frac{|c|^2}{x}c \sum_{m=1}^{\infty} x^m \sum_{k=1}^{n+m-1} S_{nmk(n+m-k)}^{(0)}.
			\label{eq:c-p_equation_uncomplete}
		\end{multline}
		The closure of the ansatz requires matching the dependence on $n$ on both sides --- the powers of $n$ and the sequence $F_n$. This, together with considering $b$ and $c$ as independent variables, leads to the following constraints for the coefficients, with $n,m\geq 1$;
		\begin{align}		
			&S_{n0n0}^{(2)} = \left(\ax n + \bx\right) \h_n, \label{eq:Condition_0}\\
			&\sum_{k=1}^{n-1} S_{n0k(n-k)}^{(1)} = \left(\ah n + \bh\right) \h_n, \label{eq:Condition_1}\\
			&\sum_{m=1}^{\infty} S_{nm0(n+m)}^{(1)} x^m = \left(\at(x)n + \bt(x)\right) \h_n, \label{eq:Condition_2}\\
			&\sum_{k=1}^{n+m-1} S_{nmk(n+m-k)}^{(0)} = \left(\beta_0 + \beta_1(n+m) + \beta_2 n m\right) \h_n \h_m, \label{eq:Condition_3}
		\end{align}
		where $\ax$, $\bx$, $\ah$, $\bh$, $\beta_i \in \mathbb{R}$ are parameters and $\at(x)$, $\bt(x)\in\mathbb{R}$ are functions of $x$. The symmetry between $n$ and $m$ in the last expression originates from the symmetry $n\leftrightarrow m$ of $S_{nmkj}$. 
		
		\noindent \underline{\em Solution to the constraints:} The above conditions are necessary for the manifold to be invariant but they must be supplemented with a condition on $\h_n$ to be sufficient. To do so, we now solve the constraints and derive an equation for $\h_n$:
		\begin{itemize}
			\item $\mathbf{S_{0000}^{(4)}}$: This coefficient is unconstrained by conditions (\ref{eq:Condition_0})-(\ref{eq:Condition_3}), becoming a free parameter.
			
			\item $\mathbf{S_{n0n0}^{(2)}}$: Eq.~(\ref{eq:Condition_0}) forces the form $S_{n0n0}^{(2)} = \left(\ax n + \bx\right) \h_n$. The parameters $\ax$ and $\bx$, however, can take any value via the symmetry transformation $a_n(t) \to a_n(t)\exp\left[-i (\delta_1 n + \delta_2)\int_0^t |a_0(\tau)|^2d\tau\right]$ and $S^{(2)}_{n0n0} \to S^{(2)}_{n0n0} - (\delta_1 n+ \delta_0)\h_n $ with real $\delta_{1}$ and $\delta_2$.  We can therefore fix, without loss of generality, $\ax = \aone$ and $\bx=\bone$ to match the upcoming expression for the couplings $S_{nm0(n+m)}^{(1)}$.
			
			\item $\mathbf{S_{nm0(n+m)}^{(1)}}$: We now use the constraints in Eqs.~(\ref{eq:Condition_1})-(\ref{eq:Condition_2}) to write these coefficients in terms of the sequence $\h_n$. We formulate  the problem in terms of generating functions, 
			\beq
			F(x) = \sum_{n=1}^{\infty} \h_n x^n \qquad \text{and} \qquad S(x,y) = \sum_{n=1}^{\infty}\sum_{m=1}^{\infty}x^m y^n S_{nm0(n+m)}^{(1)}.
			\label{eq:Generating_Functions_F_and_S}
			\eeq
			The third condition is reformulated as follows, after multiplying by $y^n$ and summing over $n$:
			\beq
			S(x,y) = \sum_{n=1}^{\infty} y^n \sum_{m=1}^{\infty} S_{nm0(n+m)}^{(1)} x^m = \at(x) y F'(y) + \bt(x) F(y).
			\label{eq:Generating_Function_Condition_3}
			\eeq
			The symmetry $n\leftrightarrow m$ in the coefficients leads to $S(y,x)=S(x,y)$, constraining $\at(x)$ and $\bt(x)$ to be
			\beq
			\at(x) = \aone F(x) + \atwo x F'(x),\qquad \bt(x) = \bone F(x) + \aone xF'(x), 
			\label{eq:mu_nu_functions_expressions}
			\eeq
			where $\aone$, $\atwo$, $\bone\in \mathbb{R}$. Combining (\ref{eq:Generating_Function_Condition_3})-(\ref{eq:mu_nu_functions_expressions}) in the form
			\begin{multline}
				S(x,x) = \sum_{n=1}^{\infty} \sum_{m=1}^{\infty}  x^{n+m} S_{nm0(n+m)}^{(1)} = \atwo \left(x F'(x)\right)^2+2 \aone x F(x) F'(x) + \bone F(x)^2 \\=  \sum_{n=1}^{\infty} \sum_{m=1}^{\infty} x^{m+n} \h_n \h_m \left[\aone (n+m) + \atwo n m+\bone\right],
				\label{eq:S(x,x)_intermediate_equation}
			\end{multline}
			one can write the expression for the coefficients by matching powers on both sides:
			\beq
			S_{nm0(n+m)}^{(1)} =  \left[\aone (n+m) + \atwo n m+\bone\right] \h_n \h_m.
			\label{eq:Smnij_extression_V2}
			\eeq
			Including the symmetries $n\leftrightarrow m$, $k\leftrightarrow j$ and $(n,m)\leftrightarrow(k,j)$, and $\h_0=1$, the coefficients with a single zero among the indices can be written in the symmetric form:
			\beq
			S_{nmkj}^{(1)} =   \left[\aone (n+m) + \atwo (n m + k j)+\bone\right] \frac{\h_n \h_m\h_k \h_j}{\h_{n+m}}.
			\label{eq:Smnij_extression}
			\eeq
			This expression can be used for $S_{n0n0}^{(2)}$ as well, since it matches their required expression. We are making a slight abuse of notation because some entries of $\h_{n+m}$ might vanish, however, when these coefficients are evaluated, the denominator is always compensated with the same factor in the numerator, avoiding conflicts.
			
			\item $\mathbf{S_{nmkj}^{(0)}}$: The constraint in Eq.~(\ref{eq:Condition_3}) is directly solved by writing $S_{nmnm}^{(0)}$ in terms of the other coefficients, 
			presenting infinitely many solutions for any choice of $\{\beta_0,\beta_1,\beta_2\}$. Rewriting the constraint with all terms on the right-hand side of the expression except $S_{nmnm}^{(0)}$ yields:
			\begin{itemize}
				\item For $n+m$ even:
				\begin{multline}
					S_{nmnm}^{(0)} = \frac{1}{2-\delta_{m,n}}\Bigg{[}  \left(\beta_0 +\beta_1(n+m) + \beta_2 n m\right) \h_n \h_m \\ - (1-\delta_{m,n})S^{(0)}_{nm\frac{n+m}{2}\frac{n+m}{2}} - \sum_{k=1}^{m-1}S^{(0)}_{nmk(n+m-k)} - \sum_{k=1}^{n-1}S^{(0)}_{nmk(n+m-k)}\\
					- (1-\delta_{m,n-2}) \sum_{k = m+1}^{\frac{n+m-2}{2}}S^{(0)}_{nmk(n+m-k)} - (1-\delta_{n,m-2}) \sum_{k = n+1}^{\frac{n+m-2}{2}}S^{(0)}_{nmk(n+m-k)} \Bigg{]}. 
					\label{eq:Snmnm_even_V0}
				\end{multline}
				\item For $n+m$ odd:
				\begin{multline}
					S_{nmnm}^{(0)} = \frac{1}{2}\Bigg{[} \left(\beta_0 +\beta_1(n+m) + \beta_2 n m\right) \h_n \h_m - \sum_{k=1}^{m-1}S^{(0)}_{nmk(n+m-k)} - \sum_{k=1}^{n-1}S^{(0)}_{nmk(n+m-k)} \\
					- (1-\delta_{m,n-1}) \sum_{k = m+1}^{\frac{n+m-1}{2}}S^{(0)}_{nmk(n+m-k)} - (1-\delta_{n,m-1}) \sum_{k = n+1}^{\frac{n+m-1}{2}}S^{(0)}_{nmk(n+m-k)}\Bigg{]}.
					\label{eq:Snmnm_odd_V0}
				\end{multline}
			\end{itemize}
		\end{itemize}

		\noindent \underline{\em A recursive relation for $\h_n$:} 	We now show that equations (\ref{eq:Condition_1})-(\ref{eq:Condition_2}) constrain the sequence $\h_n$. We use the symmetry $(n,m)\leftrightarrow(k,j)$ in the coefficients $S_{nmkj}$ to note that the generating function introduced in (\ref{eq:Generating_Functions_F_and_S}) satisfies
		\beq
		S(x,x) = \sum_{n=1}^{\infty}\sum_{m=1}^{\infty}x^{n+m} S_{nm0(n+m)}^{(1)} = \sum_{n=1}^{\infty}\sum_{m=1}^{\infty}x^{n+m} S_{(n+m)0nm}^{(1)} = \sum_{M=1}^{\infty}x^{M} \sum_{k=1}^{M-1}S_{M0k(M-k)}^{(1)}.
		\eeq
		The last term is precisely the condition in (\ref{eq:Condition_2}) after multiplying by $x^n$ and summing in $n$:
		\beq
		S(x,x) = \sum_{n=1}^{\infty}  x^n \sum_{k=1}^{n-1} S_{n0k(n-k)}^{(1)} = \ah  xF'(x) + \bh F(x).
		\label{eq:Generating_Function_Condition_2}
		\eeq
		A differential equation for $F(x)$ is obtained after combining this expression with (\ref{eq:S(x,x)_intermediate_equation}):
		\beq
		\atwo (x F')^2 + \left(2\aone F - \ah\right)x F' +  (\bone F -\bh) F = 0, \qquad \text{with} \qquad F(0)=0.
		\label{eq:F_differential_equation}
		\eeq	
		The initial condition follows from the power-series form of $F(x)$ in (\ref{eq:Generating_Functions_F_and_S}). The sequence $\h_{n\geq1}$ is then constructed from the solution to this equation, $\h_{n\geq 1} = F^{(n)}(0)/n!$. Alternatively, one can construct the sequence recursively by writing $F(x)$ in terms of $\h_n$:
		\beq
		(\ah n + \bh) \h_n = \begin{cases}
			0 & \text{for } n = 1,\\
			\sum\limits_{k=1}^{n-1} \left(\atwo  (n-k)k  + \aone n + \bone \right) \h_{k}\h_{n-k}& \text{for } n \geq 2.
		\end{cases}
		\label{eq:fn_iterative_equation}
		\eeq 
		  The set of parameters $\{\atwo, \aone,\bone,\ah,\bh,\h_1\}$ can be reduced through a series of transformations that allow us to fix $\ah=1$, $\bh=-1$, $\h_1=1$, and one of the other three parameters $\{\atwo, \aone,\bone\}$ to one without reducing the space of nontrivial solutions. The resulting expression is 
		\beq
		\h_{n\geq 2} = \frac{1}{n-1}\sum_{k=1}^{n-1} \left(\atwo  (n-k)k  + \aone n + \bone \right) \h_{k}\h_{n-k}, \qquad \text{with} \qquad \h_1 = 1,
        \label{eq:fn_iterative_equation_2}
		\eeq 
        and its solutions can be mapped to the nontrivial solutions of Eq.~(\ref{eq:fn_iterative_equation}).
          
        First, we note that nontrivial sequences in Eq.~(\ref{eq:fn_iterative_equation}) require $\bh=-\ah s$ with $s=1,2,3,...$, and the linear dependence on the parameters allows us to set $\ah=1$. The first nontrivial element in the sequence is therefore $\h_s$, which can be fixed to one through the transformation $\h_n\to \delta^{n} \h_n$. Scaling the parameters $\{\aone,\atwo,\bone\}\to \eta\{\aone,\atwo,\bone\}$ leads to another transformation $F_n\to \eta^{n-1}F_n$ that allows us to fix one of these parameters to one (different choices of the fixed parameter lead to different sets of solutions, since the zero value would be excluded). Finally, solutions for the generating function $F(x)$ with $s>1$ and parameters $\{\aone,\atwo,\bone\}$ can be written in terms of  the solution  with $\aone\to s\aone$, $\atwo\to s^2\atwo$, and $s\to1$ (whose coefficients are denoted by $\h_n^{(1)}$) in the form $F(x)  = \sum_{n=1}^{\infty} \h_n^{(1)} x^{sn}$. We then conclude that solutions for $\ah=1$, $\bh=-1$, $\h_1=1$ and $\{\atwo, \aone,\bone\}$ with one of these parameters set to one, namely solutions to Eq.~(\ref{eq:fn_iterative_equation_2}), can be mapped to any solution of (\ref{eq:fn_iterative_equation}) with other values of the parameters.

		To sum up, in this proof we have identified the constraints the coefficients $S_{nmkj}$ must satisfy to leave the manifold under discussion invariant. Then, we have solved these constraints by writing the coefficients in terms of the sequence $\h_n$, some additional parameters, and coefficients that remain unconstrained. Finally, a recursive equation for the sequence $\h_n$ has been obtained, together with a differential equation for its generating function.  
		
	\end{proof}

	
	\section{The sequence $\h_n$}
	\label{subsec:Fn_sequence}
	
	We proceed to classify the properties of the sequence $\h_n$  in terms of the parameters $\{\atwo,\aone,\bone\}$, defined by:
	\beq
	\mathcal{D} := \{\atwo,\aone,\bone \in \mathbb{R}: \quad  \atwo \geq 0, \quad \aone\geq  -\atwo, \quad \text{and} \quad \bone>-2\aone-\atwo\},
	\label{eq:domain_of_parameters}
	\eeq
	together with $\h_1>0$. It guarantees that the coefficients in the construction of $\h_n$ in (\ref{eq:fn_iterative_equation}) are positive and, consequently, that $\h_n$ is positive for all $n$. This restriction allows us to focus on models of the form (\ref{eq:Resonant_Equation}), which have the typical structure for the resonant approximation of weakly nonlinear Hamiltonian PDEs \cite{E2,Kuksin3}.
	
	Before presenting the results, we recall some properties that follow from the proof of Proposition~\ref{prop:existence_invariant_manifold} and will be used from now on. One is the scaling symmetry: $\{\atwo,\aone,\bone\}\to\delta\{\atwo,\aone,\bone\}$ and $F_n\to \delta^{n-1}F_n$, in the construction of the Hamiltonian systems. It allows us to set to one any of the parameters $\{\atwo,\aone,\bone\}$ with a nonvanishing value, introducing a natural classification of systems that we use restricted to the domain $\{\atwo,\aone,\bone\}\in\mathcal{D}$:
	\begin{itemize}
		\item \underline{\em Type I systems:} $\ \ \qquad \atwo=\aone=0$ and $\bone=1$.
		\item \underline{\em Type II systems:} $\ \qquad \{\atwo,\aone,\bone\}\in\mathcal{D} \quad$ with $ \quad  \atwo=0$ and $\aone=1.$
		\item \underline{\em Type III systems:} $\qquad \{\atwo,\aone,\bone\}\in\mathcal{D} \quad$ with $ \quad \atwo\neq 0$.
	\end{itemize}
    
	Another important element in the subsequent derivations is the generating function
	\beq
	F(x) = \sum_{n=1}^{\infty} \h_n x^n,
	\label{eq:generating_function_Fn_V2}
	\eeq
	and its differential equation
	\beq
	\atwo (x F'(x))^2 + \left(2\aone F(x) - 1\right)x F'(x) +  (\bone F(x) + 1) F(x) = 0 \quad \text{with} \quad F(0)=0.
	\label{eq:F_differential_equation_V2}
	\eeq

	\begin{proposition} \label{prop:Fn_sequence_properties} {\bf (Asymptotic behavior of $\h_n$):} Consider the sequence $\h_n$ that solves Eq.~(\ref{eq:fn_iterative_equation_consturction_method_ONLY}) with parameters $\{\atwo,\aone,\bone\}\in\mathcal{D}$. Then, its asymptotic behavior admits the following classification:
		\beq
		\h_{n\gg 1} \sim n^{\gamma} x_c^{-n} \qquad \text{with} \qquad  \gamma = \begin{cases}
			\ \ \ 0, & \text{ for \ \ Type I:\ \ \ \  $\atwo=\aone=0$ \ and \ $\bone=1$},\\
			-3/2, & \text{ for \ \ Type II:\ \ \ $\atwo=0$\ and \ $\aone = 1$},\\
			-5/2, & \text{ for \ \ Type III: \ $\atwo\neq 0$}.
		\end{cases} 
		\eeq
		$x_c$ is given by the singularity of $F(x)$ closest to the origin, which can be calculated from the implicit integration of Eq.~(\ref{eq:F_differential_equation_V2}) as $x_c = x(F_c)$ with 
		\beq
		F_c =
		\begin{cases}
			\qquad \qquad \qquad \infty, & \text{ for \ \ Type I:\ \ \ \  $\atwo=\aone=0$ \ and \ $\bone=1$},\\[0.4em]
			
			\qquad \qquad \qquad \dfrac{1}{2\aone}, & \text{ for \ \ Type II:\ \ \ $\atwo=0$\ and \ $\aone = 1$},\\[0.8em]
			
			\dfrac{1}{2\left(\mu_1+\mu_2 \pm \sqrt{\mu_2(\mu_0+\mu_2+2\mu_1)}\right)}, & \text{ for \ \ Type III: \ $\atwo\neq 0$}.
		\end{cases}
		\eeq
	\end{proposition}
	
	
	\begin{proof} By Darboux theorem \cite{Darboux}, the large-order asymptotics of the  Taylor series (\ref{eq:generating_function_Fn_V2}) is controlled by the singularity of $F(x)$ closest to the origin.
		The proof is separated according to the different types of systems:
		
		\begin{itemize}
			\item {\bf Type I ($\atwo =\aone=0$ and $\bone = 1$):} The generating function of $\h_n$ has a simple solution to its differential equation in (\ref{eq:F_differential_equation_V2}) and all the properties follow directly:
			\beq
			F(x) = \frac{x}{1-x} \qquad \Rightarrow \qquad \h_{n} = 1 \qquad \Rightarrow \qquad \gamma = 0, \  \ x_c = 1, \ \text{and} \ F_c = \infty.
			\label{eq:Fx_1}
			\eeq
			
			\item {\bf Type II ($\atwo = 0$ and $\aone \neq 0$):} Two cases must be considered separately.
			\begin{itemize}
				\item For $\bone = 0$ and $\aone=1$, the differential equation (\ref{eq:F_differential_equation_V2}) can be integrated for $x(F)$, indicating that $F(x)$ can be written in terms of the Lambert $W$-function \cite{W-Lambert_function, W-Lambert_function_Wolfram}:
				\beq
				x(F) = F e^{-2 F} \qquad \Rightarrow \qquad F(x) = -\frac{1}{2} W(-2 x).
				\label{eq:Fx_2}
				\eeq
				The sequence $\h_n$ follows from the series representation of that function,
				\beq
				\h_{n\geq 1} = \frac{(2 n)^{n - 1}}{n!} \underset{n\gg1}{\sim} n^{-3/2}(2e)^n \qquad \Rightarrow \qquad  \gamma = -3/2 \ \ x_c=(2 e)^{-1}, \ \text{and} \ F_c = \frac{1}{2}.
				\eeq 
				
				\item For $\bone=1$,  the differential equation (\ref{eq:F_differential_equation_V2}) takes the form:
				\beq
				xF' = \frac{(1+F) F}{1-2\aone F} \quad \Rightarrow \quad x(F) = \frac{F}{\left(1+F\right)^{2\aone+1}},
				\label{eq:F_fuss_catalan}
				\eeq
				which is the generating function for a particular class of Fuss-Catalan numbers \cite{BG, FussCatalan}:
				\beq
				\h_n = A_n^{(2\aone+1,1)}\underset{n\gg1}{\sim} n^{-3/2} x_c^{-n} \quad \text{with} \quad \gamma= -3/2, \ x_c = \frac{(2\aone)^{2\aone}}{(2\aone+1)^{2\aone+1}}, \ \text{and} \ F_c = \frac{1}{2\aone}
				\eeq
				with the general expression
				\beq
				A_n^{(s,r)} = \frac{r}{sn+r}\binom{sn+r}{n}.
				\eeq
			\end{itemize}
			
			\item {\bf Type III ( $\atwo \neq 0$):} The asymptotic behavior of $\h_n$ can be extracted from the differential equation for its generating function in (\ref{eq:F_differential_equation_V2}). We first write the equation in the form 
			\beq
			2 \mu_2 x\frac{dF}{dx} = 1 -2 \mu_1 F - \sqrt{\left(1-\frac{F}{F_-}\right) \left(1-\frac{F}{F_+}\right)}  \qquad \text{and} \qquad F(0)=0,
			\label{eq:F_differential_equation_square_root}
			\eeq
			with
			\beq
			F_{\pm} = \frac{1}{2\left(\mu_1+\mu_2 \pm \sqrt{\mu_2 (\mu_0+\mu_2+2 \mu_1)}\right)}.
			\eeq
			The minus sign in front of the square root of (\ref{eq:F_differential_equation_square_root}) is the branch compatible with the structure of our generating function in (\ref{eq:Generating_Functions_F_and_S}), as it keeps $F'(0)$ finite. 
			
			In the range of parameters we are considering, $\{\atwo,\aone,\bone\}\in \mathcal{D}$, it can be shown that $F_- > F_+ > 0$ and $F(x)$ grows in  $F\in(0,\ F_+)$, while its second and higher derivatives become singular at $F_+$. This singular behavior implies the asymptotic form $\h_n x^{n} \underset{n\gg 1}{\sim} n^{\gamma} (x/x_c)^n$ with $x_c=x(F_{+})$, and therefore $F_c=F_+$. Thus, the power $\gamma$ can be inferred from the behavior of $F(x)$ near that point:
			\beq
			F(x)= F_+ + C_1 (1-x/x_c) + C_{3/2}(1-x/x_c)^{3/2}+...,
			\eeq
			where $C_j$ are $x$-independent coefficients. The presence of a the noninteger power $3/2$ in the above expression requires the asymptotic behavior for the series $\h_{n\gg 1}\sim n^{-5/2}x_c^{-n}$, as one extracts from the Polylogarithm function \cite{Polylogarithm_function}
			\beq
			\text{Li}_{-\gamma} = \sum_{n=1}^{\infty} n^\gamma z^n \underset{|z|\sim 1^-}{\sim} \zeta (-\gamma ) -\zeta (-\gamma -1) \left(1 - z\right)  +\Gamma (\gamma +1) \left(1 - z\right)^{-(\gamma+1)} + ...,
			\eeq
			with noninteger $\gamma < -1$, and $\Gamma$ and $\zeta$ being the gamma and zeta functions respectively. The explicit value $x_c=x(F_+)$  can be calculated from an implicit integration of equation (\ref{eq:F_differential_equation_square_root}) to obtain $x(F)$. 
			
			These results apply to any combination of parameters in the domain $\mathcal{D}$ with $\atwo\neq 0$, concluding that Type III systems present:
			\beq
			\gamma = -5/2, \qquad \text{and} \qquad F_c = \frac{1}{2\left(\mu_1+\mu_2 + \sqrt{\mu_2 (\mu_0+\mu_2+2 \mu_1)}\right)}. 
			\eeq
		\end{itemize}
	\end{proof}

	
	\begin{proposition} \label{prop:Fn_explicit_expressions} {\bf (Explicit expressions for $\h_n$):} The sequence $\h_n$ defined by the recursive relation in Eq~(\ref{eq:fn_iterative_equation_consturction_method_ONLY}) admits explicit expressions for Type I and Type II systems, and for a subset of Type III systems:
		\begin{itemize} 
			\item \underline{Type I ($\atwo =\aone= 0$ and $\bone = 1$)}: 
			\beq
			\h_{n\geq1} = 1 \qquad \text{and} \qquad x_c = 1.
			\eeq
			\item \underline{Type II ($\atwo=0$ and $\aone\neq 0$)}:
			\begin{eqnarray}
				\{\aone = 1,\ \bone = 0\} & \Rightarrow& \h_{n\geq1} = \frac{(2 n)^{n - 1}}{n!}, \qquad \text{and} \qquad x_c = (2e)^{-1}\\ [6pt]
				\{\aone > 0,\ \bone = 1\} & \Rightarrow& \h_{n\geq1} = A_{n}^{(2\aone+1,1)} \qquad \text{and} \qquad x_c = \frac{(2\aone)^{2\aone}}{(2\aone+1)^{2\aone+1}}.
			\end{eqnarray}
			where $A_n^{(s,1)}$ is a particular set of Fuss-Catalan numbers:
			\beq
			A_n^{(s,r)} = \frac{r}{sn+r} \binom{sn+r}{n}. 
			\label{eq:Fuss_Catalan_numbers_general_form}
			\eeq

			\item \underline{Type III ($\atwo \neq 0$)}:
			\begin{align}
				\{\atwo = 1,\ \aone = 0,\ \bone = 0\}
				&  \quad\;\Rightarrow\; \quad
				\h_{n\geq 1}
				= 2\frac{(2n)^{n-2}}{n!}, \label{eq:Fn_expression_for_mu2_neq_0_mu1_0_mu0_0}
				\\[6pt]
				\{\atwo = \aone^2,\ \aone = \frac{s-1}{2},\ \bone = 1\}
				& \quad\;\Rightarrow\; \quad
				\h_{n\geq 1}
				= \frac{2}{(s-1)n+2} A_n^{(s,1)} \label{eq:F_n_expression_-5over2_system_1},\\[6pt]
				\left\{\atwo = \frac{1}{8\aone} - \frac{\aone}{2},\ \aone = \frac{1}{4s-2},\ \bone = -2\aone\right\}
				& \quad\;\Rightarrow\; \quad
				\h_{n\geq 1}
				= \frac{s-1}{sn-1} A_n^{(s,1)} \label{eq:F_n_expression_-5over2_system_2},\\[6pt]
				\left\{\atwo = \frac{1}{16\aone} - \aone,\ \aone = \frac{1}{8s-4},\ \bone = -\aone\right\}
				& \quad\;\Rightarrow\; \quad
				\h_{n\geq 1}
				= \frac{4s-2}{2sn-1} A_n^{(s,1/2)}, \label{eq:F_n_expression_-5over2_system_3}
			\end{align}
			where $A_n^{(s,r)}$ are Fuss-Catalan numbers given in (\ref{eq:Fuss_Catalan_numbers_general_form}), and $x_c = (2e)^{-1}$ for Eq.~(\ref{eq:Fn_expression_for_mu2_neq_0_mu1_0_mu0_0}) and $x_c=(s-1)^{(s-1)}/s^s$ for the other expressions.
			
		\end{itemize}
	\end{proposition}
	
	\begin{proof} The results for the Type I and II cases follow directly from the proof of Proposition~\ref{prop:Fn_sequence_properties}. We here focus on Type III systems:
		\begin{itemize}
			\item When $\bone = \aone = 0$ but $\atwo = 1$, an implicit integration of Eq.~(\ref{eq:F_differential_equation_square_root}) shows that the generating function can be written in terms of Lambert W-function \cite{W-Lambert_function, W-Lambert_function_Wolfram} :
			\beq
			x(F) = \frac{1-\sqrt{1-\frac{F}{F_+}}}{2 } \exp\left({\textstyle -1 +  \sqrt{1-\frac{F}{F_+}}}\right) \quad \Rightarrow \quad F(x) = -\frac{W(-2 x)}{4} \left[2 + W(-2 x)\right],
			\eeq
			leading to the sequence
			\beq
			\h_{n\geq 1} = 2\frac{(2  n)^{n - 2}}{n!} \underset{n\gg 1}{\sim} n^{-5/2} (2e)^n,
			\eeq
			which matches our general analysis: $\gamma = -5/2$ and $x_c=x(F_+)=(2e)^{-1}$.
			
			\item The remaining expressions for $\h_n$, 
			(\ref{eq:F_n_expression_-5over2_system_1})-(\ref{eq:F_n_expression_-5over2_system_3}), 
			follow from the same procedure. The idea is to express the generating function
			\beq
			F(x) = \sum_{n=1}^{\infty} \h_n x^n
			\label{eq:expression1}
			\eeq
			and its derivative in terms of the generating function for the Fuss-Catalan numbers
			\beq
			G(x) = \sum_{n=1}^{\infty} A_n^{(s,r)} x^n,
			\label{eq:expression2}
			\eeq
			which satisfies \cite{FussCatalan}
			\begin{equation}
				(G(x)+1)^{1/r} - 1 = x (G(x)+1)^{s/r}
				\quad \text{and} \quad
				xG'(x) = r\frac{(G(x)+1)^{\frac{r+1}{r}}-(G(x)+1)}{s+(1-s)(G(x)+1)^{1/r}}.
				\label{eq:Fuss_Catalan_Equation_and_Generalization}
			\end{equation}
			Substituting the resulting expressions in the differential equation for $F(x)$ in (\ref{eq:F_differential_equation_V2}) shows that the equation is satisfied for the corresponding values of $\{\atwo,\aone,\bone\}$, confirming that the sequence $\h_n$ solves Eq.~(\ref{eq:fn_iterative_equation_consturction_method_ONLY}).
			
			We illustrate this procedure for the case of Eq.~(\ref{eq:F_n_expression_-5over2_system_2}), where
			\beq
			\h_n = \frac{s-1}{sn-1} A_n^{(s,1)} \qquad \Rightarrow \qquad (s-1)G(x) = s x F'(x) - F(x).
			\label{eq:G-F_relation}
			\eeq
			We now write $F(x)$ as a function of $G(x)$ by defining $\tilde{F}(G) := F(x(G))$, where the relation $x(G)$ comes from (\ref{eq:Fuss_Catalan_Equation_and_Generalization}). With that, we derive the differential equation
			\begin{equation}
				\frac{s(G+1) G}{(s-1)\bigl(1-(s-1)G\bigr)} \frac{d\tilde{F}}{dG}
				= \frac{\tilde{F}}{s-1} + G,
			\end{equation}
			which admits the solution
			\begin{equation}
				\tilde{F}(G) = \frac{G}{G+1}
				- \frac{(s-1)^2}{2s-1}\frac{G^2}{G+1},
			\end{equation}
			where the integration constant has been fixed to match the structure of the generating functions (\ref{eq:expression1})-(\ref{eq:expression2}) near $x=0$. We can then write
			\begin{equation}
				x \frac{dF(x)}{dx} = \frac{G(x)}{1 + G(x)}
				+ \frac{s-1}{2s-1}\,\frac{G(x)^2}{1 + G(x)}.
			\end{equation}
			A direct substitution in Eq.~(\ref{eq:F_differential_equation_V2}) verifies that the differential equation for $F(x)$ is satisfied for the parameter values in (\ref{eq:F_n_expression_-5over2_system_2}). It then confirms that the sequence $\h_n$ solves equation (\ref{eq:fn_iterative_equation_consturction_method_ONLY}), completing the proof.				
		\end{itemize}
		
	\end{proof}


	Before concluding this section, we provide a useful result to work with the inverse function $x(F)$ in the next section, and the upper bound on the asymptotic behavior of the sequence $\h_n$ used to introduce the invariant manifold in Eq.~(\ref{eq:invariant_manifold_V2_section_construction}).
	
	\begin{proposition} \label{prop:F_grows}
		Let $F(x)$ be the generating function defined in (\ref{eq:generating_function_Fn_V2}) that solves equation (\ref{eq:F_differential_equation_V2}) for the parameters $\{\atwo,\aone,\bone\}\in \mathcal{D}$. Then, $F(x)$ and $x F'(x)$ are smooth positive growing functions in $x\in[0,x_c)$, where $x_c$ is the location of the nearest singularity.
	\end{proposition}
	\begin{proof}
		For parameters $\{\atwo,\aone,\bone\}\in \mathcal{D}$, the coefficients of the generating function are positive, and from Proposition~\ref{prop:Fn_sequence_properties} its asymptotic behavior is $\h_{n\gg 1} \sim n^{\gamma}x_c^{-n}$. Then, $F(x)$ is a positive smooth and growing function in $x\in[0,x_c)$.
	\end{proof}

	\begin{corollary}\label{corollary:asymptotic_behavior_Fn_any_parameters} Let $\h_n$ be the sequence that solves the recursive relation (\ref{eq:fn_iterative_equation_consturction_method_ONLY}) in Proposition~\ref{prop:existence_invariant_manifold}. Such a sequence grows at most exponentially:
		$|\h_n|\leq C n^{\gamma} \delta^{n}$,
		where $C$, $\gamma$, and $\delta$ are constants that depend at most on the parameters $\{\atwo,\aone,\bone\}$ and $\h_1$.
	\end{corollary}
	\begin{proof}
		This results follows from the recursive relation for $\h_n$ in Eq.~(\ref{eq:fn_iterative_equation_consturction_method_ONLY}), which leads to
		\beq
		|\h_{n\geq 2}| \leq \frac{1}{n-1}\sum_{k=1}^{n-1} \left(|\atwo|  (n-k)k  + |\aone| n + |\bone| \right) |\h_{k}||\h_{n-k}|.
		\eeq
		Since the right-hand side has positive parameters, $|\h_{n\geq2}|$ is bounded by a sequence that falls within the scope of Proposition~\ref{prop:Fn_sequence_properties}, with parameters in the domain $\mathcal{D}$ defined in (\ref{eq:domain_of_parameters}). It has been proved in that proposition that those sequences grow at most exponentially, completing the proof. 
	\end{proof}

	
	\section{Dynamics on the invariant manifold}
	\label{subsec:dynamics_within_the_manifold}
	
	The dynamics on the invariant manifold (\ref{eq:invariant_manifold_V2_section_construction}) can be classified for systems with parameters $\{\atwo,\aone,\bone\}\in\mathcal{D}$ defined in (\ref{eq:domain_of_parameters}). We recall that the family of models constructed in Proposition~\ref{prop:existence_invariant_manifold} is very large, and we now focus on models structurally akin to those arising as weakly nonlinear approximations for nonlinear dispersive Hamiltonian PDEs.
	
	\begin{proposition} \label{prop:Dynamics_on_the_manifold} Consider the models constructed in Proposition~\ref{prop:existence_invariant_manifold} with $\{\atwo,\aone,\bone\}\in\mathcal{D}$ defined in (\ref{eq:domain_of_parameters}). Then, their dynamics on the invariant manifold (\ref{eq:invariant_manifold_V2_section_construction}) are entirely determined by the parameters 
		$\{\atwo,\aone,\bone,\beta_2,\beta_1,\beta_0,S_{0000}\}$. Such dynamics present locked phases:
		\beq
		\arg \q_{n\geq 1}(t) = (n-1)\arg p(t) + \arg c(t),
		\eeq
		and three possible behaviors for the energy spectrum: 
		\begin{itemize}
			\item \underline{Stationary solutions:} corresponding to a static energy spectrum, 
			\beq
			E_n(t) = \mathrm{const}_n.
			\eeq
			
			\item \underline{Time-periodic solutions:} corresponding to an oscillatory energy spectrum, 
			\beq
			E_n(t) = E_n(t+ kT) \quad \text{with} \quad k=1,2,3,...
			\eeq
			
			\item \underline{Energy cascade solutions:} corresponding to a power-law formation in the asymptotic tail of the energy spectrum,
			\beq
			E_{n\gg1}(t)\sim |c(t)|^2 \, n^{\gamma} \left(\frac{x(t)}{x_c}\right)^{n-1} \xrightarrow[x\to x_c]{} |c|^2 n^{\gamma}.
			\eeq
		\end{itemize}

		\medskip
		
		\noindent Energy cascades further split into three types entirely determined by the parameters $\{\bone,\aone,\atwo\}$, while their existence depends additionally on the parameters
		$\{\beta_0,\beta_1,S_{0000}\}$:
		\begin{itemize}
			\item \textbf{Cascade I ($\atwo = \aone = 0$, $\bone \neq 0$):} Infinite-time formation of a positive power law, inducing unbounded growth of Sobolev norms:
			\beq
			\frac{E_{n\geq 1}(t)}{|c(t)|^2} \xrightarrow[t\to\infty]{} n, \qquad \text{and} \qquad H^{s>1/2} \xrightarrow[t\to\infty]{} \infty. 
			\eeq
			
			\item \textbf{Cascade II ($\atwo = 0$, $\aone \neq 0$):} 
			Finite-time formation of a power law $n^{-1/2}$, inducing blow-up of Sobolev norms:
			\beq
			\frac{E_{n\gg1}(t)}{|c(t)|^2} \xrightarrow[t\to T]{} n^{-1/2}, \qquad \text{and} \qquad 	H^{s>1/2} \xrightarrow[t\to T]{} \infty.
			\eeq
			
			\item \textbf{Cascade III ($\atwo \neq 0$):} 
			Finite-time formation of a power-law $n^{-3/2}$, inducing blow-up of Sobolev norms:
			\beq
			\frac{E_{n\gg1}(t)}{|c(t)|^2}
			\xrightarrow[t\to T]{} n^{-3/2}, \qquad \text{and} \qquad H^{s\geq 3/4} \xrightarrow[t\to T]{} \infty.
			\eeq
			
		\end{itemize}
		
	\end{proposition}
	
	\begin{proof}
		The first part of the proposition is proven by deriving the equations of motion on the invariant manifold:
		\begin{align}
			& i\dot{p} = p \left[2 \aone
			N+\beta_2
			E+ (\aone F + \atwo x F') \frac{2 b \bar{c} p}{x}+(\beta_1-2 \aone) E \frac{F}{xF'}\right] + \bar{b} c,  \label{eq:pdot}\\
			& i\dot{b} = b \left[C_{0000} \left(N-\frac{F}{x F'} E\right) +\left(\aone + \bone \frac{F}{xF'}\right) 2 E \right]+\left(1-\frac{F}{xF'}\right) \frac{\bar{p}c}{x} E,  \label{eq:bdot}\\		
			& i\dot{c} = c
			\left[(\beta_1+\beta_2)E+ (\beta_0+\beta_1-2
			(\aone+\bone))E\frac{F}{x F'}+2 N (\aone+\bone)\right] \nonumber \\
			& \qquad +\left(\aone+\atwo + (\aone+\bone)\frac{F}{x F'}\right)  2 b p E. \label{eq:cdot}
		\end{align}
		They follow directly from combining Eqs.~(\ref{eq:b_equation_uncomplete})-(\ref{eq:c-p_equation_uncomplete}) with the conditions for the coefficients in Eqs.~(\ref{eq:Condition_0})-(\ref{eq:Condition_3}). For these expressions, we used $x:=|p|^2$, the generating function $F(x) = \sum_{n=1}^{\infty} \h_n x^n$ and the conserved quantities in (\ref{eq:conserved_quantities}) written in terms of the manifold variables:
		\beq
		N = |b|^2 + \frac{F(x)}{xF'(x)}E, \qquad \text{and} \qquad E = |c|^2 F'(x).
		\label{eq:conserved_quantities_within_manifold}
		\eeq
		We retained $xF'(x)$ to make the expressions more compact, although this can be entirely written in terms of $F(x)$ through its differential equation in (\ref{eq:F_differential_equation}). 
		
		The external dependence of the above system of equations is entirely on the parameters 	$\{\atwo,\aone,\bone,\beta_2,\beta_1,\beta_0,S_{0000}\}$. Consequently, all Hamiltonian systems constructed in Proposition~\ref{prop:existence_invariant_manifold} with the same values of these parameters  present the same behavior on the invariant manifold. This concludes the first part of the proof.
		
		The proofs of the second and third parts of the proposition ---the classification of the dynamics on the invariant manifold and the further separation of energy cascades into three types--- are analogous for all three kinds of systems. We present the argument for Type I systems and then highlight the key differences that arise in the cases of Type II and Type III systems.

		\noindent {\bf - Type I ($\atwo=\aone=0$ and $\bone = 1$):} The key observation is that the energy spectrum is entirely controlled by $x(t)$ for initial conditions on the invariant manifold:
		\begin{equation}
			E_{n\geq 1}(t) = n\, x(t)^{n-1}\, (1-x(t))^2 E.
		\end{equation}
		This expression follows by substituting $|c(t)|^2$ in terms of $x(t)$ with \eqref{eq:conserved_quantities_within_manifold} and $F(x)=x/(1-x)$ coming from (\ref{eq:Fx_1}). 
		
		The study of the dynamics is therefore reduced to deriving and analyzing an evolution equation for $x(t)$. Such equation $\dot{x}(t) = \bar{p} \dot{p} + p \dot{\bar{p}}$ comes from combining $\dot{p}$ in \eqref{eq:pdot} with the conserved quantities to rewrite some expressions in terms of those quantities and $x(t)$; see \cite{BE,BG} for similar developments. The resulting equation has the form of a one-dimensional particle moving in a quartic potential:
		\begin{equation}
			\left(\frac{dx}{dt} \right)^2 + V(x) = 0,
			\label{eq:xdot}
		\end{equation}
		\begin{equation}
			V(x) = C_0 + C_1 (1-x) + C_2 (1-x)^2 + C_3 (1-x)^3 + C_4 (1-x)^4,
		\end{equation}
		with coefficients
		\begin{align}
			&C_0 = \frac{W^2}{4},\nonumber\\
			&C_1 = \left((S_{0000}-2) N - \beta_1 E\right) W, \nonumber\\
			&C_2 = \frac{1}{4} \left(\left(2 (S_{0000}-2) N-2 \beta_1 E\right)^2-2 E W
			\left(\beta_0+S_{0000}-4\right)\right)-4 N E, \nonumber\\
			&C_3 = 4 E (N+E) + \left(\beta_0+S_{0000}-4\right) \left(\beta_1 E-(S_{0000}-2)
			N\right) E, \nonumber\\
			&C_4 = \frac{1}{4} E^2 \left(\left(\beta_0+S_{0000}-4\right)^2-16\right).
		\end{align}
		The quantity $W$ is derived from the Hamiltonian, $\mathcal{H} = \frac{1}{2}(S_{0000}N^2 + E (W+\beta_2))$, with
		\beq
		W = \left((4 - 2S_{0000}) N + 2\beta_1 E\right)(1-x)  +  \left(\beta_0+S_{0000}-4\right)\, E\, (1-x)^2  + 2(b \bar{c} p+\bar{b} c \bar{p}).
		\eeq
		This quantity is conserved on the invariant manifold, provided that the coupling $S_{0000}$ and the parameters $\{\atwo,\aone,\bone,\beta_2,\beta_1,\beta_0\}$ remain constant. Its conservation is independent of whether the couplings $S_{nmkl}^{(0)}$ (those with no zero indices) evolve in time. This follows directly from the conservation of the Hamiltonian on the invariant manifold, proved in Appendix~\ref{app:Hamiltonian_conservation}. This observation is the key ingredient that extends this proposition to systems with time-dependent couplings.

		Equation (\ref{eq:xdot}) restricts the evolution of $x(t)$ to negative values of the potential $V(x)$ on the admissible domain $x\in[0,1)$, the domain that guarantees finite conserved quantities $N$ and $E$. This leads to three possible behaviors regarding the shape of the potential as illustrated in Figure~\ref{fig:tres_potenciales}:

        \begin{figure}
            \centering
            \includegraphics[width=\textwidth]{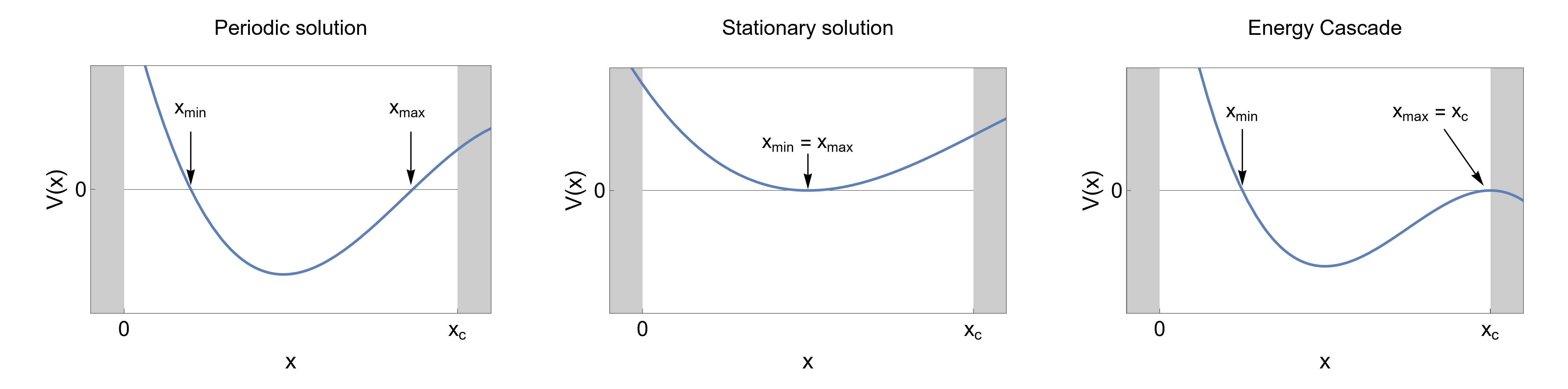}
            \caption{Representative forms of $V(x)$ associated with periodic, stationary, and cascade solutions.}
            \label{fig:tres_potenciales}
        \end{figure}

		\begin{itemize}
			\item {\it Stationary solutions:} These correspond to time-independent energy spectra, attained when $x(t)$ coincides with a vanishing equilibrium point of the potential: $V(x)=V'(x)=0$. The existence of such solutions is trivial by the single-mode initial data $a_n(0) = \delta_{n,0}$, which belongs to the manifold ($b(0)=1$, $c(0)=p(0)=0$) and is an stationary solution for all systems of the form (\ref{eq:Resonant_Equation}). More elaborate solutions can be constructed by solving the constraints $V(x) = V'(x)=0$; see \cite{BE,BG,GG} for examples.
			
			\item {\it Time-periodic solutions:} These arise when $x(0)$ is between two distinct roots of $V(x)$ in $[0,1)$ and is negative between them. It induces $x(t)$ to oscillate between those roots, and consequently the energy spectrum becomes time-periodic. This is a common situation, so their existence can be seen by a simple evaluation of the potential for a couple of initial conditions on the manifold.
			
			\item {\it Energy cascades:} These correspond to solutions that starting with $x(0)\in[0,1)$ it evolves to $x(t)\to 1$, causing the exponential decay of the energy spectrum to disappear in favor of a power-law distribution:
			\beq
			\frac{E_{n\geq 1}(t)}{|c(t)|^2} \ =\ n\, x(t)^{n-1} \ \underset{x\to 1}{\longrightarrow}\ n.
			\eeq
			For such asymptotic power, Sobolev norms $H^{s>1/2}$ grow to infinity, as can be seen from their expression that we restate here:
            \beq
	       H^{s} := \left(\sum_{n=0}^{\infty}(n+1)^{2s}|\alpha_n|^{2}\right)^{1/2}.
	       \eeq
            By contrast, lower norms remain finite, including the conserved quantities $H^{1/2} = \sqrt{N+E}$. The conservation of $E$ forces the amplitude $|c(t)|^2$ to decay to zero at a rate that compensates the divergence of the series; see \cite{BG} for a detailed discussion in deterministic systems. Consequently, norms below $s=1/2$ remain finite as well.
			
			We now analyze the existence of such energy cascades. The potential at the critical point is $V(x=1)=W^2/4\geq 0$. Therefore, cascades can only exist for
			\begin{equation}
				W = 0.
				\label{eq:Energy_cascade_condition_type_I}
			\end{equation}
			Such condition yields $V'(1)=0$ and therefore
			\begin{equation}
				V(x) = (1-x)^2\left[C_2 + C_3\,(1-x) + C_4\,(1-x)^2\right]_{W=0}.
				\label{eq:potential_cascade_Type_I}
			\end{equation}
			Cascades are then associated with an equilibrium point of the potential at $x=1$, which must be negative immediately to the left of that point (a property that we have not proved yet). As a consequence, the power-law spectrum can only be reached in infinite time:
			\begin{equation}
				x(t)\underset{t\to\infty}{\longrightarrow}1
				\qquad\Rightarrow\qquad
				\frac{E_{n\geq1}(t)}{|c(t)|^2}
				\underset{t\to\infty}{\longrightarrow} n.
			\end{equation}
			
			Initial data satisfying condition~(\ref{eq:Energy_cascade_condition_type_I}) can be constructed by expressing $W$ in terms of $b$, $c$, and $p$, and solving for $c$:
			\beq
			c = b p (1-x) \frac{-2 x+\sqrt{2 x} e^{i \lambda } \sqrt{2 \beta_1 (S_{0000}-2)+2 + \left(\left(\beta_0+2\right)
					S_{0000}-2\left(\beta_0+1\right)-S_{0000}^2\right)(1-x)}}{x \left(2 \beta_1+\left(\beta_0-S_{0000}\right)(1-x)\right)}
			\label{eq:c_conditions_Type_I}
			\eeq
			with $\lambda\in[0,2\pi)$ and non-negative discriminant.
			
			Once restricted to these initial conditions, the existence of cascades further requires one to analyze $C_2$, $C_3$, and $C_4$ to guarantee that $x(t)\to 1$ can be reached from the initial value. It introduces a complicated dependence on the parameters $S_{0000}$, $\beta_1$, and $\beta_0$, which we do not analyze in full generality, but it can be easily solved for concrete values. For instance, cascades exist for $S_{0000}=1$ and $\beta_1=\beta_0=0$. The initial data constructed with Eq.~(\ref{eq:c_conditions_Type_I}) and $x_0\in(0,1)$ leads to a potential of the form \eqref{eq:potential_cascade_Type_I}, with $V(0)\geq 0$, $V(x_0)\leq 0$, and $C_4 < 0$. The shape of the potential implies that there is no other root between $x_0$ and $x=1$, so $x(t)\to 1$. A family of models with $S_{0000}=1$ and $\beta_1=-\beta_0=(1-\beta)$ has been studied in \cite{BE}, where further details can be found. This completes the proof for Type I systems. 
		\end{itemize}

		\noindent {\bf - Type II  ($\atwo =0$ and $\aone>0$) and Type III ($\atwo > 0$):} The analysis of these systems follows the same logic as before,  with the difference that $F(x)$ does not admit an explicit expression in general, so we work with $F(t):=F(x(t))$ instead of $x(t)$.
		
		$F(x)$ is a growing function in the interval $x\in[0,x_c)$ (Proposition~\ref{prop:F_grows}), where $x_c$ controls the asymptotic behavior of the sequence $\h_{n\gg 1}\sim n^{\gamma} x_c^{-n}$ (Proposition~\ref{prop:Fn_sequence_properties}). Consequently, the inverse function $x(F)$ is well defined in $F\in[0,x_c)$ where we call $F_c=F(x_c)$. The energy spectrum is then written in the form 
		\beq
		E_n = n\, \h_n\, x(F)^{n}\  \frac{E}{xF'(x)},
		\eeq
		which can be entirely written in terms of $F(t)$ using the equations in (\ref{eq:F_differential_equation_V2}) and (\ref{eq:F_differential_equation_square_root}). An equation for this variable can be derived by combining $\dot{x} = \bar{p}\dot{p} + p \dot{\bar{p}}$ with the equation for $x F'(x)$ from (\ref{eq:F_differential_equation_V2}):
		\beq
		\left(\frac{dF(t)}{dt} \right)^2 + \tilde{V}(F(t)) = 0.
		\label{eq:Fdot}
		\eeq
		The potential takes the general form
		\begin{align}
			\tilde{V}(F) =& \,\frac{(2 \atwo x F' + 2\aone F - 1)^2}{4}\, \Bigg(16 (E^2 F-N E xF') + \frac{(xF')^4}{(F-xF')^2}\Bigg[4 \aone N-W \\
			 &+ \left[(2\beta_1-4 \aone)E+\left(4\bone - 2S_{0000}\right) N \right] \frac{F}{xF'} + (\beta_0+S_{0000}-4 \bone) E \left(\frac{F}{xF'}\right)^2\Bigg]^2\Bigg).\nonumber
		\end{align}
		and the conserved quantity $W$ comes from the  Hamiltonian, $\mathcal{H} = \frac{1}{2}(S_{0000}N^2 + E (W+\beta_2))$,
		\begin{align}
			W = &4 \aone N +\left[(2\beta_1-4 \aone)E+ (4 \bone-2S_{0000}) N\right]\frac{F}{xF'}\\
& \hspace{2cm}+(\beta_0+S_{0000}-4 \bone)E\left(\frac{F}{xF'}\right)^2+\left(1-\frac{F}{xF'}\right)\frac{2 (b\bar{c}p+\bar{b}c\bar{p})}{x}.\nonumber
		\end{align}
		
		The problem is again reduced to the motion of a one-dimensional particle moving in a time-independent potential. This leads to three possible behaviors for the energy spectrum on the manifold: stationary solutions, corresponding to initial conditions coinciding with vanishing equilibrium points of the potential ($\tilde{V}(F_{eq})=\tilde{V}'(F_{eq})=0$); time-periodic solutions, where $F(t)$ evolves between two roots of $\tilde{V}(F)$ in the interval $F\in[0,F_c)$; and energy cascades, corresponding to $F(t)$ approaching the critical value $F_c$, which makes the energy spectrum to develop a power-law tail:
		\beq
		\frac{E_{n\gg 1}(t)}{|c(t)|^2} \sim n^{\gamma+1} \left(\frac{x(F(t))}{x_c}\right)^{n-1} \ \underset{F \to F_c}{\longrightarrow} \ n^{\gamma+1},
		\eeq
		where numerical factors have been omitted. The powers are $\gamma+1=-1/2$ and $-3/2$ for Type II and Type III systems respectively, as established by Proposition~\ref{prop:Fn_sequence_properties}.
		
		Further analysis of energy cascades proceeds separately for each type of system:
		\begin{itemize}
			\item \underline{\em Type II ($\atwo = 0$, $\aone>0$):} The differential equation for the generating function $F(x)$ is
			\beq
			xF'(x) = \frac{F(1+\bone F)}{1-2\aone F},
			\eeq
			which reduces the potential to a quartic polynomial 
			\beq
			\tilde{V}(F) = C_0 + C_1 (F_c-F) + C_2 (F_c-F)^2 + C_3 (F_c-F)^3 + C_4 (F_c - F)^4
			\eeq
			where $F_c = 1/(2\aone)$ and the coefficients $C_i$ have particularly long expressions but follow directly from the above substitution. Energy cascades necessitate $\tilde{V}(F_c) = 0$, which is achieved for
			\beq
			W = 4 \aone N.
			\label{eq:Energy_cascade_condition_Type_II}
			\eeq
			However, for $E,N>0$ the potential satisfies $\tilde{V}'(F_c)>0$, indicating that energy cascades develop in finite time:
			\begin{equation}
				F(t)\underset{t\to T}{\longrightarrow} F_c
				\qquad\Rightarrow\qquad
				\frac{E_{n\geq1}(t)}{|c(t)|^2}
				\underset{t\to T}{\longrightarrow} n^{-1/2}.
			\end{equation}
			
			Proceeding as in the Type I case, condition (\ref{eq:Energy_cascade_condition_Type_II}) is satisfied provided
			\beq
			c = b p \frac{(8\aone + 4 \bone)F + e^{i\lambda} \sqrt{\Delta}}{2 F \left(S_{0000}+4 \mu_1 -\beta_0-2 \beta_1+\left(2 \beta_0 \mu_1-2 \beta_1 \mu_0-2
				S_{0000} \mu_1+4 \mu_1 \mu_0\right) F \right)},
			\label{eq:Type_II_initial_conditions}
			\eeq
			with
			\begin{multline}
				\Delta = \left(\frac{2 \mu_0}{\mu_1}+4\right)^2 + 	\frac{8 \left(\mu_0+2 \mu_1\right) \left(S_{0000} \left(\beta_1-2 \mu_1\right)-2 \left(\beta_1 \mu_0-2 \mu_1 \mu_0+\mu_0+2 \mu_1\right)\right)}{\mu_1} \left(F_c - F\right)\\ + \left(\left(4 \mu_0+8 \mu_1\right)^2-16 \left(S_{0000}-2 \mu_0\right)
				\left(-\beta_0 \mu_1+2 \left(\mu_0+\mu_1\right) \left(\beta_1-2 \mu_1\right)+S_{0000} \mu_1\right)\right) \left(F_c - F\right)^2\\
				+ 32 \mu_1 \left(S_{0000}-2 \mu_0\right) \left(\mu_0 \left(\beta_1-2 \mu_1\right)-\beta_0 \mu_1+S_{0000} \mu_1\right) \left(F_c - F\right)^3
			\end{multline}
			and $\lambda\in[0,2\pi)$ and $\Delta\geq 0$. 
			This condition alone is not sufficient to establish the existence of energy cascades. Additional restrictions on $\{S_{0000},\beta_1,\beta_0\}$ are required. For example, let $S_{0000}=1$, $\bone=1$, and $\beta_0\in[-1,7]$. In this case
			\beq
			\tilde{V}(F_c)= 0, \quad \tilde{V}'(F_c)> 0, \quad \tilde{V}(0)\geq 0, \quad \text{and} \quad \tilde{V}(-1)\leq 0.
			\eeq
			It follows that the potential has a unique root in the interval $F\in[0,F_c)$, while $\tilde{V}(F)<0$ between this root and $F_c$. Consequently, the initial conditions~(\ref{eq:Type_II_initial_conditions}) with the above parameter values generate energy cascades. Such particular example confirms the existence of energy cascades. However, many other parameter choices producing energy cascades can be characterized using the same strategy. A complete description of the parameter constraints leads to considerably more involved expressions, which we do not present here.

			\item \underline{\em Type III ($\atwo > 0$):} This case follows the same logic as for Type II systems with a few differences, so we omit most of the expressions and just explain those differences. We first recall that the generating function satisfies equation~\eqref{eq:F_differential_equation_square_root}:
			\beq
			2 \mu_2 x\frac{dF}{dx} = 1 -2 \mu_1 F - \sqrt{\left(1-\frac{F}{F_-}\right) \left(1-\frac{F}{F_+}\right)}  \qquad \text{and} \qquad F(0)=0,
			\eeq
			with
			\beq
			F_{\pm} = \frac{1}{2\left(\mu_1+\mu_2 \pm \sqrt{\mu_2 (\mu_0+\mu_2+2 \mu_1)}\right)},
			\eeq
			and the value where the generating function loses its smoothness is $F_c = F_+$ for the range of parameters we are considering $\{\atwo,\aone,\bone\}\in\mathcal{D}$.
			
            In this case, the potential does not have a polynomial form in $F$; however, it does satisfy $\tilde{V}(F_c)=0$ for any $W$. The constraint on $W$ comes from requiring that $\tilde{V}'(F_c)>0$, which is quadratic in $W$ --- indicating that Type III systems admit a range of values of $W$ where energy cascades may exist for given values of $N$ and $E$. This is an important difference with the previous cases, where cascades existed only for a specific value of $W$.
            
            Denoting $W_{\pm}$ the roots of $\tilde{V}'(F_c)$, a necessary condition for the existence of cascades is either $W\in[W_-,W_+]$ or $W\in\mathbb{R} - (W_-,W_+)$, depending where $\tilde{V}'(F_c)$ becomes positive. This, in turn, is dictated by the value of the parameters $\{\atwo,\aone,\bone,S_{0000},\beta_1,\beta_0\}$. Systems of this type have been studied in \cite{BG}, reporting energy cascades for the case $S_{0000} = 1$, $\{\atwo = \aone^2,\ \aone = \frac{s-1}{2},\ \bone = 1\}$, and $\beta_2=\beta_1=\beta_0=0$. This particular example confirms the existence of energy cascades. However, many other choices of the parameters lead to the same behavior. Note that for $W$ values where $\tilde{V}(F_c)>0$, the energy spectrum develops a power-law tail in finite time. 
		\end{itemize}
		
		In summary, we have shown the following: the dynamics on the invariant manifold are entirely governed by the parameters 	$\{\atwo,\aone,\bone,\beta_2,\beta_1,\beta_0,S_{0000}\}$; solutions on the manifold can be categorized as either stationary, time-periodic, or energy cascades; energy cascades can be further divided into three types based on the power-law formation time and asymptotic exponent.
		
	\end{proof}
	
	
	\section{Minimally structured Hamiltonian systems}
	\label{subsec:construction_predominantly_random_Hamiltonian_systems}
	
	In the previous sections, we presented a systematic construction and analysis of Hamiltonian systems with an invariant manifold and phase-coherent turbulent energy dynamics. We now show that this construction supports a large-scale randomization of the couplings without disrupting the invariant manifold and its energy cascade solutions.
	
	The key observation is that the existence of the invariant manifold constrains only a small subset of the nonlinear couplings, while the remaining couplings may be chosen arbitrarily, even as stochastic processes, without affecting the dynamics on the invariant manifold.
	
	Proposition~\ref{prop:existence_invariant_manifold} leaves unconstrained the six parameters $\{\atwo,\aone,\bone,\beta_2,\beta_1,\beta_0\}$, the coupling $S_{0000}$, and all couplings $S_{nmkj}^{(0)}$ except for the diagonal entries $S_{nmnm}^{(0)}$. Among these `degrees of freedom,' we choose the parameters $\{\atwo,\aone,\bone,\beta_2,\beta_1,\beta_0, S_{0000}\}$ to be deterministic and time-independent, since they determine the dynamics on the invariant manifold. By contrast, the remaining offdiagonal couplings $S_{nmkj}^{(0)}$ may be specified arbitrarily, provided they are real and satisfy the symmetries $n\leftrightarrow m$, $k\leftrightarrow j$, and $(n,m)\leftrightarrow (k,j)$. In particular, these couplings may be sampled from arbitrary probability distributions or evolve according to arbitrary stochastic processes. No assumptions are required on the underlying distributions or stochastic dynamics, since Propositions~\ref{prop:existence_invariant_manifold}, \ref{prop:Fn_sequence_properties}-\ref{prop:F_grows}, and \ref{prop:Dynamics_on_the_manifold} depend only on algebraic relations among the couplings. Consequently, the dynamics on the invariant manifold are identical for deterministic, random, or stochastic realizations of $S_{nmkj}^{(0)}$. The behavior of the system away from the invariant manifold, however, depends on the particular realization of these couplings.

	This family of Hamiltonian systems is constructed according to the following protocol:
	\begin{enumerate}[label=\textbf{Step \arabic*:},
		leftmargin=1.6cm]
		\item Specify the parameters $\{\atwo,\aone,\bone,\beta_2,\beta_1,\beta_0, S_{0000}\}$.
		\item Construct the sequence $\{\h_n\}$ using either the recursive relation
		(\ref{eq:fn_iterative_equation_consturction_method_ONLY}) or one of its explicit expressions given in Proposition~\ref{prop:Fn_explicit_expressions}.
		\item Determine all couplings with one or two vanishing indices through
		\beq
		S_{nmkj} = \left[
		\atwo (nm + kj) + \aone (n+m) + \bone
		\right]
		\frac{\h_n \h_m \h_k \h_j}{\h_{n+m}} 
		\label{eq:S1_expression_deterministic_couplings}
		\eeq 
		for $nmkj=0$ with $n+m\neq0$.
		\item Choose the couplings $S_{nmkj}^{(0)}$, except the diagonal ones $S_{nmnm}^{(0)}$, from any probability distribution or stochastic process, subject only to the permutation symmetries in the indices given in (\ref{eq:Resonant_System_S_form_section_construction}).
		
		\item Determine the remaining diagonal entries $S_{nmnm}^{(0)}$ from the linear constraints (\ref{eq:Condition_3_ONLY}), thereby inheriting the deterministic or random character of the previous couplings:
		\begin{itemize}
			\item For $n+m$ even:
			\begin{multline}
				S_{nmnm}^{(0)} = \frac{1}{2-\delta_{m,n}}\Bigg{[}  \left(\beta_0 +\beta_1(n+m) + \beta_2 n m\right) \h_n \h_m \\ - (1-\delta_{m,n})S^{(0)}_{nm\frac{n+m}{2}\frac{n+m}{2}} - \sum_{k=1}^{m-1}S^{(0)}_{nmk(n+m-k)} - \sum_{k=1}^{n-1}S^{(0)}_{nmk(n+m-k)}\\
				- (1-\delta_{m,n-2}) \sum_{k = m+1}^{\frac{n+m-2}{2}}S^{(0)}_{nmk(n+m-k)} - (1-\delta_{n,m-2}) \sum_{k = n+1}^{\frac{n+m-2}{2}}S^{(0)}_{nmk(n+m-k)} \Bigg{]}. 
				\label{eq:Snmnm_even}
			\end{multline}
			\item For $n+m$ odd:
			\begin{multline}
				S_{nmnm}^{(0)} = \frac{1}{2}\Bigg{[} \left(\beta_0 +\beta_1(n+m) + \beta_2 n m\right) \h_n \h_m - \sum_{k=1}^{m-1}S^{(0)}_{nmk(n+m-k)} - \sum_{k=1}^{n-1}S^{(0)}_{nmk(n+m-k)} \\
				- (1-\delta_{m,n-1}) \sum_{k = m+1}^{\frac{n+m-1}{2}}S^{(0)}_{nmk(n+m-k)} - (1-\delta_{n,m-1}) \sum_{k = n+1}^{\frac{n+m-1}{2}}S^{(0)}_{nmk(n+m-k)}\Bigg{]}.
				\label{eq:Snmnm_odd}
			\end{multline}
		\end{itemize}
		\item In case one considers systems of the form (\ref{eq:Resonant_Equation}), as we do, an additional step must be taken to retrieve the coefficients:
		\beq
		C_{nmkj} = \frac{S_{nmkj}}{\sqrt{\h_n\h_m\h_k\h_j}}.
		\eeq
	\end{enumerate} 
	
	This construction results in nonlinear random Hamiltonian systems with just an infinitesimal fraction of rigidly prescribed couplings. Such a difference in the size of the random and deterministic sets of couplings is quantified by counting\footnote{The Kronecker deltas arise because $S_{1111}^{(0)}$ and $S_{2121}^{(0)}$ are deterministic, unlike the remaining diagonal entries.} the different deterministic and random entries of $S_{nmk(n+m-k)}$ at each resonant level $M:=n+m\ge1$:
	\begin{equation}
		\mathscr{D}_M = \begin{cases}
			\frac{M}{2}+1 + \delta_{M,2} & \text{even } M,\\
			\frac{M+1}{2} + \delta_{M,3}   & \text{odd } M,
		\end{cases}
		\qquad \qquad
		\mathscr{R}_M = \begin{cases}
			\frac{M(M+2)}{8} - \delta_{M,2} & \text{even } M,\\
			\frac{M^2-1}{8} - \delta_{M,3} & \text{odd } M.
		\end{cases}
		\label{eq:counting_freedoms_and_constraints}
	\end{equation}
	Denoting as $\mathscr{N}_D$ and $\mathscr{N}_R$ the number of deterministic and random couplings across all resonant levels up to $M$ and taking the limit $M\to\infty$, the above linear versus quadratic growth confirms that the set of deterministic couplings is a negligible fraction of the total:
	\beq
	\mathscr{F} = \frac{\mathscr{N}_D}{\mathscr{N}_R + \mathscr{N}_D} \underset{M\to \infty}{\longrightarrow} 0.
	\eeq
	Consequently, the systems constructed in this section provide explicit examples showing that only an infinitesimal fraction of structured nonlinear couplings suffices to guarantee organized energy dynamics, and in particular coherent energy cascades. This conclusion is established on the invariant manifold by Propositions~\ref{prop:existence_invariant_manifold} and \ref{prop:Dynamics_on_the_manifold}, and in the next section, we extend it to the full Hamiltonian system through numerical simulations.

	
	\section{Energy cascades beyond the invariant manifold}
	\label{sec:numerics}
	
	Now that we have established coherent energy cascades on the invariant manifold of our disordered Hamiltonian systems, we will demonstrate that they are not exclusive of those analytically tractable conditions. This is particularly important to study the robustness of coherent cascades, since on the manifold the randomness of the couplings cancels out, while off the manifold the dynamics are under direct influence of the random couplings. We initialize all simulations with random initial data, which are of a completely different nature from the analytically tractable configurations on the manifold and constitute a stringent test for the emergence of coherence.
	
	We consider three classes of systems: {\bf (i)} We first simulate the disordered models constructed in Section~\ref{subsec:construction_predominantly_random_Hamiltonian_systems}, observing the emergence of coherent dynamics, and in particular, coherent energy cascades; {\bf (ii)} Then, we demonstrate that the resonance condition $n+m=k+j$ for the mode couplings in our models is by itself not enough to induce coherent dynamics. To do so, we simulate systems with fully random couplings, without any structured subsets; {\bf (iii)} Finally, we demonstrate that systems with structured subsets of couplings other than the ones presented in the previous sections also give rise to organized energy transfer starting with disordered initial data. 
	
	We perform this exploration starting with a family of random initial data that concentrates the energy in low modes, distributes the phases uniformly on $[0,2\pi)$, and whose amplitudes  are given by independent Gaussian distributions multiplied by a decaying exponential factor for modes beyond an arbitrary cut-off:
	\beq
	\alpha_{n}(0) = \begin{cases}
		g_n & \text{for } n \leq \mathcal{M},\\
		g_n\ e^{-\rho (n-\mathcal{M})} & \text{for } n > \mathcal{M},
	\end{cases} 
	\label{eq:random_initial_data}
	\eeq
	where $\rho>0$ controls the initial exponential suppression of high modes, $g_n$ are independent complex gaussian random variables with mean $0$ and variance $1$, and typically $\mathcal{M}=10$.

	\subsection{Minimally structured systems}
	\label{subsec:numerical_simulations_minimally_structured_systems}
	
	We consider two classes of systems coming from  Section~\ref{subsec:construction_predominantly_random_Hamiltonian_systems}: {\bf (i)} systems whose couplings $C_{nmkj}^{(0)}$ are randomly distributed but fixed in time for each realization and {\bf (ii)} systems with those couplings evolving as stochastic processes. The randomization is introduced by taking the couplings to be real random variables,
    \beq
    C_{nmkj}^{(0)} = \xi_{nmkj},
    \eeq
    where the $\xi_{nmkj}$ are independent random variables subject to the symmetries: $C_{nmkj}^{(0)}=C_{mnkj}^{(0)}=C_{nmjk}^{(0)}=C_{kjnm}^{(0)}$. For the first class of systems, these variables are drawn from normal distributions for each realization
	\beq
	\xi_{nmkj} \sim \mathcal{N}(0, 1),
	\eeq
	and for the second class they follow Ornstein-Uhlenbeck stochastic processes \cite{Stochastic_Processes}, 
	\beq
	d \xi_{nmkj}(t) = -\, \xi_{nmkj}(t)\, dt + \lambda\, dW_{nmkj}(t), \qquad \text{and} \qquad \xi_{nmkj}(0) \sim \mathcal{N}(0, 1),
	\eeq
	where $W_{nmkj}(t)$ are independent Wiener processes, $\lambda$ sets the size of random fluctuations, and each Ornstein-Uhlenbeck process is initialized following a random Gaussian distribution. Note, however, that we are presenting a particular construction of random couplings, but there are infinitely many ways to introduce them as shown in Section~\ref{subsec:construction_predominantly_random_Hamiltonian_systems}. The goal here is not to explore in detail the different properties of every possible choice, but to demonstrate through specific choices that coherent energy cascades arise in minimally structured models even from unfavorable scenarios for coherence due to the random initial data.

	We perform the simulations in truncated versions of the system of equation (\ref{eq:Resonant_Equation}) with $400$ and $800$ modes and evolve with a 4th-order Runge-Kutta scheme. For time-dependent stochastic couplings, these were updated at each time step (while kept constant at internal substeps) using the exact discrete-time update of Ornstein-Uhlenbeck processes \cite{Stochastic_Processes}
	\beq
	\xi_{nmkj}(t+\Delta t)
	=
	e^{-\Delta t} \, \xi_{nmkj}(t)
	+
	\lambda	\sqrt{\frac{1 - e^{-2 \Delta t}}{2}} \ \eta_{nmkj},
	\qquad
	\eta_{nmkj} \sim \mathcal{N}(0,1).
	\eeq
	We use $\lambda = |S_{0000}|$ to have random fluctuations of similar size to low mode couplings.
	The perspective we have outlined follows from these simulations, with analogous behaviors for both random fixed-in-time couplings and stochastic ones.
	
	{\bf \noindent Coherent energy cascades from random initial data:} a significant part of Type III systems ($\atwo > 0$) exhibit coherent cascades structurally akin to those observed on the invariant manifold. As illustrated in Figures~\ref{fig:examples_coherent_energy_cascade_III} and \ref{fig:stochastic_map_Type_3}, high modes undergo phase locking and the energy spectrum approaches a power-law distribution similar to the one calculated in Proposition~\ref{prop:Dynamics_on_the_manifold}:
	\beq
	E_{n\gg 1} \sim n^{-3/2} \qquad \text{and} \qquad \arg \alpha_{n\gg 1} \sim a_1 n + a_0.
	\eeq
	This asymptotic behavior is quantified by fitting the energy spectrum to an asymptotic ansatz of the form $E_{n}(t) = A(t)\,  n^{\gamma(t)} e^{-\rho(t) n}$ with $A,\rho,\gamma \in \mathbb{R}$, for sufficiently large modes, usually from $n = 150$ to $n=\mathcal{N}_{\max}-50$ in simulations of $\mathcal{N}_{\max}$ modes. Modes close to the cutoff are excluded to reduce the influence of truncation effects. The power-law formation corresponds to $\rho(t)$ going to zero as observed in Figure~\ref{fig:stochastic_map_Type_3}. However, since the exponential suppression of the spectrum weakens, truncation effects become larger and the simulation progressively loses track of the actual system. Such limitations are behind the spurious fluctuations observed in the value of $\gamma$ in Figure~\ref{fig:stochastic_map_Type_3} when $\rho$ becomes small --- at those moments a visual inspection of the energy spectrum is more effective to estimate the power-law, as one can also appreciate from Figure~\ref{fig:stochastic_map_Type_3}; see also \cite{Biasi} for the observation of these effects when comparing with analytic solutions.  
	
	Phase locking is quantified over time by the average deviation of the phases for high modes from a linear profile:
	\beq
	\Phi(t):=
	\left<\, 
	\left|
	\, \arg_{[-\pi,\pi]}\left(\alpha_n(t) e^{-i(\phi_1(t)n+\phi_0(t))} \right)\,
	\right|\, \right>_n,
	\eeq
    where $\arg_{[-\pi,\pi]}$ means phases mapped to that interval, a decision that provides two clearly distinct values of $\Phi(t)$ for locked ($\Phi(t) = 0$) and uniformly distributed random phases ($\Phi(t) \approx \pi/2$). The slope and constant term are estimated by averaging the local estimators
	\beq
    \phi_1=\left\langle\arg(\alpha_{n+1}\bar{\alpha}_n)\right\rangle_n,\qquad
	\phi_0=\left\langle\arg(\alpha_n \alpha_n \bar{\alpha}_{2n})\right\rangle_n,
	\eeq
	where the averages are taken over the interval $n=150,\ldots,\mathcal{N}_{\max}-50$ for the first quantity, and $n=150,\ldots,(\mathcal{N}_{\max}-50)/2$ for the second one, in order to reduce truncation and low-mode effects. For an exact linear phase profile the above expressions recover the exact slope and constant term, therefore, $\Phi(t)\to0$ as the phases become increasingly aligned.
	
	A first approximation to the probability of observing a coherent cascade within $0.1$ time units is provided in the heat map of Figure~\ref{fig:stochastic_map_Type_3}. We have performed twenty simulations up to $0.1$ units of time at each point marked by a white spot and counted the number of these simulations for which the exponent $\rho(t)$  became less than $10^{-2}$. For some initial data, energy cascades could emerge at longer times than the ones we have simulated. Thus, the actual probabilities to observe energy cascades in our systems can be larger.
	
	By contrast, systems of Type II exhibit this behavior less robustly, as can be appreciated in Figures~\ref{fig:examples_coherent_energy_cascade_II} and \ref{fig:stochastic_map_Type_2}, while for systems of Type I we did not observe the formation of coherent cascades from random initial data; see Figures~\ref{fig:examples_coherent_energy_cascade_I} and \ref{fig:stochastic_map_Type_1}. Moreover, in early stages of the evolution, the three types of systems (Types I-III) tend to undergo phase alignment in high modes. However, Type I and II systems are not generally able to maintain such structure, and the dynamics become increasingly unstructured, as observed in Figures~\ref{fig:stochastic_map_Type_3} - \ref{fig:stochastic_map_Type_1}.

	The difference in the capacity of the three types of systems to develop coherent cascades is already observed on the invariant manifold, where Type III systems exhibit cascades for a larger domain of initial conditions than the other two, and it seems that this behavior is also extended outside the manifold. In particular, for fixed values the conserved quantities $E$ and $N$, Type III systems admit cascade solutions on the invariant manifold for a range of values of the Hamiltonian, $\mathcal{H}$, while Type II and I only for a specific value.

	\begin{figure}[h!]
		\centering
		\includegraphics[width = 8cm]{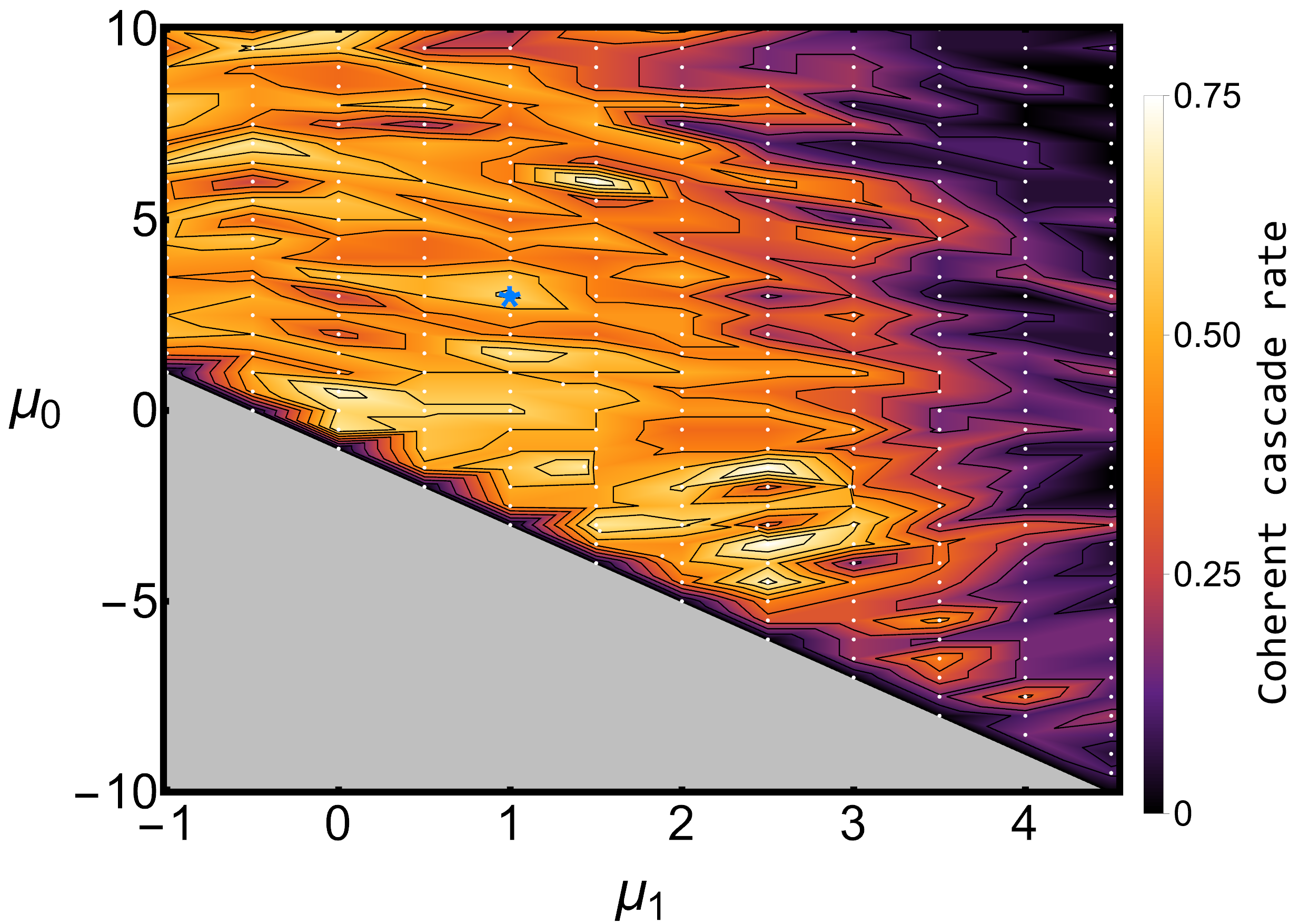}

        \vspace{0.2cm}
        
        \includegraphics[width = 5.6cm]{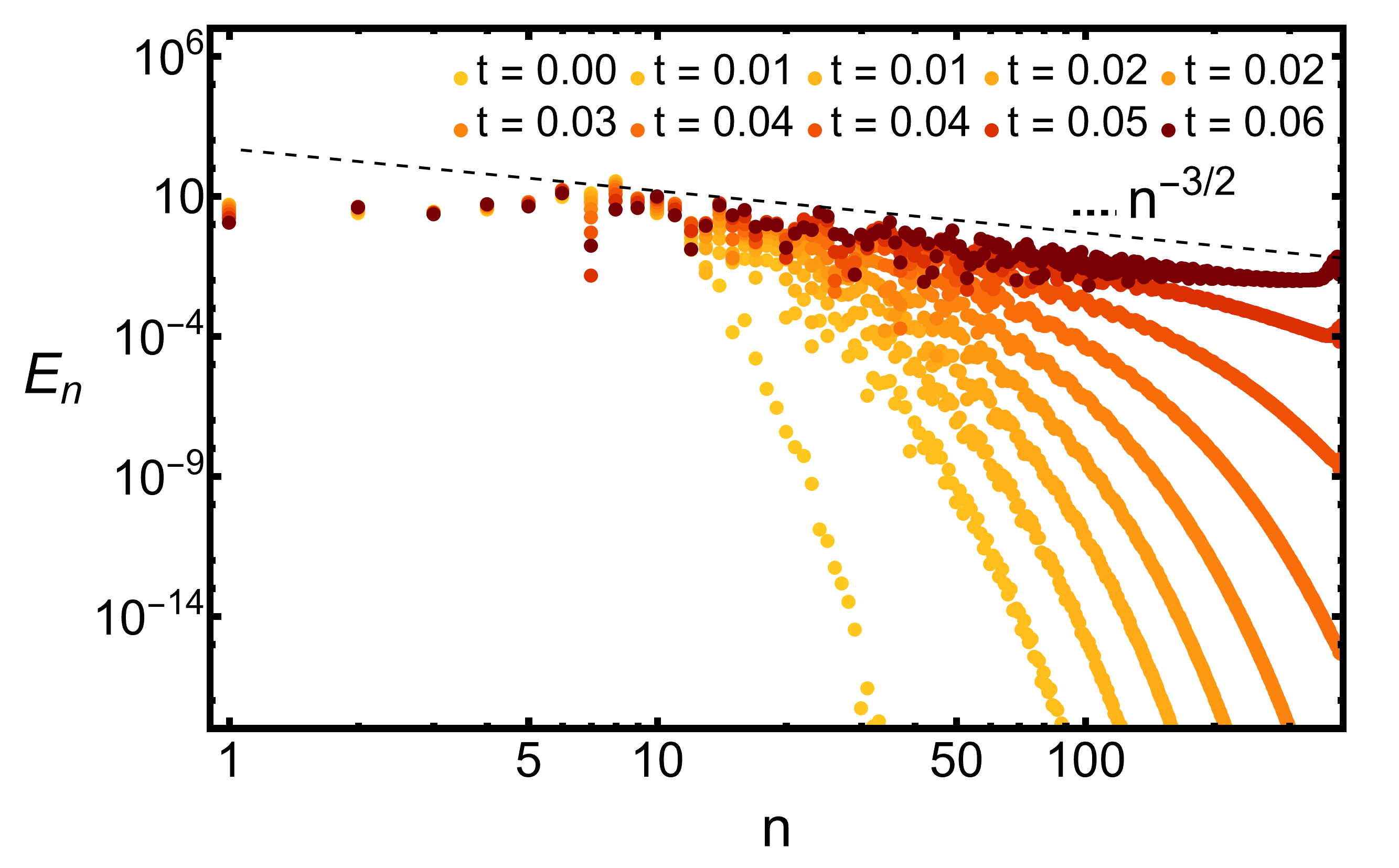}
        \includegraphics[width = 5.6cm]{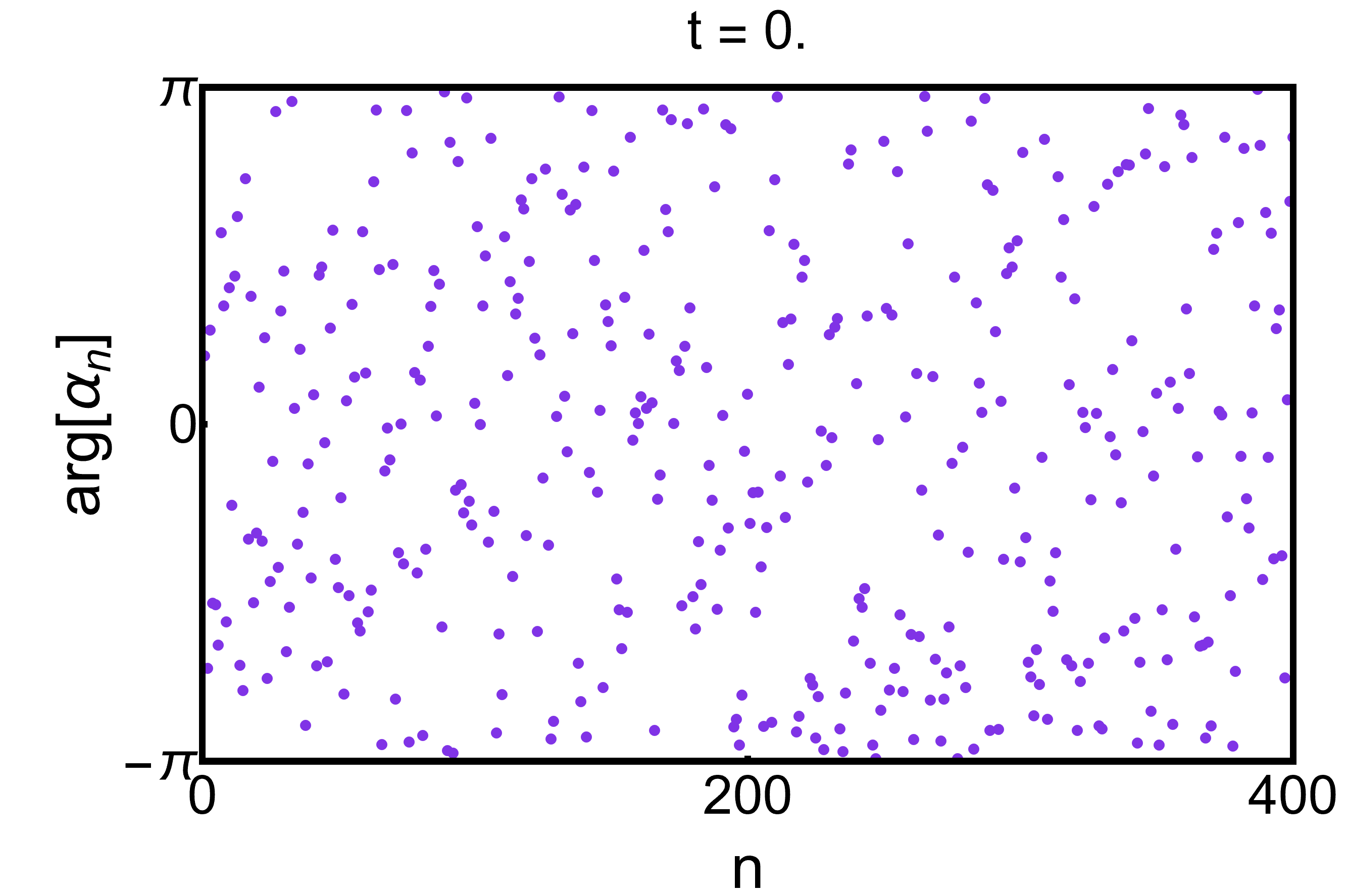}
		\includegraphics[width = 5.6cm]{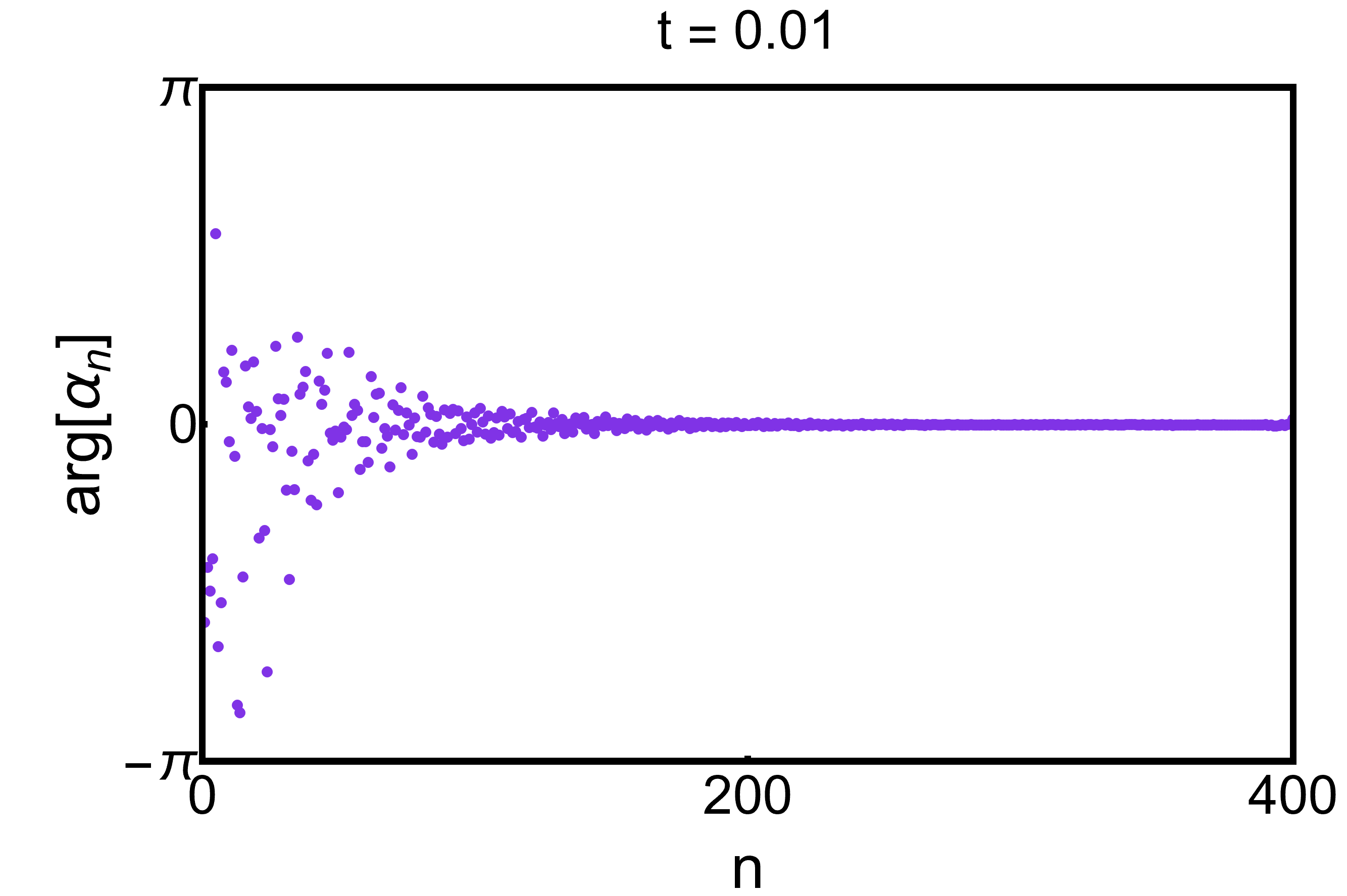}

        \vspace{0.2cm}
        
        \includegraphics[width = 5.6cm]{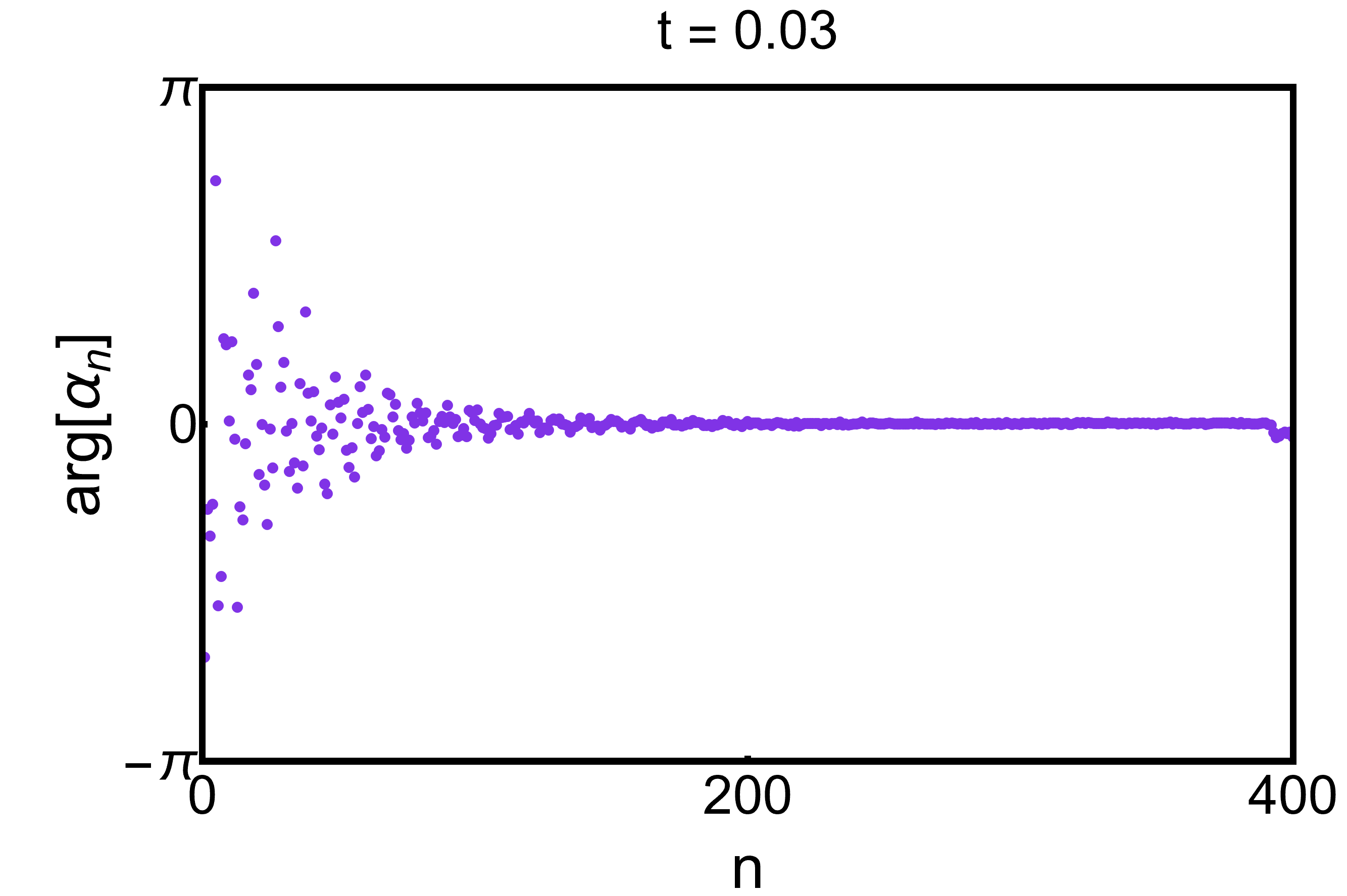}
        \includegraphics[width = 5.6cm]{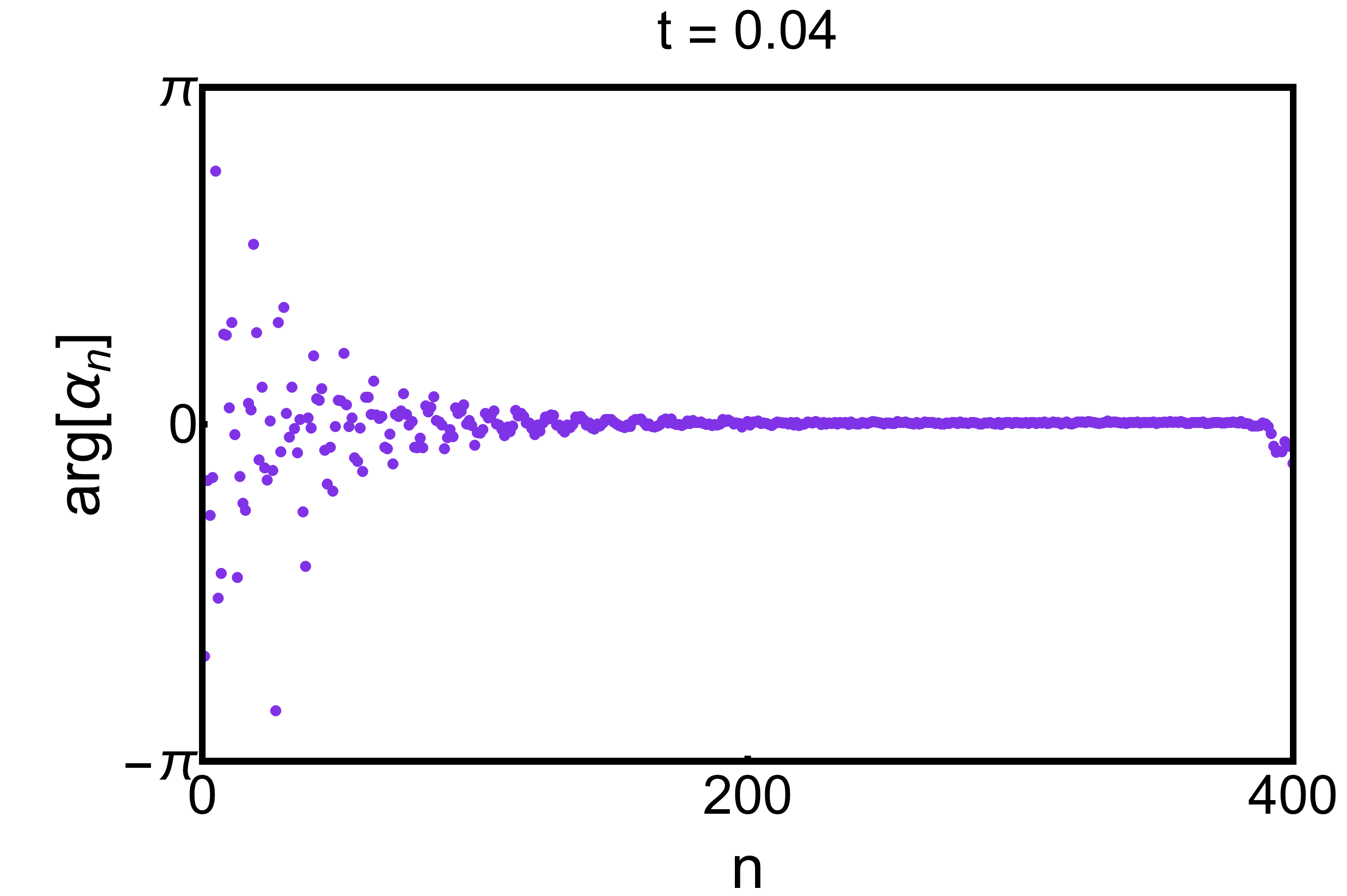}
        \includegraphics[width = 5.6cm]{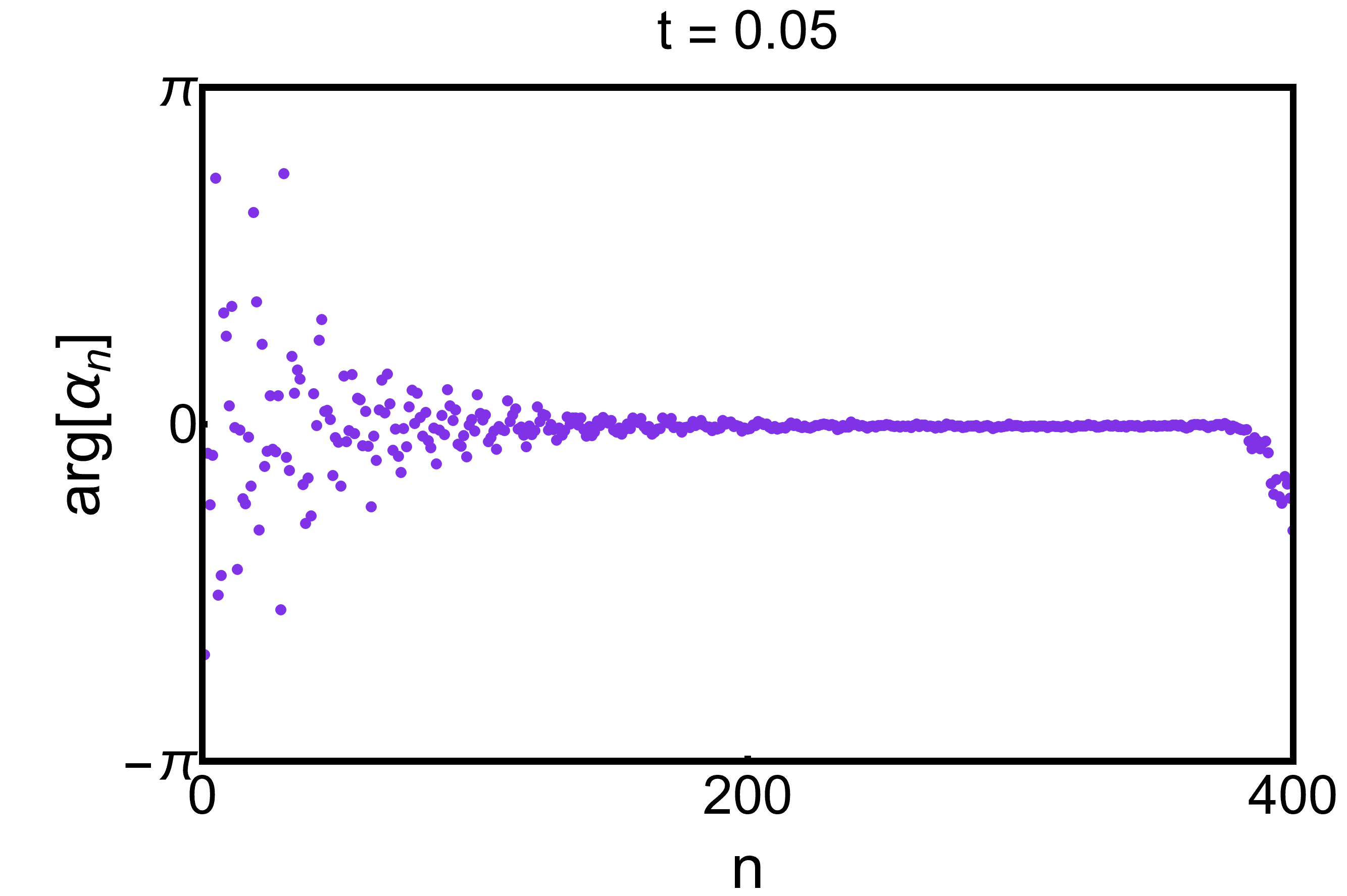}
    
        \vspace{0.2cm}
        
        \includegraphics[width = 5.6cm]{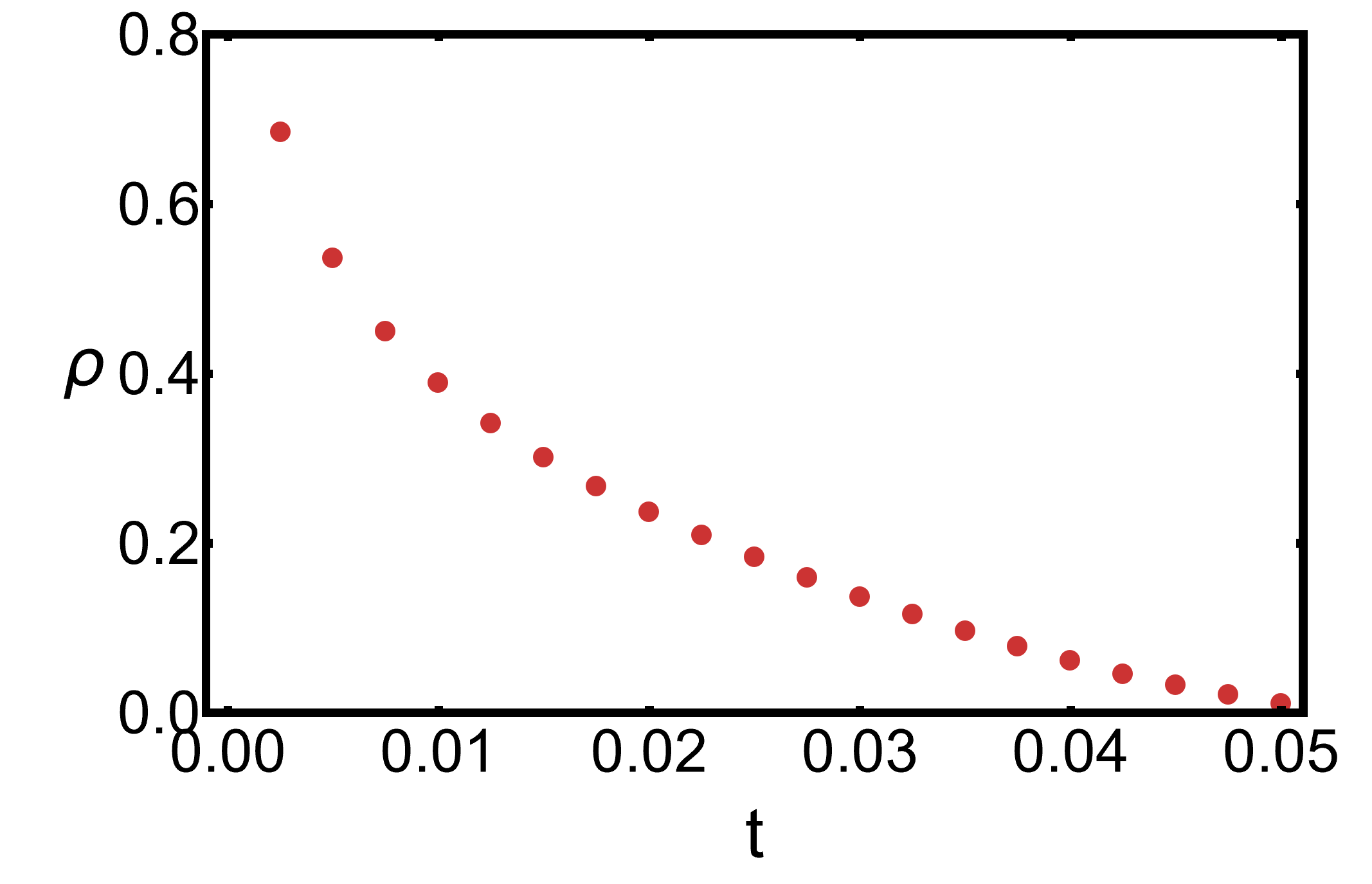}
        \includegraphics[width = 5.6cm]{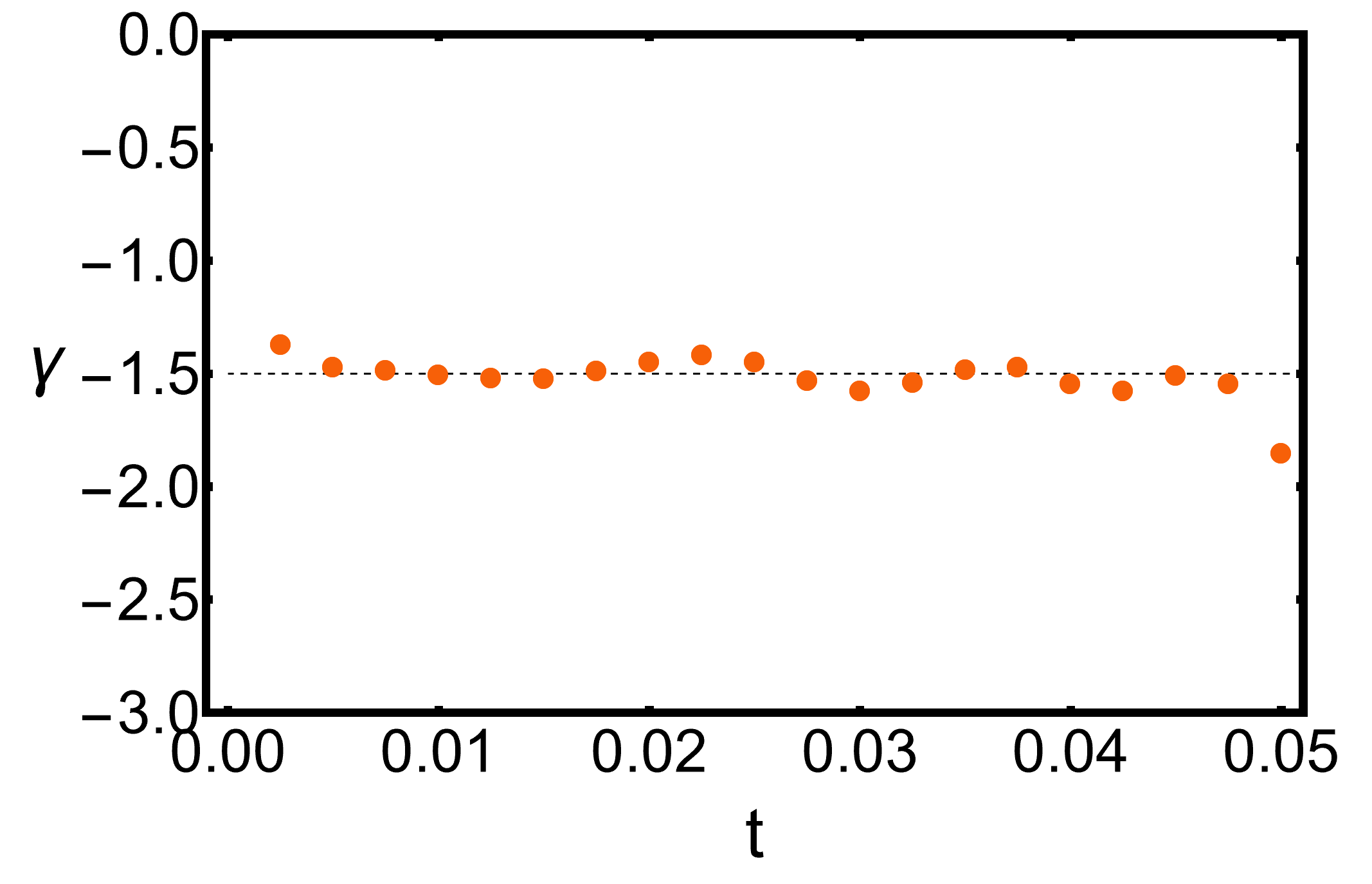}
        \includegraphics[width = 5.6cm]{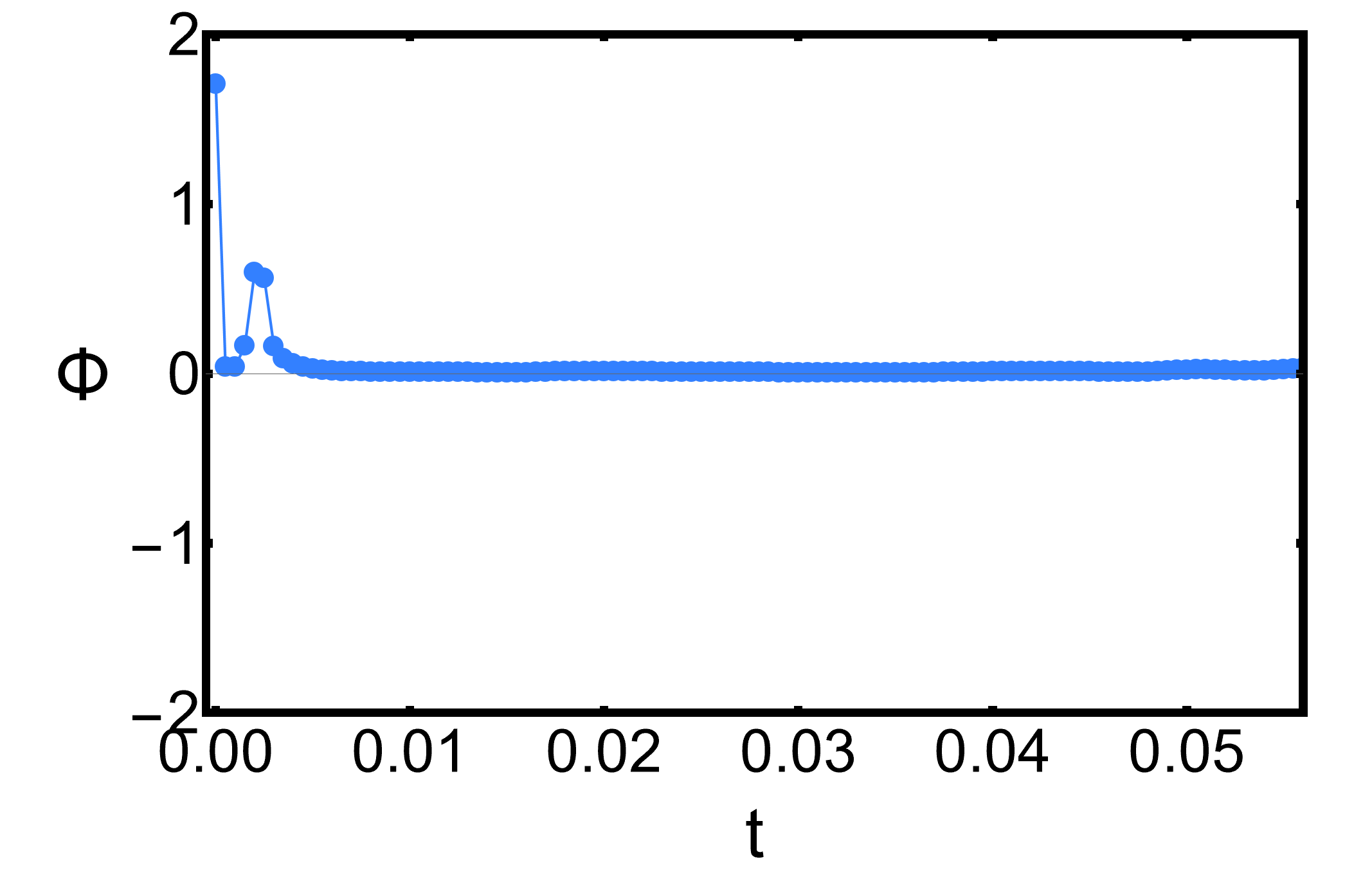}
			\caption{Coherent cascades in Type III systems with $\atwo = 1$ and $\beta_2=\beta_1=\beta_0=0$, stochastic couplings and random initial data. {\em First row:} Occurrence rate of coherent cascades within a time interval $t\in[0,\,0.1]$. The gray region lies outside the parameter domain $\mathcal{D}$, white points indicate the computational grid. {\em Second and third rows:} Evolution of the energy spectrum and a sequence of phase distributions corresponding to a system indicated by the blue star in the previous plot. Deviations from the linear trend on the right-hand side of the phase plots are due to truncation effects from simulating a truncated system of $400$ modes and the gradual loss of the exponential suppression of high modes. {\em Fourth row:} Evolution of the exponent $\rho(t)$, the power-law exponent $\gamma(t)$, and the deviation from a linear phase relation $\Phi(t)$, associated with the previous simulation.}
		\label{fig:stochastic_map_Type_3}
	\end{figure}

    \begin{figure}
		\centering
		\includegraphics[width = 8cm]{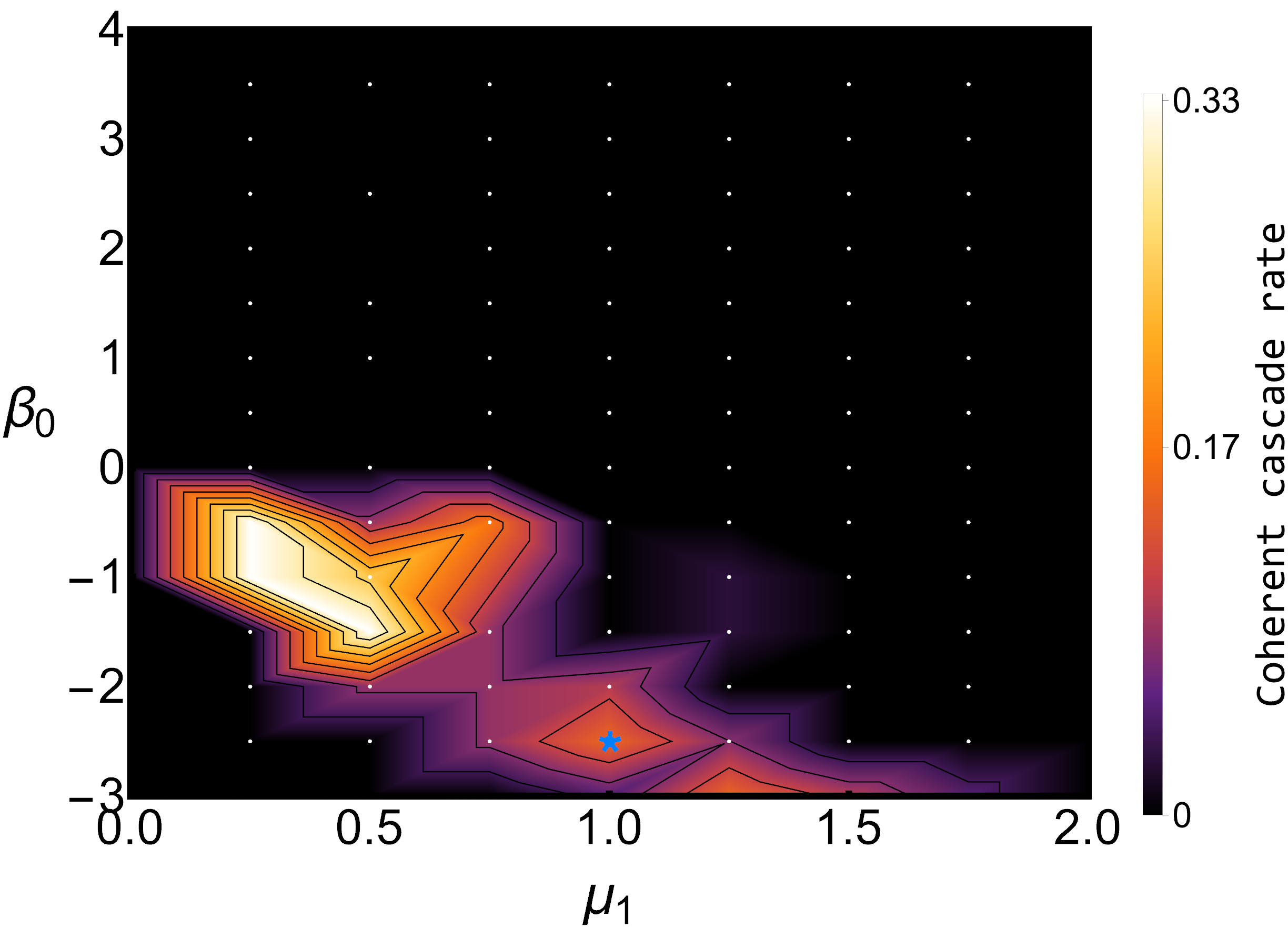}

        \vspace{0.5cm}
        
		\includegraphics[width = 5.6cm]{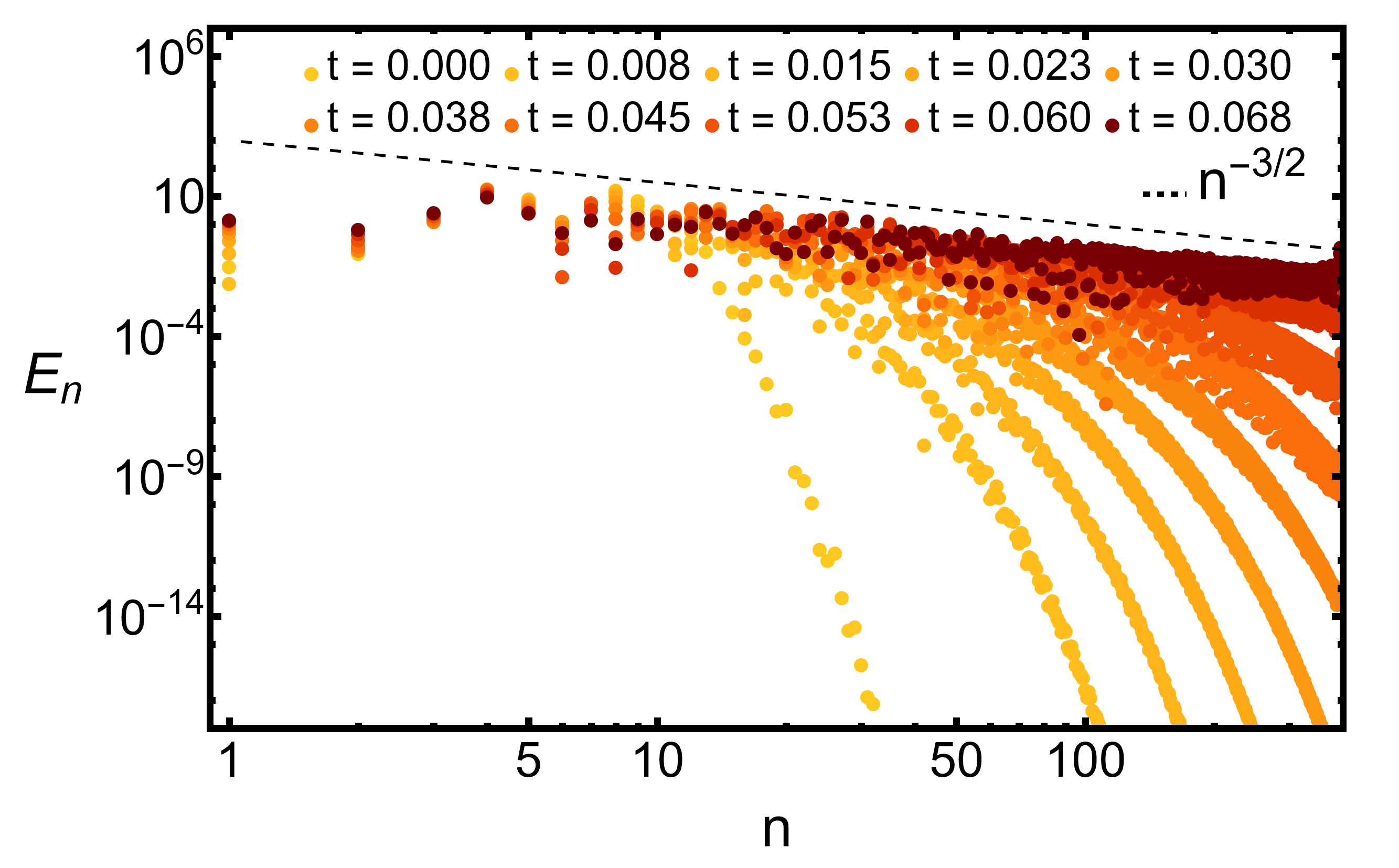}
		\includegraphics[width = 5.6cm]{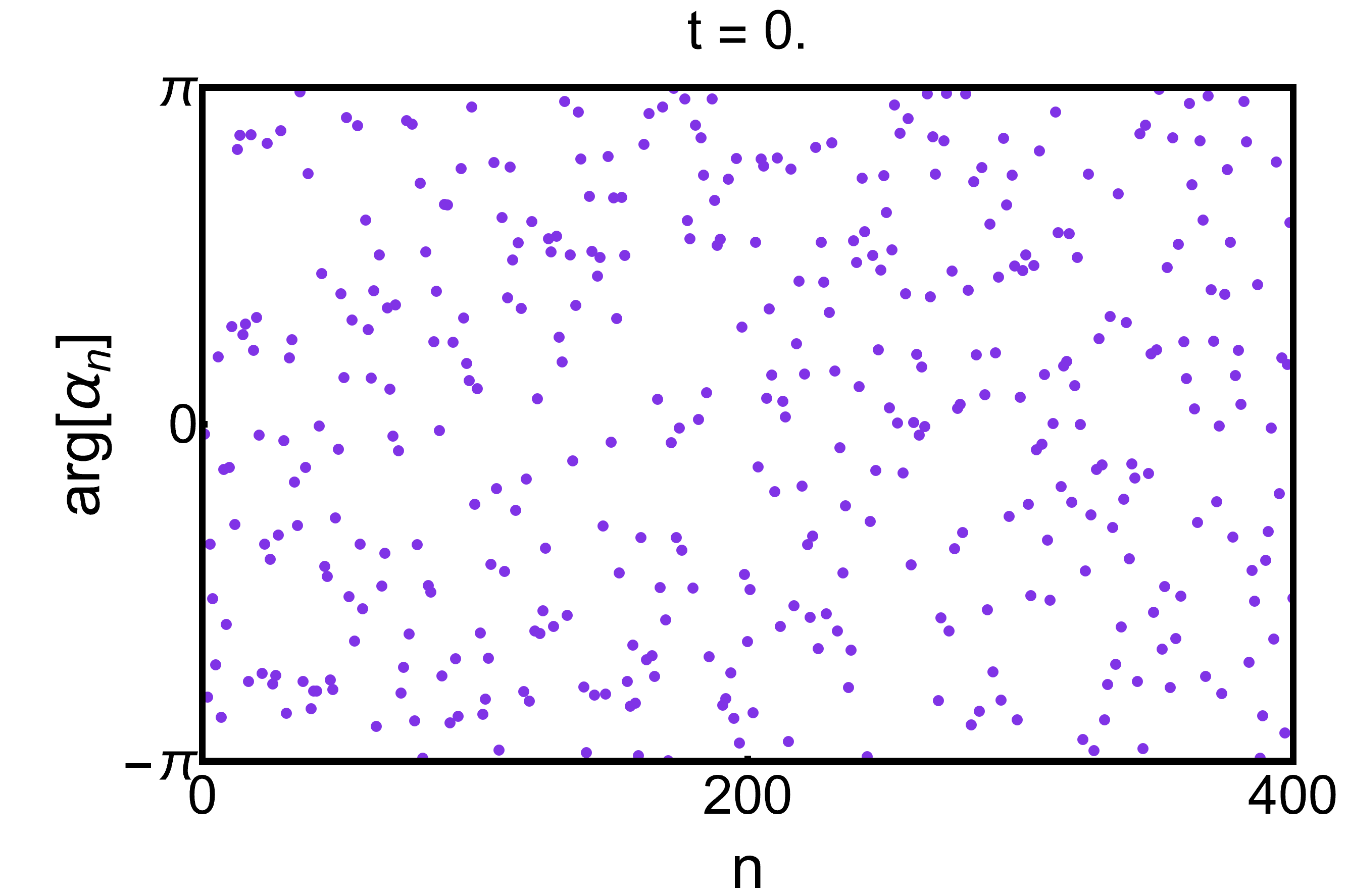}
        \includegraphics[width = 5.6cm]{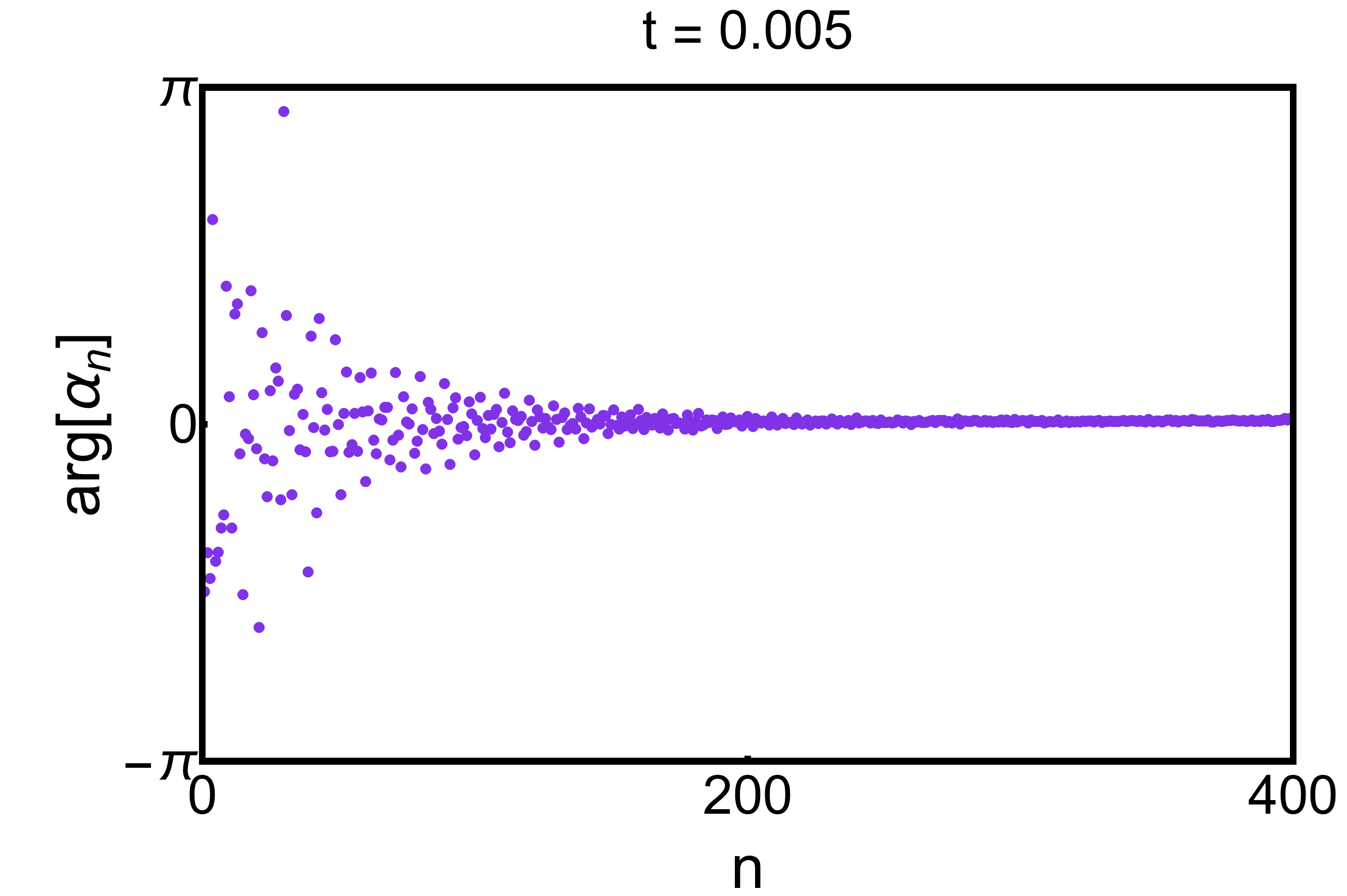}
		
        \vspace{0.5cm}
        
		\includegraphics[width = 5.6cm]{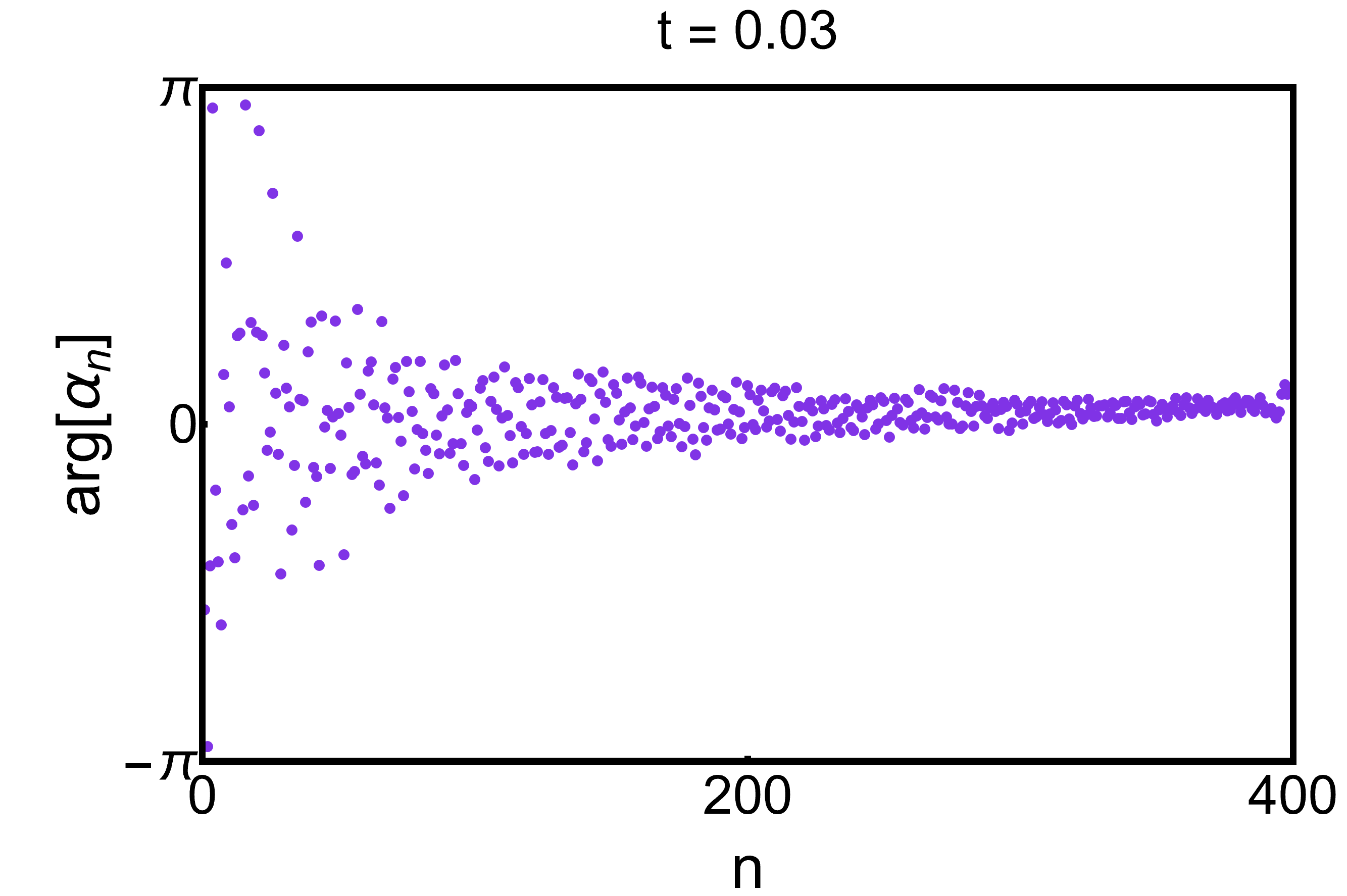}
        \includegraphics[width = 5.6cm]{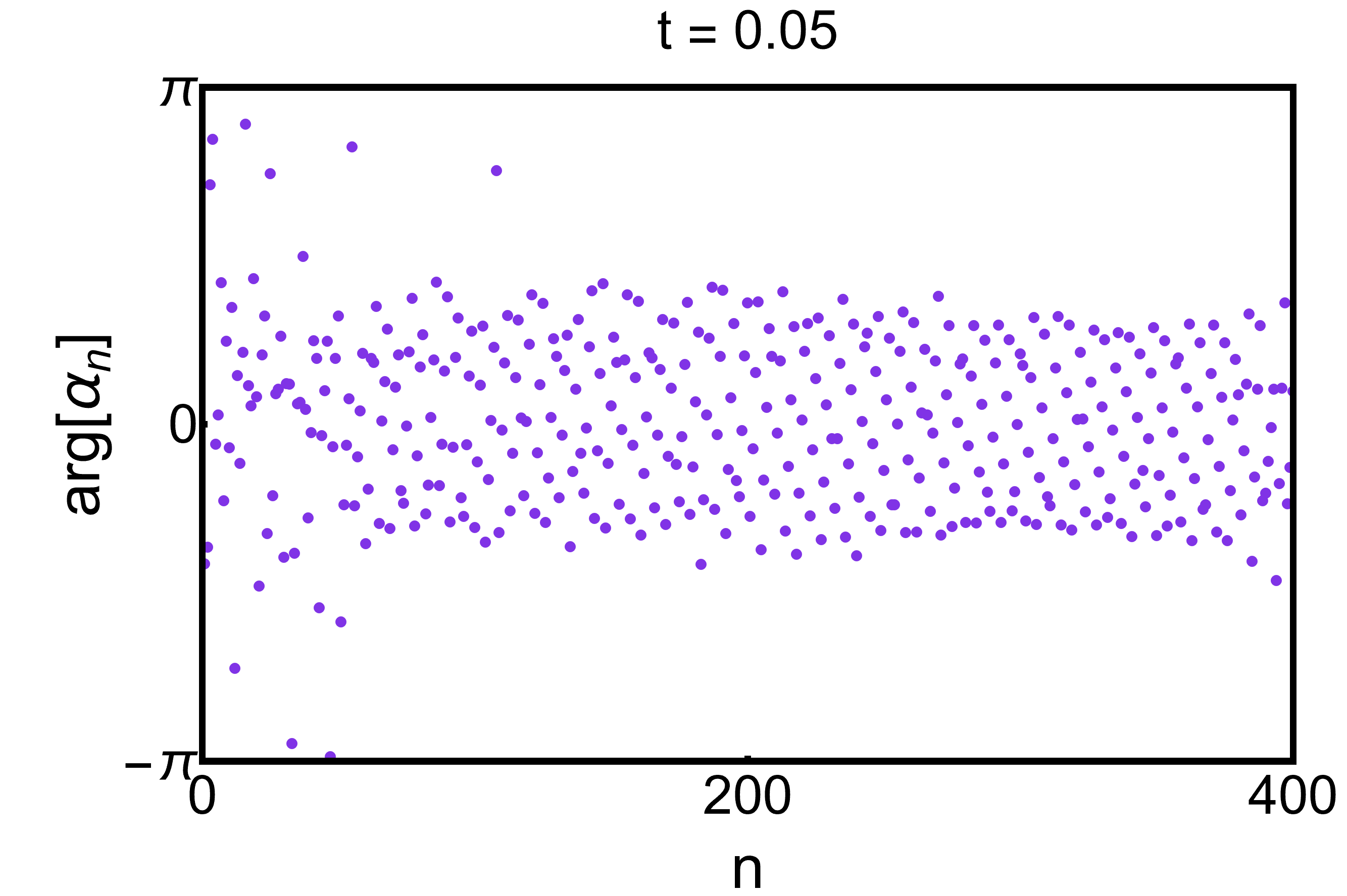}
        \includegraphics[width = 5.6cm]{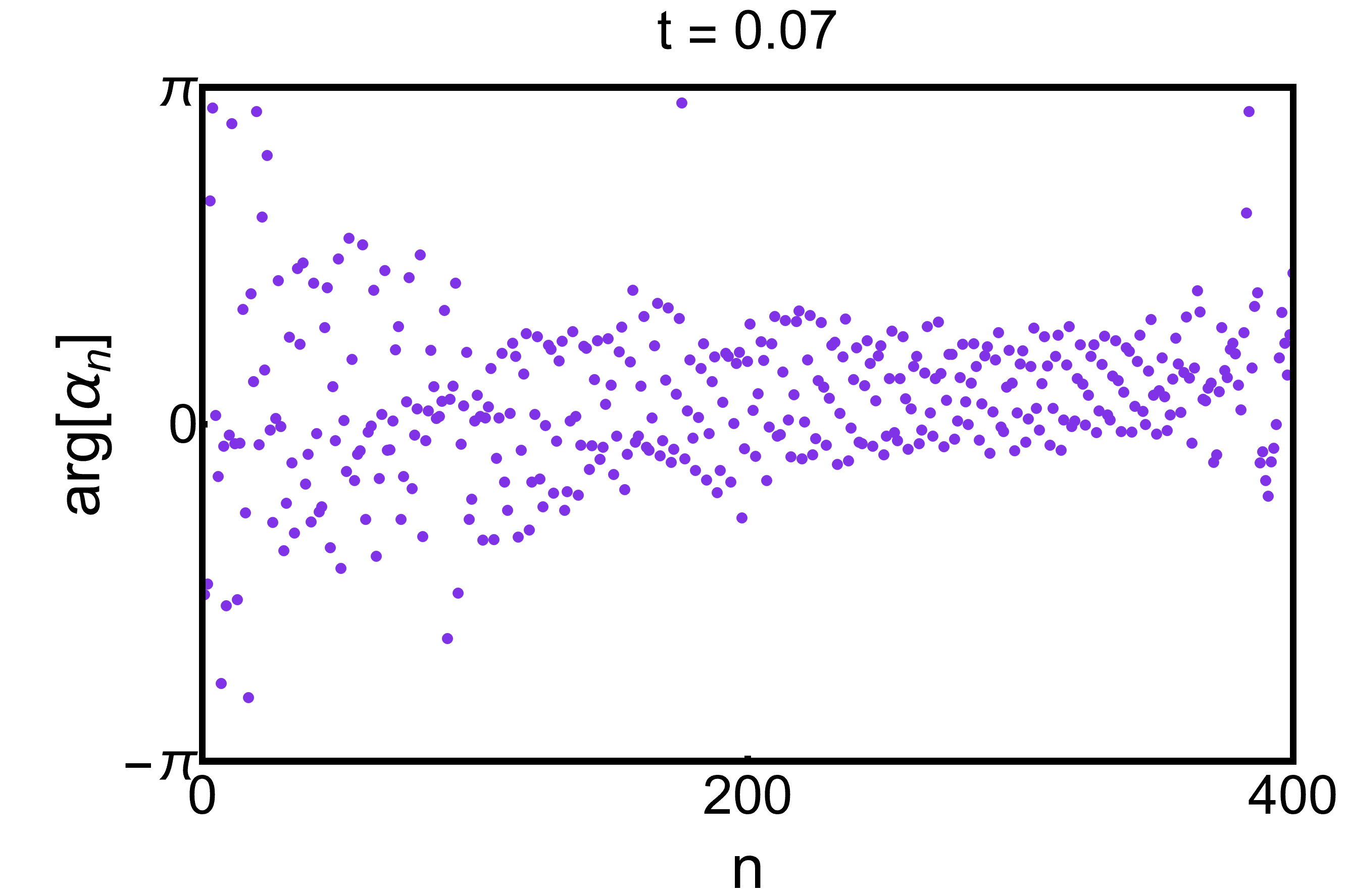}
        
        \vspace{0.5cm}
		
        \includegraphics[width = 5.6cm]{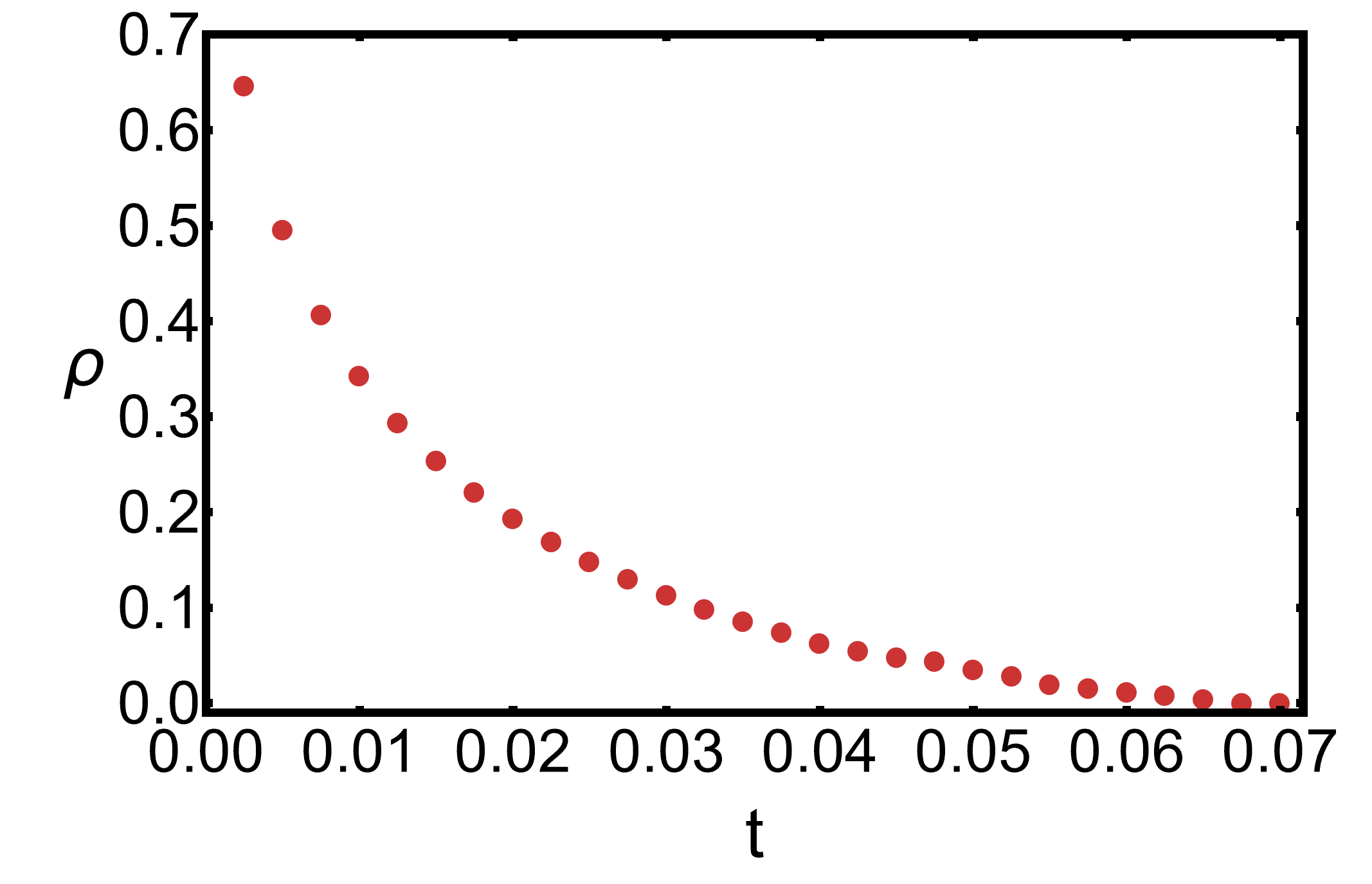}
		\includegraphics[width = 5.6cm]{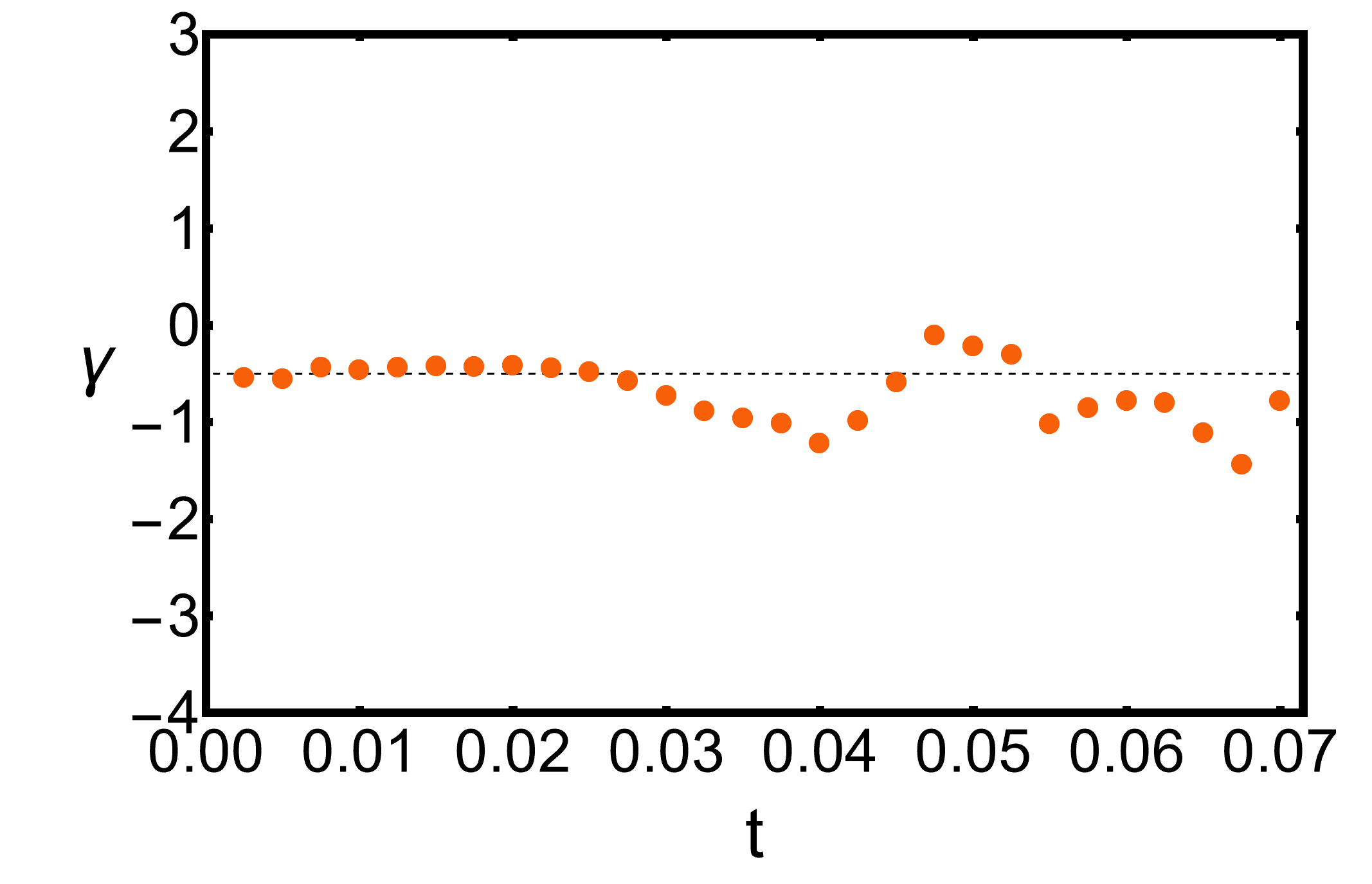}
        \includegraphics[width = 5.6cm]{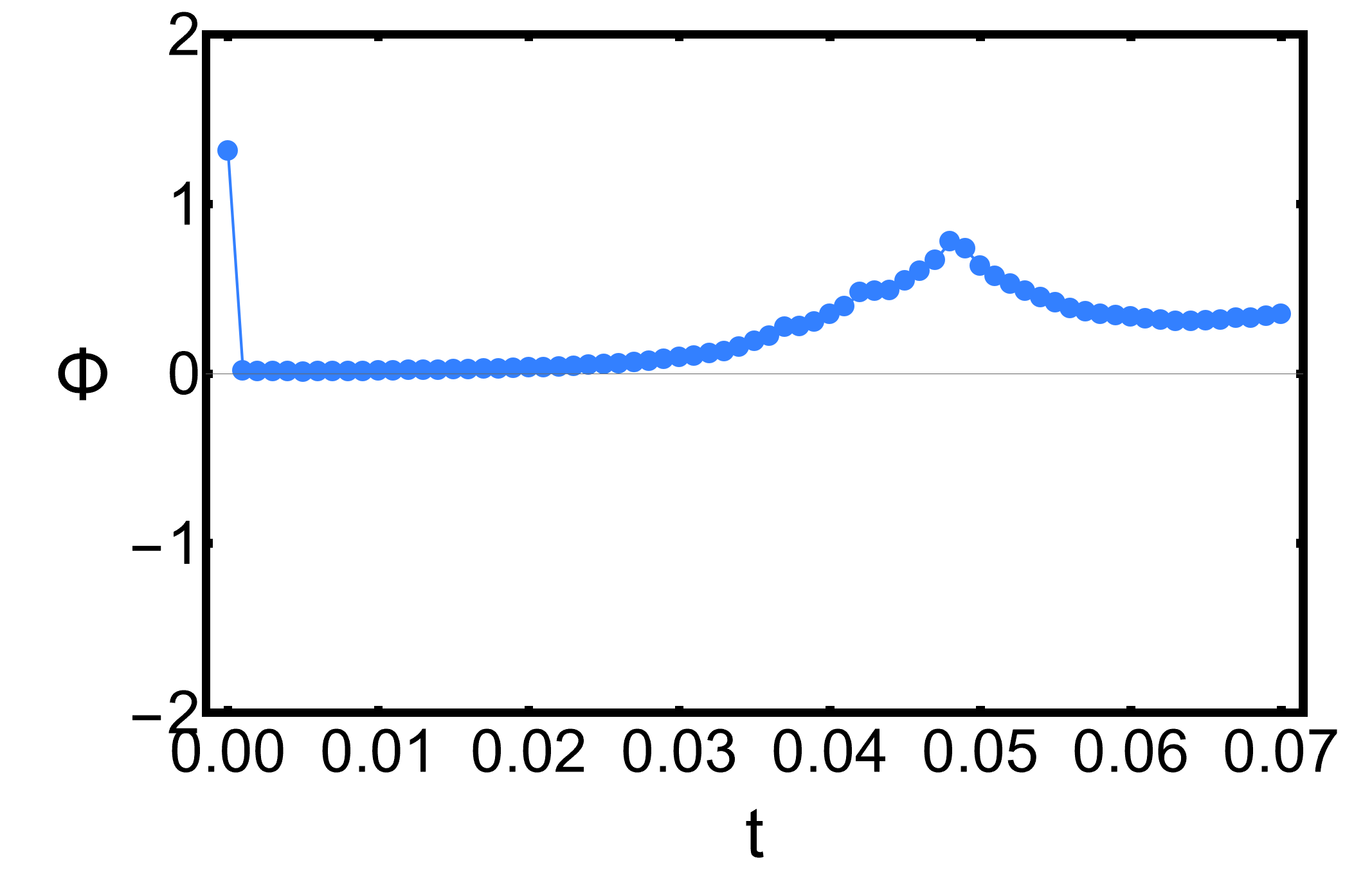}
            \caption{Coherent cascades in Type II systems with $\atwo = 0$, $\bone=1$, $\beta_2=0$,  and $\beta_1=-\beta_0$, stochastic couplings and random initial data. {\em First row:} Occurrence rate of coherent cascades within a time interval $t\in[0,\,0.1]$. The white points indicate the computational grid. {\em Second and third rows:} Evolution of the energy spectrum and a sequence of phase distributions corresponding to a system indicated by the blue star in the previous plot. {\em Fourth row:} Evolution of the exponent $\rho(t)$, the asymptotic power-law exponent $\gamma(t)$, and the deviation from a linear phase relation $\Phi(t)$, associated with the previous simulation.}
		\label{fig:stochastic_map_Type_2}
	\end{figure}

        \begin{figure}
		\centering
        
		\includegraphics[width = 5.6cm]{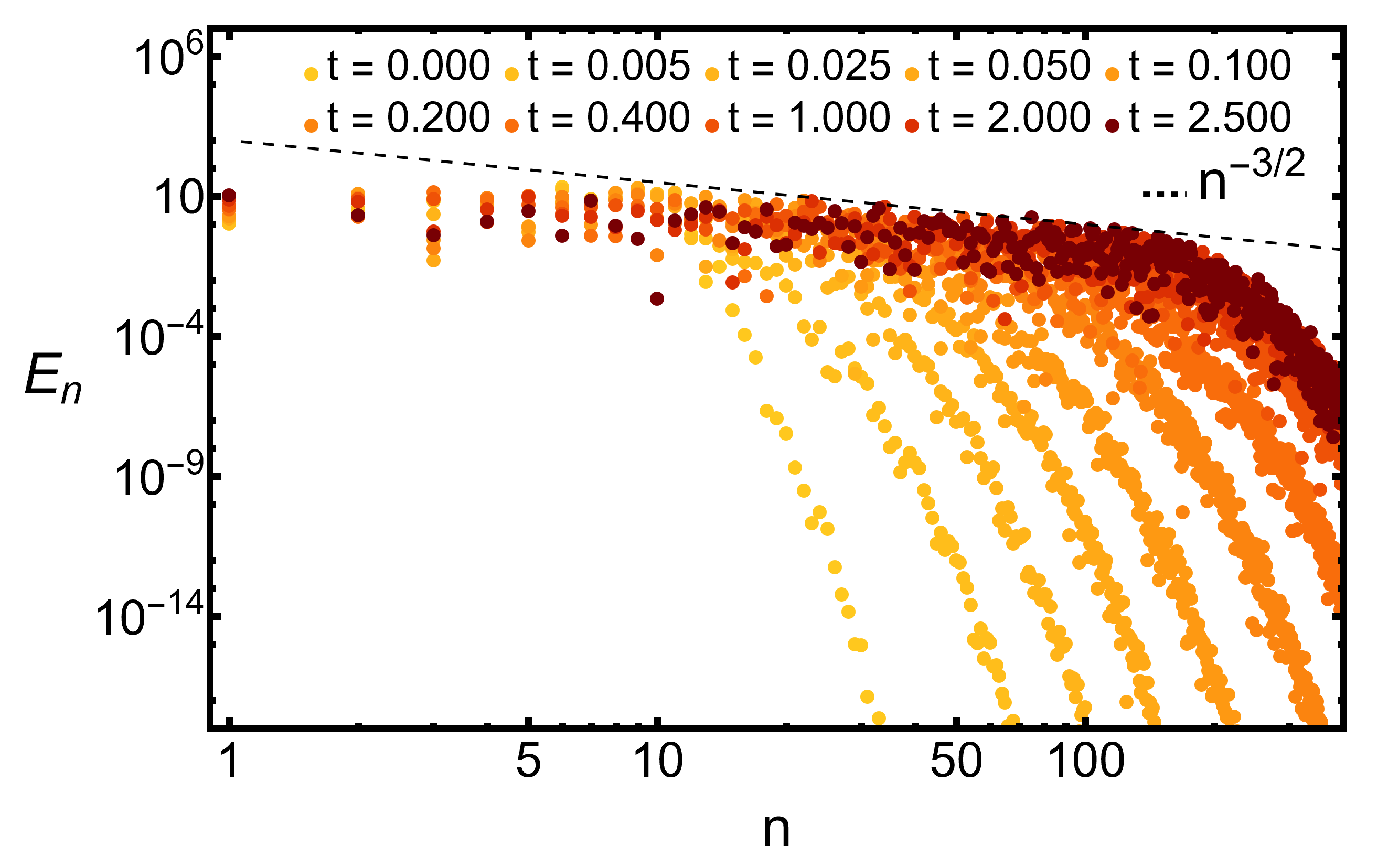}
		\includegraphics[width = 5.6cm]{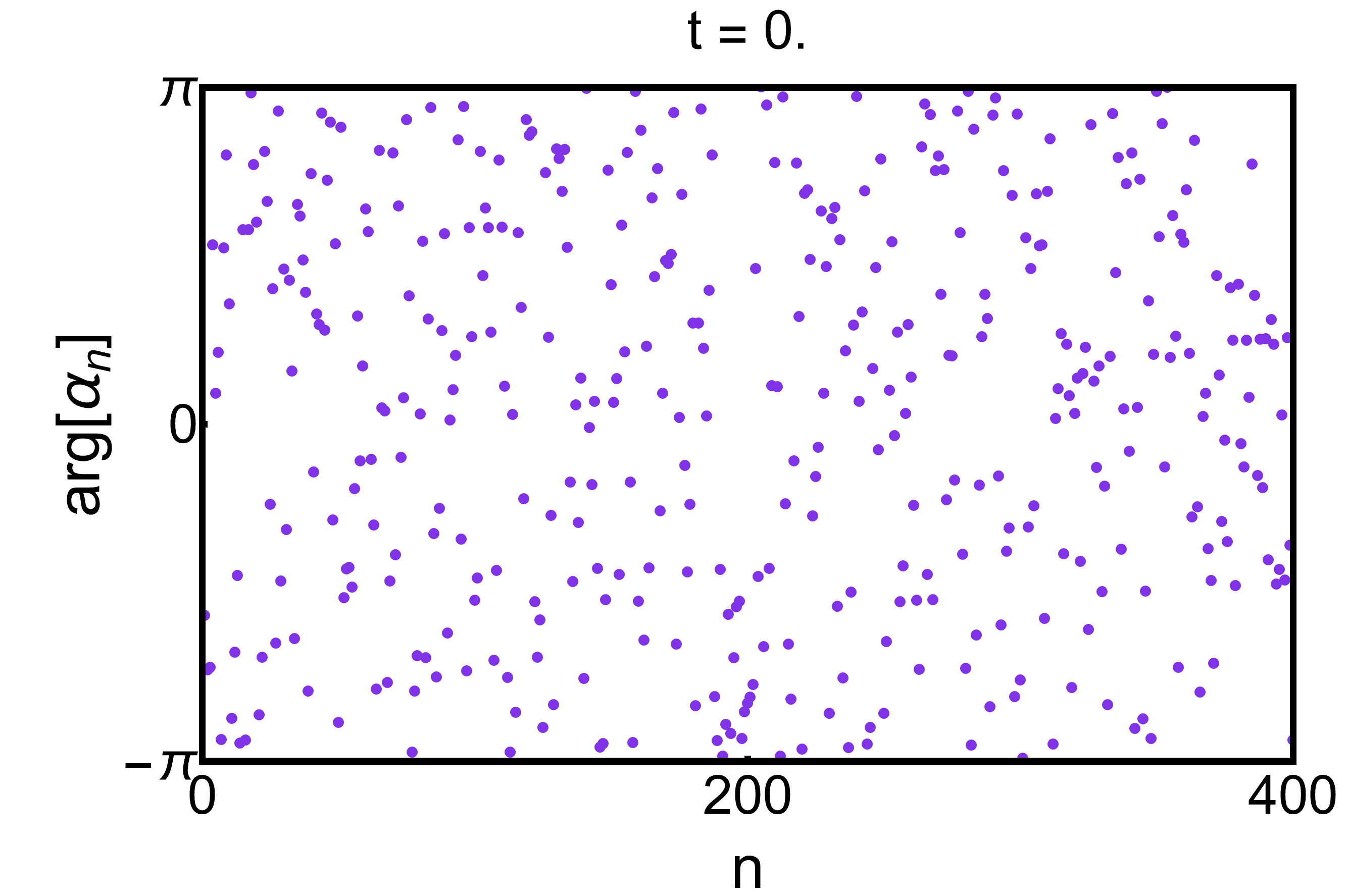}
        \includegraphics[width = 5.6cm]{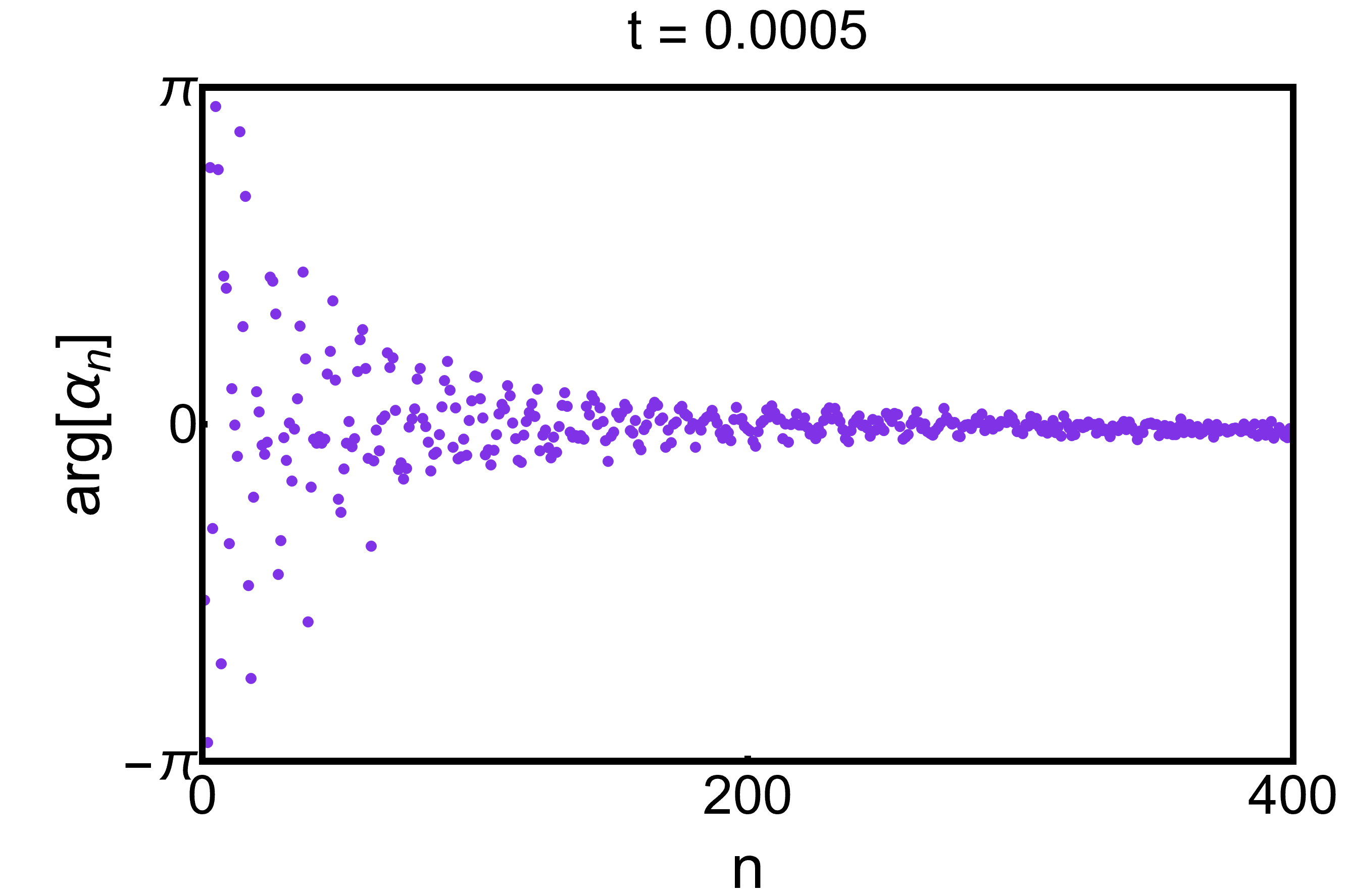}
		
        \vspace{0.5cm}
        
		\includegraphics[width = 5.6cm]{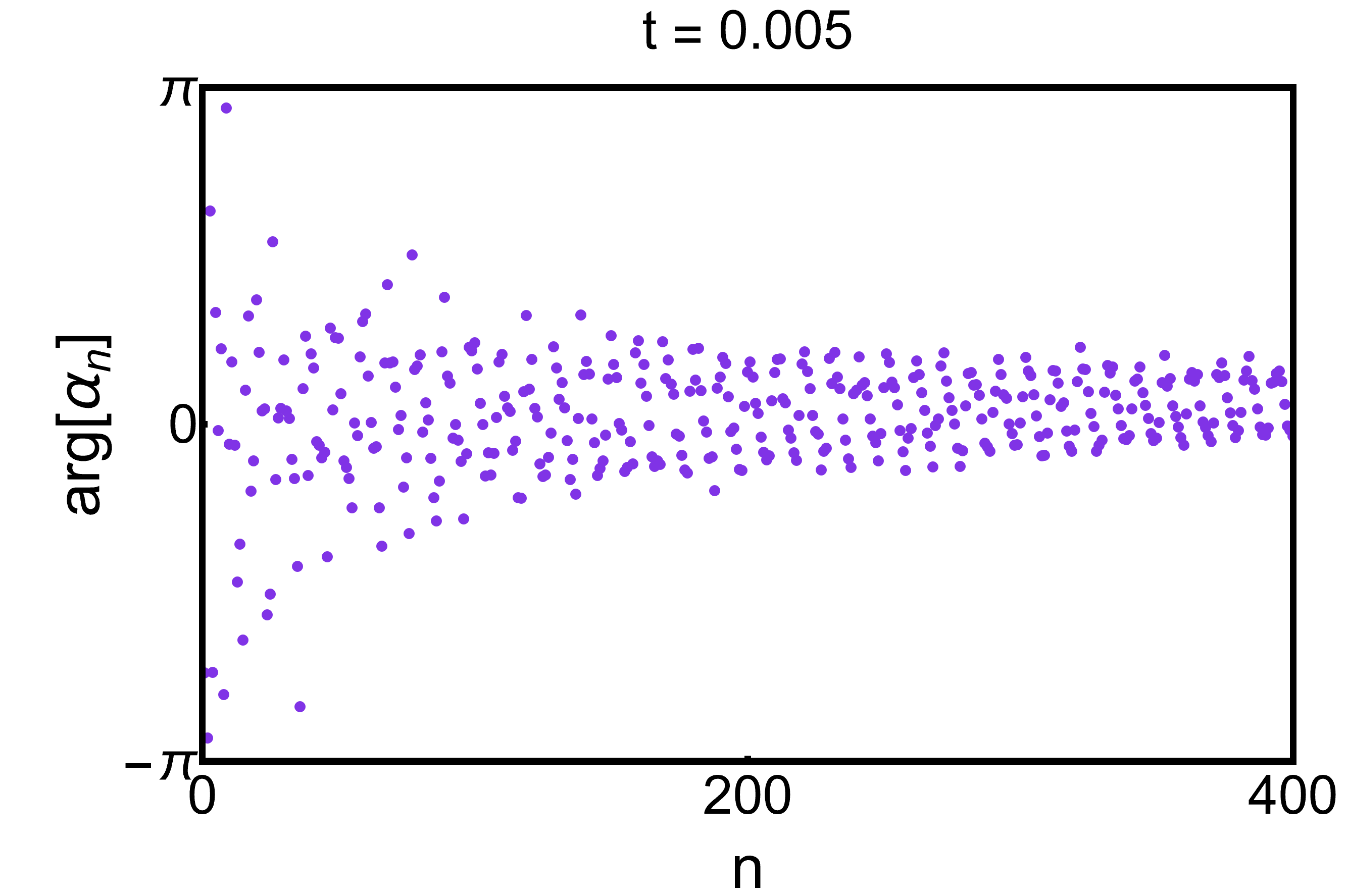}
        \includegraphics[width = 5.6cm]{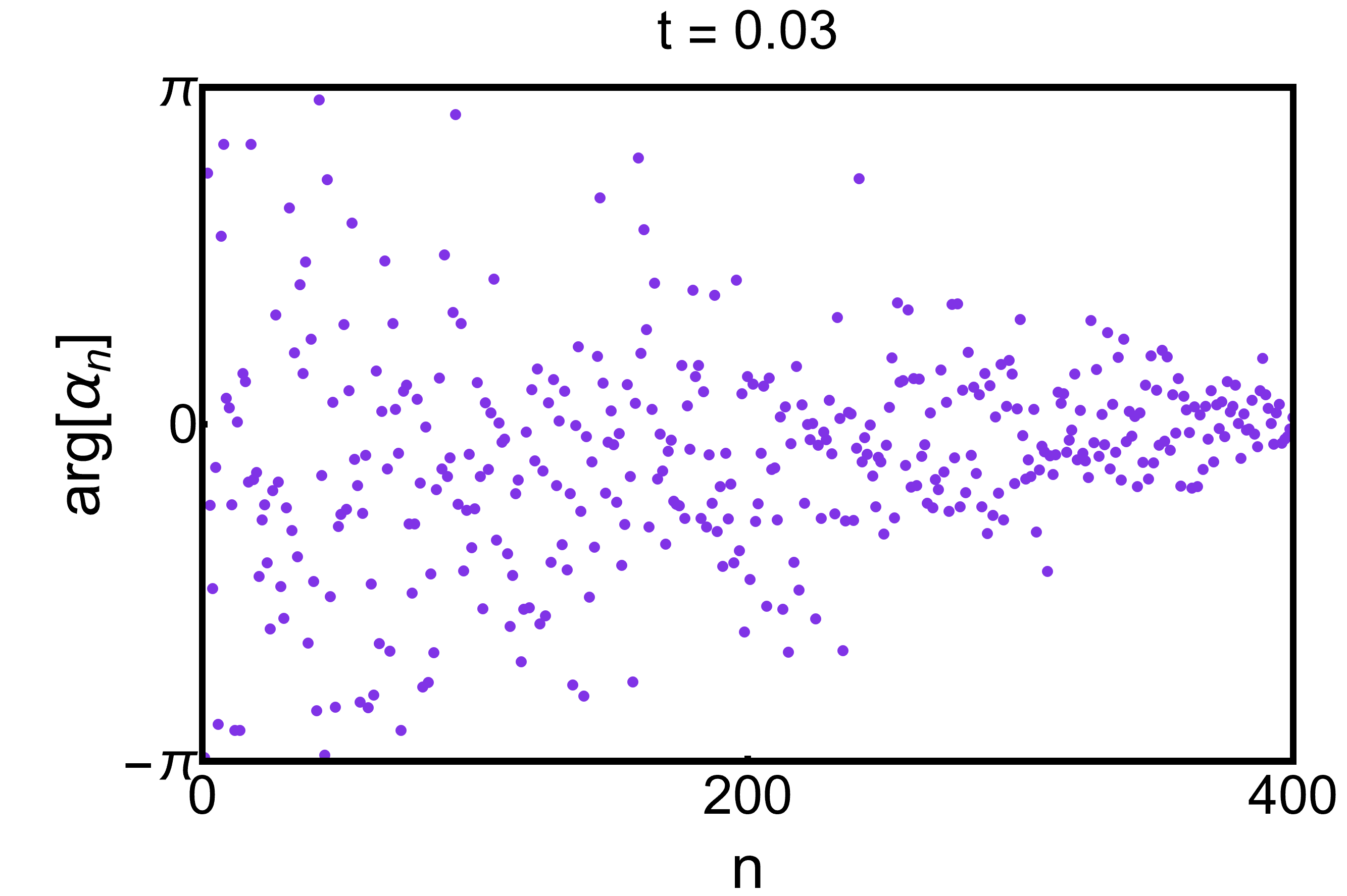}
        \includegraphics[width = 5.6cm]{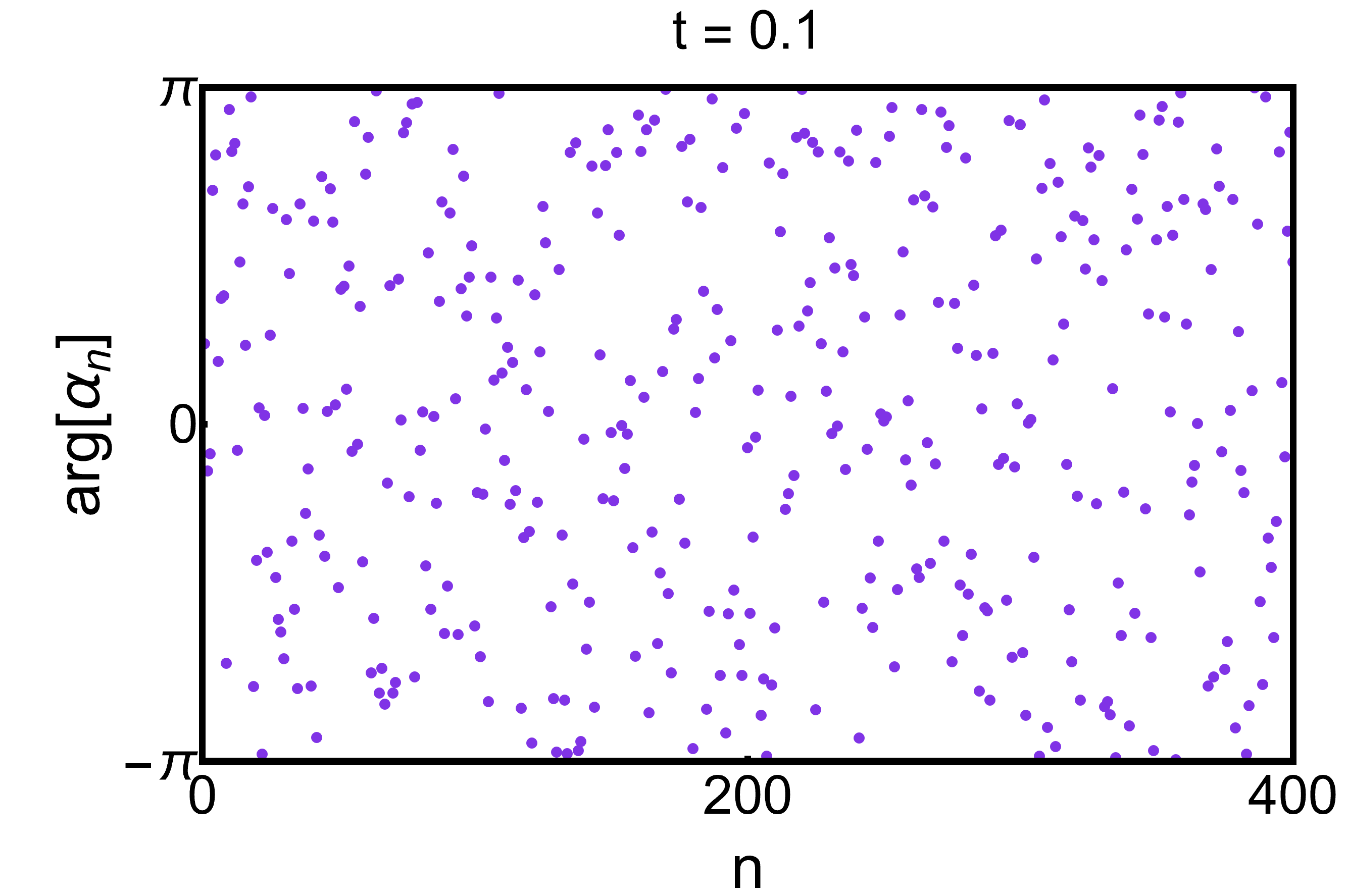}
        
        \vspace{0.5cm}
		
        \includegraphics[width = 5.6cm]{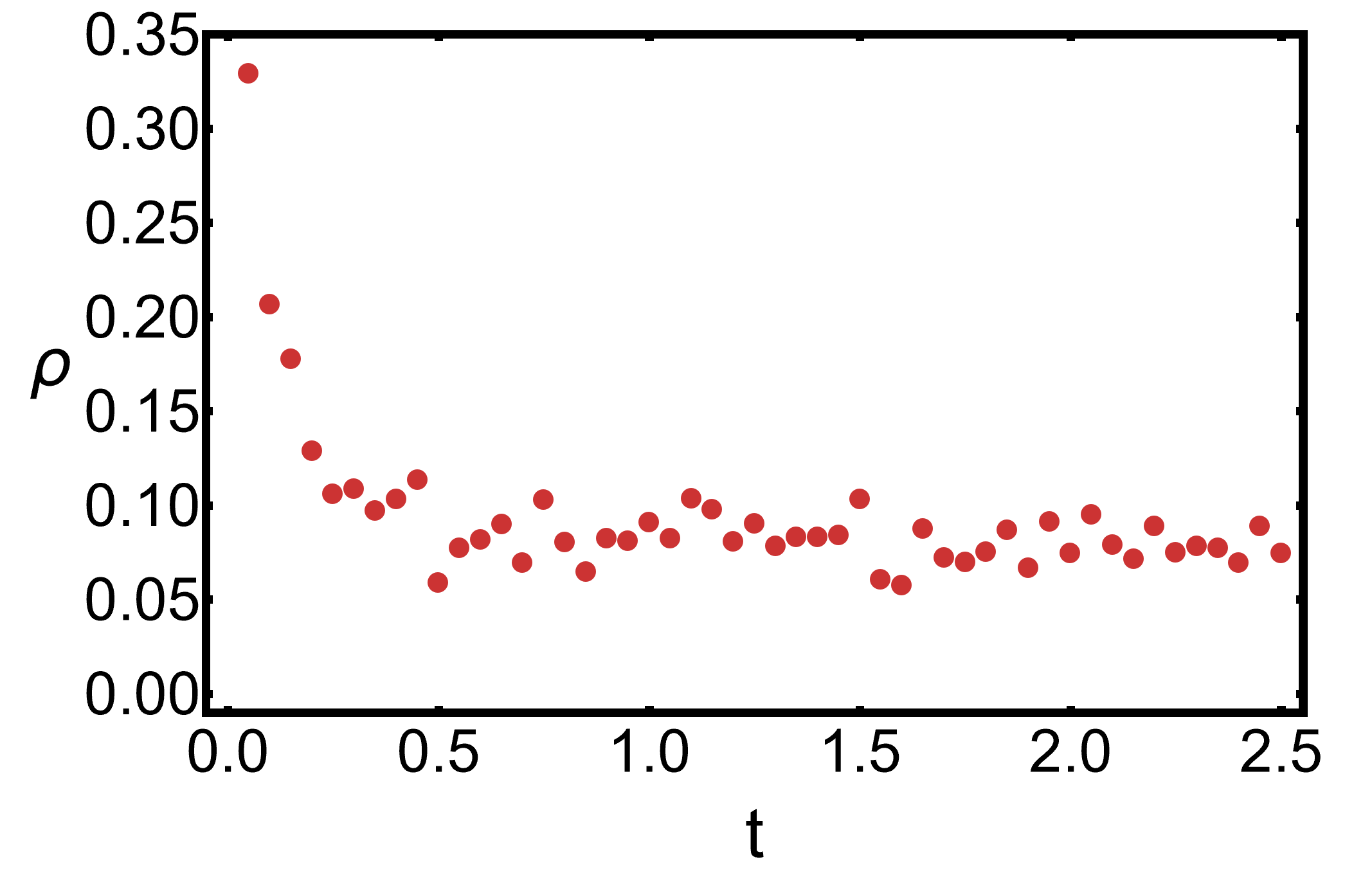}
		\includegraphics[width = 5.6cm]{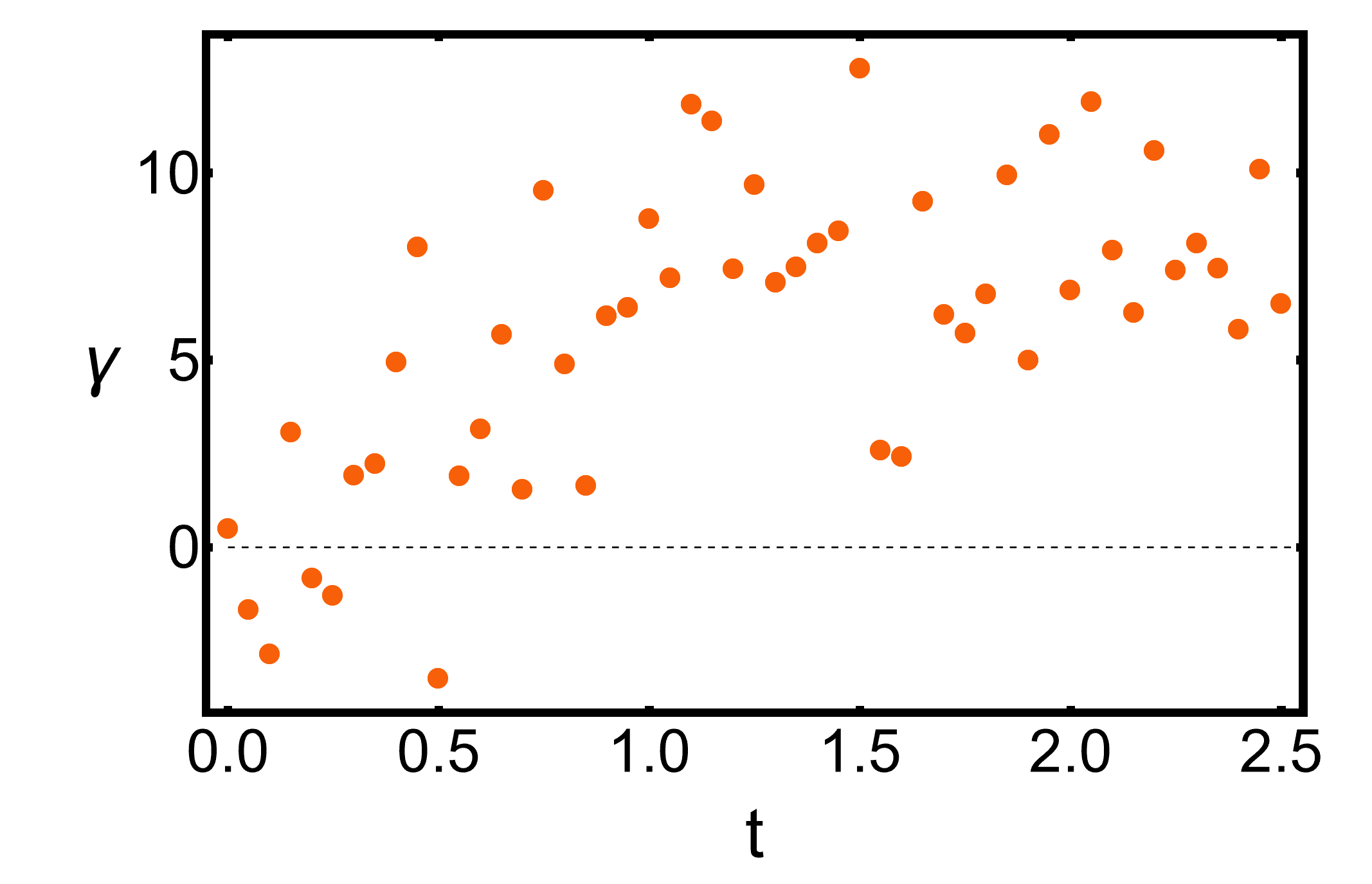}
        \includegraphics[width = 5.6cm]{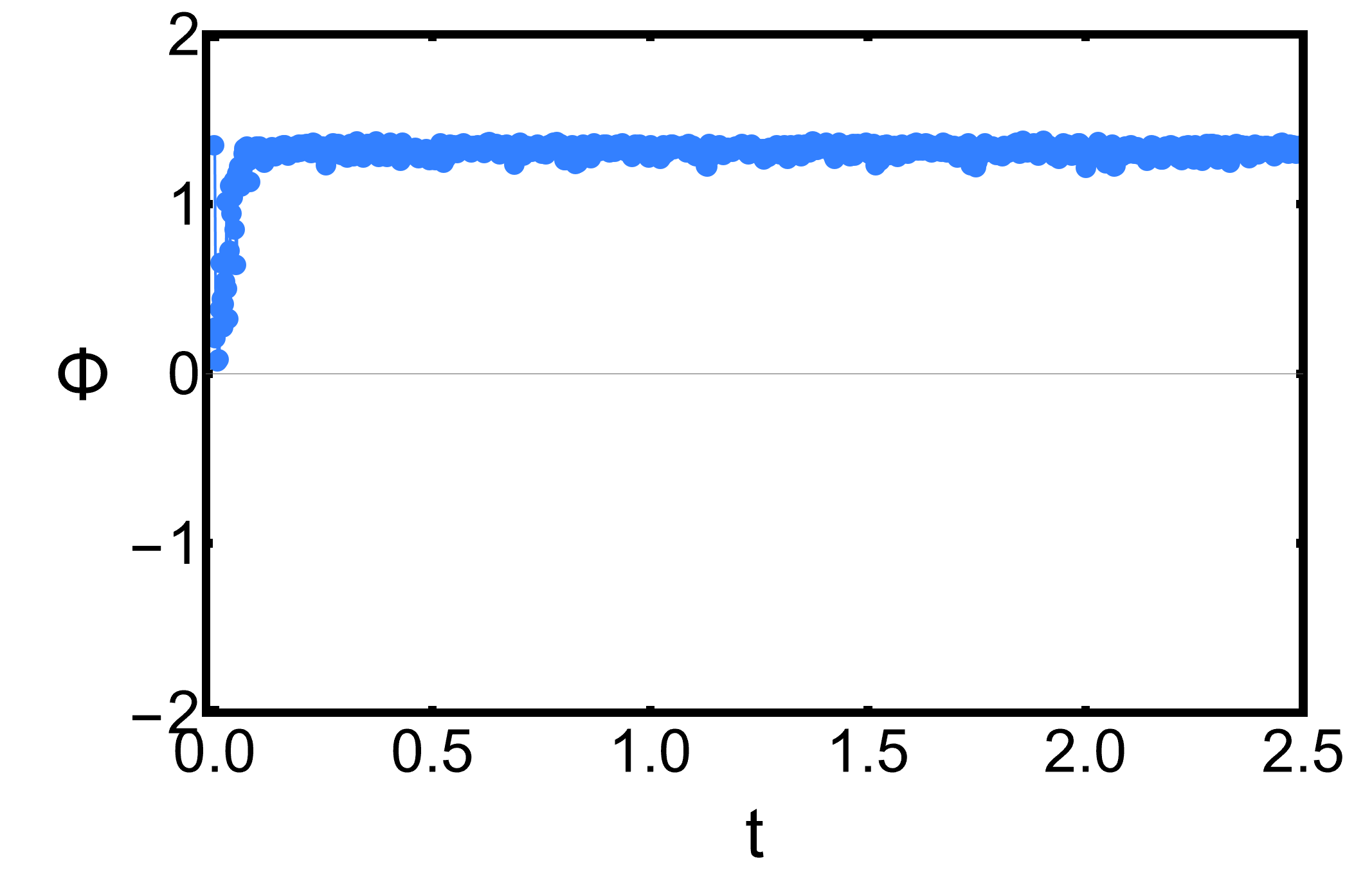}
            \caption{Non-observation of coherent cascades in Type I systems with $\atwo = \aone = 0$, $\bone = 1$, $\beta_2=0$ and $(\beta_0,\beta_1)\in [-5,5]\times[-5,5]$, stochastic couplings and random initial data. {\em First and second rows:} Evolution of the energy spectrum and a sequence of phase distributions corresponding to a simulation with $(\beta_0,\beta_1)=(-3.5,-2.5)$. {\em Third row:} Evolution of the exponent $\rho(t)$, the asymptotic power-law exponent $\gamma(t)$, and the deviation from a linear phase relation, $\Phi(t)$.}
		\label{fig:stochastic_map_Type_1}
	\end{figure}
	
	\newpage
	
	\subsection{Fully random Hamiltonian systems}
	
	Resonant Hamiltonian systems present two elements in their structure that can be a source of order or disorder in the dynamics: the nonlinear couplings $C_{nmkj}$ and the dispersion relation $\omega_n$ that determines the resonant interactions through the condition $\w_n+\w_m = \w_k + \w_j$. Our models present random nonlinear couplings but an organized web of resonances through the linear dispersion relation that leads to the condition $n+m=k+j$. This raises the question of whether such structure is entirely responsible for the emergence of coherent dynamics, namely, whether this alone would be sufficient. Numerical simulations provide evidence to rule out this option in the class of systems (\ref{eq:Resonant_Equation}), enhancing the role played by the structured subset of couplings in the organization of dynamics. 
	
	We have performed simulations of Hamiltonian systems with no structured subset of couplings: all couplings are independent identically distributed Gaussian random variables, up to symmetries,
	\beq
	C_{nmkj} \sim \mathcal{N}(0, 1).
	\eeq
	As shown in Figure~\ref{fig:fully_random_Hamiltonians}. These systems show no clear traces of phase organization when evolved from random initial data (\ref{eq:random_initial_data}) and rapidly lose their initial phase alignment when evolved from coherent data. It suggests that coherence is not a generic feature of fully resonant Hamiltonian systems, but instead requires the presence of a certain structured subset of couplings. A slow migration of energy to high modes is observed, as predicted in \cite{EvninMelonicTurbulence}. However, whether such transfer is maintained in time and ultimately develops a power law is an open question. Such cascades would be driven by noncoherent dynamics, contrary to the ones we have reported in this work.

	\begin{figure}
		\centering
		\includegraphics[width = 5.6cm]{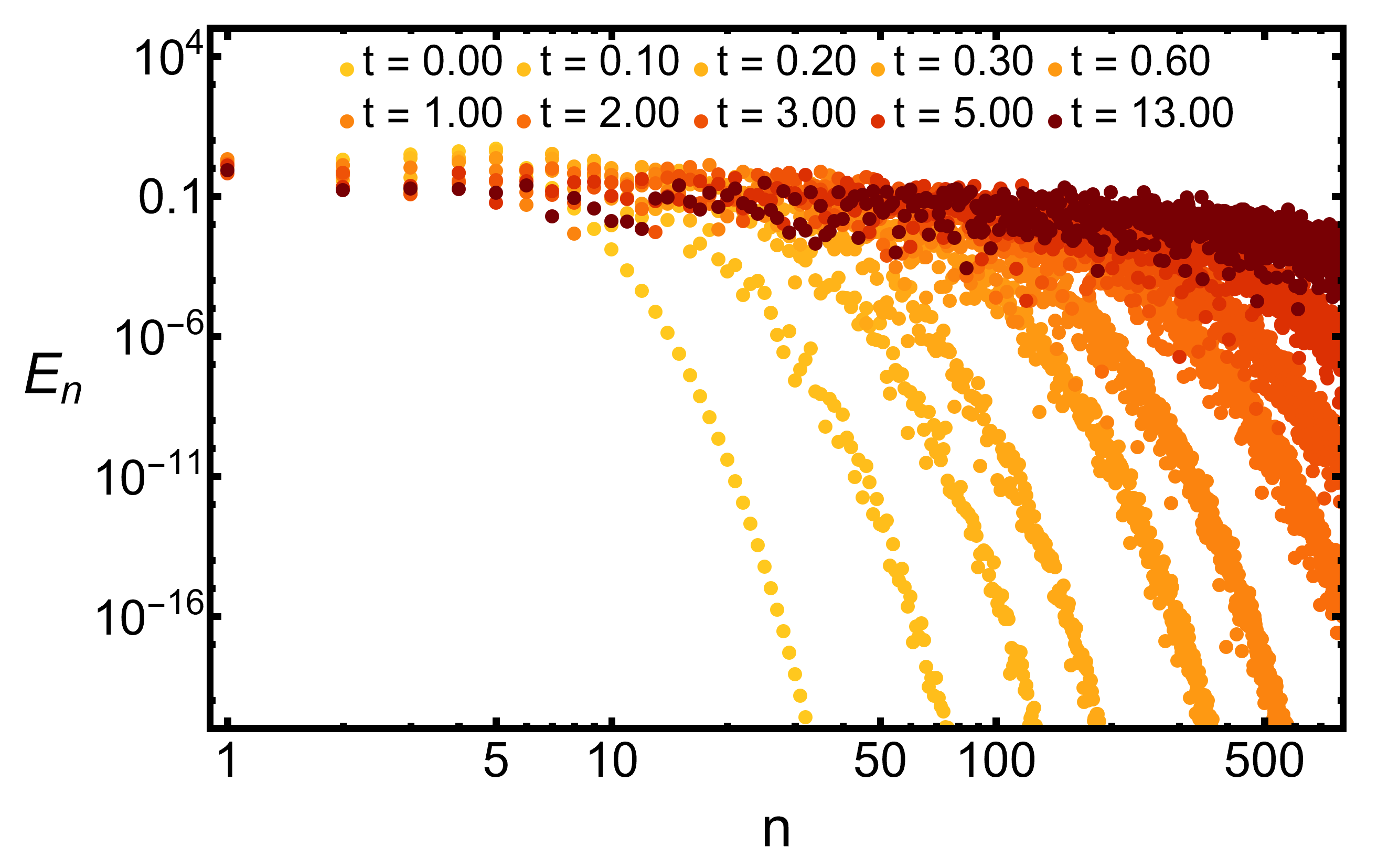}
		\includegraphics[width = 5.6cm]{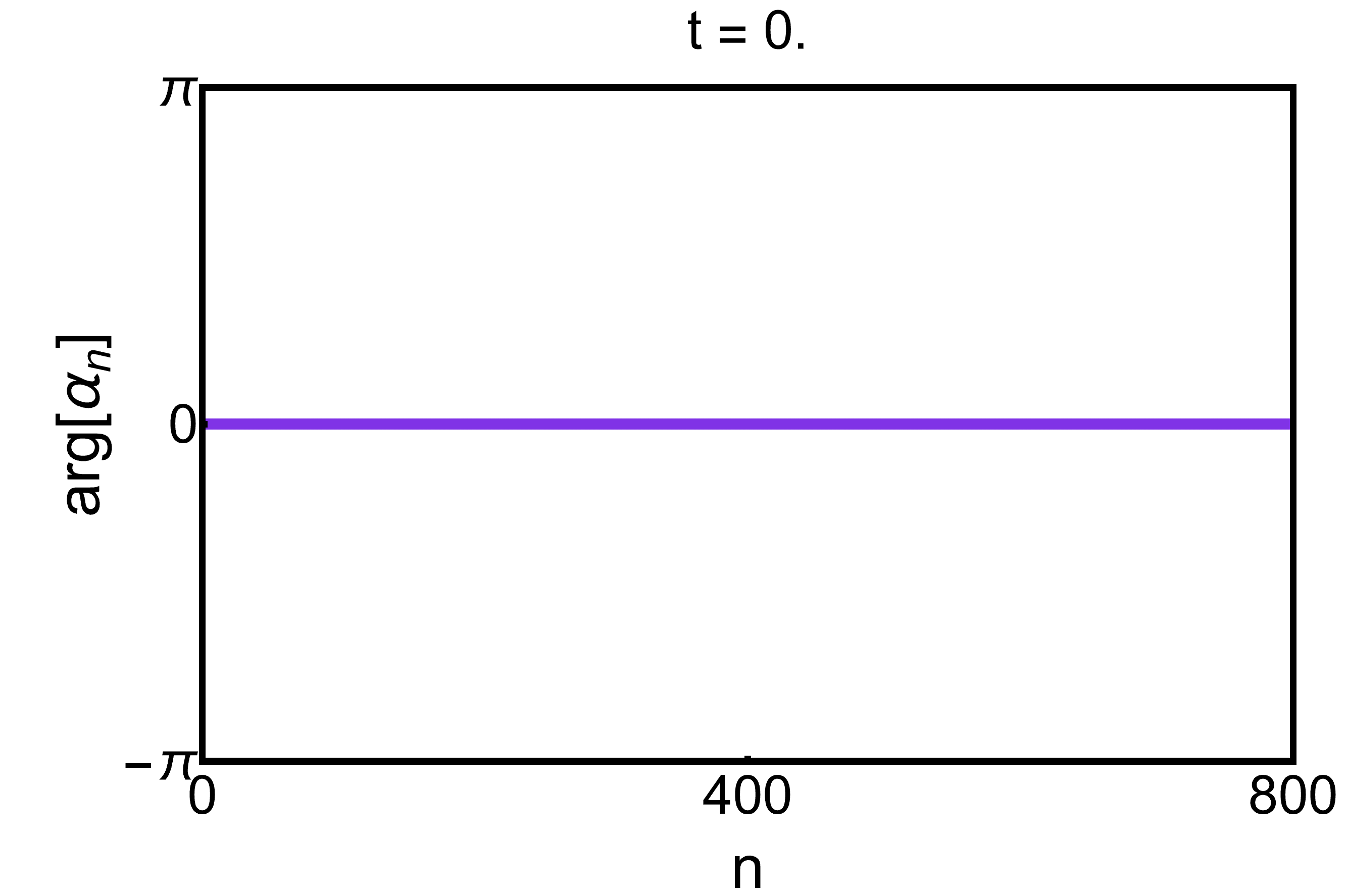}
		\includegraphics[width = 5.6cm]{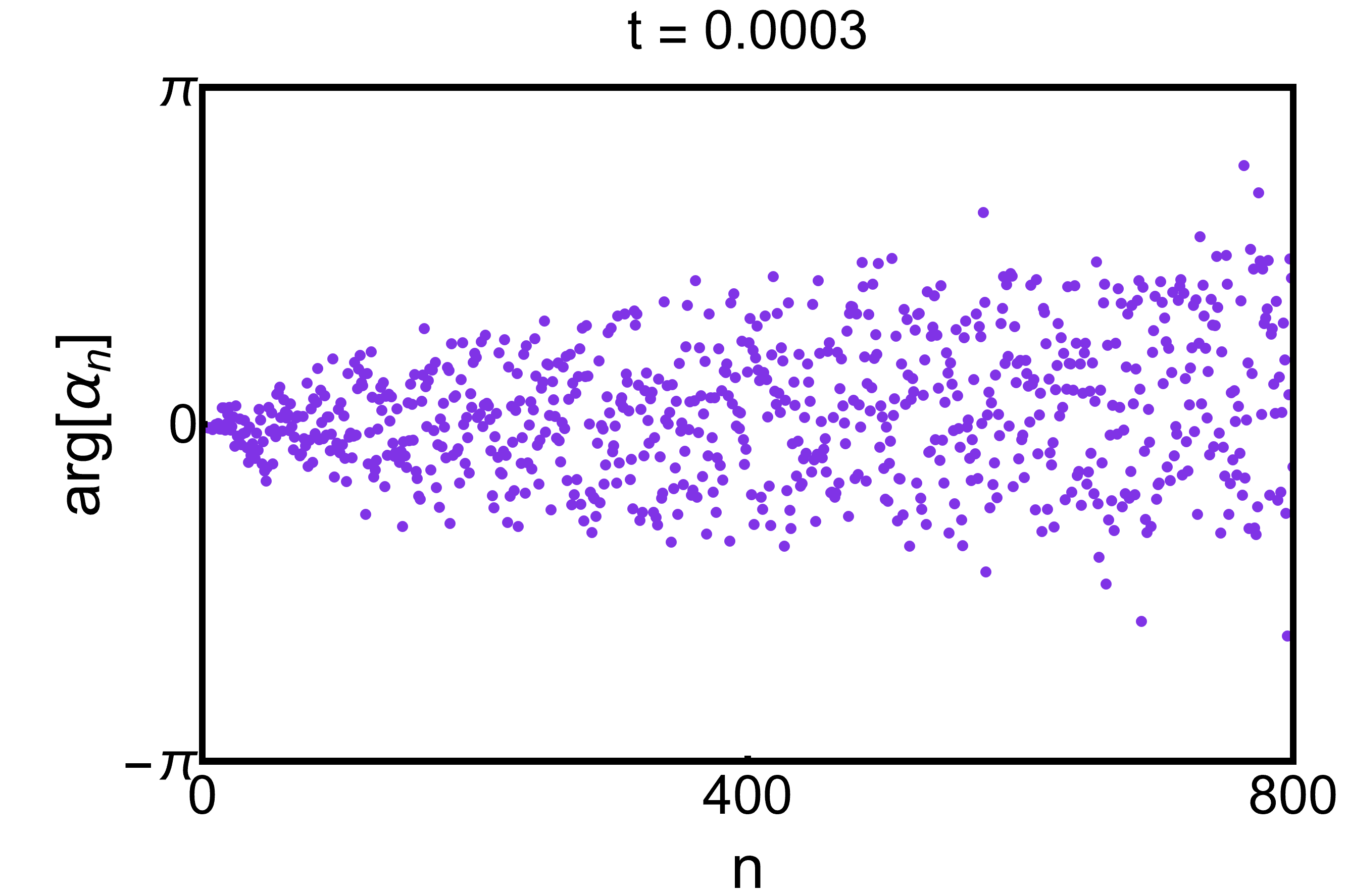}
		
		\vspace{0.1cm}
		\includegraphics[width = 5.6cm]{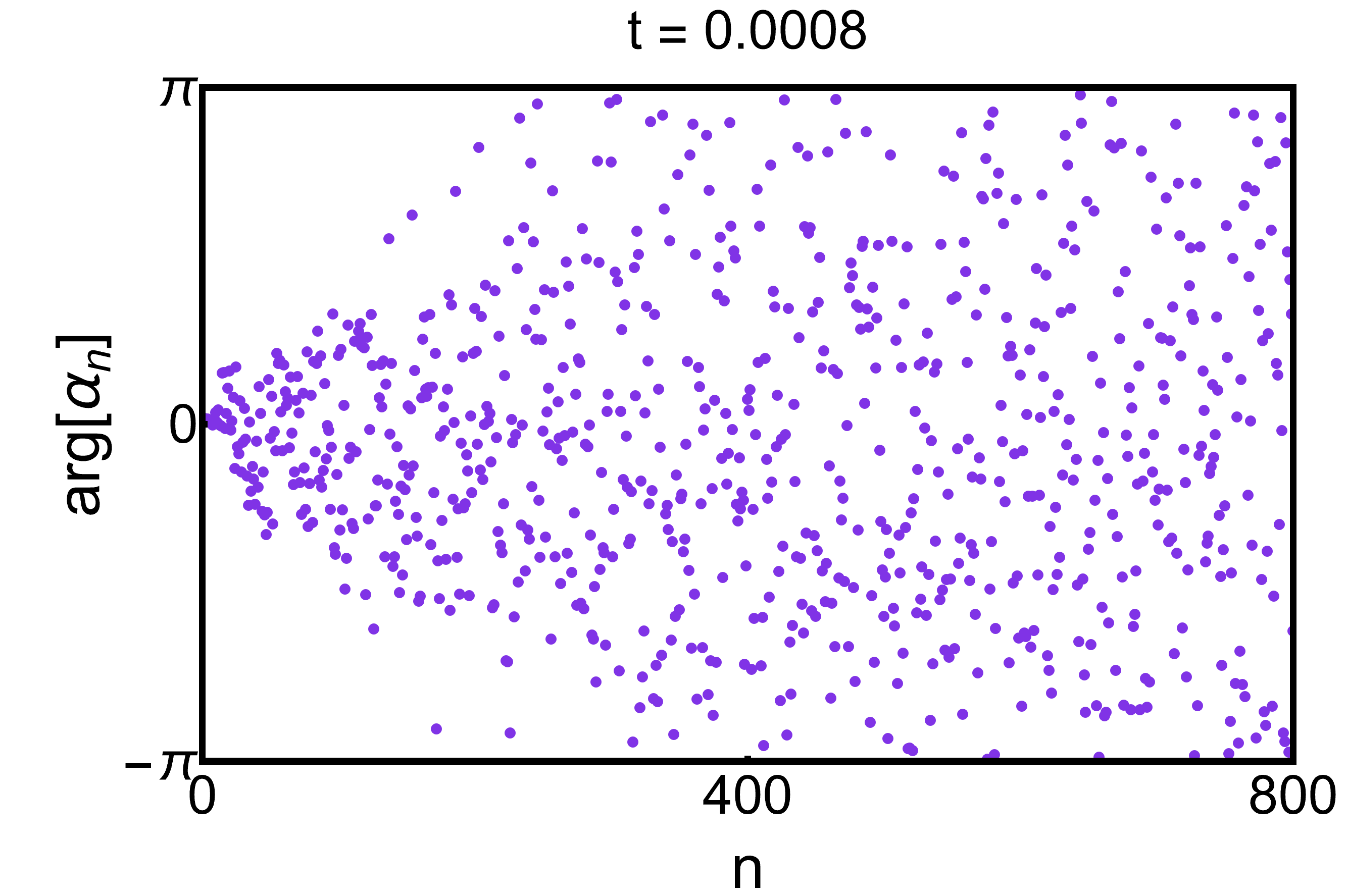}
		\includegraphics[width = 5.6cm]{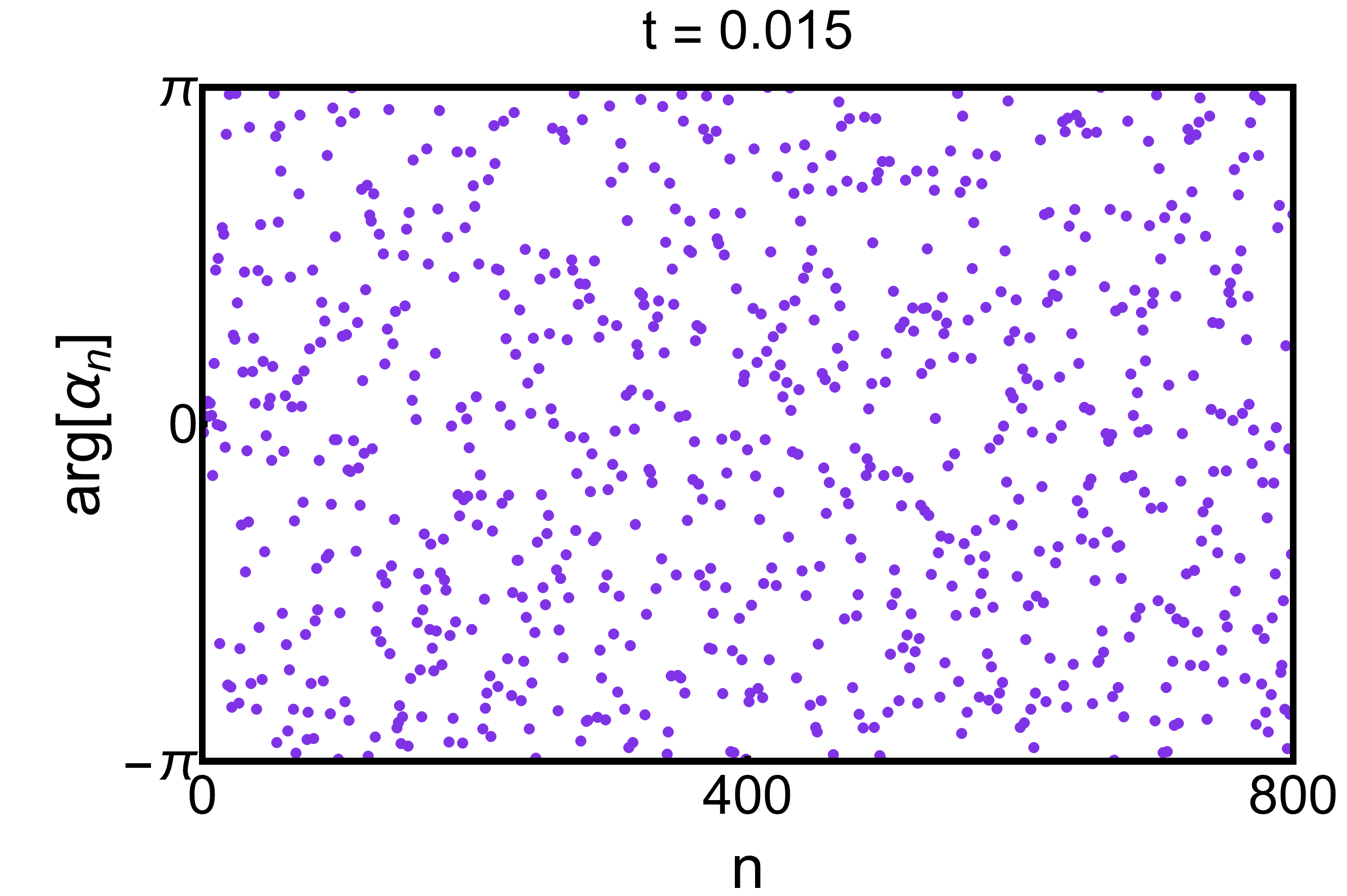}
		\includegraphics[width = 5.6cm]{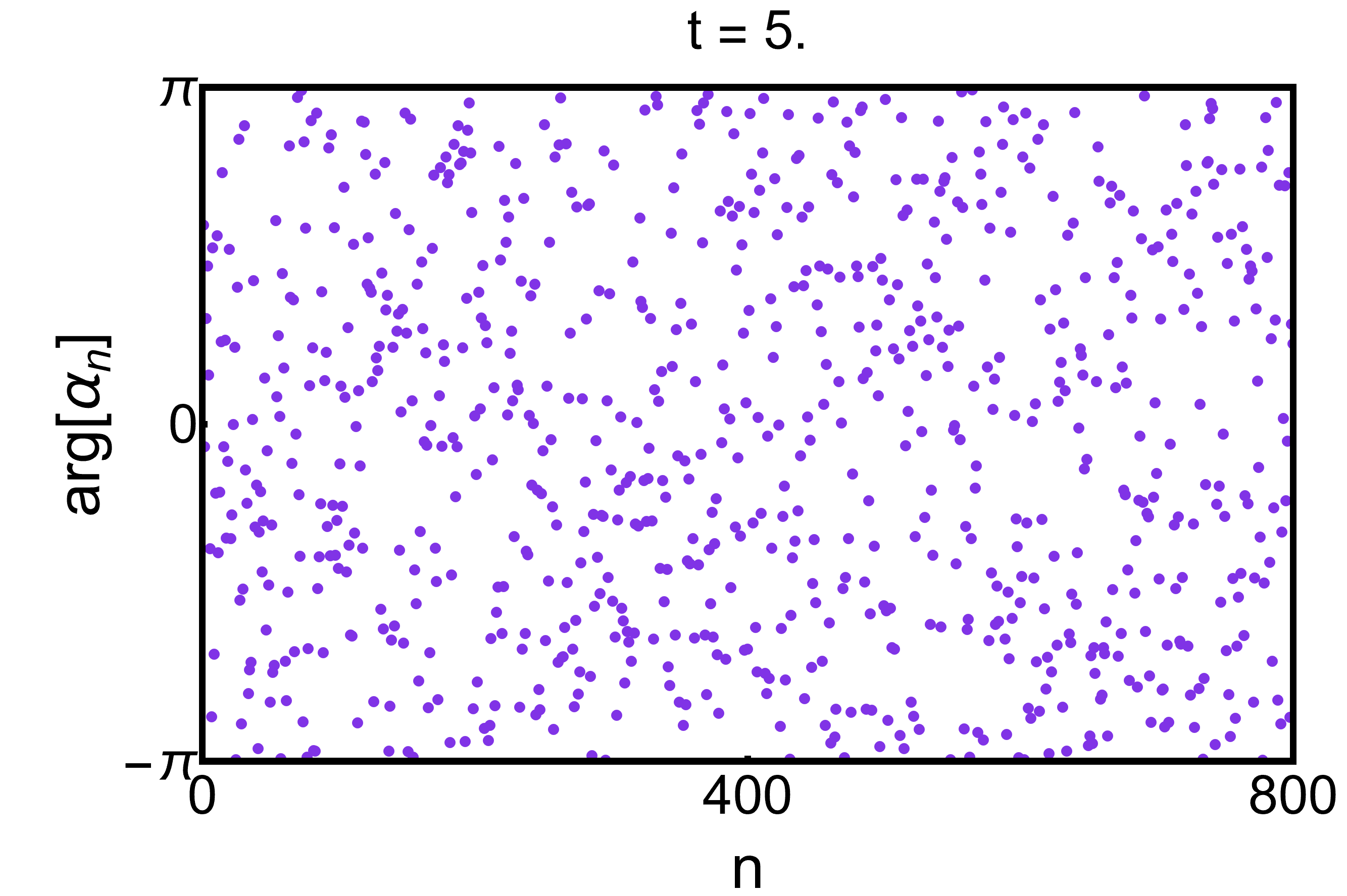}
		
		\vspace{0.2cm}
		
		\includegraphics[width = 5.6cm]{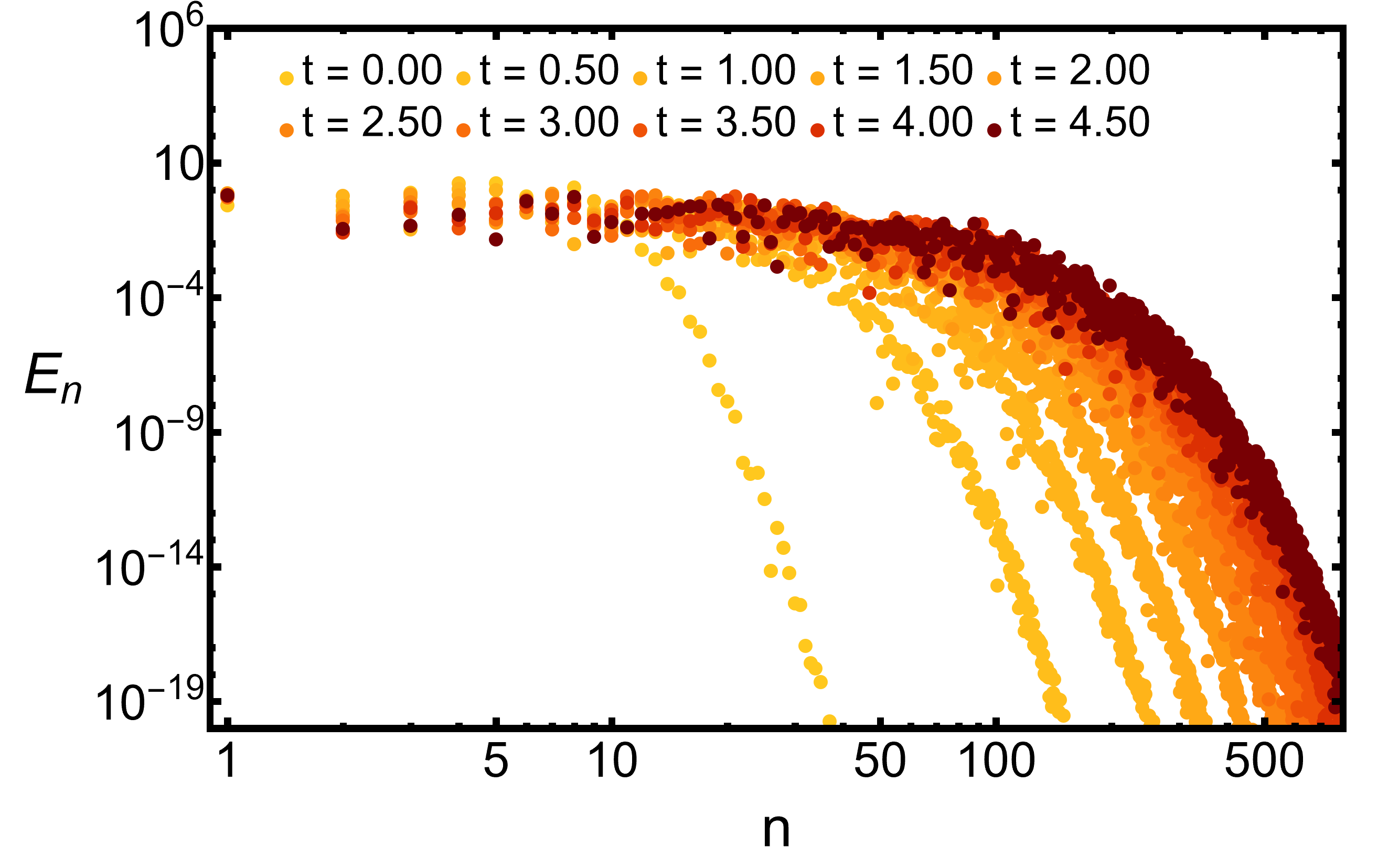}
		\includegraphics[width = 5.6cm]{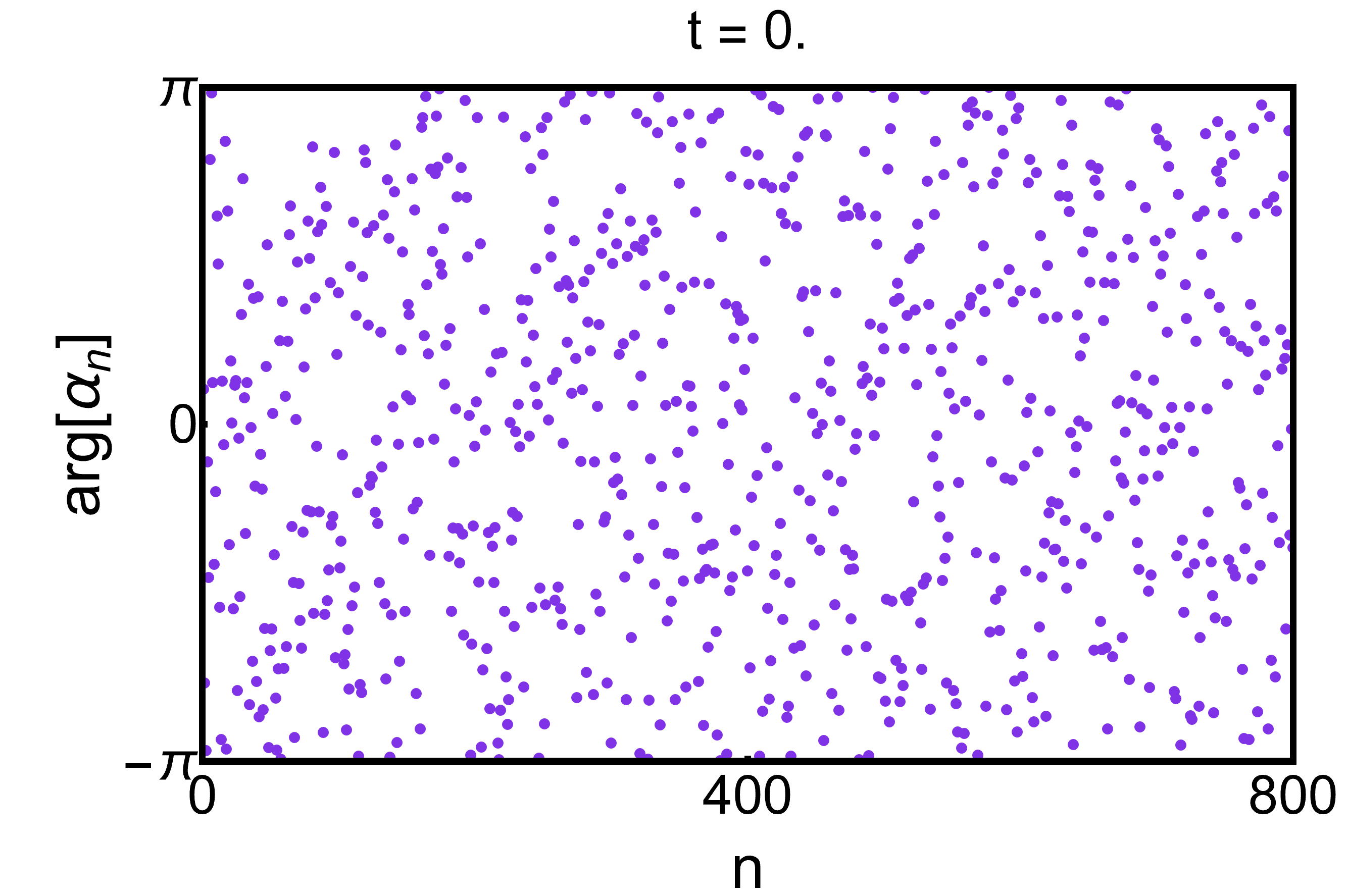}
		\includegraphics[width = 5.6cm]{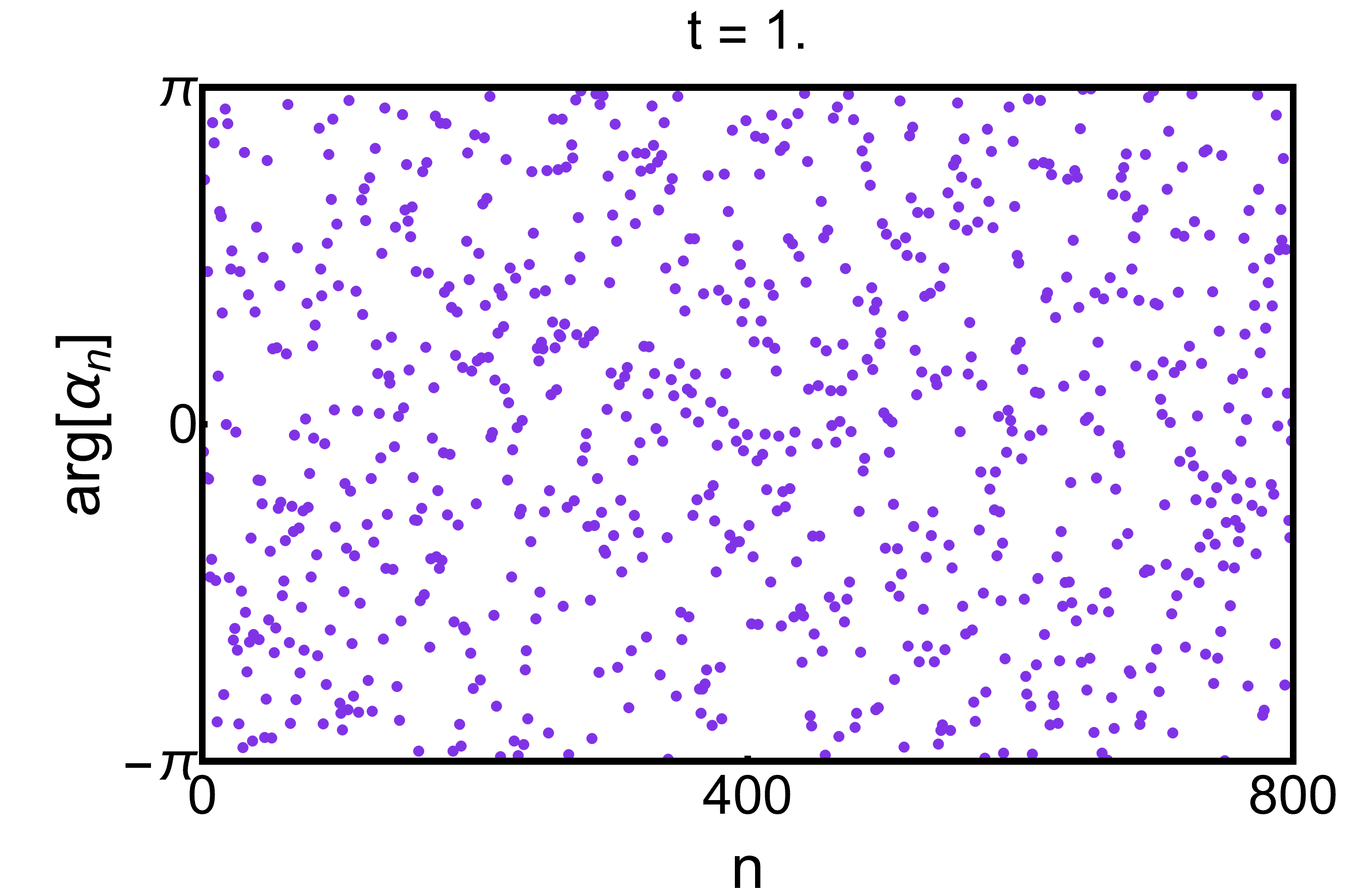}
		
		\vspace{0.1cm}
		\includegraphics[width = 5.6cm]{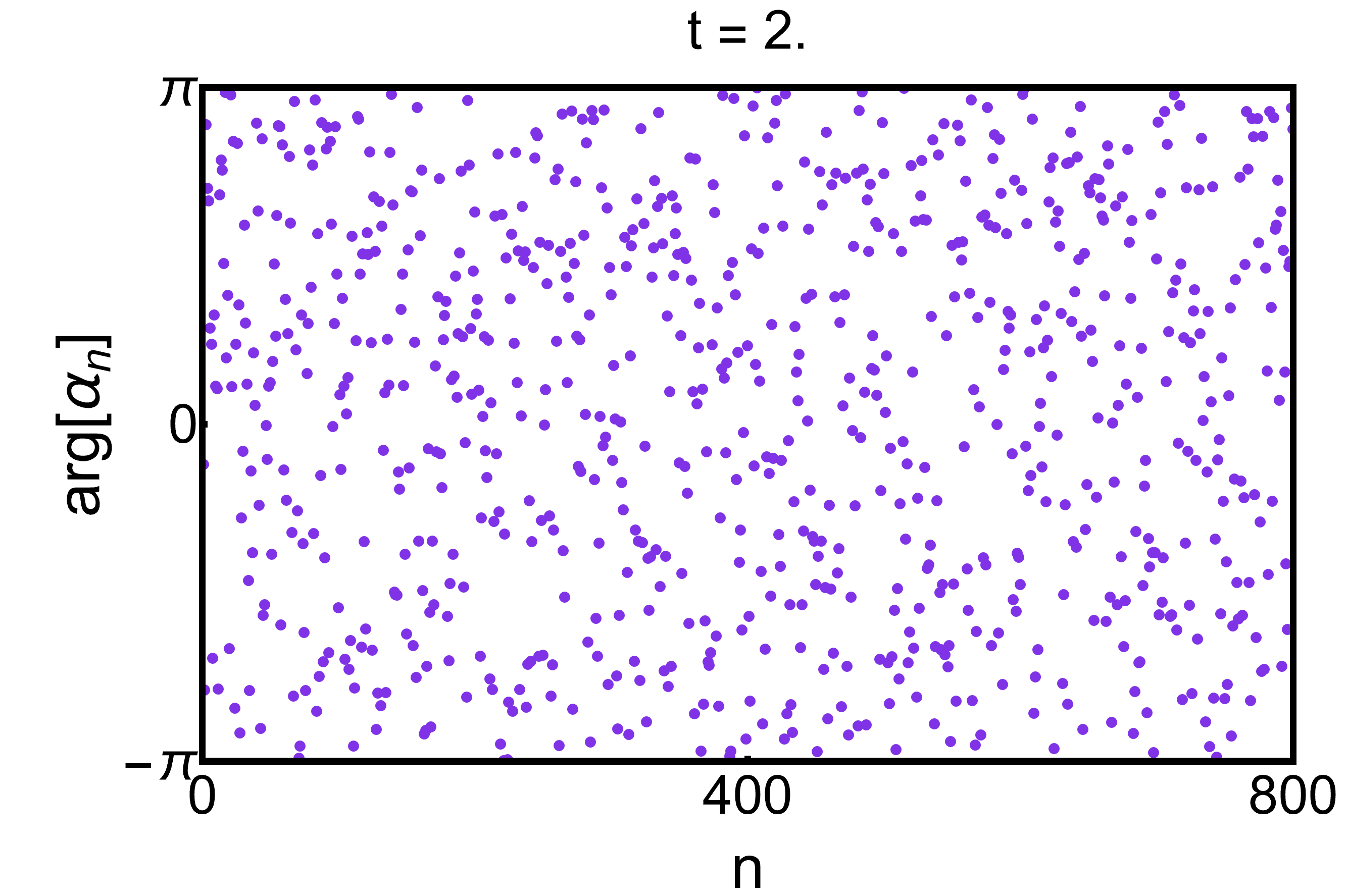}
		\includegraphics[width = 5.6cm]{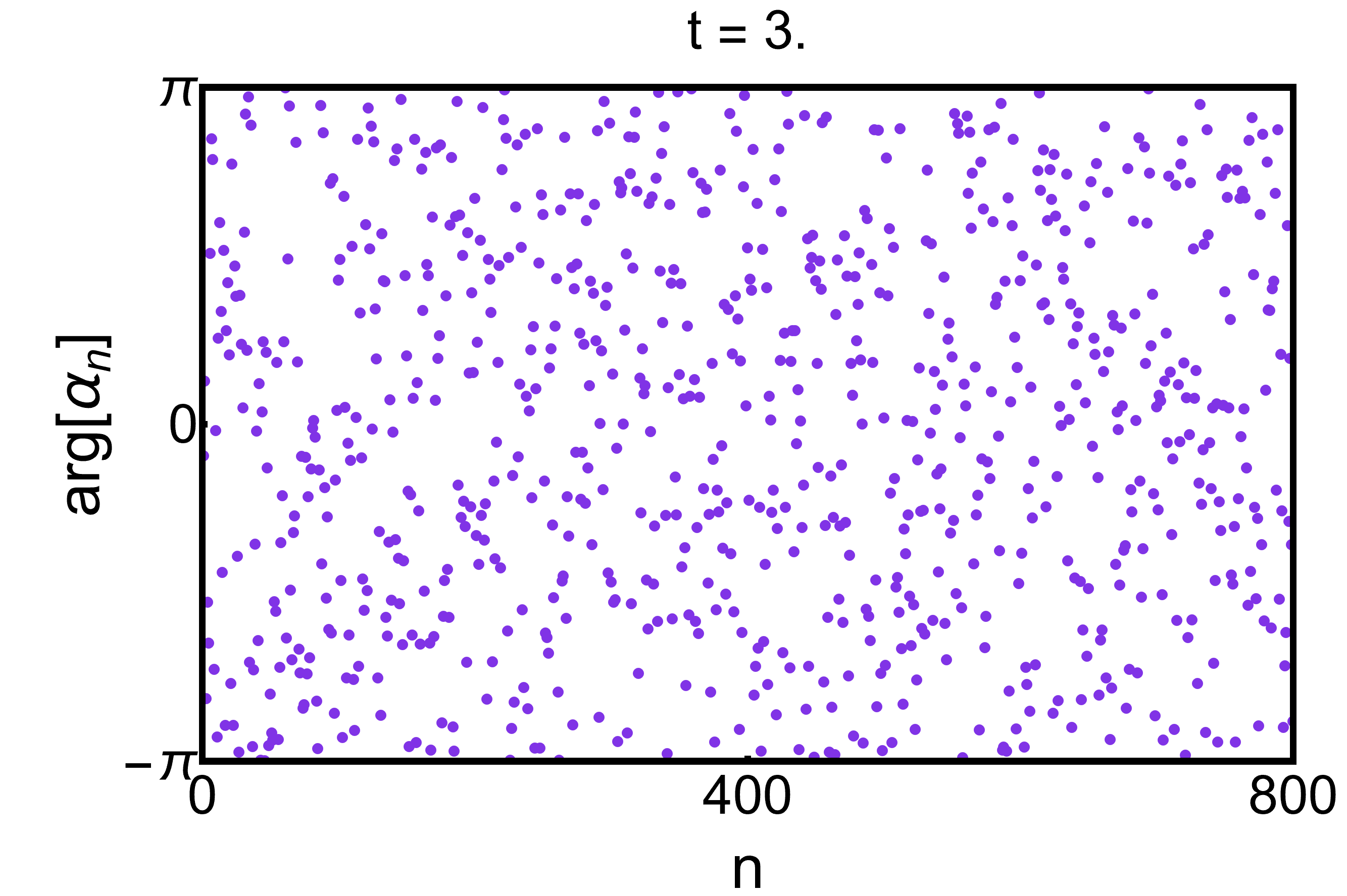}
		\includegraphics[width = 5.6cm]{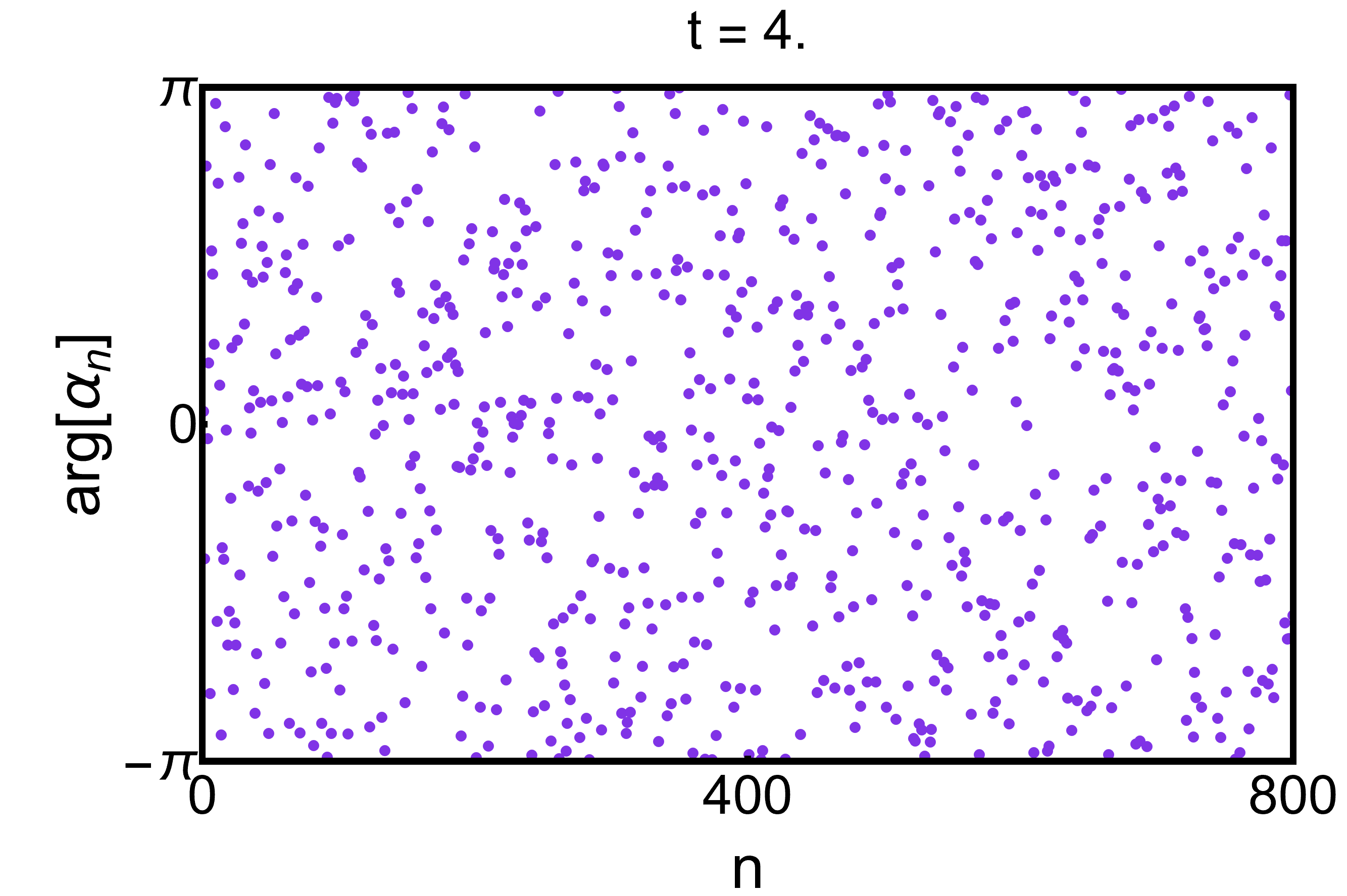}
		\caption{Evolution of a fully random Hamiltonian system initialized with coherent data (two first rows) and random data (two last rows). The first plot shows the evolution of the energy spectrum, while the other plots the distribution of the phases at different times. No emergence of coherence is appreciated.}
		\label{fig:fully_random_Hamiltonians}
	\end{figure}
	

	\subsection{Other structured subsets of couplings}
    \label{subsec:other_minimally_structured_Hamiltonians}

	Simulations of fully random systems suggest that the structured subset of couplings we have introduced in the previous models plays an essential role in the emergence of coherence. A natural question is then whether its specific algebraic structure, which leads to the invariant manifold, is also necessary. To address this point, we consider Hamiltonian systems that retain a similar structured subset of mode interactions as in Section~\ref{subsec:numerical_simulations_minimally_structured_systems}, but whose nonlinear couplings no longer satisfy the algebraic conditions of Proposition~\ref{prop:existence_invariant_manifold}. In particular, we take
	\beq\label{eqFig10}
	C_{nmkj} = \begin{cases}
		\ \ \ \sqrt{n+m+1} & \text{for } \quad   n m k j = 0,\\
		\xi_{nmkj} \ \sqrt{n+m+1} & \text{for } \quad nmkj\neq 0,
	\end{cases}
	\eeq
	where $\xi_{nmkj}$ are the random variables introduced in Section~\ref{subsec:numerical_simulations_minimally_structured_systems}. By construction, these coefficients do not satisfy the hypotheses of Proposition~\ref{prop:existence_invariant_manifold}, and the invariant manifold studied previously is absent. Nevertheless, as shown in Figure~\ref{fig:other_minimally_structured_Hamiltonians}, the energy spectra approach a power-law profile, suggesting an efficient transfer of energy toward high modes, while phase locking is preserved for coherent initial data, and it also emerges and remains robust for random initial data. In contrast to the phase-incoherent behavior observed for systems with fully random nonlinear couplings, these results further emphasize the organizing role played by even a small fraction of structured nonlinear couplings. Although only a simple family of examples was considered, the simulations suggest that coherent cascades may persist in a considerably broader class of random Hamiltonian systems than those covered by the invariant-manifold analysis developed in this work.
	\begin{figure}
		\centering
		\includegraphics[width = 5.6cm]{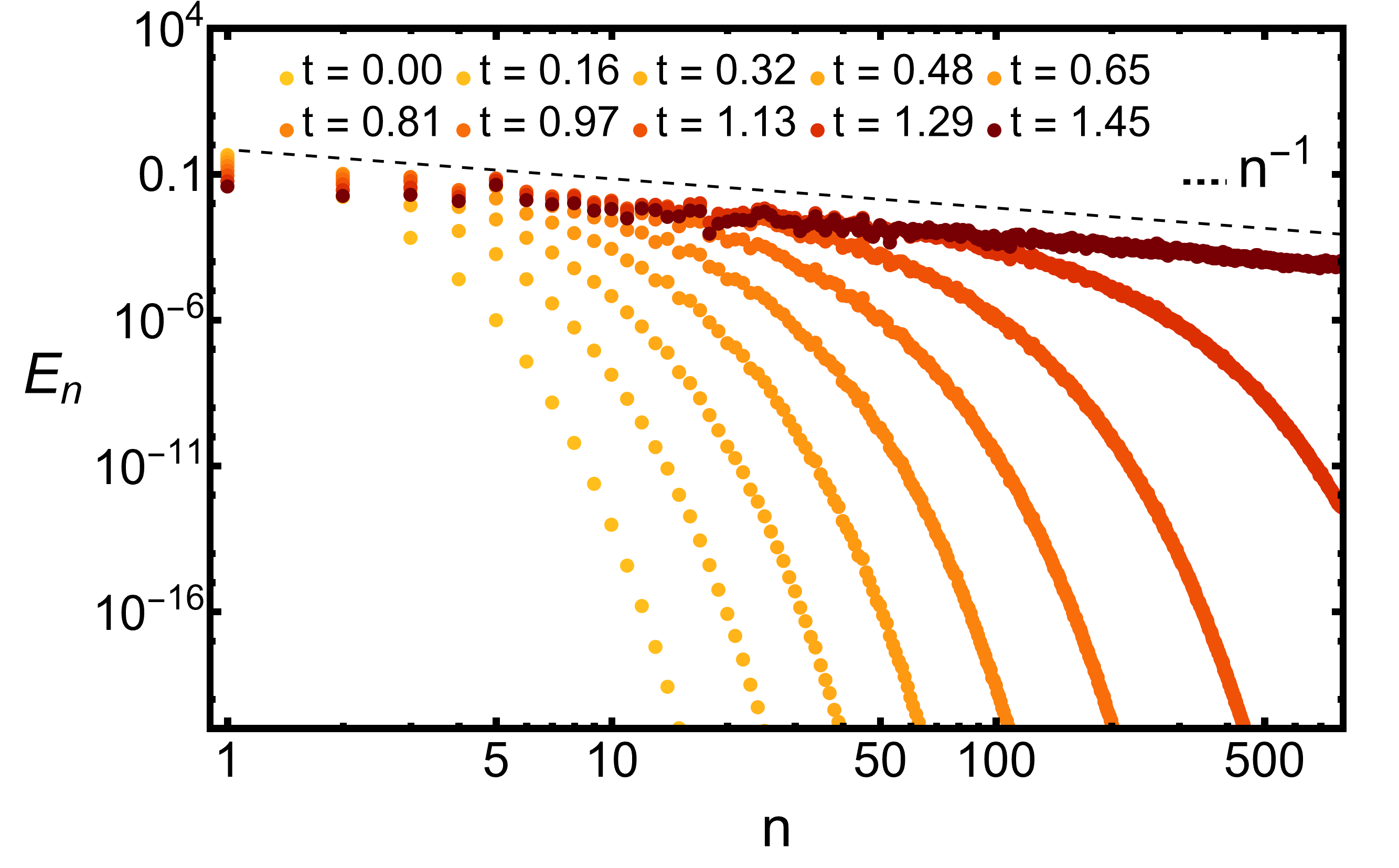}
		\includegraphics[width = 5.6cm]{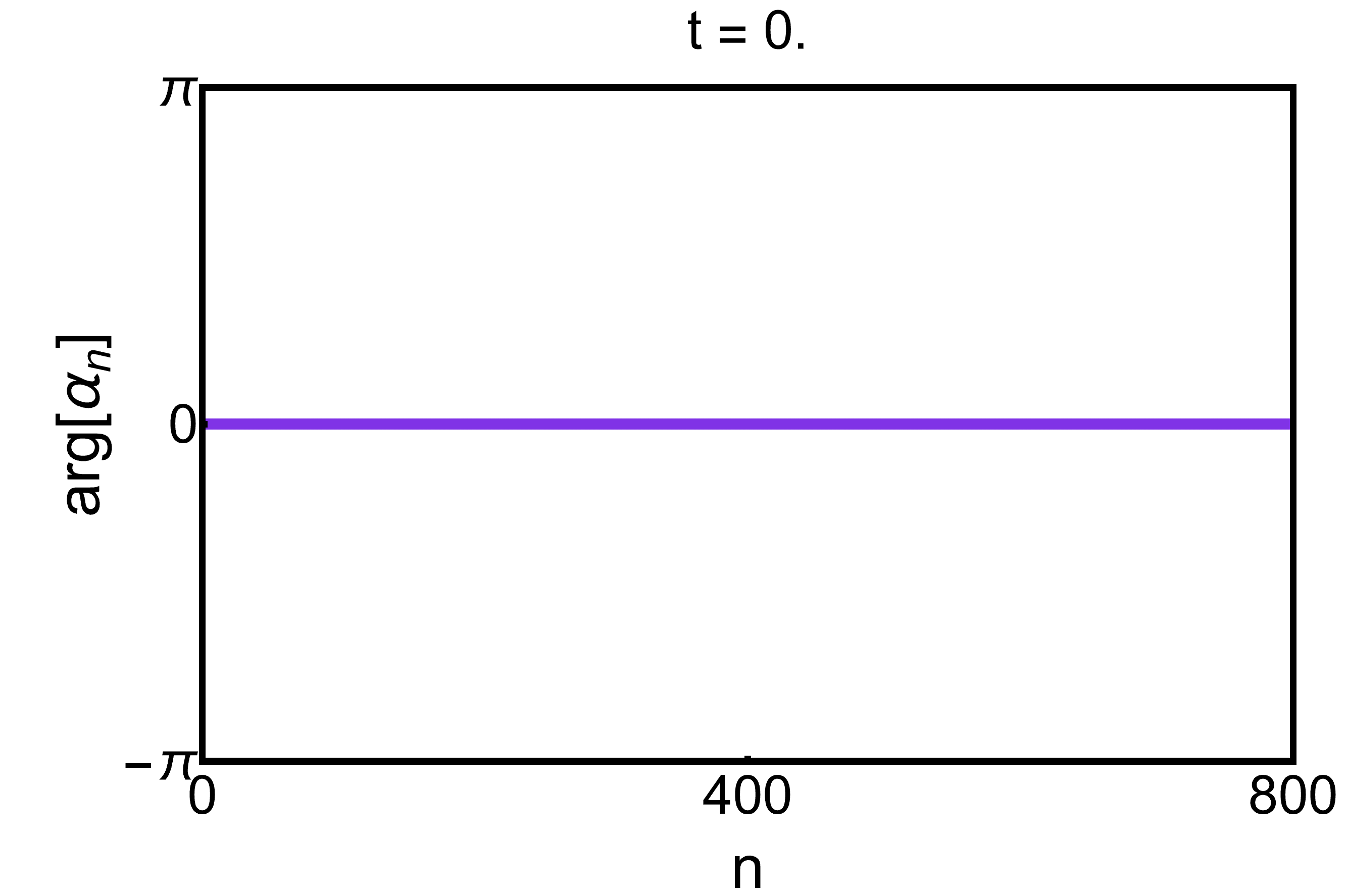}
		\includegraphics[width = 5.6cm]{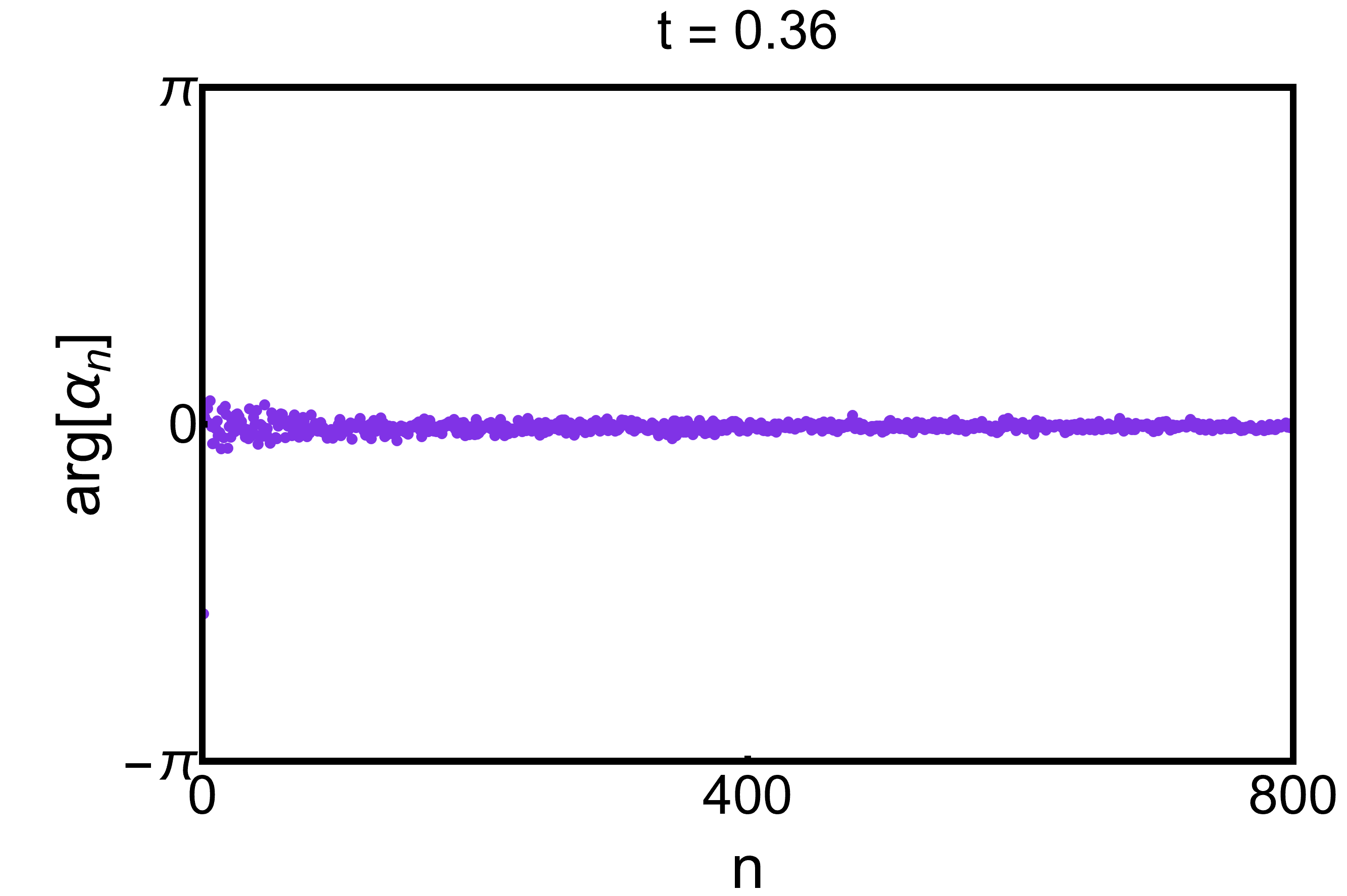}
		
		\vspace{0.1cm}
		\includegraphics[width = 5.6cm]{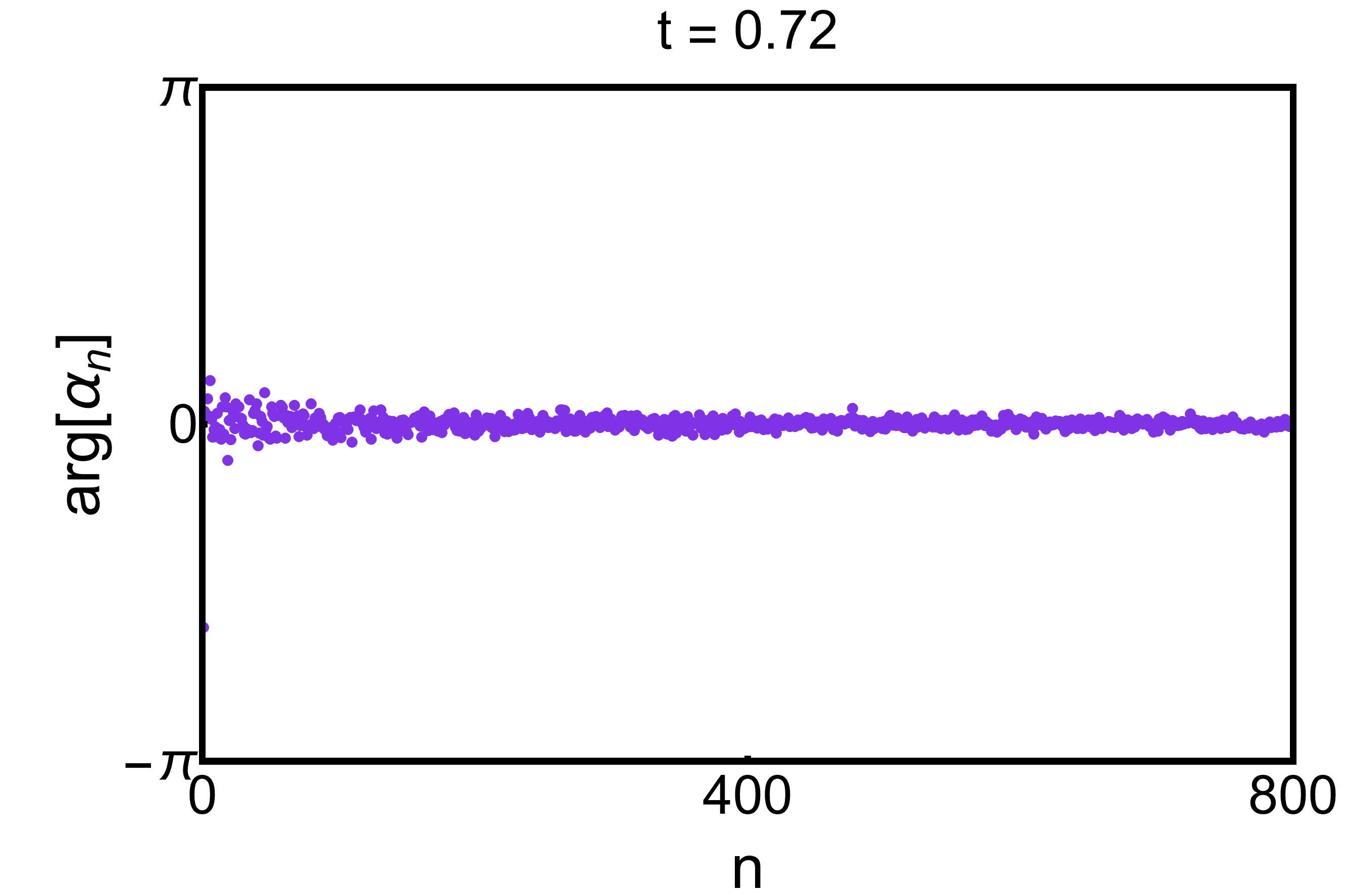}
		\includegraphics[width = 5.6cm]{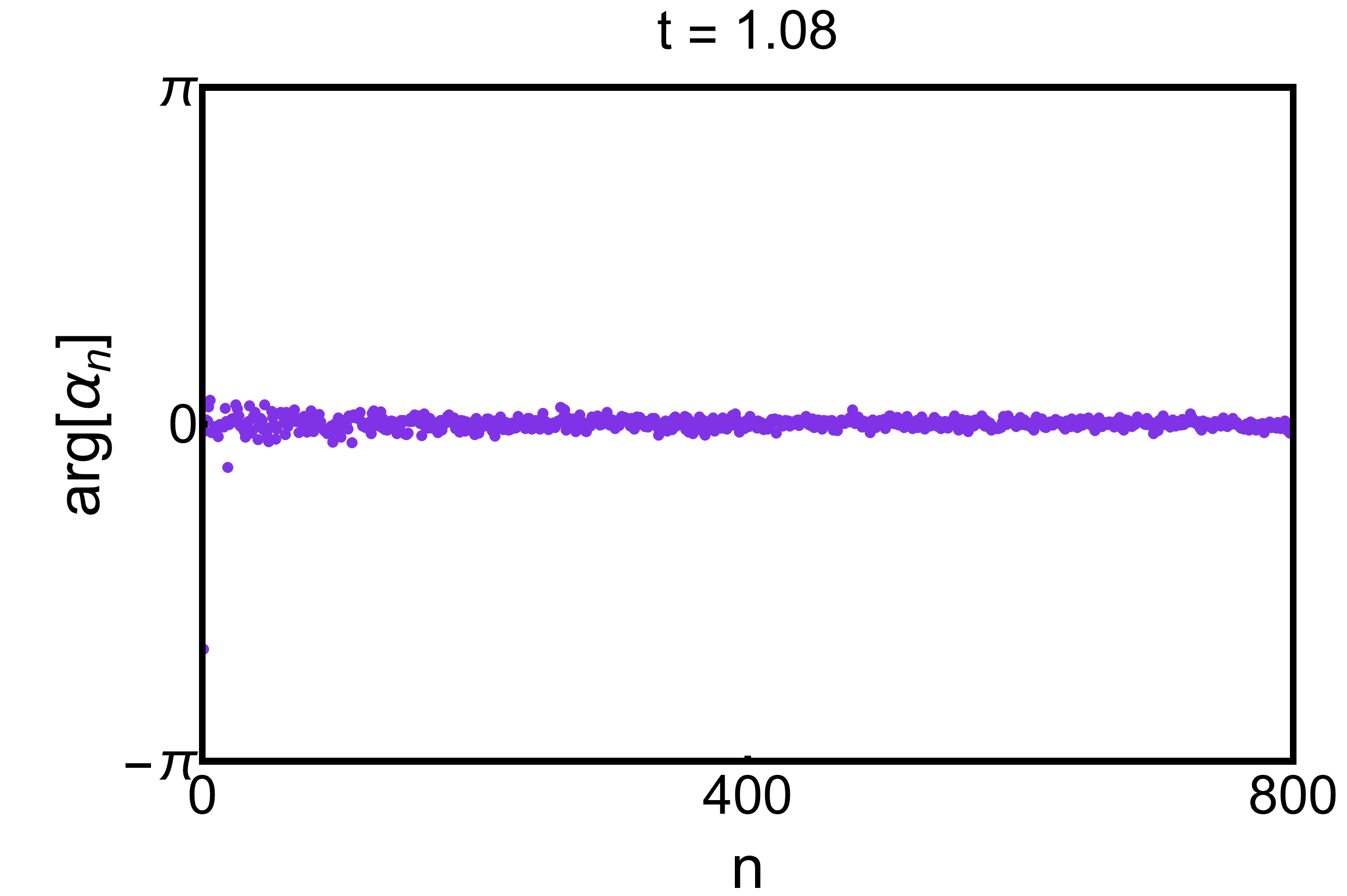}
		\includegraphics[width = 5.6cm]{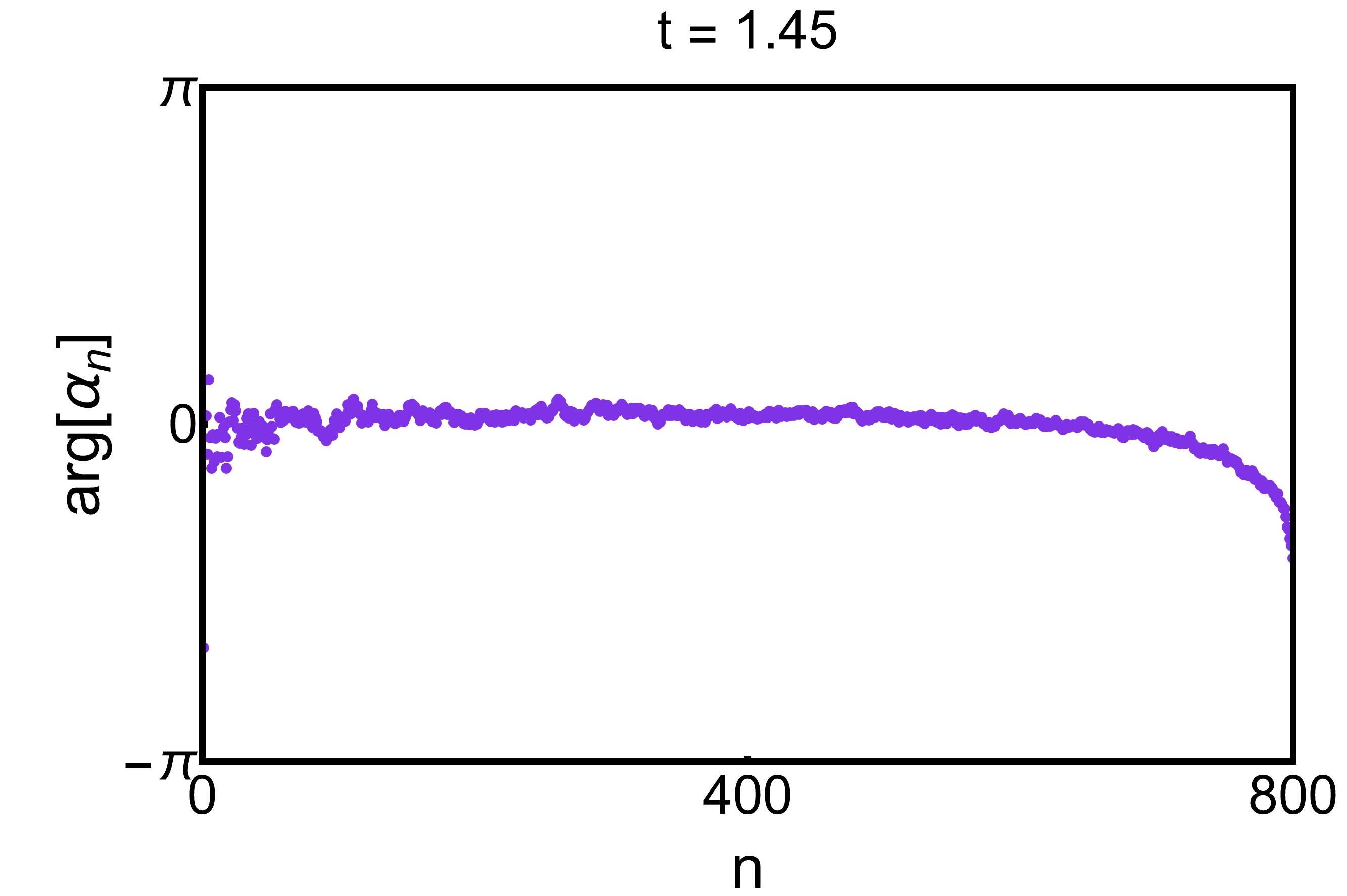}
		
		\vspace{0.2cm}
		
		\includegraphics[width = 5.6cm]{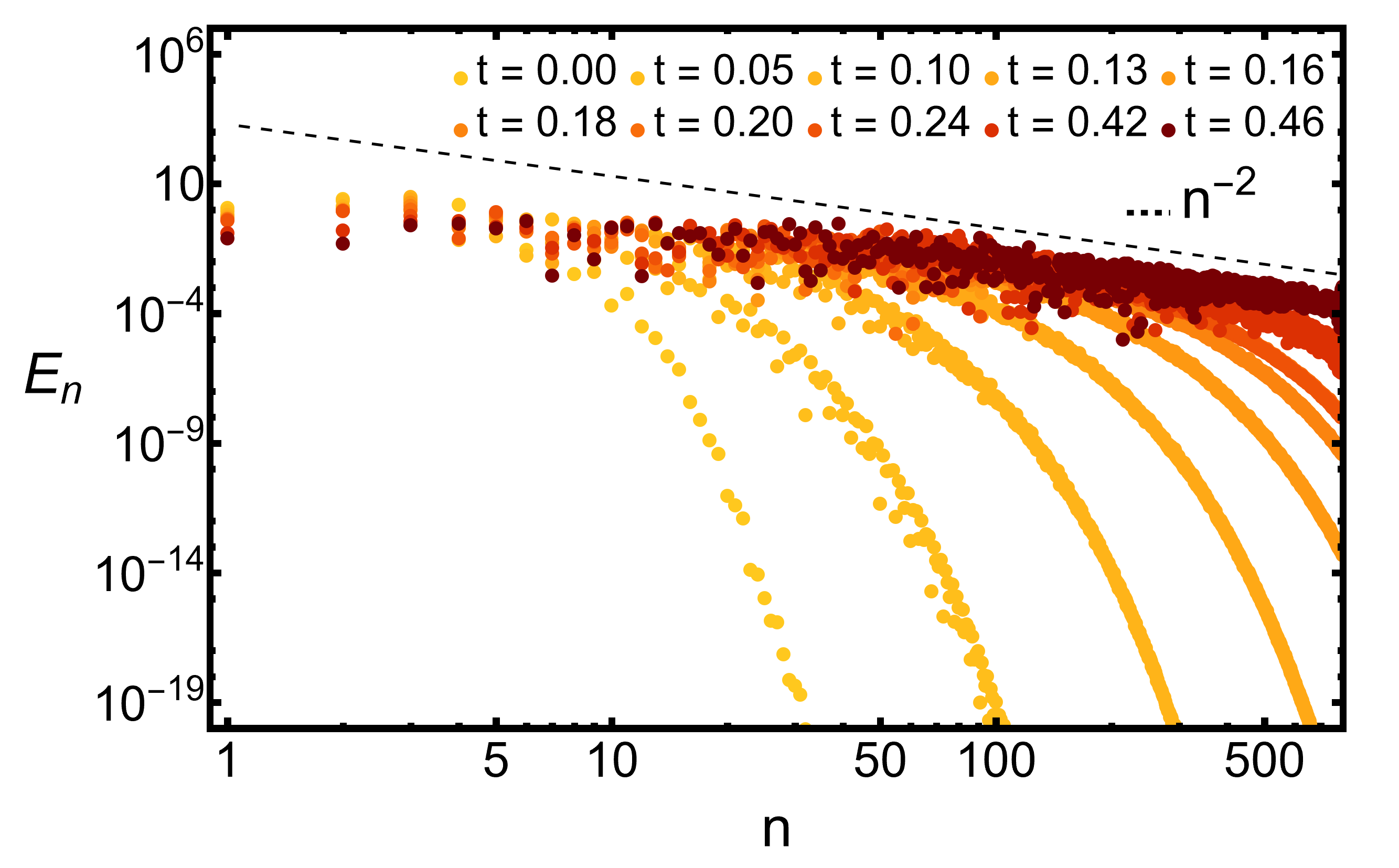}
		\includegraphics[width = 5.6cm]{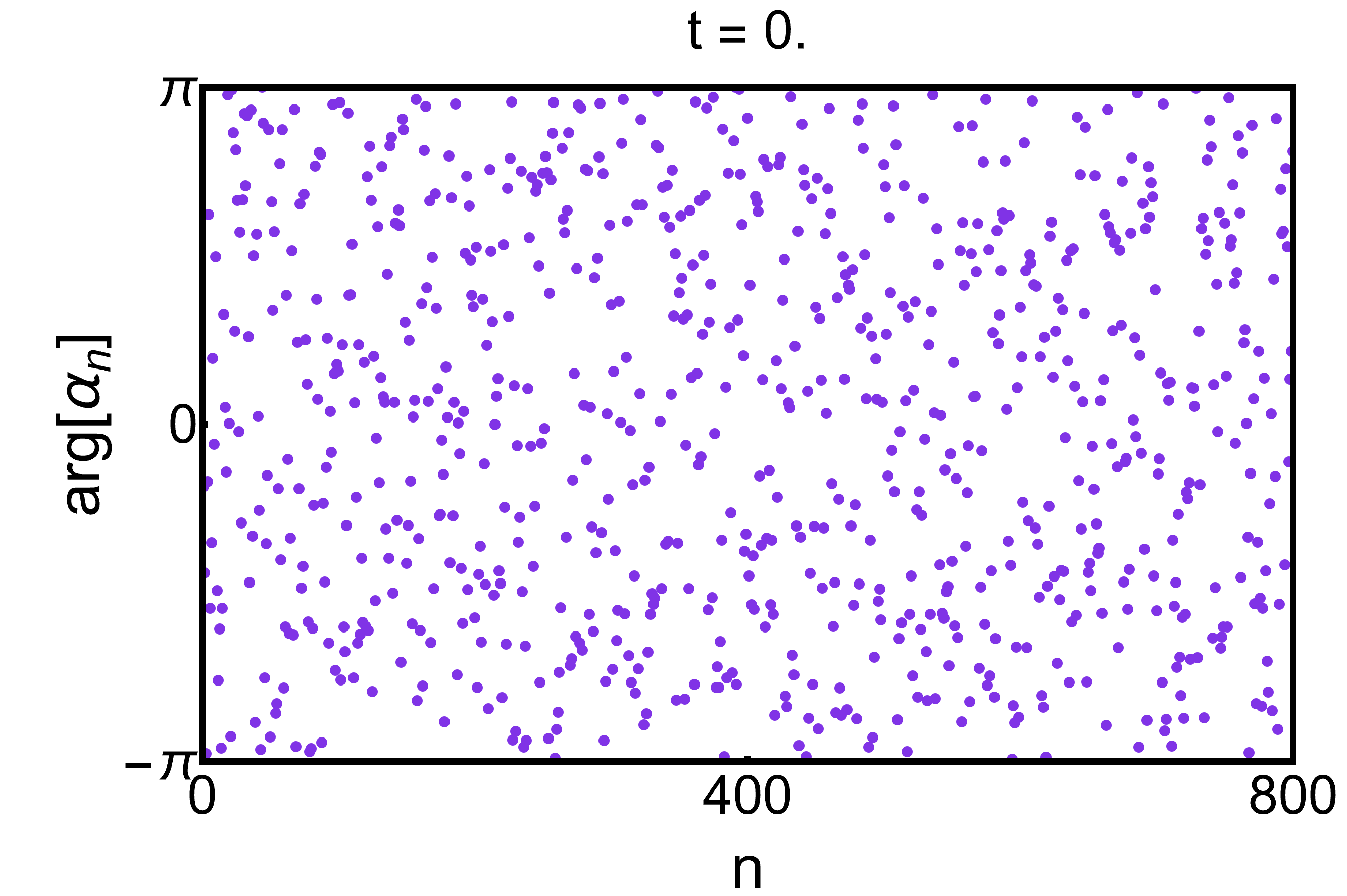}
		\includegraphics[width = 5.6cm]{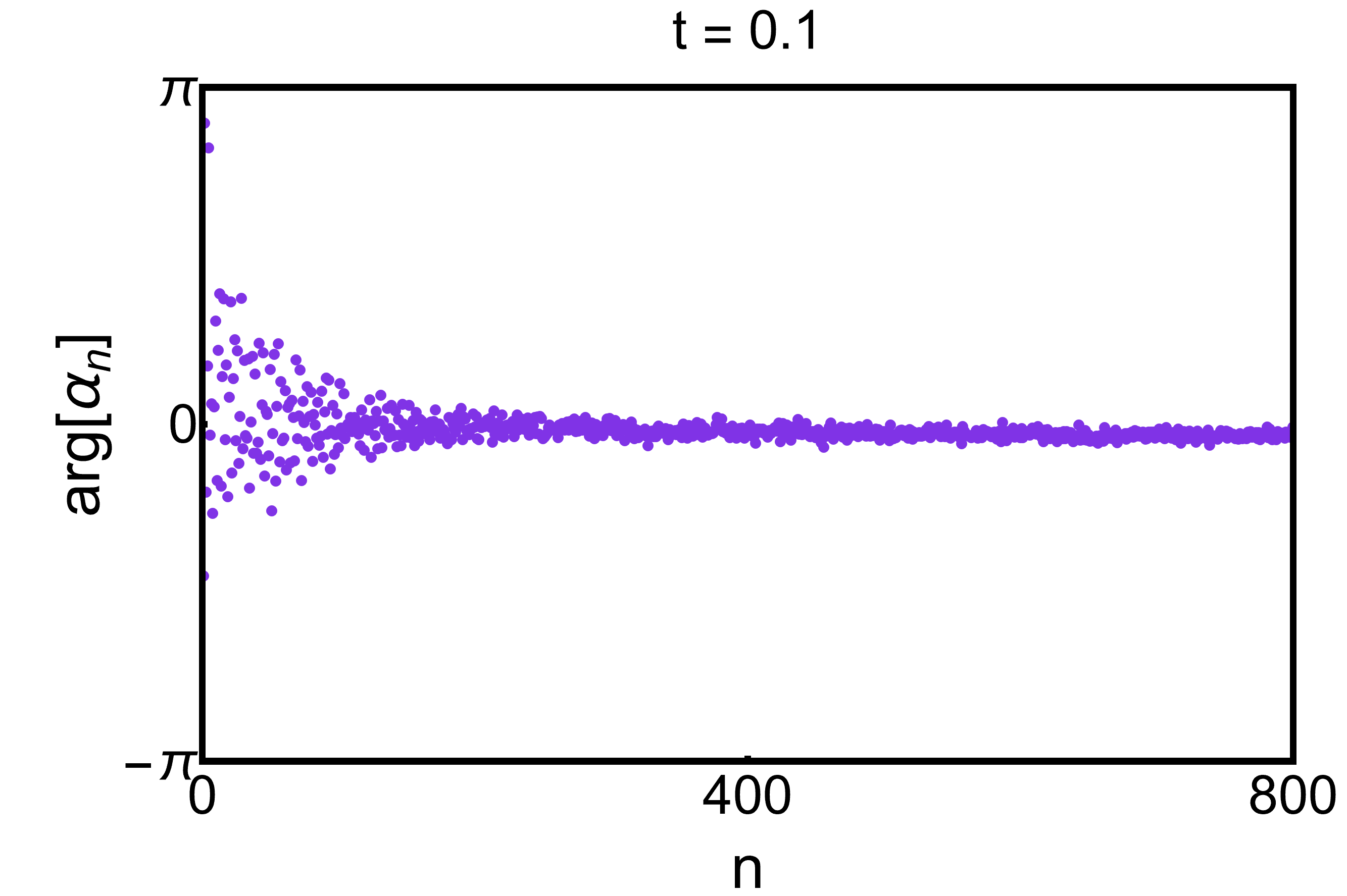}
		
		\vspace{0.1cm}
		\includegraphics[width = 5.6cm]{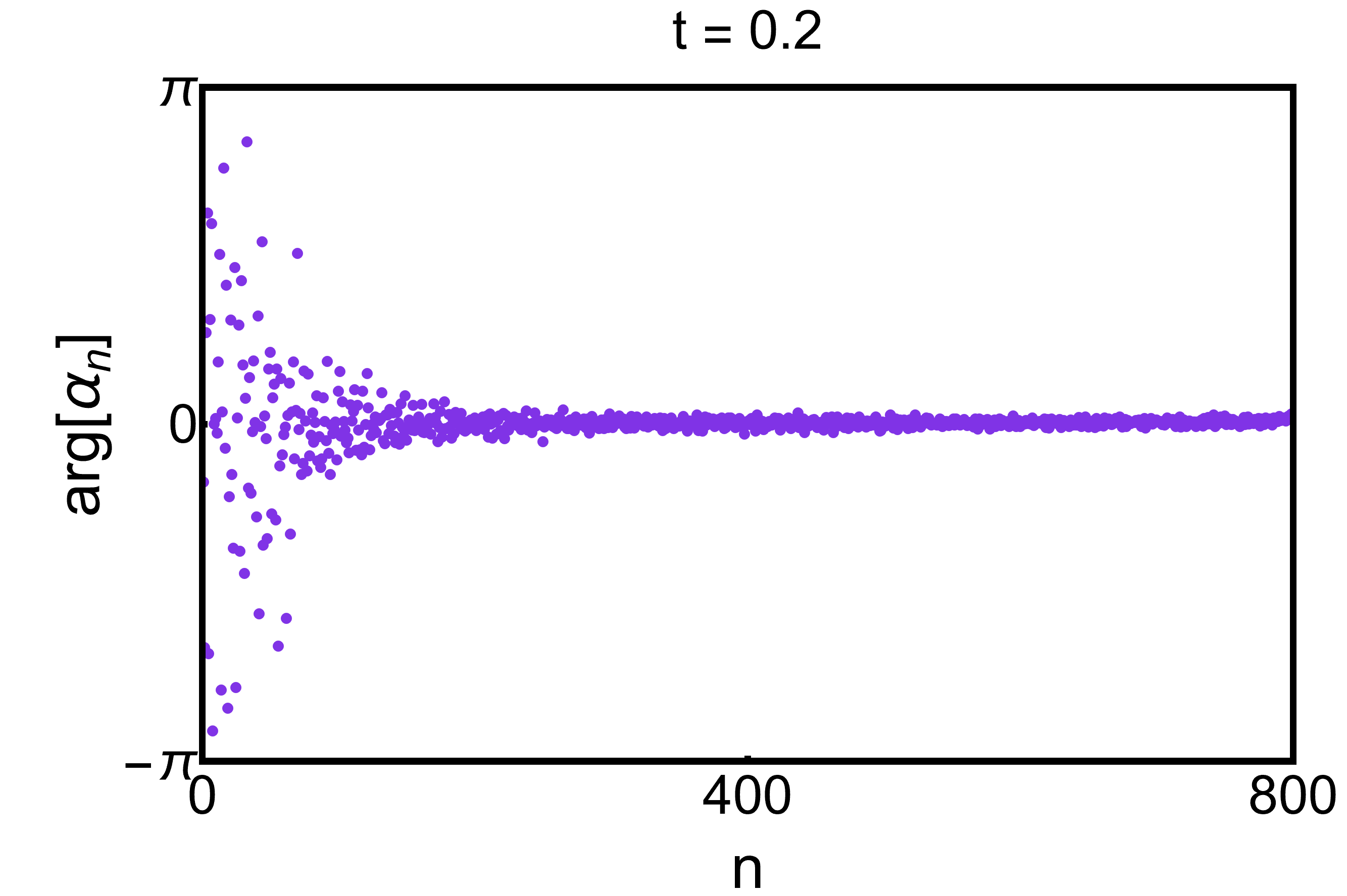}
		\includegraphics[width = 5.6cm]{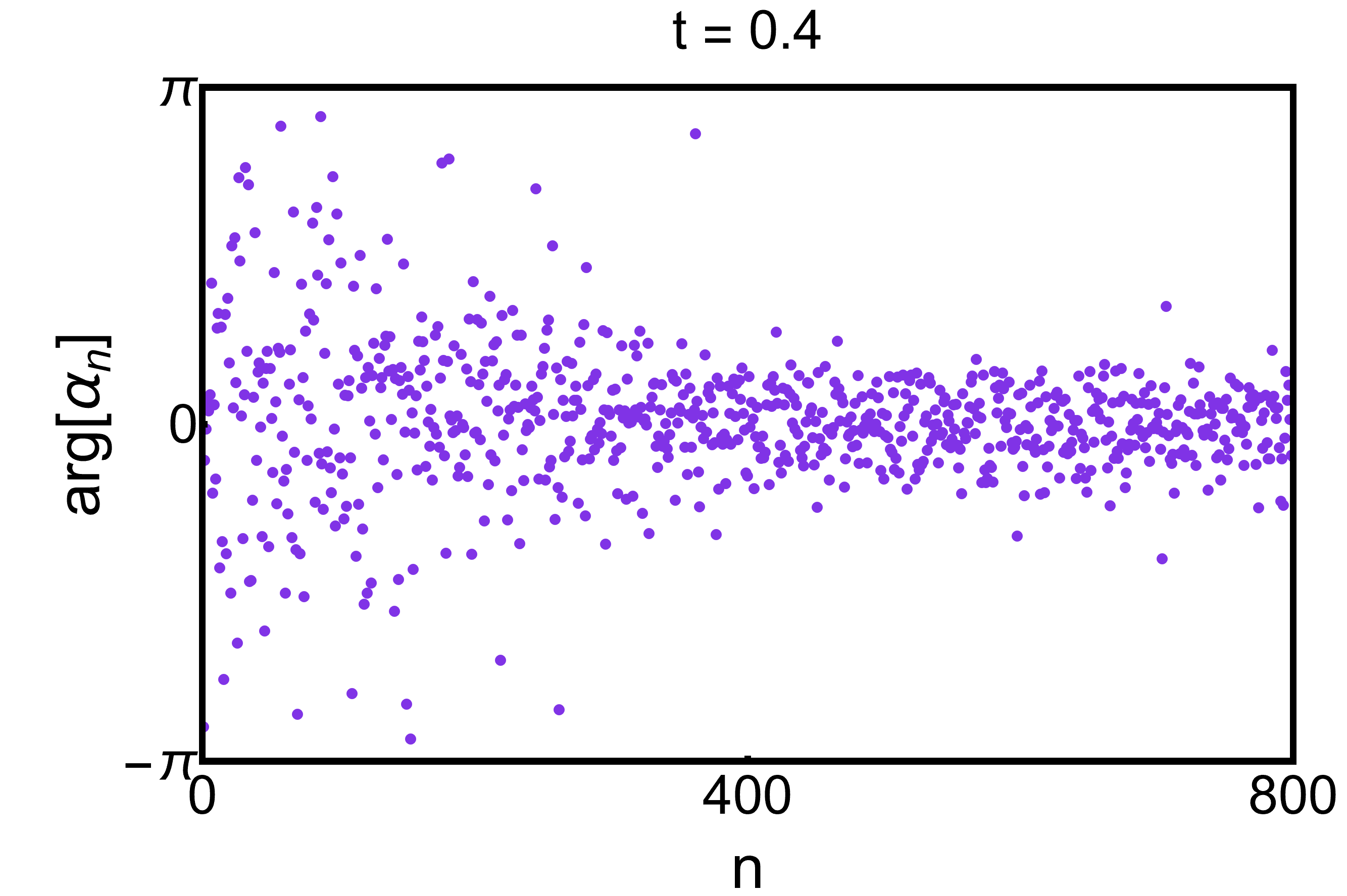}
		\includegraphics[width = 5.6cm]{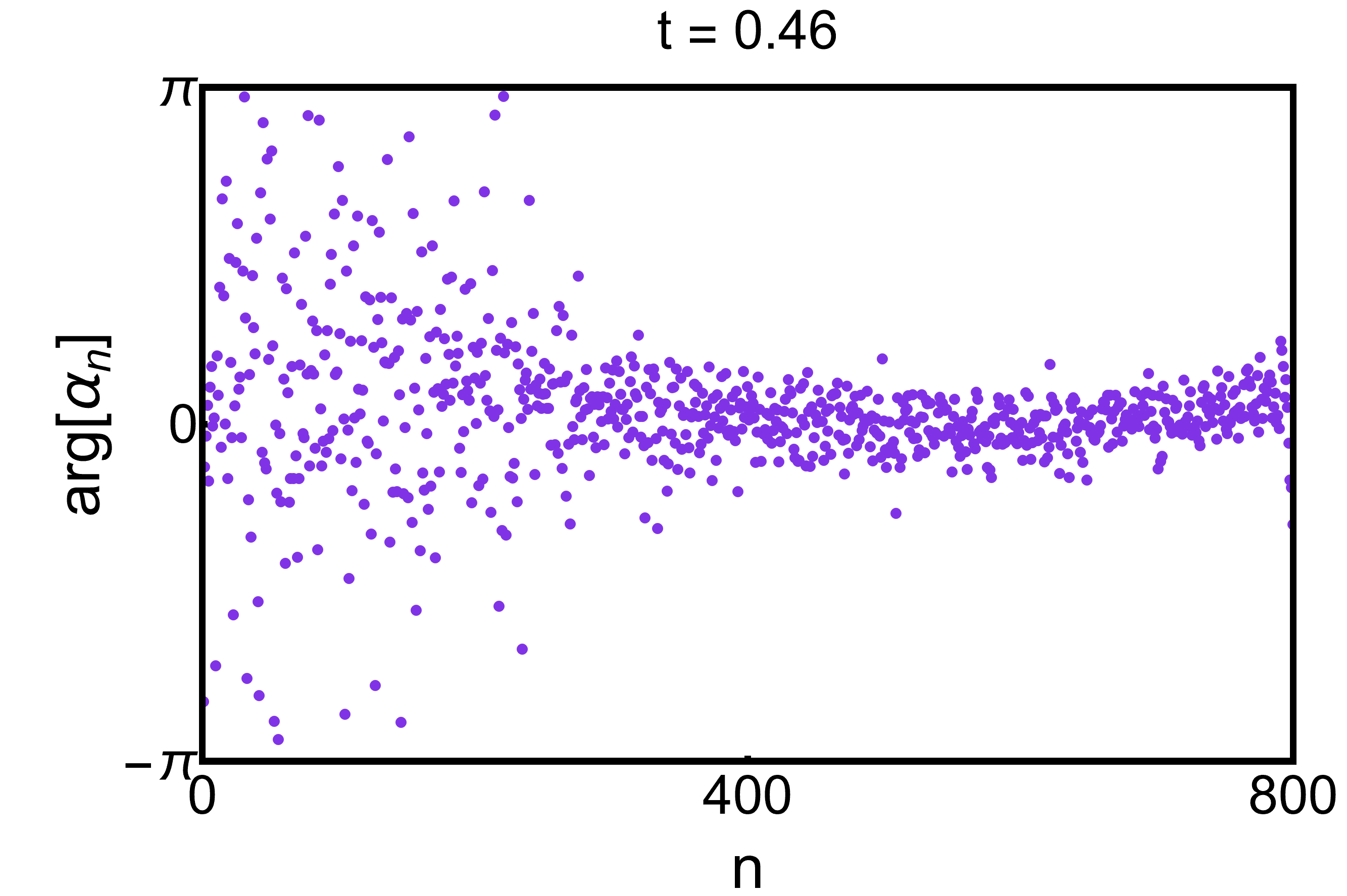}
		\caption{Evolution of the disordered Hamiltonian system (\ref{eqFig10}) initialized with coherent data (two first rows) and random data (two last rows). The first plot shows the energy spectrum at different times and the other plots the distribution of the phases. Both cases demonstrate coherence and strong transfer of energy to high modes. Deviations from the linear trend on the right-hand side of the last phase plots are due to truncation effects.}
		\label{fig:other_minimally_structured_Hamiltonians}
	\end{figure}


    
	\section{Solvable models with explicit expressions}
	\label{sec:Explicit_systems}
	
	Before concluding, we summarize the Hamiltonian systems that arise from our construction with explicit expressions for the couplings, and explain where the previously known models fit among them.
	
	\subsection{Explicit models}
	\label{subsec:Explicit_Coefficients_summary}
	
	\begin{itemize} 
		\item \underline{Type I ($\atwo =\aone= 0$ but $\bone \neq 0$)}:
		\beq
		C_{nmkj} = \begin{cases}
			1 & \text{for } nmkj = 0, \\
			\frac{\beta_1 (n+m) + \beta_0}{n+m-1} + \beta_2 \min(n,m,k,j) & \text{for } nmkj \neq 0.
		\end{cases} \qquad \text{and} \qquad \h_n = 1.
		\label{eq:Explicit_expression_mu2_0_mu1_0}
		\eeq
		
		\item \underline{Type II ($\atwo=0$ but $\aone\neq 0$)}:
		\beq
		C_{nmkj} = \begin{cases}
			[\aone (n+m) + \bone] \ \frac{\sqrt{\h_n \h_m\h_k \h_j}}{\h_{n+m}} & \text{for } nmkj = 0, \\
			\frac{\beta_1 (n+m)+ \beta_0}{n+m  - 1} \  [\aone (n+m) + \bone] \ \frac{\sqrt{\h_n \h_m\h_k \h_j}}{\h_{n+m}} & \text{for } nmkj \neq 0,
		\end{cases}
		\eeq 
		with sequences
		\begin{eqnarray}
			\{\aone > 0,\ \bone = 1\} &\Rightarrow& \h_n = A_{n}^{(2\aone+1,1)},\\[6pt]
			\{\aone = 1,\ \bone = 0\} & \Rightarrow& \h_n = \frac{(2 n)^{n - 1}}{n!}.
		\end{eqnarray}
		where $A_n^{(s,1)}$ is a particular set of Fuss-Catalan numbers \cite{FussCatalan}:
		\beq
		A_n^{(s,r)} = \frac{r}{sn+r} \binom{sn+r}{n}. 
		\label{eq:Fuss_Catalan_numbers_general_form_v2}
		\eeq

		\item \underline{Type III ($\atwo \neq 0$)}:
		\beq
		C_{nmkj} =  \begin{cases}
			\left[
			\atwo (nm + kj) + \aone (n+m) + \bone
			\right]\
			\frac{\sqrt{\h_n \h_m\h_k \h_j}}{\h_{n+m}} & \text{for } nmkj = 0, \\[6pt]
			(1-\beta) \frac{\left[\aone (n+m) + \atwo n m + \bone\right] \left[\aone (n+m)  + \atwo kj+ \bone\right]}{n+m-1}\ \frac{\sqrt{\h_n \h_m\h_k \h_j}}{\h_{n+m}} & \text{for } nmkj \neq 0,
		\end{cases}
		\label{eq:deformation_S_mu2_activated}
		\eeq
		with $\beta_2 = (1-\beta)\atwo$, $\beta_1 = (1-\beta)\aone$, $\beta_0 = (1-\beta)\bone$ and $\beta \in \mathbb{R}$. Explicit expressions for the sequence have been obtained for the parameters:
		\begin{align}
			\{\atwo = 1,\ \aone = 0,\ \bone = 0\}
			&\;\Rightarrow\;
			\h_{n\geq 1}
			= 2\frac{(2n)^{n-2}}{n!}
			\\[6pt]
			\{\atwo = \aone^2,\ \aone = \frac{s-1}{2},\ \bone = 1\}
			&\;\Rightarrow\;
			\h_{n\geq 1}
			= \frac{2}{(s-1)n+2} A_n^{(s,1)}\\[6pt]
			\left\{\atwo = \frac{1}{8\aone} - \frac{\aone}{2},\ \aone = \frac{1}{4s-2},\ \bone = -2\aone\right\}
			&\;\Rightarrow\;
			\h_{n\geq 1}
			= \frac{s-1}{sn-1} A_n^{(s,1)}\\[6pt]
			\left\{\atwo = \frac{1}{16\aone} - \aone,\ \aone = \frac{1}{8s-4},\ \bone = -\aone\right\}
			&\;\Rightarrow\;
			\h_{n\geq 1}
			= \frac{4s-2}{2sn-1} A_n^{(s,1/2)}
		\end{align}
		where $A_n^{(s,r)}$ are the Fuss-Catalan numbers given in (\ref{eq:Fuss_Catalan_numbers_general_form_v2}).
		
	\end{itemize}


	\subsection{Previously known models} 
	

	{\bf Systems with Lax pair structures:} The cubic Szeg\H{o}-equation \cite{GG}, the $\beta$-Szeg\H{o}-equation \cite{BE}, and the min-Hamiltonian system \cite{minHamiltonian} can be written in the form (\ref{eq:Resonant_Equation})-(\ref{eq:C_symmetries}) with coefficients:
	\beq
	C_{nmkj}^{(Sz)} = 1, \qquad C_{nmkj}^{(\beta)} = \begin{cases}
		1 & \text{for }  n m k j =0\\
		1-\beta & \text{for }  n m k j \neq 0
	\end{cases} \qquad \text{and} \qquad C_{nmkj}^{(\text{min})} = \min(n,m,k,j) + 1.
	\eeq
	These models were first constructed following {\it ad hoc} processes, while now arising from our systematic construction for the parameters $\{\atwo,\aone,\bone\} = \{0,0,1\}$ and particular choices of $\{\beta_2,\beta_1,\beta_0\}$ in (\ref{eq:Explicit_expression_mu2_0_mu1_0}) for each of them. The $\alpha$-Szeg\H{o}-equation \cite{Xu} and the damped-Szeg\H{o} equation \cite{GG2,GGH} are linear deformations of the original cubic Szeg\H{o} equation: $i\dot{\alpha_0} \to i\dot{\alpha_0}  \pm \alpha_0$ and $i\dot{\alpha_0} \to i\dot{\alpha_0}  \pm i\alpha_0$, which can be easily incorporated in our construction.

	Interestingly, all these models admit Lax pair structures. This is not an inherent ingredient in our construction; however, given that the framework generates infinitely many systems, it is not implausible that additional examples with Lax pair structures will emerge.

	
	\noindent {\bf Finite-time cascade systems:} Other two classes of Hamiltonian systems have been recently introduced in \cite{Biasi} and \cite{BG} in connection to the problems of coherent condensation and structure formation. They are given by the coefficients
	\beq
	C_{nmkj}^{(Z)} = \left[\aone(n+m)+1\right]\frac{\sqrt{F_nF_mF_kF_j}}{F_{n+m}}
	\eeq
	and
	\beq
	C_{nmkj}^{(Y)} = \begin{cases}
		\left[\atwo (n m+kj) + \aone (n+m) + 1\right]\frac{\sqrt{F_nF_mF_kF_j}}{F_{n+m}} & \text{for  } n m k j =0\\
		0 & \text{otherwise},
	\end{cases}
	\eeq
	which arise from our construction for the parameters $\aone > 0$ and
	\begin{align}
		&\{\atwo = 0,\ \ \bone = 1\} \qquad \text{and} \qquad \{\beta_2 = 0,\ \beta_1 = 1,\ \beta_0 = -1\},\\
		&\{\atwo = \aone^2,\ \bone = 1\} \qquad \text{and} \qquad \{\beta_2 = 0,\ \beta_1 = 0,\ \beta_0 = 0\}.
	\end{align}


	\section{Discussion}
	\label{sec:Discussion}
	
	In this work, we have addressed the question of whether coherent energy cascades can emerge and persist in disordered Hamiltonian systems. In particular, we have demonstrated that fully structured mode interactions are not essential to produce such cascades, but an infinitesimal fraction of structured nonlinear couplings --- embedded in a cloud of random or stochastic ones --- suffices to sustain phase-organized energy transfer to arbitrarily high modes. We now discuss the implications of these findings and their connections to related topics.
	
	\noindent {\bf Scope:} The models presented in this work have specific algebraic constraints in their couplings. This is, however, a consequence of our proof strategy rather than an intrinsic necessity of coherent cascades: in order to establish their existence rigorously, we leveraged a solvable structure that survived large-scale randomization of the couplings. The numerical evidence presented in Section~\ref{sec:numerics} --- in particular the emergence of coherent cascades from random initial data through a process of phase locking --- suggests that the phenomenon is considerably more general than what we have proved here. Therefore, extending the study to broader classes of disordered Hamiltonian systems would be an interesting direction.
	
	\noindent {\bf Structural requirements:} The above direction motivates the question of what the minimal structural requirements of Hamiltonian systems actually are to produce coherent cascades. Numerical simulations of fully random systems — with no structured subsets of couplings but with the resonant condition $n+m=k+j$ maintained — demonstrate that the organized resonance web alone is insufficient to produce coherent energy dynamics, highlighting the essential role of the structured subset of couplings. However, it remains unclear what minimal conditions that subset must satisfy. Whether the ratio of structured to random couplings we have obtained can be reduced, whether the rigid algebraic form of the structured couplings can be relaxed, and what asymptotic behaviors of the interaction coefficients are necessary to support coherent cascades are all natural questions that this work opens up, and that we hope will be pursued in future studies.
	
	\noindent {\bf The growth of Sobolev norms:} The unbounded growth of Sobolev norms --- the rigorous verification of energy transfer from low to arbitrarily high modes --- has been the subject of intense attention over the last three decades, with relatively few rigorous examples identified. Our results provide explicit examples of this behavior in largely disordered systems, with Sobolev norms growing to infinity both in finite and infinite time. These examples are rigorously established on the invariant manifold, while numerical simulations provide evidence that such behaviors extend to broader classes of initial data. We are therefore optimistic that unbounded growth of Sobolev norms holds for broad classes of initial data in the models presented here --- and in `nearby' Hamiltonian systems, both random and deterministic.
	
	\noindent {\bf Phase locking:} Underlying all of the above is a mechanism that our results reveal but do not yet fully explain: the emergence of phase locking from random initial data in systems with a large majority of random couplings. This is striking in the systems we consider, since there is no obvious algebraic mechanism forcing phase alignment when the initial conditions are far from the invariant manifold. The fact that phase locking nonetheless emerges suggests that the minimal structured subset acts as an organizing seed capable of overcoming the disordering effect of the random majority of the nonlinear couplings. This behavior has a significant effect in position space compared to random phase cascades observed in kinetic regimes. While both settings are characterized  by the energy transfer to small spatial scales, phase locking results in energy localization, and statistically independent phases result in spatial roughness (noise-like spatial profiles). Understanding the progenitor mechanism for phase locking more precisely, and in particular whether it can be captured by a reduced description involving only the structured couplings, is therefore an important problem that goes beyond coherent energy cascades, and involves the emergence of other coherent structures such as condensates \cite{Biasi} and rogue waves \cite{RicardoGrandeRogueWaves}, among others. 
	
	\noindent {\bf Weakly nonlinear dispersive PDEs with disorder:} These questions are all closely connected to understanding whether the general observations of this work have realizations in weakly nonlinear waves in disordered media and canonical models of physics, such as nonlinear Schr\"odinger equations. The resonant systems studied here arise as leading-order approximations of weakly nonlinear Hamiltonian PDEs in the regime where nonlinear effects are small but accumulate over long times. In this context, the interaction coefficients $C_{nmkj}$ encode the nonlinear part of the original PDE. The disorder we introduce in these coefficients therefore corresponds to a form of disorder in the nonlinear part of the original PDE, distinct from the more commonly studied linear disorder introduced via random potentials \cite{PicozziDramaticAceleration,PicozziStrongDisorder}. It is natural to ask whether experimental realizations of disordered models similar to ours with coherent energy cascades are achievable, for instance by introducing spatially dependent factors in the nonlinear part of nonlinear Schr\"odinger equations with the harmonic potential, or in similar dispersive equations. While it may be challenging to find ways to finetune individual mode couplings in a reailstic experiment, the subject is attractive, and also connects to the question whether resonant systems with prescribed (not necessarily random) couplings can be engineered in a lab (possible applications can extend far beyond the topics of this article, for example, to the domain of quantum simulations). Numerical results from Section~\ref{subsec:other_minimally_structured_Hamiltonians} suggest that the presence of coherent cascades in disordered systems is plausible; a key step is in finding a way to maintain a pivotal subset of structured nonlinear couplings while conducting large-scale randomization. 

    \noindent {\bf Turbulence in weakly nonlinear PDEs without disorder:} A different way to view introducing disorder into the couplings is as the first step in exploring the neighborhood of a given system in the space of interaction parameters. This may help one bound the differences between solutions of two different systems, for example, a solvable and an analytic intractable one. Interesting unproved conjectures exist in the literature in relation to turbulent energy transfer in PDE dynamics, for example, the Bizo\'n-Rostworowski AdS instability conjecture \cite{BR} (see \cite{E2} for a review). The latter instability can be analyzed \cite{CEV1,BMR} in terms of a resonant Hamiltonian system of the form (\ref{eq:Resonant_Equation}) with very complicated mode couplings \cite{CEV1} that make it analytically intractable.
    Could it be that this, or another similar system, derived from a physically motivated PDE, is sufficiently close in the coupling parameter space to one of the solvable models considered here, while the turbulent cascades are sufficiently robust (as they were in the random coupling explorations presented here) to survive the coupling modifications from one system to the other? A reassuring feature is that, heuristically, the large mode-number asymptotics of the coupling coefficients must play a crucial role in the formation of turbulent cascades, and the solvable resonant systems constructed here display power-law asymptotics similar to what has been seen in resonant Hamiltonian systems arising from realistic PDEs \cite{UVpwr1,UVpwr2,UVpwr3}. Moreover, Hamiltonian models with solvable structures similar to those considered here arise naturally from nonlinear Schr\"odinger and wave equations \cite{BBE1,Evnin,BBE2}.
    
	\noindent {\bf Outlook:} Our results suggest that coherent energy cascades are not a fragile consequence of highly organized nonlinear interactions, but can persist in predominantly disordered Hamiltonian systems provided that a small structured core remains. Understanding how little structure is needed to organize such turbulent dynamics is an important question for future work.  This points toward a perspective in which coherence and disorder coexist rather than collide.


	\noindent \underline{\bf Acknowledgments:} We are grateful to Patrick G\'erard, Daniel Eceizabarrena, and Miguel A. Mu\~noz Mart\'inez for useful discussions on this work. The project that gave rise to these results received the support of a fellowship
	from the ``la Caixa'' Foundation (ID 100010434), with fellowship code LCF/BQ/PI24/12040029, the Mar\'ia de Maeztu grant CEX2023-001318-M funded by MICIU/AEI /10.13039/501100011033, the Xunta de Galicia (CIGUS Network of Research Centres and grant ED431C-2025/11), and an IGFAE Summer Fellowships 2025. This work benefited from the use of the infrastructures provided by the Galician Supercomputing Center (CESGA).


    \noindent \underline{\bf Author contributions:} A.B. and O.E. conceived the initial idea of this project. A.B. developed the analytic and numerical frameworks and wrote the paper. O.E. contributed further to the writing and contextualization. A.I.J. conducted the numerical exploration of random time-independent systems and assisted in analytic developments. B.C. conducted the numerical exploration of stochastic systems and adapted the code accordingly.

	
	\appendix
	
	\section{Hamiltonian conservation on the invariant manifold}
	\label{app:Hamiltonian_conservation}
	\begin{lemma}
		Consider the family of Hamiltonian systems constructed in Proposition~\ref{prop:existence_invariant_manifold} where $S_{nmkj}^{(0)}(t)$ are the only time-dependent couplings. Then the Hamiltonian $\mathcal H$ in (\ref{eq:Hamiltonian}) is conserved for trajectories on the invariant manifold.
	\end{lemma}
	
	\begin{proof} The key observation for this proof is that the Hamiltonian (\ref{eq:Hamiltonian}) has no explicit dependence on time  when restricted to the invariant manifold, being conserved.
		
		The only explicit time dependence of the Hamiltonian arises through the couplings $S_{nmkj}^{(0)}(t)$. It is therefore sufficient to show that this dependence disappears when the Hamiltonian is restricted to the invariant manifold. To show that, we use that  $a_k(t)a_{n+m-k}(t) = a_n(t)a_{m}(t)$ on the invariant manifold, and the constraint (\ref{eq:Condition_3_ONLY}).
		
		The couplings $S_{nmkj}^{(0)}(t)$ only appear in the following combination:
		\begin{multline}
			\sum_{n=1}^{\infty}\sum_{m=1}^{\infty}\sum_{k=1}^{n+m-1}
			S_{nmk,n+m-k}^{(0)}(t)\,
			\bar{a}_n(t) \bar{a}_m(t) a_k(t) a_{n+m-k}(t)=  \sum_{n=1}^{\infty}\sum_{m=1}^{\infty}
			|a_n(t)|^2 |a_m(t)|^2
			\sum_{k=1}^{n+m-1} S_{nmk,n+m-k}^{(0)}(t)\\
			= \sum_{n=1}^{\infty}\sum_{m=1}^{\infty}
			|a_n(t)|^2 |a_m(t)|^2 (\beta_2 nm+\beta_1(n+m)+\beta_0)\, \h_n \h_m.
		\end{multline}
		which is independent of time. Consequently, the Hamiltonian restricted to the invariant manifold has no explicit time dependence, concluding the proof.
	\end{proof}



\begin{thebibliography}{999}
		
		\bibitem{Book_Zakharov}  V. E. Zakharov, V. S. L'vov, and G. Falkovich, {\em Kolmogorov spectra of turbulence I: Wave turbulence}, \doi{Springer (2012)}{https://doi.org/10.1007/978-3-642-50052-7}. 
		
		\bibitem{Book_Nazarenko} S. Nazarenko, {\em Wave Turbulence}, Lectures Notes in Physics, \doi{Springer (2011)}{https://doi.org/10.1007/978-3-642-15942-8}.
		
		\bibitem{Book_Galtier} S. Galtier, {\em Physics of wave turbulence}, \doi{Cambridge University Press (2022)}{https://doi.org/10.1017/9781009275880}.
		
		\bibitem{BR} P. Bizo\'n, and A. Rostworowski, {\em Weakly turbulent instability of anti–de Sitter spacetime}, \doi{Phys.\ Rev.\ Lett. {\bf 107}, 031102 (2011)}{10.1103/PhysRevLett.107.031102}, \arXiv{1104.3702} [gr-qc].	
		
		\bibitem{BMR} P. Bizo\'n, M. Maliborski, and A. Rostworowski, {\em Resonant dynamics and the instability of anti-de Sitter spacetime}, \doi{Phys.\ Rev.\ Lett. {\bf 115}, 081103 (2015)}{10.1103/PhysRevLett.115.081103}, \arXiv{1506.03519} [gr-qc].
		
		\bibitem{Biasi} A.~Biasi, {\em Exact solutions for a coherent phenomenon of condensation in conservative Hamiltonian systems}, \doi{Phys. Rev. E {\bf 110}, 034107 (2024)}{10.1103/PhysRevE.110.034107}, \arXiv{2401.15083} [cond-mat.stat-mech].
		
		\bibitem{BG} A.~Biasi, and P.~G\'erard, {\em Energy cascades and condensation via coherent dynamics in Hamiltonian systems}, \doi{Comm.\ Math.\ Phys. {\bf 406}, 264 (2025)}{10.1007/s00220-025-05449-5}, \arXiv{2412.03663} [math-ph].
		
		\bibitem{Staffilani2010} J.~Colliander, M.~Keel, G.~Staffilani, H.~Takaoka, and T.~Tao, {\em Transfer of energy to high frequencies in the cubic defocusing nonlinear Schr\"odinger equation}, \doi{Inv.\ math. {\bf 181}, 39 (2010)}{10.1007/s00222-010-0242-2}, \arXiv{0808.1742} [math.AP].
		
		\bibitem{GG} P. G\'erard, and S. Grellier, {\em The cubic Szeg\H{o} equation}, \doi{Ann. Sci. \'ENS {\bf 43}, 761 (2010)}{10.24033/asens.2133}, \arXiv{0906.4540} [math.CV].
		
		\bibitem{Kuksin1} S.~B.~Kuksin, {\em Growth and oscillations of solutions of nonlinear Schr\"odinger equation}, \doi{Comm. Math. Phys. {\bf 178}, 265 (1996)}{10.1007/BF02099448}.
		
		\bibitem{Kuksin2} S.~B.~Kuksin, {\em Oscillations in space-periodic nonlinear Schr\"odinger equations}, \doi{Geom.\ Func.\ Anal. {\bf 7}, 338 (1997)}{10.1007/PL00001622}.	
		
		\bibitem{Guardia} M.~Guardia, and V.~Kaloshin, {\em Growth of Sobolev norms in the cubic defocusing nonlinear Schr\"odinger equation}, \doi{J. Eur. Math. Soc. {\bf 17}, 71 (2015)}{10.4171/JEMS/499}, \arXiv{1205.5188} [math.AP].
		
		\bibitem{Hani} Z.~Hani, {\em Long-time instability and unbounded Sobolev orbits for some periodic nonlinear Schr\"odinger equations}, \doi{Arch.\ Rat.\ Mech.\  Anal. {\bf 211}, 929 (2014)}{10.1007/s00205-013-0689-6}, \arXiv{1210.7509} [math.AP].
		
		\bibitem{Hani2} Z. Hani, B. Pausader, N. Tzvetkov, and N. Visciglia, {\em Modified scattering for the cubic Schr\"odinger equation on product spaces and applications}, \doi{Forum Math. Pi {\bf 3}, 63 (2015)}{10.1017/fmp.2015.5}, \arXiv{1311.2275} [math.AP].
	
		\bibitem{HausProcesi} E.~Haus, and M.~Procesi, {\em Growth of Sobolev norms for the quintic NLS on $\mathbb{T}^2$}, \doi{Anal. PDE {\bf 8}, 883 (2015)}{10.2140/apde.2015.8.883}, \arXiv{1405.1538} [math.AP].

		\bibitem{GHP} M. Guardia, E. Haus, and M. Procesi, {\em Growth of Sobolev norms for the analytic NLS on $\mathbb{T}^2$}, \doi{Adv. Math. {\bf 301}, 615 (2016)}{10.1016/j.aim.2016.06.018}, \arXiv{1503.02468} [math.AP].

		\bibitem{Maspero} M.~Guardia, Z.~Hani, E.~Haus, A.~Maspero, and M.~Procesi, {\em Strong nonlinear instability and growth of Sobolev norms near quasiperiodic finite gap tori for the 2D cubic NLS equation}, \doi{J.\ Eur.\ Math.\ Soc. {\bf 25}, 1497 (2022)}{10.4171/JEMS/1200}, \arXiv{1810.03694} [math.AP].
		
		\bibitem{Maspero2} A. Maspero, and F. Murgante, {\em One dimensional energy cascades in a fractional quasilinear NLS}, \doi{Arch. Ration. Mech. Anal. {\bf 250}, 3 (2026)}{10.1007/s00205-025-02159-z}, \arXiv{2408.01097} [math.AP].
		
		\bibitem{Bourgain2} J. Bourgain, {\em On the growth in time of higher Sobolev norms of smooth solutions of Hamiltonian PDE}, \doi{Int.\ Math.\ Res.\ Not. {\bf 1996}, 277 (1996)}{10.1155/S1073792896000207}.

		\bibitem{GGwaveequation} P.~G\'erard, and S.~Grellier,  {\em Effective integrable dynamics for a certain nonlinear wave equation}, \doi{Anal. PDE {\bf 5}, 1139 (2012)}{10.2140/apde.2012.5.1139}, \arXiv{1110.5719} [math.AP].
		
		\bibitem{GerardLenzmannPocovnicuRaphael} P.~G\'erard, E.~Lenzmann, O.~Pocovnicu, and P.~Rapha\"el, {\em A two-soliton with transient turbulent regime for the cubic half-wave equation on the real line}, \doi{Anal.\ PDE {\bf 4}, 7 (2018)}{10.1007/s40818-017-0043-7}, \arXiv{1611.08482} [math.AP].	
		
		\bibitem{GG2} P.~G\'erard, and S.~Grellier, {\em On a damped Szeg\H{o} equation (with an appendix in collaboration with Christian Klein)}, \doi{SIAM J. Math. Anal. {\bf 52}, 4391 (2020)}{10.1137/19M1299189}, \arXiv{1912.10933} [math.AP].
		
		\bibitem{GGH} P.~G\'erard, S.~Grellier, Z.~He, {\em Turbulent cascades for a family of damped Szeg\H{o} equations}, \doi{Nonlinearity {\bf 35}, 4820 (2022)}{10.1088/1361-6544/ac7e13}, \arXiv{2111.05247} [math.AP].
		
		\bibitem{Xu} H.~Xu, {\em Large-time blowup for a perturbation of the cubic Szeg\H{o} equation}, \doi{Anal. \ PDE {\bf 7}, 717 (2014)}{10.2140/apde.2014.7.717}, \arXiv{1307.5284} [math.AP]; {\em The cubic Szeg\H{o} equation with a linear perturbation}, \arXiv{1508.01500} [math.AP].
		
		\bibitem{GL}  P.~G\'erard, and E.~Lenzmann, {\em The Calogero-Moser derivative nonlinear Schr\"odinger equation}, \doi{Comm. Pur. App. Math. {\bf 77}, 4008 (2024)}{10.1002/cpa.22203}, \arXiv{2208.04105} [math.AP]. 
		
		\bibitem{BE} A. Biasi, and O. Evnin, {\em Turbulent cascades in a truncation of the cubic Szeg\H{o} equation and related systems}, \doi{Anal. \ PDE {\bf 15}, 217 (2022)}{10.2140/apde.2022.15.217}, \arXiv{2002.07785} [math.AP].
		
		\bibitem{minHamiltonian} B. Craps, B., M. De Clerck, O. Evnin, P. Hacker, and M. Pavlov, {\em Bounds on quantum evolution complexity via lattice cryptography}, \doi{SciPost Phys. {\bf 13}, 090 (2022)}{10.21468/SciPostPhys.13.4.090}, \arXiv{2202.13924} [quant-ph].
		
		\bibitem{Thomann} 	V. Schwinte, and L. Thomann, {\em  Growth of Sobolev norms for coupled lowest Landau level equations},  \doi{Pure App. Anal. {\bf 3}, 189 (2021)}{10.2140/paa.2021.3.189}, \arXiv{2006.01468} [math.AP]. 
		
		\bibitem{Pocovnicu} O.~Pocovnicu, {\em Traveling waves for the cubic Szeg\H{o} equation on the real line}, \doi{Anal. \ PDE {\bf 4}, 379 (2011)}{10.2140/apde.2011.4.379}, \arXiv{1001.4037} [math.AP].
		
		\bibitem{GP} P.~G\'erard, and A.~Pushnitski, {\em  An inverse spectral problem for Hankel operators and turbulent solutions of the cubic Szeg\H{o} equation on the line}, \doi{J. Eur. Math. Soc. {\bf 27}, 4591 (2025)}{10.4171/JEMS/1457}, \arXiv{2202.03783} [math.AP].
		
		\bibitem{Giuliani} F.~Giuliani, {\em Sobolev instability in the cubic NLS equation with convolution potentials on irrational tori}, \doi{J. Diff. Eq. {\bf 416}, 1 (2025)}{10.1016/j.jde.2024.09.044}, \arXiv{2308.13468} [math.AP].
		
		\bibitem{GuardiaGiuliani} F.~Giuliani, and M.~Guardia, {\em Sobolev norms explosion for the cubic NLS on irrational tori}, \doi{Nonlin. Anal. {\bf 220}, 112865 (2022)}{10.1016/j.na.2022.112865 }, \arXiv{2110.15845} [math.AP].
		
		\bibitem{GiulianiScandone} F.~Giuliani, and R.~Scandone,  {\em Energy cascade and Sobolev norms inflation for the quantum Euler equations on tori}, \doi{Adv. Math. {\bf 479}, 110453 (2025)}{10.1016/j.aim.2025.110453}, \arXiv{2410.21080} [math.AP].
		
		\bibitem{MasperoWaterWaves} B.~Langella, A.~Maspero, F.~Murgante, and S.~Terracina, {\em Transfer of energy for pure-gravity water waves with constant vorticity}, \arXiv{2604.08343}.

		\bibitem{MBox1} M.~Maliborski, {\em Instability of flat space enclosed in a cavity}, \doi{Phys. Rev. Lett. {\bf 109}, 221101 (2012)}{10.1103/PhysRevLett.109.221101}, \arXiv{1208.2934} [gr-qc].
		
		\bibitem{Jalmuzna} P. Bizo\'n, and J. Ja\l mu\.zna, {\em Globally regular instability of AdS$_3$}, \doi{Phys.\ Rev.\ Lett. {\bf 111}, 1306 (2013)}{10.1103/PhysRevLett.111.041102}, \arXiv{1306.0317} [gr-qc].

		\bibitem{E2} O. Evnin, {\em Resonant Hamiltonian systems and weakly nonlinear dynamics in AdS spacetimes}, \doi{Class. Quant. Grav. {\bf  38}, 203001 (2021)}{10.1088/1361-6382/ac1b46}, \arXiv{2104.09797} [gr-qc].
		
		\bibitem{MBox2} J. Kurzweil, and M. Maliborski, {\em Resonant dynamics and the instability of the box Minkowski model}, \doi{Phys.\ Rev. D {\bf 106}, 124020 (2022)}{10.1103/PhysRevD.106.124020}, \arXiv{2209.05608} [gr-qc].
		
		\bibitem{KehleMoschidis} C.~Kehle, and G.~Moschidis, {\em Weakly turbulent dynamics on Schwarzschild-AdS black hole spacetimes}, \arXiv{2604.12118} [gr-qc].

		\bibitem{Freivogel} B.~Freivogel and I-S.~Yang,
		{\it Coherent cascade conjecture for collapsing solutions in global AdS,}
		\doi{Phys. Rev. D \textbf{93}, 103007 (2016)}{10.1103/PhysRevD.93.103007},
		\arXiv{1512.04383} [hep-th].

		\bibitem{EvninMelonicTurbulence} S.~Dartois, O.~Evnin, L.~Lionni, V.~Rivasseau, and G.~Valette, {\em Melonic turbulence}, \doi{Comm.\ Math.\ Phys. {\bf 374}, 1179 (2020)}{10.1007/s00220-020-03683-7}, \arXiv{1810.01848} [math-ph].
		
		\bibitem{PicozziDramaticAceleration} A. Fusaro, J. Garnier, K. Krupa, G. Millot, A. Picozzi,
		{\em Dramatic acceleration of wave condensation mediated
		by disorder in multimode fibers}, 
		\doi{Phys. Rev. Lett. {\bf 122}, 123902 (2019)}{10.1103/PhysRevLett.122.123902}, \arXiv{2011.05111} [physics.optics].
		
		\bibitem{HuRosenhaus}X.~Y.~Hu and V.~Rosenhaus, {\it Random coupling model of turbulence as a classical Sachdev-Ye-Kitaev model,} \doi{Phys. Rev. E \textbf{108}, 054132 (2023)}{10.1103/PhysRevE.108.054132}, \arXiv{2303.03421} [hep-th].
		
		\bibitem{Frahm} K. M. Frahm and  D. L. Shepelyansky, {\em Random matrix model of Kolmogorov-Zakharov turbulence}, \doi{Phys. Rev. E {\bf 109}, 044201 (2024)}{10.1103/PhysRevE.109.044201}, \arXiv{2401.11545} [cond-mat.stat-mech].
		
		\bibitem{rnd} J.-P.~Bouchaud, L.~Cugliandolo, J.~Kurchan and M.~M\'ezard, {\it Mode coupling
			approximations, glass theory and disordered systems,} \doi{Physica A {\bf 226}, 243 (1996)}{10.1016/0378-4371\%2895\%2900423-8},
		\arXiv{cond-mat/9511042}. 

		\bibitem{NazarenkoKrstulovicZhuSemisalov} Y.~Zhu, B.~Semisalov, G.~Krstulovic, and S.~Nazarenko, {\em Direct and inverse cascades in turbulent Bose-Einstein condensate}, \doi{Phys. Rev. Lett. {\bf 130}, 133001 (2023)}{10.1103/PhysRevLett.130.133001}, \arXiv{2208.09279} [cond-mat.quant-gas].
		
		\bibitem{HaniLiNahmodStaffilani} Z.~Hani, Y.~Li, A.~Nahmod, and G.~Staffilani, {\em Non-equilibrium steady state for a three-mode energy cascade model}, \arXiv{2505.16018} [math.PR].
		
		\bibitem{FaouCarles} R.~Carles, and E.~Faou, {\em A toy model for frequency cascade in the nonlinear Schrodinger equation}, \doi{Nonlinearity {\bf 39}, 075016 (2026)}{10.1088/1361-6544/ae83fc}, \arXiv{2506.15226} [math.AP].

		\bibitem{Kuksin3} S. Kuksin and A. Maiocchi, {\em The effective equation method}, in New Approaches to Nonlinear Waves, Springer (2016), \arXiv{1501.04175} [math-ph].

		\bibitem{FaouGermainHani} E.~Faou, P.~Germain, and Z.~Hani, {\em The weakly nonlinear large-box limit of the 2D cubic nonlinear Schr\"odinger equation}, \doi{J. Am. Math. Soc. {\bf 29}, 915  (2016)}{10.1090/jams/845}, \arXiv{1308.6267} [math.AP].

		\bibitem{GHT} P. Germain, Z. Hani and L. Thomann, {\em On the continuous resonant equation for NLS: I. Deterministic analysis}, \doi{J. Math. Pur. App. {\em 105}, 131 (2016)}{10.1016/j.matpur.2015.10.002}, \arXiv{1501.03760} [math.AP].
		
		\bibitem{BBCE} A. Biasi, P. Bizo\'n, B. Craps, and O. Evnin, {\em Exact lowest-Landau-level solutions for vortex precession in BoseEinstein condensates}, \doi{Phys.\ Rev. A {\bf 96}, 053615 (2017)}{10.1103/PhysRevA.96.053615}, \arXiv{1705.00867} [cond-mat.quant-gas]; {\em Two infinite families of resonant solutions for the Gross-Pitaevskii equation}, \doi{Phys.\ Rev. E {\bf 98}, 032222 (2018)}{10.1103/PhysRevE.98.032222}, \arXiv{1805.01775} [cond-mat.quant-gas].
	
		\bibitem{GGT} P. G\'erard, P. Germain and L. Thomann, {\em On the cubic lowest Landau level equation}, \doi{Arch. Rat. Mech. Anal. {\bf 231}, 1073 (2019)}{10.1007/s00205-018-1295-4}, \arXiv{1709.04276} [math.AP].
		
		\bibitem{CF} P. Bizo\'n, B. Craps, O. Evnin, D. Hunik, V. Luyten and M. Maliborski, {\em Conformal flow on $S^3$ and weak field integrability in AdS$_4$}, \doi{Comm. Math. Phys. {\bf 353}, 1179 (2017)}{10.1007/s00220-017-2896-8}, \arXiv{1608.07227} [math.AP].
		
		\bibitem{CEL} B. Craps, O. Evnin and V. Luyten, {\em Maximally rotating waves in AdS and on spheres}, \doi{JHEP {\bf 1709} 059 (2017)}{10.1007/JHEP09\%282017\%29059}, \arXiv{1707.08501} [hep-th].
		
		\bibitem{CEV1} B. Craps, O. Evnin and J. Vanhoof, {\em Renormalization group, secular term resummation and AdS (in)stability}, \doi{JHEP {\bf 2014}, 1 (2014)}{10.1007/JHEP10\%282014\%29048}, \arXiv{1407.6273} [gr-qc]; {\it Renormalization, averaging, conservation laws and AdS (in)stability,} \doi{JHEP \textbf{01}, 108 (2015)}{10.1007/JHEP01(2015)108},\arXiv{1412.3249} [gr-qc].
		
		\bibitem{BEF} P. Bizo\'n, O. Evnin and F. Ficek, {\em A nonrelativistic limit for AdS perturbations}, \doi{JHEP {\bf 12}, 113 (2018)}{10.1007/JHEP12\%282018\%29113}, \arXiv{1810.10574} [gr-qc].

		\bibitem{Stochastic_Processes} S. S\"arkk\"a, and A. Solin, {\em Applied stochastic differential equations}, Cambridge University Press,(2019).

		\bibitem{Darboux}
		G.~V.~Dunne,
		{\it Introductory lectures on resurgence,}
		\doi{Eur. Phys. J. ST \textbf{235}, 3891 (2026)}{10.1140/epjs/s11734-026-02135-y},
		\arXiv{2511.15528} [hep-th].
		
		\bibitem{W-Lambert_function} R.~M.~Corless, G.~H.~Gonnet, D.~E.~G.~Hare, D.~J.~Jeffrey, and D.~E.~Knuth, {\em On the LambertW function}, \doi{Adv. Comput. Math. {\bf 5}, 329 (1996)}{https://doi.org/10.1007/BF02124750}.
		
		\bibitem{W-Lambert_function_Wolfram} E.~W.~Weisstein, {\em Lambert W-Function}, From MathWorld--A Wolfram Resource.\\ \url{https://mathworld.wolfram.com/LambertW-Function.html}.

		\bibitem{FussCatalan} R.~L.~Graham, D.~E.~Knuth, and O.~Patashnik, {\it Concrete Mathematics: A Foundation for Computer Science,}
		Addison-Wesley (1994).

		\bibitem{Polylogarithm_function} Digital Library of Mathematical Functions: \url{https://dlmf.nist.gov/25.12}.

		\bibitem{RicardoGrandeRogueWaves} M.~Berti, R.~Grande, A.~Maspero, and G.~Staffilani, {\em Rogue waves and large deviations for 2D pure gravity deep water waves}, \arXiv{2510.15159} [math.AP].
		
		\bibitem{PicozziStrongDisorder} N. Berti, K. Baudin, A. Fusaro, G. Millot, A. Picozzi, J. Garnier, {\em Interplay of thermalization and strong disorder: wave turbulence theory, numerical simulations, and experiments in multimode optical fibers}, \doi{Phys. Rev. Lett. {\bf 129}, 063901 (2022)}{10.1103/PhysRevLett.129.063901}, \arXiv{2207.08680} [physics.optics].
	
		\bibitem{UVpwr1}
		B.~Craps, O.~Evnin and J.~Vanhoof,
		{\it Ultraviolet asymptotics and singular dynamics of AdS perturbations,}
		\doi{JHEP \textbf{10}, 079 (2015)}{10.1007/JHEP10(2015)079},
		\arXiv{1508.04943} [gr-qc].
		
		\bibitem{UVpwr2}
		O.~Evnin and P.~Jai-akson,
		{\it Detailed ultraviolet asymptotics for AdS scalar field perturbations,}
		\doi{JHEP \textbf{04}, 054 (2016)}{10.1007/JHEP04(2016)054},
		\arXiv{1602.05859} [hep-th].
		
		\bibitem{UVpwr3}
		A.~F.~Biasi, J.~Mas and A.~Paredes,
		{\it Delayed collapses of Bose-Einstein condensates in relation to anti-de Sitter gravity,}
		\doi{Phys. Rev. E \textbf{95}, 032216 (2017)}{10.1103/PhysRevE.95.032216},
		\arXiv{1610.04866} [nlin.PS].
	
		\bibitem{BBE1} A. Biasi, P. Bizo\'n and O. Evnin, {\em Solvable cubic resonant systems}, \doi{Comm. Math. Phys. {\bf 369}, 433 (2019)}{10.1007/s00220-019-03365-z
}, \arXiv{1805.03634} [nlin.SI].
		
		\bibitem{BBE2} A. Biasi, P. Bizo\'n, and O. Evnin, {\em Complex plane representations and stationary states in cubic and quintic resonant systems}, \doi{J. Phys. A {\bf 52}, 435201 (2019)}{10.1088/1751-8121/ab4406}, \arXiv{1904.09575} [math-ph].
		
		\bibitem{Evnin} O. Evnin, {\em Breathing modes, quartic nonlinearities and effective resonant systems}, \doi{SIGMA {\bf 16}, 034 (2020)}{10.3842/SIGMA.2020.034}, \arXiv{1912.07952} [math-ph].

	\end{thebibliography}
\end{document}